\documentclass[aps,pra,10pt,twocolumn,superscriptaddress,showkeys,letterpaper]{revtex4-2}

\pdfoutput=1

\usepackage[T1]{fontenc}
\usepackage[utf8]{inputenc}
\usepackage{color}
\usepackage{babel}
\usepackage{mathtools}
\usepackage{amsmath}
\usepackage{amsthm}
\usepackage{amssymb}
\usepackage{graphicx}
\usepackage[pdfusetitle,
bookmarks=true,bookmarksnumbered=false,bookmarksopen=false,
 breaklinks=false,pdfborder={0 0 0},pdfborderstyle={},backref=false,colorlinks=true]
 {hyperref}
\hypersetup{
 linkcolor=magenta, urlcolor=blue, citecolor=blue}

\usepackage{ragged2e}
\usepackage{caption}

\usepackage{subcaption}

\makeatletter
\theoremstyle{plain}
\newtheorem{thm}{\protect\theoremname}
\theoremstyle{plain}
\newtheorem{lyxalgorithm}{\protect\algorithmname}
\theoremstyle{remark}
\newtheorem{rem}{\protect\remarkname}
\theoremstyle{definition}
\newtheorem{problem}{\protect\problemname}
\theoremstyle{plain}
\newtheorem{cor}{\protect\corollaryname}
\theoremstyle{plain}
\newtheorem{lem}{\protect\lemmaname}
\theoremstyle{plain}
\newtheorem*{thm*}{\protect\theoremname}

\DeclareMathOperator{\Tr}{Tr}

\DeclareMathOperator{\sech}{sech}
\DeclareMathOperator{\signum}{sgn}
\DeclareMathOperator{\spec}{spec}
\DeclareMathOperator{\poly}{poly}
\DeclareMathOperator{\relu}{ReLU}
\DeclareMathOperator{\silu}{SiLU}
\DeclareMathOperator{\gelu}{GeLU}
\DeclareMathOperator{\grelu}{GReLU}
\DeclareMathOperator{\erf}{erf}

\allowdisplaybreaks

\makeatother

\providecommand{\algorithmname}{Algorithm}
\providecommand{\corollaryname}{Corollary}
\providecommand{\lemmaname}{Lemma}
\providecommand{\problemname}{Problem}
\providecommand{\remarkname}{Remark}
\providecommand{\theoremname}{Theorem}

\begin{document}
\title{Fermi--Dirac machines as quantizations of neurons}

\author{Alexander He}

\affiliation{Department of Physics, Cornell University, Ithaca, New York 14850, USA}
\author{Nana Liu}
\affiliation{Institute of Natural Sciences, Shanghai Jiao Tong University, Shanghai 200240, China}
\affiliation{School of Mathematical Sciences, Shanghai Jiao Tong University, Shanghai 200240, China}
\affiliation{Ministry of Education Key Laboratory in Scientific and Engineering Computing,
Shanghai Jiao Tong University, Shanghai 200240, China}
\affiliation{Global College, Shanghai Jiao Tong University, Shanghai 200240, China}

\author{Mark M. Wilde}
\affiliation{School of Electrical and Computer Engineering,
Cornell University, Ithaca, New York 14850, USA}

\begin{abstract}
Fermi--Dirac machines have been proposed recently as an approach
to solving semidefinite optimization problems on quantum computers.
Here, we reinterpret them as canonical quantizations
of classical neurons. By viewing a classical neuron as an activation
function applied to a parameterized classical Hamiltonian, we quantize
this model by replacing classical variables with operators whose
eigenvalues encode their possible values.
 This follows the standard approach to
canonical quantization in quantum mechanics. Crucially, when the Hamiltonian consists
of commuting operators, our construction reduces exactly to a classical
neuron. More generally, our approach yields an activation
observable, defined as an activation function applied to a parameterized
quantum Hamiltonian. The output of this quantized neuron is a random
variable with expectation value equal to that of the activation observable
with respect to an input state. We develop efficient hybrid quantum--classical algorithms
for evaluating outputs and gradients of our quantized neurons, enabling
evaluation and training. These algorithms rely on basic primitives that include
random sampling, Hamiltonian simulation, and the Hadamard test. We also
quantize a whole host of other activation functions, including the smooth
rectified linear unit (ReLU), sigmoid linear unit, Gaussian-smoothed
ReLU, and Gaussian error linear unit (GeLU), which
are known to be useful for deep learning applications.
Numerical experiments indicate that neurons based on quantum Hamiltonians can learn functions
that classical neurons cannot.
We further define
a computational decision problem based on Fermi--Dirac neurons and
prove that it is BQP-complete, providing complexity-theoretic evidence
against efficient classical simulation.  Finally, we generalize our
approach to continuous quantum variables and sketch two different
ways of composing these neurons into networks.
\end{abstract}

\maketitle

\tableofcontents

\section{Introduction}

\subsection{Background and motivation}

A key goal of quantum information science is to understand the differences
between the classical and quantum theories of information, with effects
like superposition, entanglement, and quantum uncertainty not being
present in the classical theory, while enabling technological capabilities
well beyond that of classical information processing systems~\cite{Long2025,Awschalom2025}.
These differences have been explored extensively in the context of
quantum computation~\cite{Dalzell2025}, communication~\cite{Gisin2007},
and sensing~\cite{Degen2017}, and more recently there has been a
growing effort to understand the differences between classical and
quantum machine learning~\cite{chang2025primer} and optimization
\cite{Abbas2024}.

In a recent paper~\cite{liu2026fermidirac}, Fermi--Dirac machines
were proposed as a method for solving semidefinite optimization problems
involving quantum measurements, and they were proven to be optimal
for Fermi--Dirac entropic regularizations of this task. Furthermore,
hybrid quantum--classical algorithms were proposed for solving these
measurement optimization problems, thus leading to a novel method
for performing optimization on quantum computers that appears to resist
efficient classical simulation. To review this concept briefly, let
$J\in\mathbb{N}$, let $\left(H_{1},\ldots,H_{J}\right)$ be a tuple
of Hamiltonians, let $\left(\theta_{1},\ldots,\theta_{J}\right)$
be a parameter vector, and let $T>0$ be a temperature. An important
constituent of a Fermi--Dirac machine is the following measurement
operator:
\begin{equation}
f_{T}(H(\theta))\coloneqq\left(e^{-H(\theta)/T}+I\right)^{-1},
\end{equation}
which is realized by applying the Fermi--Dirac function
\begin{equation}
f_{T}(x)\coloneqq\left(e^{-x/T}+1\right)^{-1},\label{eq:fd-function}
\end{equation}
also known as the sigmoid or logistic function~\cite{Hastie2009},
to the following parameterized Hamiltonian:
\begin{equation}
H(\theta)\coloneqq\sum_{j=1}^{J}\theta_{j}H_{j}.
\end{equation}
One can then evaluate the expectation of this Fermi--Dirac observable
with respect to a quantum state $\rho$, leading to the expected value
$\Tr\!\left[f_{T}(H(\theta))\rho\right]$.

\subsection{Canonical quantization of neurons}

\label{subsec:Canonical-quantization}

In this paper, we show how Fermi--Dirac machines admit a natural
interpretation as quantizations of classical neurons, leading to a
novel approach distinct from prior proposals, which we call \textit{Fermi--Dirac neurons}~\cite{Kapoor2016,cao2017,Wan2017,Hu2018QuantumNeuron,Killoran2019,Yan2020,Kristensen2021,Monteiro2021,Singh2024,barney2025,Roncallo2025}.
Throughout our paper, we use the terms Fermi--Dirac machine and Fermi--Dirac
neuron interchangeably. We also quantize well known neurons like the
smooth rectified linear unit~\cite{Dugas2000,glorot2011deep} and
the sigmoid linear unit~\cite{Elfwing2018,Ramachandran2017}, both
of which have played a key role in deep learning applications.

To gain an initial sense of our approach, recall the following basic
model for a neuron~\cite{Rosenblatt1958Perceptron,Werbos1974Beyond}:
\begin{equation}
\varphi(w^{T}z+b),\label{eq:1st-order-neuron}
\end{equation}
where $\varphi\colon\mathbb{R}\mapsto\mathbb{R}$ is a nonlinear activation
function, $w\in\mathbb{R}^{n}$ is a weight vector, $z\coloneqq\left(z_{1},\ldots,z_{n}\right)\in\left\{ -1,1\right\} ^{n}$
is a vector of input spin variables, $b\in\mathbb{R}$ is a bias,
and $w^{T}z$ is the standard inner product. The scalar function input
to $\varphi$ can be written as follows: 
\begin{equation}
w^{T}z+b=\sum_{i=1}^{n}w_{i}z_{i}+b,\label{eq:neuron-hamiltonian}
\end{equation}
and thus understood as a classical noninteracting Hamiltonian over
the spin variables in $z$.

Following the standard first quantization procedure 
(also known as canonical quantization)~\cite{Shankar1994}, we quantize the model in~\eqref{eq:1st-order-neuron}
by promoting the $i$th classical spin variable to a Pauli-$Z$ observable
$\sigma_{Z}^{\left(i\right)}$, representing the $i$th quantum spin
in a quantum spin model. As such, the scalar function in~\eqref{eq:neuron-hamiltonian}
is promoted to the following parameterized Hamiltonian acting on the
Hilbert space of $n$ qubits:
\begin{equation}
H_{C}(\theta)\coloneqq\sum_{i=1}^{n}w_{i}\sigma_{Z}^{\left(i\right)}+bI^{\otimes n},\label{eq:ham-q-class}
\end{equation}
where $\theta\coloneqq\left(w,b\right)$ and $\sigma_{Z}^{\left(i\right)}$
acts on the $i$th spin.

By the functional calculus, $\varphi\!\left(H_{C}(\theta)\right)$
is itself an observable, which is a function of the observable $H_{C}(\theta)$.
We refer to $\varphi\!\left(H_{C}(\theta)\right)$ as an \textit{activation
observable} and note that its expectation can be evaluated with respect
to an arbitrary input state $\rho$. In this case, since the Hamiltonian
in~\eqref{eq:ham-q-class} is classical (i.e., every term in~\eqref{eq:ham-q-class}
commutes), the expectation $\Tr\!\left[\varphi\!\left(H_{C}(\theta)\right)\rho\right]$
simplifies as follows:
\begin{equation}
\Tr\!\left[\varphi\!\left(H_{C}(\theta)\right)\rho\right]=\sum_{z\in\left\{ -1,1\right\} ^{n}}p(z)\varphi(w^{T}z+b),
\end{equation}
where $p(z)\coloneqq\langle z|\rho|z\rangle$ is a probability distribution,
with
$
|z\rangle\equiv|z_{1}\rangle\otimes\cdots\otimes|z_{n}\rangle
$
and $|z_{i}\rangle$ a $\pm1$-eigenstate of $\sigma_{Z}^{\left(i\right)}$.

As such, in this case, the expectation $\Tr\!\left[\varphi\!\left(H_{C}(\theta)\right)\rho\right]$
reduces to that of a classical neuron; indeed, this is the same expected
value that results from picking $z$ according to the probability
distribution $p(z)$ and sending the value~$z$ through the classical
neuron in~\eqref{eq:1st-order-neuron}. Thus, crucially, we immediately
see how our quantization $\varphi\!\left(H_{C}(\theta)\right)$ reduces
to a classical neuron for a classical Hamiltonian of the form in~\eqref{eq:ham-q-class}.
Indeed, by applying the functional calculus, we find in this case
that
\begin{equation}
\varphi\!\left(H_{C}(\theta)\right)=\sum_{z\in\left\{ -1,1\right\} ^{n}}\varphi(w^{T}z+b)|z\rangle\!\langle z|.
\end{equation}

\subsection{Summary of contributions}

\label{subsec:Summary-of-contributions}

In our paper, we generalize the approach outlined in Section~\ref{subsec:Canonical-quantization}
to an arbitrary parameterized quantum Hamiltonian $H_{Q}(\theta)$,
while also allowing for second-order interactions between the quantum
spins. We also quantize a number of activation functions that have
played a prominent role in deep learning. See Figure~\ref{fig:concept}
for a depiction of this conceptual contribution of our paper. For
example, we can set
\begin{multline}
H_{Q}(\theta)\coloneqq\sum_{\alpha,\beta\in\left\{ x,y,z\right\} }\sum_{i,j=1}^{n}\Omega_{\left(\alpha,i\right),\left(\beta,j\right)}\sigma_{\alpha}^{\left(i\right)}\otimes\sigma_{\beta}^{\left(j\right)}\\
+\sum_{\alpha\in\left\{ x,y,z\right\} }\sum_{i=1}^{n}\omega_{\alpha,i}\sigma_{\alpha}^{\left(i\right)}+bI^{\otimes n},\label{eq:fully-q-ham-1}
\end{multline}
where $\theta\coloneqq\left(\Omega,\omega,b\right)$, $\Omega\in\mathbb{R}^{3n\times3n}$,
$\omega\in\mathbb{R}^{3n}$, and the interacting and noninteracting
terms include all Pauli-$X$, $Y$, and $Z$ observables indexed by
$x$, $y$, and $z$. Particular examples of~\eqref{eq:fully-q-ham-1}
include the transverse-field Ising model \cite[Chapter~5]{Sachdev2011}
and Heisenberg model~\cite{Mattis2006,Goldschmidt2011} with tunable
parameters.

\begin{figure*}
\begin{centering}
\includegraphics[width=0.75\textwidth]{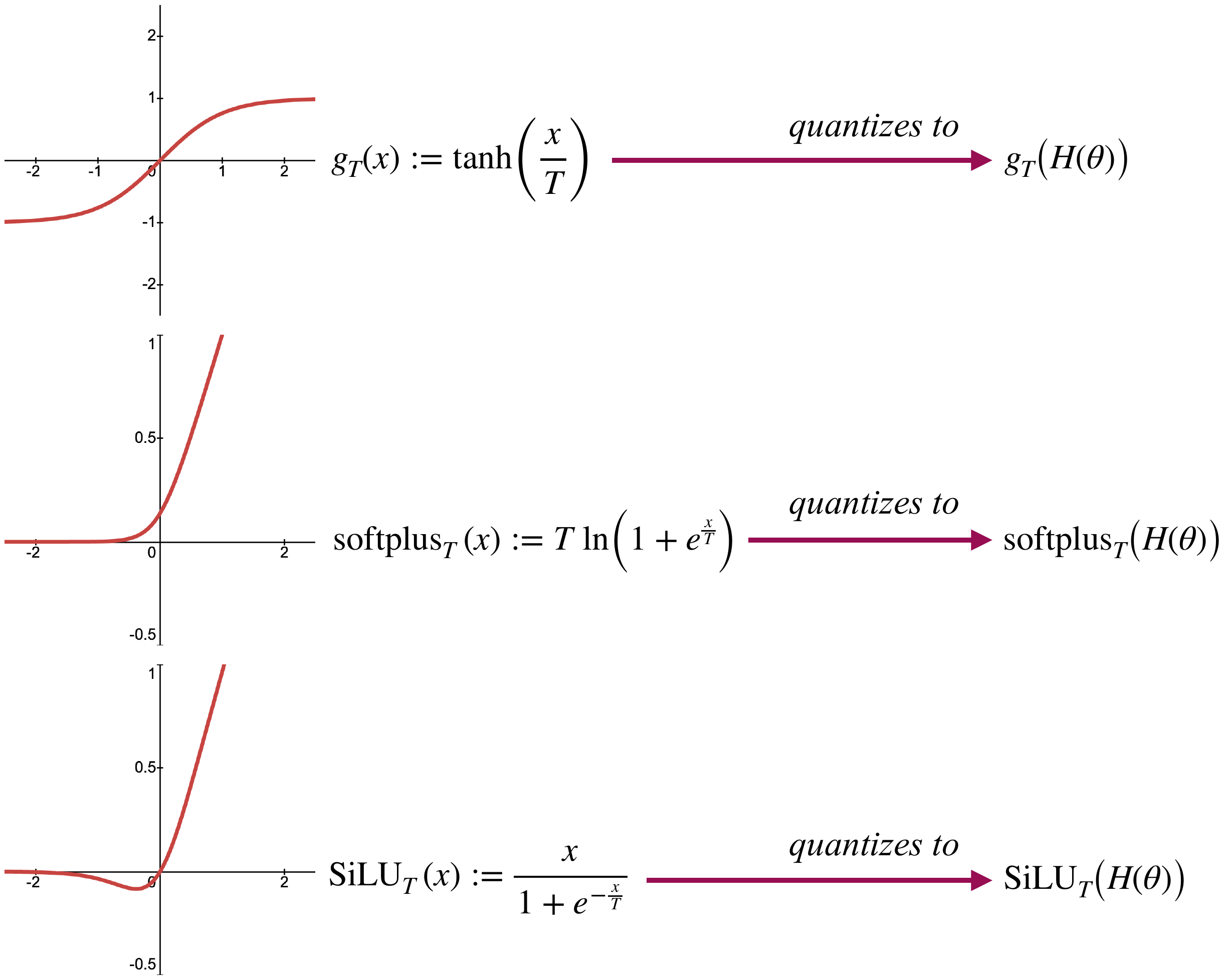}
\par\end{centering}
\caption{This figure illustrates one of the main conceptual contributions of
our paper. We quantize various common activation functions~\cite{Lederer2021,Dubey2022},
including hyperbolic tangent, softplus, and sigmoid linear unit (SiLU),
by following the canonical quantization procedure established by the
founders of quantum mechanics. Classically, an activation function
$\varphi(x)$ acts on the output $x=w^{T}z+b$ of a classical energy
function (Hamiltonian), where $z$ is a vector input to the neuron.
We replace the classical Hamiltonian with a parameterized quantum
Hamiltonian $H(\theta)=\sum_{j}\theta_{j}H_{j}$ and subsequently
apply an activation function to it, leading to an activation observable
$\varphi(H(\theta))$. Some of our technical contributions include
quantum algorithms for evaluating the expectation value $\Tr\!\left[\varphi(H(\theta))\rho\right]$
and its gradient, where $\rho$ is an input state and $\varphi$ is
one of the activation functions depicted.}\label{fig:concept}

\end{figure*}

We then construct the activation observable $\varphi(H_{Q}(\theta))$,
where $\varphi$ is a nonlinear activation function as before. The
result is a quantum machine learning model that can be used for binary
classification or function approximation, generalizing how single
neurons can be employed for such tasks in the classical case~\cite{Rosenblatt1958Perceptron,Werbos1974Beyond}.
We show how both the objective function $\theta\mapsto\Tr\!\left[\varphi\!\left(H_{Q}(\theta)\right)\rho\right]$
and its gradient can be efficiently estimated by hybrid quantum--classical
algorithms for various choices of $\varphi$, thus enabling efficient
evaluation and training of our quantized neurons. The training algorithms
involve random classical sampling, Hamiltonian simulation~\cite{lloyd1996universal,childs2018toward,Motta2022Emerging},
and the Hadamard test~\cite{Cleve1998}, thus bearing a similarity
to those recently proposed for training quantum Boltzmann machines
\cite{Patel2025a,Patel2025,Minervini2026,Liu2025}; we remark later
on the duality between Fermi--Dirac and quantum Boltzmann machines.
We also develop quantum algorithms for single-shot evaluation of $\Tr\!\left[\varphi\!\left(H_{Q}(\theta)\right)\rho\right]$,
which rely on algorithms that we refer to as quantum convolution and
multiplication and in turn build on quantum algorithms known as the
power of one qumode~\cite{Liu2016} and Schr\"odingerization~\cite{Jin2024a, Jin2023,Jin2024}. Table~\ref{tab:algorithm-summary} provides a brief summary of all of the algorithms introduced in our paper.

Other contributions of our paper are as follows. We perform small-scale
numerical experiments to support the case that classical neurons cannot
simulate our quantized neurons. We also generalize the whole framework
by quantizing neurons that have continuous inputs, by means of continuous
quantum variables~\cite{Serafini2023}, and we outline two different
ways of composing our quantized neurons to realize canonical quantizations
of neural networks. Finally, with the goal of providing complexity-theoretic
evidence that Fermi--Dirac machines cannot be simulated classically,
we define a computational decision problem based on them and prove
that it is BQP-complete, meaning that this problem can be solved efficiently
using a quantum computer and every decision problem efficiently solvable
by a quantum computer can be efficiently reduced to this problem.

\begin{table*}[t]
\centering
\renewcommand{\arraystretch}{1.25}
\begin{tabular}{p{0.39\textwidth} p{0.59\textwidth} }
\hline
\textbf{Algorithm} & \textbf{Purpose / Output}  \\
\hline

Algorithm~\ref{alg:grad-est}:
Fermi--Dirac gradient estimation
&
Estimates the partial derivative
$\frac{\partial}{\partial\theta_{j}}
\Tr[g_{T}(H(\theta))\rho]$.

\\

Algorithm~\ref{alg:obj-func-est}:
Fermi--Dirac objective estimation
&
Estimates the expectation value
$\Tr[g_{T}(H(\theta))\rho]$.

\\

Algorithm~\ref{alg:convolution-on-A}:
Quantum convolution algorithm
&
Outputs a random variable with expected value
$\Tr[s(A)\rho]$, where $s=q*r$.

\\

Algorithm~\ref{alg:FD-thermal-alg}:
Fermi--Dirac neuron simulation
&
Simulates the activation observable
$g_T(H(\theta))$ associated with a
Fermi--Dirac neuron.

\\

Algorithm~\ref{alg:ReLU-one-shot}:
Smooth ReLU neuron simulation
&
Realizes the smooth ReLU activation observable
$r_T(H(\theta))$.

\\

Algorithm~\ref{alg:convolution-and-mult-on-A}:
Q.~convolution \& multiplication
&
Outputs a random variable with expectation value
$
\Tr[s_1(At_1)s_2(At_2)\rho],
$
where $s_i = q_i * r_i$, for $i\in \{1,2\}$.
\\
Algorithm~\ref{alg:SiLU-one-shot}:
SiLU neuron simulation
&
Realizes the SiLU activation observable
$\silu_T(H(\theta))$.

\\

Algorithm~\ref{alg:grad-est-squared-loss}:
Squared-loss gradient estimation
&
Estimates the partial derivative
$\frac{\partial}{\partial\theta_j}\mathcal{L}^{(2)}(\theta)$
for the squared-loss objective.

\\

Algorithm~\ref{alg:grad-log-loss}:
Logistic-loss gradient estimation
&
Estimates the quantity $\zeta_j$
appearing in the gradient of the logistic-loss objective.

\\

Algorithm~\ref{alg:log-loss-obj-func-est}:
Logistic-loss objective estimation
&
Estimates the quantity $\zeta$
associated with the logistic-loss objective.

\\

Algorithm~\ref{alg:grad-smooth-relu}:
Smooth-ReLU gradient estimation
&
Estimates the quantity $\zeta_j$
appearing in the gradient of the smooth-ReLU objective.

\\

Algorithm~\ref{alg:softplus-obj-func-est}:
Softplus objective estimation
&
Estimates the quantity $\zeta$
associated with the softplus objective function.

\\

Algorithm~\ref{alg:grad-silu}:
SiLU gradient estimation
&
Estimates the quantity $\zeta_j$
appearing in the gradient of the SiLU objective function.

\\

Algorithm~\ref{alg:silu-obj-func-est}:
SiLU objective estimation
&
Estimates the quantity $\zeta$
associated with the SiLU objective function.

\\

Algorithm~\ref{alg:obj-func-est-1}:
Cross-entropy estimation
&
Estimates the cross entropy
$\Xi(\eta\|\rho(\theta))$.

\\
\hline
\end{tabular}
\caption{Summary of the hybrid quantum--classical algorithms introduced in this work.}
\label{tab:algorithm-summary}
\end{table*}

\subsection{Paper organization}

The rest of our paper is organized as follows.
In Section~\ref{sec:Fermi=002013Dirac-machines-quants-neurs}, we
extend the development in Section~\ref{subsec:Canonical-quantization},
by showing how to quantize second-order neurons and, more generally,
by developing quantized neurons that involve arbitrary parameterized
Hamiltonians (Section~\ref{subsec:Quantizing-second-order-neurons}).
Focusing specifically on Fermi--Dirac machines in Section~\ref{subsec:FD-objective-function-grad},
we establish formulas for the Fermi--Dirac objective function (Theorem
\ref{thm:obj-func-formula}) and its gradient (Theorem~\ref{thm:gradient-obj-func})
that are amenable to efficient estimation by means of hybrid quantum--classical
algorithms. We provide detailed specifications of these algorithms
in Section~\ref{subsec:HQCs-FD-obj-func-grad}. In applications, it
is important to have an algorithm that simulates the firing of a Fermi--Dirac
neuron, and we develop such an algorithm (Algorithm~\ref{alg:FD-thermal-alg})
in Section~\ref{subsec:one-shot-FD-neuron}, which is based on a more
general quantum convolution algorithm (Algorithm~\ref{alg:convolution-on-A})
that may be of independent interest. In Section~\ref{subsec:binary-class-func-approx},
we show how to train Fermi--Dirac neurons for binary classification
and function approximation, by means of either squared-loss (Section~\ref{subsec:Squared-loss-minimization}) or logistic-loss (Section~\ref{subsec:Logistic-loss-minimization}) minimization.

In Section~\ref{sec:Quantizing-smooth-ReLU}, we quantize neurons
based on the rectified linear unit (ReLU), developing the basic models,
formulas for the objective functions and gradient, and algorithms
for estimating them. As mentioned previously in Section~\ref{subsec:Canonical-quantization},
ReLU-like activation functions are of particular interest, due to
their usage in deep learning applications. In Section~\ref{subsec:smooth-ReLU},
we quantize neurons based on the smooth ReLU (also called softplus)
activation function, and in Section~\ref{subsec:SiLU}, we quantize
neurons based on the sigmoid linear unit (also called swish) activation
function. Our algorithm for single-shot simulation of the sigmoid
linear unit is based on a quantum convolution and multiplication algorithm
(Algorithm~\ref{alg:convolution-and-mult-on-A}), which may be of
independent interest.

In Section~\ref{sec:Quantizing-Gaussian-activation}, we modify all of the activation functions considered
in Sections~\ref{sec:Fermi=002013Dirac-machines-quants-neurs} and
\ref{sec:Quantizing-smooth-ReLU} to be based on the Gaussian probability
density rather than the logistic probability density. This leads to
activation functions called erf, Gaussian ReLU (GReLU), and Gaussian
error linear unit (GeLU), which ``gaussify'' Fermi--Dirac, ReLU,
and SiLU quantized neurons, respectively. A key practical advantage
of this Gaussian approach is that the control qumode state used in
Algorithms~\ref{alg:FD-thermal-alg},~\ref{alg:ReLU-one-shot}, and
\ref{alg:SiLU-one-shot} can be initialized to a bosonic Gaussian
state, which is easier to prepare in experimental setups involving quantum-optical elements~\cite{Andersen2010}.

In Section~\ref{sec:gen-CV-case}, we generalize the whole framework
to continuous quantum variables, showing how to quantize neurons with
continuous inputs (Section~\ref{subsec:Canonical-quantization-of-CV-neurons})
and how to train and test them (Section~\ref{subsec:Training-and-testing-CV-case}).
Interestingly, our hybrid quantum--classical algorithms developed
for the finite-dimensional case apply nearly unchanged to the continuous-variable
case.

Section~\ref{sec:Numerical-experiments} showcases the results of
numerical experiments, in which we trained and tested our quantized
neurons for both binary classification and function approximation.
The results reported here demonstrate, for the examples considered,
that classical neurons cannot achieve the same performance as our
quantized neurons do for these tasks, whenever the input data is quantum
(i.e., non-orthogonal or entangled states). This indicates, consistent
with prior findings in quantum machine learning (see, e.g., \cite[Section~III-B]{chang2025primer}),
that our quantized neurons may demonstrate an advantage when processing
quantum data. If the input data is classical in both training and
testing, then our quantized neurons do not appear to provide better
performance over classical neurons.

In Section~\ref{sec:Proposals-for-q-neural-nets}, we provide two
proposals for quantizing neural networks. In one approach, we replace
every neuron in a neural network with a quantized neuron that takes
in quantum states but outputs a classical variable. This allows for
hybrid network, in which some of the neurons are quantized and others
remain classical. This construction appears to be most effective when
the first layer of a neural network consists of quantized neurons
and accept quantum inputs, but neurons in other layers are classical.
In another approach, we construct a quantum observable network, which
consists of recursive application of activation functions and tunable
weights to a parameterized Hamiltonian.

In Section~\ref{sec:Complexity-theoretic-evidence}, we provide complexity-theoretic
evidence that Fermi--Dirac neurons cannot be simulated classically,
by proving that a computational decision problem based on them is
BQP-complete. Our result here relies on the guided local Hamiltonian
problem~\cite{Cade2023}, which is already known to be BQP-complete.

Section~\ref{sec:Connections-with-QBMs} comments on connections between
our quantized neurons and quantum Boltzmann machines, pinpointing
a duality between these models, related to the duality between states
and observables in quantum mechanics. Both approaches are Hamiltonian-based, highlighting their common foundation in fundamental quantum-physical
principles.

In Section~\ref{subsec:HQC-cross-entropy-log-part}, we show how our
approach in Theorems~\ref{thm:obj-func-formula},~\ref{thm:obj-func-softplus},
and~\ref{thm:silu-obj-func} and Algorithms~\ref{alg:obj-func-est},
\ref{alg:softplus-obj-func-est}, and~\ref{alg:silu-obj-func-est}
can be repurposed for estimating cross entropy and the log-partition
function, which are of independent interest and should find application
in thermodynamics and quantum Boltzmann machine learning.

We finally conclude in Section~\ref{sec:Conclusion} by summarizing
the main findings of our paper and outlining directions for future
research.

\section{Fermi--Dirac machines as quantizations of neurons}

\label{sec:Fermi=002013Dirac-machines-quants-neurs}

\subsection{Quantizing second-order neurons}

\label{subsec:Quantizing-second-order-neurons}

In Section~\ref{subsec:Canonical-quantization}, we already showed
how to quantize neurons of the form in~\eqref{eq:1st-order-neuron}.
Going beyond the model presented there, one can alternatively consider
a second-order neuron~\cite{Ivakhnenko1971,rumelhart1986parallel,Giles1987},
which involves quadratic interactions between the variables, via a
weight matrix $W\in\mathbb{R}^{n\times n}$:
\begin{equation}
\varphi(z^{T}Wz+w^{T}z+b),\label{eq:second-order-neuron}
\end{equation}
where $z^{T}Wz=\sum_{i,j=1}^{n}W_{ij}z_{i}z_{j}$. Thus, the scalar
function input to $\varphi$ can be written as follows: 
\begin{equation}
z^{T}Wz+w^{T}z+b=\sum_{i,j=1}^{n}W_{ij}z_{i}z_{j}+\sum_{i=1}^{n}w_{i}z_{i}+b,\label{eq:neuron-hamiltonian-second-order}
\end{equation}
and understood as a classical interaction Hamiltonian. Following
the same approach as in Section~\ref{subsec:Canonical-quantization},
we quantize the model in~\eqref{eq:second-order-neuron}--\eqref{eq:neuron-hamiltonian-second-order}
by promoting each classical spin variable to a Pauli-$Z$ observable
$\sigma_{Z}$. The scalar function in~\eqref{eq:neuron-hamiltonian-second-order}
is then promoted to a parameterized Hamiltonian of the following form:
\begin{equation}
H_{C}(\theta)\coloneqq\sum_{i,j=1}^{n}W_{ij}\sigma_{Z}^{\left(i\right)}\otimes\sigma_{Z}^{\left(j\right)}+\sum_{i=1}^{n}w_{i}\sigma_{Z}^{\left(i\right)}+bI^{\otimes n},\label{eq:ham-q-class-1}
\end{equation}
where $\theta\coloneqq\left(W,w,b\right)$ and $\sigma_{Z}^{\left(i\right)}$
acts on the $i$th spin. Then $\varphi\!\left(H_{C}(\theta)\right)$
is an activation observable, generalizing that resulting from~\eqref{eq:ham-q-class}.

For the Hamiltonian in~\eqref{eq:ham-q-class-1}, all of the summands
simultaneously commute, so that applying the functional calculus results
in the following alternative expression for the activation observable
$\varphi\!\left(H_{C}(\theta)\right)$: 
\begin{equation}
\varphi\!\left(H_{C}(\theta)\right)=\sum_{z\in\left\{ -1,1\right\} ^{n}}\varphi(z^{T}Wz+w^{T}z+b)|z\rangle\!\langle z|.\label{eq:diagonalization-2nd-order-neuron}
\end{equation}
This implies that quantization of the classical Hamiltonian in~\eqref{eq:neuron-hamiltonian-second-order}
does not produce any distinction between the quantum and classical
cases. That is, evaluating the expectation of this observable when
the system is in the classical state $|z\rangle\!\langle z|$ results
in the value
\begin{equation}
\Tr\!\left[\varphi\!\left(H_{C}(\theta)\right)|z\rangle\!\langle z|\right]=\varphi(z^{T}Wz+w^{T}z+b).
\end{equation}
Evaluating for a general quantum input state $\rho$ results in the
following expected value:
\begin{equation}
\Tr\!\left[\varphi\!\left(H_{C}(\theta)\right)\rho\right]=\sum_{z\in\left\{ -1,1\right\} ^{n}}p(z)\varphi(z^{T}Wz+w^{T}z+b),\label{eq:classical-second-order-gen-model}
\end{equation}
where $p(z)\coloneqq\langle z|\rho|z\rangle$.

We can realize models beyond the classical case by allowing Hamiltonians
that include non-commuting terms. For example, we can set $H_{Q}(\theta)$
as in~\eqref{eq:fully-q-ham-1} and consider the activation observable
$\varphi(H_{Q}(\theta))$, where $\varphi$ is a nonlinear activation
function as before. A key difference between the more general Hamiltonian
$H_{Q}(\theta)$ in~\eqref{eq:fully-q-ham-1} and the classical Hamiltonian
$H_{C}(\theta)$ in~\eqref{eq:ham-q-class-1} is that it is no longer
possible to find a simple diagonalization of $\varphi(H_{Q}(\theta))$,
given that the dependence of the spectrum and the eigenvectors of
$H_{Q}(\theta)$ on $\theta$ is rather complex. However, in what
follows, we show how hybrid quantum--classical algorithms can be
used for efficient estimation of the expected value of $\varphi(H_{Q}(\theta))$
and its gradient, for certain choices of $\varphi$, thus enabling
efficient evaluation and training. For this purpose, in what follows,
for $J,n,k\in\mathbb{N}$, we consider a rather general parameterized
Hamiltonian of the following form:
\begin{equation}
H(\theta)\coloneqq\sum_{j=1}^{J}\theta_{j}H_{j},\label{eq:general-k-local-ham}
\end{equation}
where $\theta_{j}\in\mathbb{R}$ for all $j\in\left[J\right]\coloneqq\left\{ 1,\ldots,J\right\} $
and $H_{j}$ is a $n$-qubit Hamiltonian that acts nontrivially on
$k$ qubits, where $k$ is a constant independent of $n$.

\subsection{Fermi--Dirac objective function and its gradient}

\label{subsec:FD-objective-function-grad}

For a temperature $T>0$ and a quantum state $\rho$, consider the
following objective function: 
\begin{equation}
\Tr\!\left[g_{T}(H(\theta))\rho\right],\label{eq:obj-func-tanh}
\end{equation}
where
\begin{align}
g_{T}(x) & \coloneqq\tanh\!\left(\frac{x}{T}\right)\label{eq:tanh-g-T}
\end{align}
and $H(\theta)$ is a parameterized Hamiltonian of the form in~\eqref{eq:general-k-local-ham}.
This objective function and its gradient play an important role for
training Fermi--Dirac machines for binary classification and function
approximation, as detailed later in Section~\ref{subsec:binary-class-func-approx}.
The threshold realized by the hyperbolic tangent function becomes
sharper as $T$ decreases, and it converges to the sign function in
the limit $T\to0$. Also, note that
\begin{equation}
f_{T}(x)=\frac{1}{2}\left(1+g_{T}\!\left(\frac{x}{2}\right)\right),\label{eq:tanh-to-fd}
\end{equation}
so that all of our results apply equally well to the Fermi--Dirac
function $f_{T}(x)=\left(e^{-x/T}+1\right)^{-1}$ defined in~\eqref{eq:fd-function}.
However, as our paper focuses on spin variables taking values $\pm1$,
we find it more convenient to work with $g_{T}(x)$ instead of $f_{T}(x)$
and do so for the remainder of our paper. Owing to the simple relation
in~\eqref{eq:tanh-to-fd} between the Fermi--Dirac and hyperbolic
tangent functions, we simply refer to quantized neurons based on $g_{T}(H(\theta))$
as Fermi--Dirac machines.

Our goal now is to find expressions for the gradient $\nabla_{\theta}\Tr\!\left[g_{T}(H(\theta))\rho\right]$
and objective function $\Tr\!\left[g_{T}(H(\theta))\rho\right]$ that
are suitable for estimation on a quantum computer, assuming that we
have sample access to the state $\rho$ and the ability to perform
Hamiltonian simulation according to $H(\theta)$. Under our assumption
in~\eqref{eq:general-k-local-ham} of a $k$-local Hamiltonian, standard
techniques of Hamiltonian simulation are applicable~\cite{lloyd1996universal,childs2018toward,Motta2022Emerging}.

We begin by establishing such a formula for the gradient $\nabla_{\theta}\Tr\!\left[g_{T}(H(\theta))\rho\right]$,
i.e., for each partial derivative. The expression given in Theorem
\ref{thm:gradient-obj-func} below follows from a novel formula for
the derivative of the matrix hyperbolic tangent function, which is
proven in Appendix~\ref{app:grad-deriv} and may be of independent
interest. 
\begin{thm}
\label{thm:gradient-obj-func}Let $T>0$ be a temperature, let $g_{T}(x)$
be defined as in~\eqref{eq:tanh-g-T}, and let $\rho$ be a quantum
state. Then the following equality holds:
\begin{multline}
\frac{\partial}{\partial\theta_{j}}\Tr\!\left[g_{T}(H(\theta))\rho\right]=\\
\frac{1}{T}\mathbb{E}_{\substack{t\sim\mu,\\
s\sim\upsilon
}
}\!\left[\Re\!\left[\Tr\!\left[H_{j}e^{iH(\theta)t/T}\mathcal{U}_{st/T}^{H(\theta)}(\rho)\right]\right]\right],\label{eq:formula-deriv-obj-func}
\end{multline}
where $\mu(t)$ is the following probability density function
\begin{equation}
\mu(t)\coloneqq\frac{t}{2\sinh\!\left(\frac{\pi t}{2}\right)},\label{eq:high-peak-tent-2-def}
\end{equation}
$\upsilon$ is a uniform random variable over the unit interval $\left[0,1\right]$,
and $\mathcal{U}_{t}^{H(\theta)}$ is the following unitary quantum
channel:
\begin{equation}
\mathcal{U}_{t}^{H(\theta)}(\cdot)\coloneqq e^{-iH(\theta)t}(\cdot)e^{iH(\theta)t}.
\end{equation}
\end{thm}

\begin{proof}
See Appendix~\ref{app:grad-deriv}.
\end{proof}
Inspecting the formula in~\eqref{eq:formula-deriv-obj-func}, we see
that it is proportional to the expectation of an expression of the
form $\Re\!\left[\Tr\!\left[HU\sigma\right]\right]$, where $H$ is
a Hamiltonian, $U$ is a unitary, and $\sigma$ is state (to see this,
set $H=H_{j}$, $U=e^{iH(\theta)t/T}$, and $\sigma=\mathcal{U}_{st/T}^{H(\theta)}(\rho)$).
As such, we can employ random sampling, the Hadamard test, and Hamiltonian
simulation to estimate this formula efficiently by means of a hybrid
quantum--classical algorithm similar to those put forward in~\cite{Patel2025a,Patel2025,Minervini2026}.
As $T$ decreases (i.e., if we want a sharper threshold function),
both the Hamiltonian simulation time increases and the sampling overhead
does as well, the latter due to the prefactor $\frac{1}{T}$ in~\eqref{eq:formula-deriv-obj-func}.

By employing Theorem~\ref{thm:gradient-obj-func} along with the fundamental
theorem of calculus, we can express the objective function in~\eqref{eq:obj-func-tanh}
in terms of a telescoping sum representing integrals of its partial
derivatives, leading to the following theorem:
\begin{thm}
\label{thm:obj-func-formula}Let $T>0$ be a temperature, let $g_{T}(x)$
be as defined in~\eqref{eq:tanh-g-T}, and let $\rho$ be a quantum
state. Then the following equality holds:
\begin{multline}
\Tr\!\left[g_{T}(H(\theta))\rho\right]=\\
\frac{\left\Vert \theta\right\Vert _{1}}{T}\mathbb{E}_{\substack{j\sim q,\\
t\sim\mu,\\
s\sim\upsilon,\\
\lambda\sim\upsilon
}
}\!\left[\signum(\theta_{j})\Re\!\left[\Tr\!\left[H_{j}e^{iH(j,\lambda)t/T}\mathcal{U}_{st/T}^{H(j,\lambda)}(\rho)\right]\right]\right],\label{eq:obj-func-FTOC-formula}
\end{multline}
where $q(j)$ is the following probability distribution on $\left[J\right]$:
\begin{equation}
q(j)\coloneqq\frac{\left|\theta_{j}\right|}{\left\Vert \theta\right\Vert _{1}},\qquad\left\Vert \theta\right\Vert _{1}\coloneqq\sum_{j=1}^{J}\left|\theta_{j}\right|,
\end{equation}
$\mu$ is the probability density in~\eqref{eq:high-peak-tent-2-def},
$\upsilon$ is a uniform random variable over the unit interval $\left[0,1\right]$,
and
\begin{align}
H(j,\lambda) & \equiv H(\theta^{\left(j\right)}(\lambda)),\\
\theta^{\left(j\right)}(\lambda) & \coloneqq\left(0,\ldots,0,\lambda\theta_{j},\theta_{j+1},\ldots,\theta_{J}\right),\\
H(\theta_{1},\ldots,\theta_{J}) & \equiv H(\theta)=\sum_{j=1}^{J}\theta_{j}H_{j}.
\end{align}
\end{thm}

\begin{proof}
See Appendix~\ref{app:obj-func-deriv}.
\end{proof}
Inspecting the formula in~\eqref{eq:obj-func-FTOC-formula}, it again
is equal to the expectation of an expression of the form $\Re\!\left[\Tr\!\left[HU\sigma\right]\right]$,
so that the same aforementioned primitives can be used in a hybrid
quantum--classical algorithm for estimating it.

\subsection{Hybrid quantum--classical algorithms for estimating the Fermi--Dirac
objective function and its gradient}

\label{subsec:HQCs-FD-obj-func-grad}

Based on Theorems~\ref{thm:gradient-obj-func} and~\ref{thm:obj-func-formula},
here we delineate hybrid quantum--classical algorithms for estimating
the gradient $\nabla_{\theta}\Tr\!\left[g_{T}(H(\theta))\rho\right]$
and the objective function $\Tr\!\left[g_{T}(H(\theta))\rho\right]$,
which both make use of classical random sampling, Hamiltonian simulation,
and the Hadamard test.

We begin with Algorithm~\ref{alg:grad-est} for estimating $\frac{\partial}{\partial\theta_{j}}\Tr\!\left[g_{T}(H(\theta))\rho\right]$,
which can be implemented for all $j\in\left[J\right]$ in order to
estimate the gradient $\nabla_{\theta}\Tr\!\left[g_{T}(H(\theta))\rho\right]$.
\begin{lyxalgorithm}
\label{alg:grad-est}A hybrid quantum--classical algorithm for estimating
the $j$th partial derivative $\frac{\partial}{\partial\theta_{j}}\Tr\!\left[g_{T}(H(\theta))\rho\right]$
consists of the following steps:
\begin{enumerate}
\item Set $k\leftarrow1$, and set
\begin{equation}
K\leftarrow O\!\left(\left(\frac{\left\Vert H_{j}\right\Vert }{T\varepsilon}\right)^{2}\ln\!\left(\frac{1}{\delta}\right)\right),
\end{equation}
where $\varepsilon>0$ is the desired accuracy and $\delta\in\left(0,1\right)$
is the desired failure probability.
\item Sample $t\sim\mu$ and $s\sim\upsilon$, where the probability densities
$\mu$ and $\upsilon$ are defined in Theorem~\ref{thm:gradient-obj-func}.
\item Prepare the state $\mathcal{U}_{st/T}^{H(\theta)}(\rho)$ using one
sample of $\rho$ and Hamiltonian simulation to realize the unitary
channel $\mathcal{U}_{st/T}^{H(\theta)}$.
\item Perform the quantum circuit depicted in Figure~\ref{fig:Quantum-circuits-grad-obj-func}(a),
with measurement outcomes $Z_{k}\in\left\{ -1,1\right\} $ for the
$\sigma_{Z}$ measurement and $X_{k}\in\spec(H_{j})$ for the $H_{j}$
measurement. Set $Y_{k}\leftarrow\frac{1}{T}Z_{k}\cdot X_{k}$. Set
$k\leftarrow k+1.$
\item Repeat Steps 2-4 $K-1$ more times. Compute the average $\overline{Y_{K}}\coloneqq\frac{1}{K}\sum_{k=1}^{K}Y_{k}$
and output this value as an estimate of $\frac{\partial}{\partial\theta_{j}}\Tr\!\left[g_{T}(H(\theta))\rho\right]$.
\end{enumerate}
\end{lyxalgorithm}

By the Hoeffding inequality~\cite{Hoeffding1963} (e.g., in the form
of \cite[Theorem~1]{Bandyopadhyay2023}) and Theorem~\ref{thm:gradient-obj-func},
we are guaranteed that
\begin{equation}
\Pr\!\left[\left|\overline{Y_{K}}-\frac{\partial}{\partial\theta_{j}}\Tr\!\left[g_{T}(H(\theta))\rho\right]\right|\leq\varepsilon\right]\geq1-\delta.\label{eq:hoeffding-guarantee-FD-grad}
\end{equation}

Algorithm~\ref{alg:obj-func-est} below provides a method for estimating
the objective function $\Tr\!\left[g_{T}(H(\theta))\rho\right]$,
which makes use of the formula stated in Theorem~\ref{thm:obj-func-formula}.
\begin{lyxalgorithm}
\label{alg:obj-func-est}A hybrid quantum--classical algorithm for
estimating the objective function $\Tr\!\left[g_{T}(H(\theta))\rho\right]$
consists of the following steps:
\begin{enumerate}
\item Set $k\leftarrow1$, and set
\begin{equation}
K\leftarrow O\!\left(\left(\frac{\left\Vert \theta\right\Vert _{1}\max_{j\in\left[J\right]}\left\Vert H_{j}\right\Vert }{T\varepsilon}\right)^{2}\ln\!\left(\frac{1}{\delta}\right)\right),
\end{equation}
where $\varepsilon>0$ is the desired accuracy and $\delta\in\left(0,1\right)$
is the desired failure probability.
\item Sample $j\sim q$, $t\sim\mu$, $s\sim\upsilon$, $\lambda\sim\upsilon$.
\item Prepare the state $\mathcal{U}_{st/T}^{H(\theta^{\left(j\right)}(\lambda))}(\rho)$
using one sample of $\rho$ and Hamiltonian simulation to realize
the unitary channel $\mathcal{U}_{st/T}^{H(\theta^{\left(j\right)}(\lambda))}$.
\item Perform the quantum circuit depicted in Figure~\ref{fig:Quantum-circuits-grad-obj-func}(b),
with measurement outcomes $Z_{k}\in\left\{ -1,1\right\} $ for the
$\sigma_{Z}$ measurement and $X_{k}\in\spec(H_{j})$ for the $H_{j}$
measurement. Set $W_{k}\leftarrow\frac{\left\Vert \theta\right\Vert _{1}}{T}\cdot\signum(\theta_{j})Z_{k}\cdot X_{k}$.
Set $k\leftarrow k+1.$
\item Repeat Steps 2-4 $K-1$ more times. Compute the average $\overline{W_{K}}\coloneqq\frac{1}{K}\sum_{k=1}^{K}W_{k}$
and output this value as an estimate of $\Tr\!\left[g_{T}(H(\theta))\rho\right]$.
\end{enumerate}
\end{lyxalgorithm}

By the Hoeffding inequality, we are guaranteed that
\begin{equation}
\Pr\!\left[\left|\overline{W_{K}}-\Tr\!\left[g_{T}(H(\theta))\rho\right]\right|\leq\varepsilon\right]\geq1-\delta.
\end{equation}
Figure~\ref{fig:Quantum-circuits-grad-obj-func} depicts the quantum
circuits used in Algorithms~\ref{alg:grad-est} and~\ref{alg:obj-func-est}.

Let us note here that the objective function $\Tr\!\left[g_{T}(H(\theta))\rho\right]$
can alternatively be estimated by means of frameworks like quantum
singular value transformation~\cite{Gilyen2019} or Schr\"odingerization
\cite{Jin2024a,Jin2023,Jin2024}, which require more a complex quantum control
operation in order to be realized. However, what Algorithm~\ref{alg:obj-func-est}
demonstrates is that it is possible to eliminate this controlled-operation 
logic and instead perform a combination of classical random sampling,
Hamiltonian simulation, and the Hadamard test in order to estimate
$\Tr\!\left[g_{T}(H(\theta))\rho\right]$. Regardless, in the next
section, we show how to estimate $\Tr\!\left[g_{T}(H(\theta))\rho\right]$
by means of a generalization of Schr\"odingerization, which we call
a quantum convolution algorithm (Algorithm~\ref{alg:convolution-on-A}),
as doing so can be useful for simulating the firing of a Fermi--Dirac
neuron.

\begin{figure}
\begin{centering}
\includegraphics[width=\linewidth]{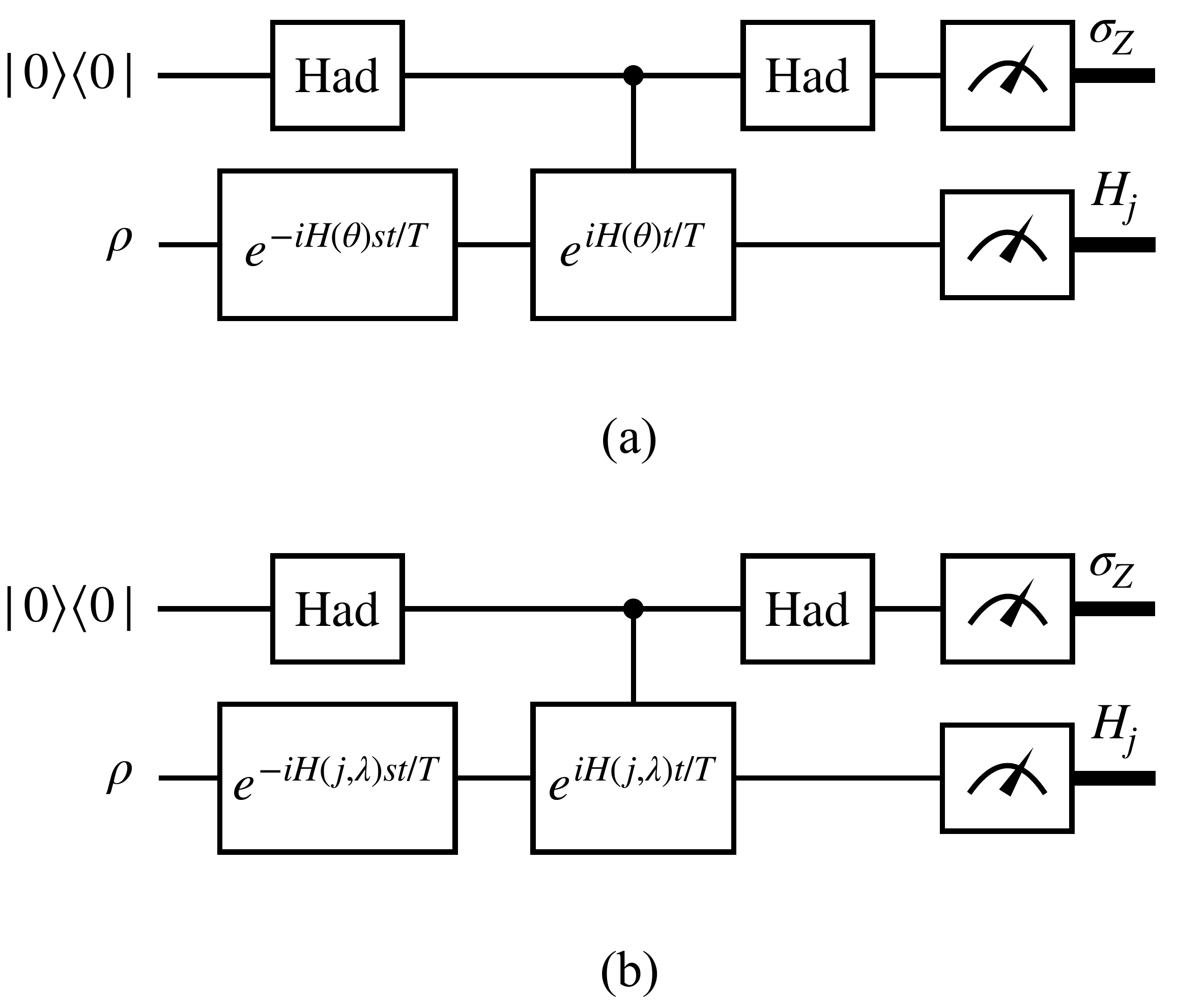}
\par\end{centering}
\caption{(a) Quantum circuit used in Algorithm~\ref{alg:grad-est} for estimating
the $j$th partial derivative $\frac{\partial}{\partial\theta_{j}}\Tr\!\left[g_{T}(H(\theta))\rho\right]$.
(b) Quantum circuit used in Algorithm~\ref{alg:obj-func-est} for
estimating the objective function $\Tr\!\left[g_{T}(H(\theta))\rho\right]$.
The notation $H(j,\lambda)$ is defined in Theorem~\ref{thm:obj-func-formula}.
Both circuits involve classical random sampling, Hamiltonian simulation,
and the Hadamard test. \textquotedblleft Had\textquotedblright{} denotes
a single-qubit Hadamard gate.}\label{fig:Quantum-circuits-grad-obj-func}
\end{figure}

\subsection{Quantum algorithm for simulating the firing of a Fermi--Dirac neuron}

\label{subsec:one-shot-FD-neuron}

In various applications, it is important to have a method that simulates
the firing of a Fermi--Dirac neuron, i.e., when just a single sample of
$\rho$ is available. For this purpose, we develop a general quantum
convolution algorithm that outputs a random variable with expected
value $\Tr\!\left[s(A)\rho\right]$, where $s$ is equal to the convolution
of an even probability density function $q$ and another function
$r$, the Hermitian operator $A$ is a Hamiltonian, and $\rho$ is
an input state. This algorithm is a generalization of \cite[Algorithm~19]{liu2026fermidirac}
and relies on quantum algorithmic primitives known as the power of
one qumode~\cite{Liu2016} and Schr\"odingerization~\cite{Jin2024a, Jin2023,Jin2024}.
It may find application beyond our purposes here.
\begin{lyxalgorithm}[Quantum convolution algorithm]
\label{alg:convolution-on-A}Let $q\colon\mathbb{R}\to\left[0,1\right]$
be an even probability density function, let $r\colon\mathbb{R}\to\mathbb{R}$
be a measurable function such that the convolution $s=q*r$ is well
defined, and let $A$ be a Hamiltonian. The algorithm for outputting
a random variable with expected value $\Tr\!\left[s(A)\rho\right]$
proceeds as follows:
\begin{enumerate}
\item Prepare a control qumode register in the state $|q\rangle\!\langle q|$,
where
\begin{equation}
|q\rangle\coloneqq\int_{-\infty}^{\infty}dp\,\sqrt{q(p)}|p\rangle\label{eq:control-qumode-1}
\end{equation}
 and $\left\{ |p\rangle\right\} _{p\in\mathbb{R}}$ denotes the momentum
quadrature basis. Prepare a data register in the state $\rho$.
\item Apply the Hamiltonian evolution $e^{i\hat{x}\otimes A}$ on the control
and data registers, where $\hat{x}$ is the position quadrature operator.
\item Measure the control register in the momentum quadrature basis $\left\{ |p\rangle\right\} _{p\in\mathbb{R}}$,
obtaining outcome $p\in\mathbb{R}$.
\item Output $r(p)$.
\end{enumerate}
\end{lyxalgorithm}

Figure~\ref{fig:q-conv-alg-circuit} depicts the quantum circuit used
in Algorithm~\ref{alg:convolution-on-A}.
\begin{thm}
\label{thm:gen-conv-alg}The expected value of the random variable
output by Algorithm~\ref{alg:convolution-on-A} is $\Tr\!\left[s(A)\rho\right]$.
\end{thm}

\begin{proof}
See Appendix~\ref{app:proof-q-conv-alg}.
\end{proof}
\begin{figure}
\begin{centering}
\includegraphics[width=\linewidth]{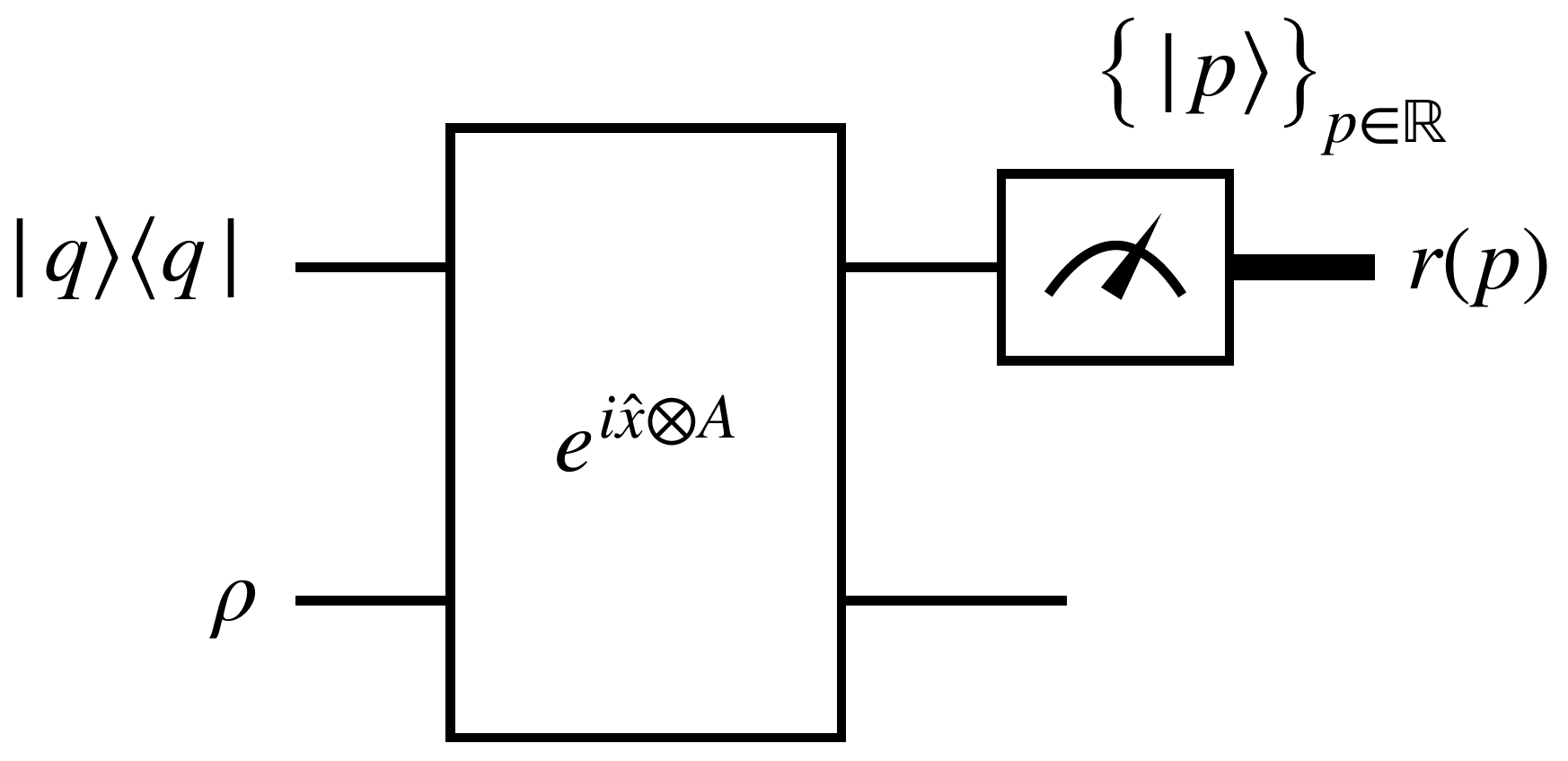}
\par\end{centering}
\caption{Quantum circuit used in the quantum convolution algorithm (Algorithm
\ref{alg:convolution-on-A}). The state $|q\rangle\!\langle q|$ of
the control qumode is defined in~\eqref{eq:control-qumode-1}, and
$\rho$ is the input state on which we would like to perform the simulation.
The measurement of the control qumode at the end is a momentum-quadrature
measurement giving an outcome $p\in\mathbb{R}$. The algorithm finally
outputs $r(p)$, and Theorem~\ref{thm:gen-conv-alg} guarantees that
the expected value of the output is equal to $\Tr\!\left[s(A)\rho\right]$,
where $s=q*r$. }\label{fig:q-conv-alg-circuit}
\end{figure}

We now show how to apply Algorithm~\ref{alg:convolution-on-A} to
simulate the firing of a Fermi--Dirac neuron. It follows from~\eqref{eq:fd-function}
and~\eqref{eq:tanh-to-fd} that the hyperbolic tangent function $g_{T}$
is related to the Fermi--Dirac function $f_{T/2}$ as follows:
\begin{align}
g_{T}(x) & =f_{T/2}(x)-\left(1-f_{T/2}(x)\right),
\end{align}
which implies that
\begin{multline}
\Tr\!\left[g_{T}(H(\theta))\rho\right]=\\
\Tr\!\left[f_{T/2}(H(\theta))\rho\right]-\Tr\!\left[\left(I-f_{T/2}(H(\theta))\right)\rho\right].\label{eq:tanh-obs-FD-meas}
\end{multline}
Note that
\begin{equation}
0\leq f_{T/2}(H(\theta))\leq I,
\end{equation}
so that $f_{T/2}(H(\theta))$ is a Fermi--Dirac thermal measurement
operator and
\begin{equation}
\left(f_{T/2}(H(\theta)),I-f_{T/2}(H(\theta))\right)\label{eq:FD-thermal-meas}
\end{equation}
is a Fermi--Dirac thermal measurement~\cite{liu2026fermidirac}.
Thus, by performing the measurement in~\eqref{eq:FD-thermal-meas}
on the state $\rho$, assigning the value $Z=+1$ if the outcome $f_{T/2}(H(\theta))$
occurs, and assigning the outcome $Z=-1$ if the outcome $I-f_{T/2}(H(\theta))$
occurs, we can simulate this Fermi--Dirac neuron using just a single
sample of $\rho$. It follows from~\eqref{eq:tanh-obs-FD-meas} that
\begin{equation}
\mathbb{E}[Z]=\Tr\!\left[g_{T}(H(\theta))\rho\right].\label{eq:expected-val-hyp-tanh}
\end{equation}

We make the following choices in Algorithm~\ref{alg:convolution-on-A}
to realize the outcome in~\eqref{eq:expected-val-hyp-tanh}. Set $T_{1},T_{2}>0$
such that $T/2=T_{1}T_{2}$. Set $q$ to be the logistic probability
density function $\ell_{T_{1}}$, where
\begin{equation}
\ell_{T_{1}}(p)\coloneqq\frac{e^{p/T_{1}}}{T_{1}\left(e^{p/T_{1}}+1\right)^{2}}=\frac{1}{4T_{1}}\sech^{2}\!\left(\frac{p}{2T_{1}}\right).\label{eq:logistic-prob-dens}
\end{equation}
Set $A$ to be $H(\theta)/T_{2}$. Set $r$ to be the following linear
combination of indicator functions:
\begin{equation}
r(p)=\mathbf{1}_{p\geq0}-\mathbf{1}_{p<0}.\label{eq:r-indicator}
\end{equation}
For completeness, we now specialize Algorithm~\ref{alg:convolution-on-A},
so that it follows the choices in~\eqref{eq:logistic-prob-dens}--\eqref{eq:r-indicator}
and becomes an algorithm for simulating the firing of a Fermi--Dirac
neuron using a single sample of $\rho$:
\begin{lyxalgorithm}
\label{alg:FD-thermal-alg}The algorithm for simulating a Fermi--Dirac
neuron described by the activation observable $g_{T}(H(\theta))$,
where $T/2=T_{1}T_{2}$, proceeds as follows:
\begin{enumerate}
\item For $T_{1}>0$, prepare a control qumode register in the state $|\ell_{T_{1}}\rangle\!\langle\ell_{T_{1}}|$,
where
\begin{align}
|\ell_{T_{1}}\rangle & \coloneqq\int_{-\infty}^{\infty}dp\,\sqrt{\ell_{T_{1}}(p)}|p\rangle,\label{eq:control-qumode}
\end{align}
 and $\left\{ |p\rangle\right\} _{p\in\mathbb{R}}$ denotes the momentum
basis. Prepare a data register in the state $\rho$.
\item For $T_{2}>0$, apply the Hamiltonian evolution $e^{i\hat{x}\otimes H(\theta)/T_{2}}$
on the control and data registers, where $\hat{x}$ is the position
quadrature operator.
\item Measure the control register in the momentum basis $\left\{ |p\rangle\right\} _{p\in\mathbb{R}}$,
obtaining outcome $p\in\mathbb{R}$.
\item Output $Z=+1$ if $p\geq0$ and $Z=-1$ if $p<0$.
\end{enumerate}
\end{lyxalgorithm}

\begin{figure}
\begin{centering}
\includegraphics[width=\linewidth]{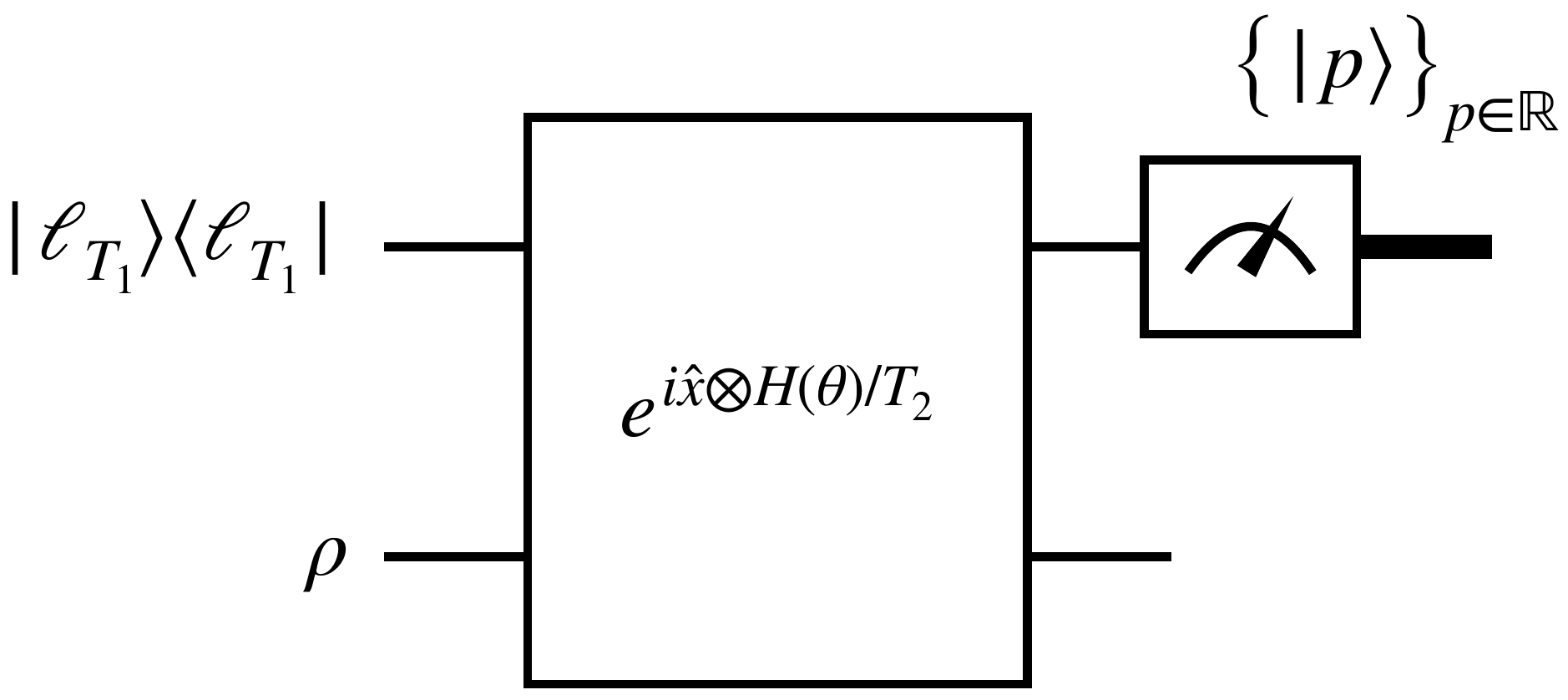}
\par\end{centering}
\caption{Quantum circuit for simulating the firing of a Fermi--Dirac neuron
using a single sample of the state $\rho$. The temperature $T/2=T_{1}T_{2}$,
where $T_{1},T_{2}>0$, as detailed in Algorithm~\ref{alg:FD-thermal-alg}.
The state $|\ell_{T_{1}}\rangle\!\langle\ell_{T_{1}}|$ of the control
qumode is defined in~\eqref{eq:control-qumode}, and $\rho$ is the
input state on which we would like to perform the simulation. The
measurement of the control qumode at the end is a momentum-quadrature
measurement giving an outcome $p\in\mathbb{R}$. The algorithm finally
outputs $+1$ if $p\protect\geq0$ and $-1$ otherwise. }\label{fig:fermi-dirac-circuit}
\end{figure}

\begin{thm}
\label{thm:FD-thermal-alg-proof}The expected value of the random
variable output by Algorithm~\ref{alg:FD-thermal-alg} is $\Tr\!\left[g_{T}(H(\theta))\rho\right]$.
\end{thm}

\begin{proof}
The correctness of Algorithm~\ref{alg:FD-thermal-alg} follows from
Theorem~\ref{thm:gen-conv-alg} and because
\begin{equation}
(\ell_{T_{1}}*r)(p/T_{2})=g_{T}(p),\label{eq:convolution-FD-firing}
\end{equation}
for all $p\in\mathbb{R}$, with $r$ set as in~\eqref{eq:r-indicator}.
We prove~\eqref{eq:convolution-FD-firing} in Appendix~\ref{app:convolution-FD-firing}.
\end{proof}
Let us note that a sharper threshold function $g_{T}(p)$, with lower
temperature $T>0$, can be realized in two different ways, as discussed
previously in \cite[Section IV-A]{liu2026fermidirac}. That is, Algorithm
\ref{alg:FD-thermal-alg} contains two temperature parameters $T_{1}$
and $T_{2}$, which allow for controlling the temperature (and thus
sharpness) of the resulting Fermi--Dirac machine. In order to lower
the temperature $T=T_1T_2$ of the resulting Fermi--Dirac machine, one can choose
to keep $T_{1}$ fixed while lowering $T_{2}$, which amounts to increasing
the time $t=1/T_2$ needed for Hamiltonian evolution while keeping the input control state with fixed $T_1$. This is the preferred approach, unless the coherence time of the controlled gate is smaller than $1/T_2$, such that $T_2$ cannot be further lowered. In cases where the coherence time of the Hamiltonian simulation cannot be increased further, one can keep $T_{2}$ fixed while lowering $T_{1}$ and instead using the input control state
$|\ell_{T_{1}}\rangle$ with a smaller $T_1$. Although $|\ell_{T_{1}}\rangle$ is not quite a squeezed bosonic Gaussian state of the form
\begin{equation}
    |G_{\sigma}\rangle\equiv (1/(2 \pi \sigma^2)^{1/4}) \int_{-\infty}^\infty  e^{-p^2/(4 \sigma^2)}|p\rangle dp,
\end{equation} (with squeezing $1/\sigma$ in the momentum quadrature), it can be very close to such a state. For example, with the choice $\sigma=\pi T_1/\sqrt{3}$, the quantum fidelity $\mathcal{F}(|\ell_{T_{1}}\rangle, |G_{\sigma}\rangle) \approx 0.9942$. Then decreasing $T_1$ means that the state is essentially more highly
squeezed with respect to the momentum quadrature and requires
more energy to prepare.

\subsection{Training Fermi--Dirac machines for binary classification and function
approximation}

\label{subsec:binary-class-func-approx}

A key task for which Fermi--Dirac machines are useful is binary classification,
similar to how perceptrons were originally proposed for this purpose
\cite{Rosenblatt1958Perceptron}. More generally, they can also be
used for function approximation. In what follows, we focus on squared and logistic-loss functions, while noting that other quantum loss functions can be considered \cite{qiu2025,qiu2026}.

\subsubsection{Squared-loss minimization}

\label{subsec:Squared-loss-minimization}

For binary classification, we suppose that training data is available
in the form of labeled quantum states, as
\begin{equation}
\left(\rho_{1},y_{1}\right),\ldots,\left(\rho_{M},y_{M}\right),\label{eq:quantum-training-data}
\end{equation}
where $M\in\mathbb{N}$, $\rho_{m}$ is a quantum state, and $y_{m}\in\left\{ -1,+1\right\} $
for all $m\in\left[M\right]$. For function approximation, we suppose
instead that $y_{m}\in\left[-1,1\right]$ for all $m\in\left[M\right]$.

Either way, for a model parameterized Hamiltonian $H(\theta)$ of
the form in~\eqref{eq:general-k-local-ham}, we can employ a squared
loss function of the following form:
\begin{equation}
\mathcal{L}^{\left(2\right)}(\theta)\coloneqq\frac{1}{M}\sum_{m=1}^{M}\left(\Tr\!\left[g_{T}(H(\theta))\rho_{m}\right]-y_{m}\right)^{2},\label{eq:loss-function}
\end{equation}
with the goal of training being to tune the parameter vector $\theta$
in order to minimize $\mathcal{L}^{\left(2\right)}(\theta)$. In Appendix
\ref{app:Squared-loss-function}, we show how $\mathcal{L}^{\left(2\right)}(\theta)$
can be expressed as follows:
\begin{equation}
\mathcal{L}^{\left(2\right)}(\theta)=\Tr\!\left[\Delta^{\left(2\right)}(\theta,y)\overline{\rho}\right],\label{eq:sq-loss-express-obs}
\end{equation}
in terms of the following squared-loss observable:
\begin{equation}
\Delta^{\left(2\right)}(\theta,y)\coloneqq\sum_{m=1}^{M}|m\rangle\!\langle m|\otimes\left(g_{T}(H(\theta))-y_{m}I\right)^{\otimes2}\label{eq:squared-loss-observable}
\end{equation}
and the labeled classical--quantum training state:
\begin{equation}
\overline{\rho}\coloneqq\frac{1}{M}\sum_{m=1}^{M}|m\rangle\!\langle m|\otimes\rho_{m}^{\otimes2}.\label{eq:cq-training-state}
\end{equation}

By a direct calculation, we find that
\begin{multline}
\frac{\partial}{\partial\theta_{j}}\mathcal{L}^{\left(2\right)}(\theta)=\\
\frac{2}{M}\sum_{m=1}^{M}\left(\Tr\!\left[g_{T}(H(\theta))\rho_{m}\right]-y_{m}\right)\frac{\partial}{\partial\theta_{j}}\Tr\!\left[g_{T}(H(\theta))\rho_{m}\right],
\end{multline}
so that Algorithms~\ref{alg:grad-est} and~\ref{alg:obj-func-est}
can be used for estimating $\frac{\partial}{\partial\theta_{j}}\mathcal{L}^{\left(2\right)}(\theta)$,
by multiplying the corresponding estimates. The resulting estimator
is a biased estimator, but it is consistent, so that it converges
to $\frac{\partial}{\partial\theta_{j}}\mathcal{L}^{\left(2\right)}(\theta)$
if sufficiently many samples are available. If it is desired to have
an unbiased estimator, the following expression for the partial derivative
$\frac{\partial}{\partial\theta_{j}}\mathcal{L}^{\left(2\right)}(\theta)$
leads to an algorithm that lends well to sampling it in an unbiased
way:
\begin{multline}
\frac{\partial}{\partial\theta_{j}}\mathcal{L}^{\left(2\right)}(\theta)=\\
\frac{2}{M}\sum_{m=1}^{M}\Tr\!\left[\left(\kappa(T,\theta,m)\otimes\frac{\partial}{\partial\theta_{j}}g_{T}(H(\theta))\right)\left(\rho_{m}\otimes\rho_{m}\right)\right],\label{eq:grad-squared-loss}
\end{multline}
where
\begin{equation}
\kappa(T,\theta,m)\equiv g_{T}(H(\theta))-y_{m}I.
\end{equation}
Appendix~\ref{app:Squared-loss-function} provides a derivation of
\eqref{eq:grad-squared-loss}, an alternative formula for $\frac{\partial}{\partial\theta_{j}}\mathcal{L}^{\left(2\right)}(\theta)$
based on Theorems~\ref{thm:gradient-obj-func} and~\ref{thm:obj-func-formula},
and a hybrid quantum--classical algorithm for estimating $\frac{\partial}{\partial\theta_{j}}\mathcal{L}^{\left(2\right)}(\theta)$
in an unbiased way.

In order to train Fermi--Dirac machines for squared-loss minimization,
we can employ the standard gradient descent algorithm with a learning
rate $\eta>0$. It updates the parameter vector $\theta$ according
to the following rule for $p\in\mathbb{N}$:
\begin{equation}
\theta^{\left(p+1\right)}\leftarrow\theta^{\left(p\right)}-\eta\left.\nabla_{\theta}\mathcal{L}^{\left(2\right)}(\theta)\right|_{\theta=\theta^{\left(p\right)}},\label{eq:grad-descent}
\end{equation}
and iterates until convergence or a maximum number of iterations is
reached.

\subsubsection{Logistic-loss minimization}

\label{subsec:Logistic-loss-minimization}

It is common to use a logistic-loss function when training a classical
neuron for binary classification (see, e.g.,~\cite{Zhang2024}). Let
us first review this concept in the classical case before moving on
to the quantum case. Suppose that training data is available in the
form of labeled vectors, as $\left(x_{1},y_{1}\right),\ldots,\left(x_{M},y_{M}\right)$,
where $M\in\mathbb{N}$, $x_{m}\in\mathbb{R}^{n}$, and $y_{m}\in\left\{ -1,+1\right\} $
for all $m\in\left[M\right]$. Then the logistic-loss function is
given by
\begin{equation}
\frac{1}{M}\sum_{m=1}^{M}T\ln\!\left(1+e^{-y_{m}\left(w^{T}x_{m}+b\right)/T}\right),
\end{equation}
where $T>0$ controls the sharpness of the loss function and $\theta=(w,b)$,
as in~\eqref{eq:1st-order-neuron}. Note that
\begin{equation}
\lim_{T\to0}T\ln\!\left(1+e^{-z/T}\right)=\max\!\left\{ 0,-z\right\} ,
\end{equation}
and that it is closely related to the fermionic free energy from \cite[Proposition~5]{liu2026fermidirac}.
The idea behind the logistic-loss function is that there is a roughly
linearly increasing penalty when the signs of $y_{m}$ and $w^{T}x_{m}+b$
are mismatched (i.e., a misclassification is occurring), while there
is little to no penalty if the signs match. One trains the model by
minimizing the logistic-loss function by means of gradient descent,
as in~\eqref{eq:grad-descent}. After doing so, one can realize the
function $x\mapsto f_{T}(w^{T}x+b)$ on a test data sample and then
threshold this output to make a binary classification decision. Alternatively,
one can classify by outputting $+1$ with probability $f_{T}\!\left(w^{T}x+b\right)$
and $-1$ with probability $1-f_{T}\!\left(w^{T}x+b\right)$.

Now suppose that training data is available in the form of labeled
probability distributions
\begin{equation}
\left(p_{1}(x),y_{1}\right),\ldots,\left(p_{M}(x),y_{M}\right),
\end{equation}
where $M\in\mathbb{N}$, $x\in\mathbb{R}^{n}$ and $y_{m}\in\left\{ -1,+1\right\} $
for all $m\in\left[M\right]$. The aforementioned deterministic case
is recovered when each $p_{1}, \ldots, p_{M}$ is a degenerate
probability distribution. The goal now is to use this training data
in order to classify probability distributions. In this case, we can
employ the following logistic-loss function:
\begin{equation}
\frac{1}{M}\sum_{m=1}^{M}\mathbb{E}_{x\sim p_{m}}\!\left[T\ln\!\left(1+e^{-y_{m}\left(w^{T}x+b\right)/T}\right)\right].\label{eq:logistic-loss-prob}
\end{equation}
One can proceed with training by using (stochastic) gradient descent
\cite{Bottou2018} on the above loss function, sampling from each
probability distribution $p_{m}(x)$ in order to do so.

This latter case now motivates a generalization to the quantum case.
Suppose that training data is available in the form of labeled quantum
states, as in~\eqref{eq:quantum-training-data}. Then, for a parameterized
Hamiltonian $H(\theta)$, the logistic-loss function for this case
is as follows:
\begin{equation}
\mathcal{L}_{T}^{\log}(\theta)\coloneqq\frac{1}{M}\sum_{m=1}^{M}T\Tr\!\left[\ln\!\left(I+e^{-y_{m}H(\theta)/T}\right)\rho_{m}\right].\label{eq:log-loss-quantum}
\end{equation}
To have a sense of this loss function, observe that we have followed
the canonical quantization procedure mentioned around~\eqref{eq:ham-q-class},
quantizing the energy function $w^{T}x+b$ in~\eqref{eq:logistic-loss-prob}
by replacing it with a general parameterized Hamiltonian $H(\theta)$
and employing the following logistic-loss observable:
\begin{equation}
T\ln\!\left(I+e^{-y_{m}H(\theta)/T}\right).
\end{equation}
 Furthermore, consider that the Hamiltonian $H(\theta)$ can be written
in terms of its spectral decomposition as
\begin{equation}
H(\theta)=\sum_{k}E_{k}^{\theta}\Pi_{k}^{\theta},
\end{equation}
where $E_{k}^{\theta}$ is an eigenvalue and $\Pi_{k}^{\theta}$ is
a spectral projection. Applying the functional calculus, the following
equality holds:
\begin{equation}
\mathcal{L}_{T}^{\log}(\theta)=\frac{1}{M}\sum_{m=1}^{M}\sum_{k}q_{m}^{\theta}(k)\,T\ln\!\left(1+e^{-y_{m}E_{k}^{\theta}/T}\right),
\end{equation}
where $q_{m}^{\theta}(k)\coloneqq\Tr\!\left[\Pi_{k}^{\theta}\rho_{m}\right]$
is a probability distribution. Thus, the classical logistic-loss function
in~\eqref{eq:logistic-loss-prob} is nearly recovered when doing so,
with a key difference being that the probability distribution $q_{m}^{\theta}(k)$
has a nontrivial dependence on $\theta$.

Similar to what we did for the squared-loss function in~\eqref{eq:sq-loss-express-obs},
we can also rewrite the logistic-loss function $\mathcal{L}_{T}^{\log}(\theta)$
as follows:
\begin{equation}
\mathcal{L}_{T}^{\log}(\theta)=\Tr\!\left[\Delta^{\log}(\theta,y)\overline{\rho}\right],
\end{equation}
in terms of the following labeled logistic-loss observable
\begin{equation}
\Delta^{\log}(\theta,y)\coloneqq\sum_{m=1}^{M}|m\rangle\!\langle m|\otimes T\ln\!\left(I+e^{-y_{m}H(\theta)/T}\right),
\end{equation}
and the labeled classical--quantum training state:
\begin{equation}
\overline{\rho}\coloneqq\frac{1}{M}\sum_{m=1}^{M}|m\rangle\!\langle m|\otimes\rho_{m}.
\end{equation}

Theorem~\ref{thm:grad-logloss} below states an expression for the
partial derivative $\frac{\partial}{\partial\theta_{j}}\mathcal{L}_{T}^{\log}(\theta)$,
which lends well to sampling it in an unbiased way.
\begin{thm}
\label{thm:grad-logloss}The following equality holds:
\begin{multline}
\frac{\partial}{\partial\theta_{j}}\left(T\Tr\!\left[\ln\!\left(I+e^{-y_{m}H(\theta)/T}\right)\rho_{m}\right]\right)\\
=-\frac{y_{m}}{2}\Tr\!\left[H_{j}\rho_{m}\right]+\\
\frac{1}{2T}\left\Vert \theta\right\Vert _{1}\mathbb{E}_{\substack{s\sim\upsilon,\\
k\sim q,\\
t\sim\gamma
}
}\!\left[s\Re\!\left[\Tr\!\left[\signum(\theta_{k})H_{k}H_{j}e^{iy_{m}H(\theta)t/T}\rho_{m}'\right]\right]\right],\label{eq:partial-deriv-logloss}
\end{multline}
where
\begin{equation}
\rho_{m}'\equiv\mathcal{U}_{st/T}^{y_{m}H(\theta)}\!\left(\rho_{m}\right),
\end{equation}
$y_{m}\in\left\{ -1,+1\right\} $, $H(\theta)$ is a parameterized
Hamiltonian of the form in~\eqref{eq:general-k-local-ham}, $T>0$,
$\gamma(t)$ is the following high-peak tent probability density function~\cite{Patel2025a}:
\begin{equation}
\gamma(t)\coloneqq\frac{2}{\pi}\ln\!\left|\coth\!\left(\frac{\pi t}{2}\right)\right|,\label{eq:high-peak-tent-gamma-main-txt}
\end{equation}
$\upsilon$ is a uniform random variable on the unit interval $\left[0,1\right]$,
$q(k)\coloneqq\frac{\left|\theta_{k}\right|}{\left\Vert \theta\right\Vert _{1}}$
is a probability distribution, and $\mathcal{U}_{t}^{y_{m}H(\theta)}$
is the following unitary quantum channel:
\begin{equation}
\mathcal{U}_{t}^{y_{m}H(\theta)}(\cdot)\coloneqq e^{-iy_{m}H(\theta)t}(\cdot)e^{iy_{m}H(\theta)t}.
\end{equation}
\end{thm}

Appendix~\ref{app:Logistic-loss-function} provides a proof of Theorem
\ref{thm:grad-logloss} and develops Algorithm~\ref{alg:grad-log-loss},
a hybrid quantum--classical algorithm for estimating the partial
derivative $\frac{\partial}{\partial\theta_{j}}\mathcal{L}_{T}^{\log}(\theta)$.
In order to train this model, we can employ the standard gradient
descent algorithm with a learning rate $\eta>0$, as in~\eqref{eq:grad-descent},
but using the gradient of $\mathcal{L}_{T}^{\log}(\theta)$ instead
of that of $\mathcal{L}^{\left(2\right)}(\theta)$.

\subsubsection{Fermi--Dirac machines as collective measurements}

As observed in~\cite{liu2026fermidirac}, we remark here that Fermi--Dirac
machines, in general, realize collective measurements (also called
joint measurements), meaning that they cannot be realized by performing
individual measurements on each qubit, followed by post-processing,
or by adaptive local measurements (see, e.g.,~\cite{Higgins2011}).
It is well known in quantum information theory that one can have a
higher performance in identifying classical messages encoded into
non-orthogonal quantum states by performing collective measurements,
with examples occurring in quantum optical communication~\cite{Giovannetti2004}
and qubit communication~\cite{Fuchs1997,King2002}.

Thus, this is the kind of scenario for which we expect Fermi--Dirac
machines to offer an improvement in binary classification. That is,
when the training data consists of non-classical states (i.e., non-orthogonal
or entangled), these prior results from quantum information theory
suggest that Fermi--Dirac machines should outperform classical neurons
for this task, given that the measurements realized by the latter
are not collective. We indeed find that Fermi--Dirac machines outperform
classical neurons in the small-scale numerical simulations reported
in Section~\ref{sec:Numerical-experiments}, while we leave it open
to establish theoretical evidence supporting this claim.

\section{Quantizing the rectified linear unit}

\label{sec:Quantizing-smooth-ReLU}

In this section, we quantize the rectified linear unit (ReLU) in two
different ways, leading to the smooth ReLU and sigmoid linear unit
(SiLU) activation observables.

ReLU is an activation function used in deep neural networks~\cite{Jarrett2009,nair2010rectified}.
Its adoption has been shown to improve optimization and alleviate
the vanishing-gradient problem that can arise during training~\cite{glorot2011deep}.
The ReLU activation function is defined as
\begin{equation}
\relu(y)\coloneqq y_{+}\coloneqq\max\left\{ 0,y\right\} =\begin{cases}
y & :y>0\\
0 & :y\leq0
\end{cases},\label{eq:relu-def}
\end{equation}
providing a simple piecewise-linear nonlinearity. Although ReLU is
not differentiable at $y=0$, several smooth approximations have been
introduced and are often found to improve gradient flow and empirical
training performance in deep architectures~\cite{glorot2011deep,Elfwing2018,Ramachandran2017}.

\subsection{Quantizing the smooth ReLU (softplus) activation function}

\label{subsec:smooth-ReLU}

One smooth approximation to ReLU is the smooth ReLU (or softplus)
activation function~\cite{Dugas2000,glorot2011deep}:
\begin{equation}
r_{T}(y)\coloneqq T\ln\!\left(1+e^{y/T}\right),\label{eq:softplus-func}
\end{equation}
where $T>0$ is a temperature parameter controlling the smoothness
of the approximation. In the zero-temperature limit,
\begin{equation}
\lim_{T\to0}r_{T}(y)=\relu(y)\quad\forall y\in\mathbb{R},
\end{equation}
so that $r_{T}$ interpolates smoothly between a differentiable activation
and the piecewise-linear ReLU function. Moreover, $r_{T}$ is closely
related to the logistic-loss function already analyzed in Section
\ref{subsec:Logistic-loss-minimization}. Consequently, it is also
closely related to the fermionic free energy objective function studied
in \cite[Proposition~5]{liu2026fermidirac}.

Given a parameterized Hamiltonian $H(\theta)$, we can quantize~\eqref{eq:relu-def}
and~\eqref{eq:softplus-func} by defining the following smooth ReLU
observable:
\begin{equation}
r_{T}(H(\theta)).
\end{equation}
Its expectation with respect to an arbitrary state $\rho$ is $\Tr\!\left[r_{T}(H(\theta))\rho\right]$.
In order to employ ReLU observables in quantum machine learning, we
should devise methods for estimating the objective function $\Tr\!\left[r_{T}(H(\theta))\rho\right]$
and its gradient $\nabla_{\theta}\Tr\!\left[r_{T}(H(\theta))\rho\right]$.

We begin with the following theorem, which states a formula for the
$j$th partial derivative of $\Tr\!\left[r_{T}(H(\theta))\rho\right]$. 
\begin{thm}
\label{thm:grad-softplus}The following equality holds:
\begin{multline}
\frac{\partial}{\partial\theta_{j}}\Tr\!\left[r_{T}(H(\theta))\rho\right]=\frac{1}{2}\Tr\!\left[H_{j}\rho\right]+\\
\frac{\left\Vert \theta\right\Vert _{1}}{2T}\mathbb{E}_{\substack{s\sim\upsilon,\\
k\sim q,\\
t\sim\gamma
}
}\!\left[s\Re\!\left[\Tr\!\left[\signum(\theta_{k})H_{k}H_{j}e^{iH(\theta)t/T}\mathcal{U}_{st/T}^{H(\theta)}\!\left(\rho\right)\right]\right]\right].\label{eq:partial-deriv-softplus}
\end{multline}
where $H(\theta)$ is a parameterized Hamiltonian of the form in~\eqref{eq:general-k-local-ham},
$T>0$, $\gamma(t)$ is the probability density function defined in
\eqref{eq:high-peak-tent-gamma-main-txt}, $\upsilon$ is a uniform
random variable on the unit interval $\left[0,1\right]$, $q(k)\coloneqq\frac{\left|\theta_{k}\right|}{\left\Vert \theta\right\Vert _{1}}$
is a probability distribution, and $\mathcal{U}_{t}^{H(\theta)}$
is the following unitary quantum channel:
\begin{equation}
\mathcal{U}_{t}^{H(\theta)}(\cdot)\coloneqq e^{-iH(\theta)t}(\cdot)e^{iH(\theta)t}.
\end{equation}
\end{thm}

\begin{proof}
See Appendix~\ref{app:softplus-grad}.
\end{proof}
The formula in~\eqref{eq:partial-deriv-softplus} is amenable to efficient
estimation on a quantum computer. In Appendix~\ref{app:HQCs-softplus},
we detail hybrid quantum--classical algorithms for estimating the
gradient $\nabla_{\theta}\Tr\!\left[r_{T}(H(\theta))\rho\right]$
and the objective function $\Tr\!\left[r_{T}(H(\theta))\rho\right]$.

We can also realize the expected value of the smooth ReLU observable,
using just one sample of the state $\rho$, by means of the following
special case of Algorithm~\ref{alg:convolution-on-A}:
\begin{lyxalgorithm}
\label{alg:ReLU-one-shot}The algorithm for realizing the smooth ReLU
neuron, described by the activation observable $r_{T}(H(\theta))$,
where $T=T_{1}T_{2}$, proceeds as follows:
\begin{enumerate}
\item For $T_{1}>0$, prepare a control qumode register in the state $|\ell_{T_{1}}\rangle\!\langle\ell_{T_{1}}|$,
defined in~\eqref{eq:control-qumode}. Prepare a data register in
the state $\rho$.
\item For $T_{2}>0$, apply the Hamiltonian evolution $e^{i\hat{x}\otimes H(\theta)/T_{2}}$
on the control and data registers, where $\hat{x}$ is the position
quadrature operator.
\item Measure the control register in the momentum basis $\left\{ |p\rangle\right\} _{p\in\mathbb{R}}$,
obtaining outcome $p\in\mathbb{R}$.
\item Output $T_{2}\relu(p)=T_{2}\max\left\{ p,0\right\} $.
\end{enumerate}
\end{lyxalgorithm}

\begin{thm}
\label{thm:implementing-ReLU}The expected value of the random variable
output by Algorithm~\ref{alg:ReLU-one-shot} is $\Tr\!\left[r_{T}(H(\theta))\rho\right]$.
\end{thm}

\begin{proof}
As a consequence of Theorem~\ref{thm:gen-conv-alg}, it suffices to
prove that
\begin{equation}
r_{T}(p)=T_{2}\left(\relu*\ell_{T_{1}}\right)\left(\frac{p}{T_{2}}\right),\label{eq:smooth-relu-conv}
\end{equation}
for all $p\in\mathbb{R}$. See Appendix~\ref{app:ReLU-one-shot-proof}
for a proof of~\eqref{eq:smooth-relu-conv}.
\end{proof}

\subsection{Quantizing the sigmoid linear unit (swish) activation function}

\label{subsec:SiLU}

Another smooth approximation to ReLU is the sigmoid linear unit (SiLU),
also known as the swish activation function. It is defined as~\cite{Elfwing2018,Ramachandran2017}
\begin{equation}
\silu_{T}(x)\coloneqq\frac{x}{1+e^{-x/T}}=xf_{T}(x),\label{eq:silu-obj}
\end{equation}
where the Fermi--Dirac function $f_{T}(x)$ is defined in~\eqref{eq:fd-function}.
Unlike ReLU, the SiLU activation is smooth and mildly non-monotonic,
combining a linear response for large positive inputs with sigmoid-based
gating near the origin. These properties have been observed to improve
gradient flow and empirical performance in deep neural networks relative
to ReLU and related piecewise-linear activations~\cite{Elfwing2018,Ramachandran2017}.

Given a parameterized Hamiltonian $H(\theta)$, we can quantize~\eqref{eq:relu-def}
and~\eqref{eq:silu-obj} by defining the following SiLU observable:
\begin{equation}
\silu_{T}(H(\theta)).
\end{equation}
Its expectation with respect to an arbitrary state $\rho$ is $\Tr\!\left[\silu_{T}(H(\theta))\rho\right]$.
In order to employ SiLU observables in quantum machine learning, we
should devise methods for estimating the objective function $\Tr\!\left[\silu_{T}(H(\theta))\rho\right]$
and its gradient $\nabla_{\theta}\Tr\!\left[\silu_{T}(H(\theta))\rho\right]$.

We begin by establishing a formula for the $j$th partial derivative
of $\Tr\!\left[\silu_{T}(H(\theta))\rho\right]$, which is amenable
to efficient estimation on a quantum computer:
\begin{thm}
\label{thm:grad-silu}The following equality holds:
\begin{multline}
\frac{\partial}{\partial\theta_{j}}\Tr\!\left[\silu_{T}(H(\theta))\rho\right]=\frac{1}{2}\Tr\!\left[H_{j}\rho\right]+\\
\frac{\left\Vert \theta\right\Vert _{1}}{2T}\mathbb{E}_{\substack{t\sim\xi,\\
s\sim\upsilon,\\
k\sim q
}
}\left[s\,\Re\left[\Tr\!\left[\signum(\theta_{k})H_{k}H_{j}e^{iH(\theta)t/\left(2T\right)}\rho'\right]\right]\right],\label{eq:partial-deriv-silu}
\end{multline}
where
\begin{equation}
\rho'\equiv\mathcal{U}_{st/\left(2T\right)}^{H(\theta)}\left(\rho\right),
\end{equation}
the objective function $\Tr\!\left[\silu_{T}(H(\theta))\rho\right]$
is defined from~\eqref{eq:silu-obj}, $\xi(t)$ is the following probability
density function
\begin{equation}
\xi(t)\coloneqq\frac{\gamma(t)+\mu(t)}{2},
\end{equation}
$\gamma(t)$ is defined in~\eqref{eq:high-peak-tent-gamma-main-txt},
$\mu(t)$ is defined in~\eqref{eq:high-peak-tent-2-def}, $\upsilon$
is a uniform random variable on the unit interval $\left[0,1\right]$,
$q$ is the following probability distribution:
\begin{equation}
q(k)\coloneqq\frac{\left|\theta_{k}\right|}{\left\Vert \theta\right\Vert _{1}},
\end{equation}
and $\mathcal{U}_{t}^{H(\theta)}$ is the following unitary quantum
channel:
\begin{equation}
\mathcal{U}_{t}^{H(\theta)}(\cdot)\coloneqq e^{-iH(\theta)t}(\cdot)e^{iH(\theta)t}.
\end{equation}
\end{thm}

\begin{proof}
See Appendix~\ref{app:grad-silu}.
\end{proof}
The formula in~\eqref{eq:partial-deriv-silu} is amenable to efficient
estimation on a quantum computer. In Appendix~\ref{app:HQCs-SiLU},
we detail hybrid quantum--classical algorithms for estimating the
gradient $\nabla_{\theta}\Tr\!\left[\silu_{T}(H(\theta))\rho\right]$
and the objective function $\Tr\!\left[\silu_{T}(H(\theta))\rho\right]$.

In order to evaluate $\silu_{T}$ using a single sample of the state
$\rho$, we propose the following generalization of Algorithm~\ref{alg:convolution-on-A}:
\begin{lyxalgorithm}[Quantum convolution and multiplication algorithm]
\label{alg:convolution-and-mult-on-A}Let $q_{1},q_{2}\colon\mathbb{R}\to\left[0,1\right]$
be even probability density functions, let $r_{1},r_{2}\colon\mathbb{R}\to\mathbb{R}$
be measurable functions such that the convolutions $s_{1}=q_{1}*r_{1}$
and $s_{2}=q_{2}*r_{2}$ are well defined, let $A$ be a Hamiltonian,
and let $t_{1},t_{2}\geq0$. The algorithm for outputting a random
variable with expected value $\Tr\!\left[s_{1}(At_{1})s_{2}(At_{2})\rho\right]$
proceeds as follows:
\begin{enumerate}
\item Prepare two control qumode registers in the state $|q_{1}\rangle\!\langle q_{1}|\otimes|q_{2}\rangle\!\langle q_{2}|$,
where
\begin{equation}
|q_{i}\rangle\coloneqq\int_{-\infty}^{\infty}dp\,\sqrt{q_{i}(p)}|p\rangle,\label{eq:control-qumode-1-1}
\end{equation}
for $i\in\left\{ 1,2\right\} $, and $\left\{ |p\rangle\right\} _{p\in\mathbb{R}}$
denotes the momentum quadrature basis. Prepare a data register in
the state $\rho$.
\item Apply the Hamiltonian evolution $e^{i\hat{x}\otimes At_{1}}$ on the
first control register and the data register, where $\hat{x}$ is
the position quadrature operator, and apply the Hamiltonian evolution
$e^{i\hat{x}\otimes At_{2}}$ on the second control register and the
data register.
\item Measure the first and second control register in the momentum quadrature
basis $\left\{ |p\rangle\right\} _{p\in\mathbb{R}}$, obtaining outcomes
$p_{1}\in\mathbb{R}$ and $p_{2}\in\mathbb{R}$.
\item Output $r_{1}(p_{1})r_{2}(p_{2})$.
\end{enumerate}
\end{lyxalgorithm}

See Figure~\ref{fig:q-conv-and-mult-alg-circuit} for a depiction
of Algorithm~\ref{alg:convolution-and-mult-on-A}.

\begin{figure}
\begin{centering}
\includegraphics[width=\linewidth]{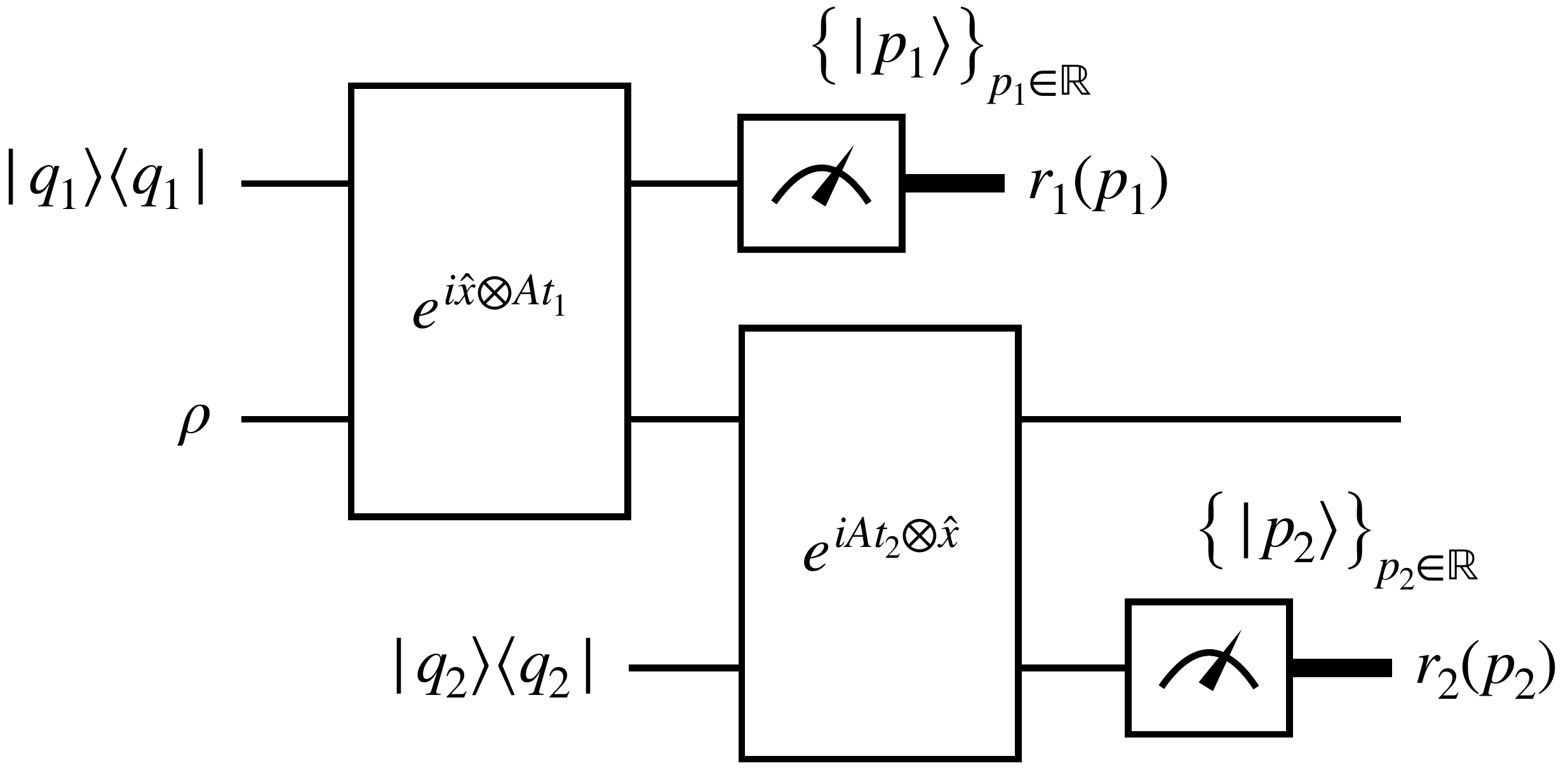}
\par\end{centering}
\caption{Quantum circuit used in the quantum convolution and multiplication
algorithm (Algorithm~\ref{alg:convolution-and-mult-on-A}). The states
$|q_{1}\rangle\!\langle q_{1}|$ and $|q_{2}\rangle\!\langle q_{2}|$
of the control qumodes are defined in~\eqref{eq:control-qumode-1-1},
and $\rho$ is the input state on which we would like to perform the
simulation. The measurements of the control qumodes at the end are
momentum-quadrature measurements, giving outcomes $p_{1},p_{2}\in\mathbb{R}$.
The algorithm finally outputs $r_{1}(p_{1})r_{2}(p_{2})$, and Theorem
\ref{thm:gen-conv-and-mult-alg} guarantees that the expected value
of the output is equal to $\Tr\!\left[s_{1}(At_{1})s_{2}(At_{2})\rho\right]$,
where $s_{i}=q_{i}*r_{i}$, for $i\in\left\{ 1,2\right\} $. }\label{fig:q-conv-and-mult-alg-circuit}
\end{figure}

\begin{thm}
\label{thm:gen-conv-and-mult-alg}The expected value of the random
variable output by Algorithm~\ref{alg:convolution-and-mult-on-A}
is $\Tr\!\left[s_{1}(At_{1})s_{2}(At_{2})\rho\right]$.
\end{thm}

\begin{proof}
See Appendix~\ref{app:proof-q-conv-and-mult-alg}.
\end{proof}
\begin{rem}
Algorithm~\ref{alg:convolution-and-mult-on-A} and Theorem~\ref{thm:gen-conv-and-mult-alg}
naturally generalize to the case of realizing the expected value $\Tr\!\left[\prod_{i=1}^{m}s_{i}(At_{i})\rho\right]$,
where, for all $i\in\left[m\right]$, $q_{i}\colon\mathbb{R}\to\left[0,1\right]$
is an even probability density function, $r_{i}\colon\mathbb{R}\to\mathbb{R}$
is a measurable function such that the convolution $s_{i}=q_{i}*r_{i}$
is well defined, and $t_{i}\geq0$.
\end{rem}

We can make the following choices in Algorithm~\ref{alg:convolution-and-mult-on-A}
to realize the expected value $\Tr\!\left[\silu_{T}(H(\theta))\rho\right]$.
Set $T_{1},T_{2}>0$ such that $T=T_{1}T_{2}$. Set
\begin{align}
q_{1}(p) & =\frac{e^{-p^{2}}}{\sqrt{\pi}},\\
q_{2} & =\ell_{T_{1}},\\
r_{1}(p_{1}) & =p_{1},\\
r_{2}(p_{2}) & =\mathbf{1}_{p_{2}\geq0}.
\end{align}
Set $At_{1}=H(\theta)$ and $At_{2}=H(\theta)/T_{2}$. For completeness,
we now specialize Algorithm~\ref{alg:convolution-and-mult-on-A},
so that it becomes the following algorithm for realizing the SiLU
neuron using a single sample of $\rho$:
\begin{lyxalgorithm}
\label{alg:SiLU-one-shot}The algorithm for realizing the SiLU neuron,
described by the activation observable $\silu_{T}(H(\theta))$, where
$T=T_{1}T_{2}$, proceeds as follows:
\begin{enumerate}
\item For $T_{1}>0$, prepare two control qumode registers in the state
$|0\rangle\!\langle0|\otimes|\ell_{T_{1}}\rangle\!\langle\ell_{T_{1}}|$,
where $|0\rangle\!\langle0|$ is the vacuum state and $|\ell_{T_{1}}\rangle\!\langle\ell_{T_{1}}|$
is defined in~\eqref{eq:logistic-prob-dens} and~\eqref{eq:control-qumode}.
Prepare a data register in the state $\rho$.
\item For $T_{2}>0$, apply the Hamiltonian evolution $e^{i\hat{x}\otimes H(\theta)}$
on the first control qumode register and the data register, and then
apply the Hamiltonian evolution $e^{i\hat{x}\otimes H(\theta)/T_{2}}$
on the second control qumode register and the data register.
\item Measure the control registers in the momentum basis $\left\{ |p\rangle\right\} _{p\in\mathbb{R}}$,
obtaining outcomes $p_{1},p_{2}\in\mathbb{R}$.
\item Output $p_{1}\mathbf{1}_{p_{2}\geq0}$.
\end{enumerate}
\end{lyxalgorithm}

\begin{thm}
\label{thm:silu-one-shot-alg}The expected value of the random variable
output by Algorithm~\ref{alg:SiLU-one-shot} is $\Tr\!\left[\silu_{T}(H(\theta))\rho\right]$.
\end{thm}

\begin{proof}
See Appendix~\ref{app:SiLU-one-shot-proof}.
\end{proof}

\section{Quantizing Gaussian activation functions}

\label{sec:Quantizing-Gaussian-activation}

In this section, we show how to quantize Gaussian activation functions,
with the main idea being to replace the logistic probability density
$\ell_{T}$ with the Gaussian probability density $\phi_{T}$ of mean
zero and standard deviation $T>0$, defined as 
\begin{equation}
\phi_{T}(x)\coloneqq\frac{1}{\sqrt{2\pi}T}\exp\!\left(-\frac{1}{2}\left(\frac{x}{T}\right)^{2}\right).\label{eq:normal-dens-mean-0-stddev-T}
\end{equation}
From a
practical perspective, it is sensible to do so, given the ease with
which experimentalists can prepare bosonic Gaussian states in the
laboratory~\cite{Andersen2010}.

We consider Gaussian variants of Fermi--Dirac (Section~\ref{sec:Fermi=002013Dirac-machines-quants-neurs}),
smooth ReLU (Section~\ref{subsec:smooth-ReLU}), and SiLU (Section
\ref{subsec:SiLU}) neurons. See Figure~\ref{fig:concept-gaussian} for a conceptual depiction of this section's contribution. The result
is an extension of Theorems~\ref{thm:FD-thermal-alg-proof},~\ref{thm:implementing-ReLU},
and~\ref{thm:silu-one-shot-alg} and Algorithms~\ref{alg:FD-thermal-alg},
\ref{alg:ReLU-one-shot}, and~\ref{alg:SiLU-one-shot} to the Gaussian case. Throughout
this section, we employ the standard Gaussian cumulative distribution,
probability density function, and error function:
\begin{align}
\Phi\!\left(x\right) & \coloneqq\frac{1}{\sqrt{2\pi}}\int_{-\infty}^{x}dy\,e^{-\frac{y^{2}}{2}},\\
\phi(x) & \coloneqq\phi_{T=1}(x)=\frac{1}{\sqrt{2\pi}}\exp\!\left(-\frac{x^{2}}{2}\right),\\
\erf(x) & \coloneqq\frac{2}{\sqrt{\pi}}\int_{0}^{x}dy\,e^{-y^{2}}=2\Phi(\sqrt{2}x)-1.
\end{align}
\begin{figure*}
\begin{centering}
\includegraphics[width=0.75\textwidth]{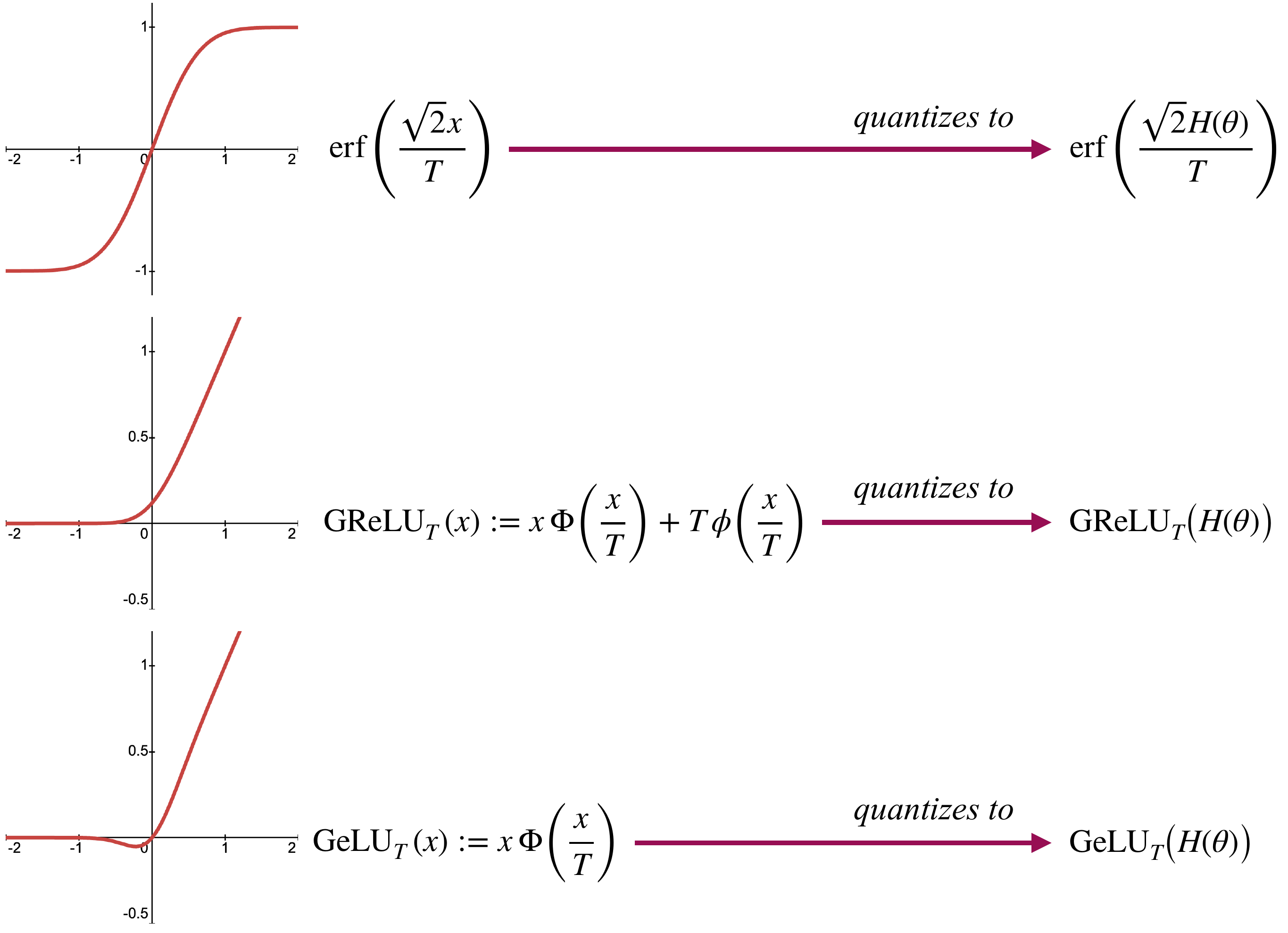}
\par\end{centering}
\caption{This figure is similar to Figure~\ref{fig:concept}, but instead showcasing
the Gaussian activation functions that we quantize in Section~\ref{sec:Quantizing-Gaussian-activation},
which include the Gaussian error function ($\erf$), the Gaussian-smoothed
rectified linear unit (GReLU), and the Gaussian error linear unit
(GeLU). }\label{fig:concept-gaussian}
\end{figure*}

\subsection{Gaussian error function (erf) activation observable}

Let us begin by ``gaussifying'' Fermi--Dirac machines and the related
Theorem~\ref{thm:FD-thermal-alg-proof} and Algorithm~\ref{alg:FD-thermal-alg}.
There, the key relation is the convolution identity in~\eqref{eq:convolution-FD-firing}.
If we replace $\ell_{T_{1}}$ therein with $\phi_{T_{1}}$, we instead
find that
\begin{equation}
(\phi_{T_{1}}*r)(p/T_{2})=2\Phi\!\left(\frac{2p}{T}\right)-1\eqqcolon\erf\!\left(\frac{\sqrt{2}p}{T}\right),\label{eq:convolution-Gaussian-with-indicators}
\end{equation}
for all $p\in\mathbb{R}$, with $r$ set as in~\eqref{eq:r-indicator}
and $T=T_{1}T_{2}$. See Appendix~\ref{app:erf-func-conv} for a proof
of~\eqref{eq:convolution-Gaussian-with-indicators}. Thus, the result
is that the gaussified version of Algorithm~\ref{alg:FD-thermal-alg}
outputs a random variable with expected value
\begin{equation}
F_{\erf}(\theta)\equiv\Tr\!\left[\erf\!\left(\sqrt{2}H(\theta)/T\right)\rho\right].
\end{equation}
The error function $\erf$ can be viewed as a Gaussian approximation
of the hyperbolic tangent function. For training such a neuron, we
need the partial derivative $\frac{\partial}{\partial\theta_{j}}F_{\erf}(\theta)$,
and employing the same proof as for Theorem~\ref{thm:gradient-obj-func},
we find that $\mu(t)$ is replaced by a Gaussian probability density.
Thus, Theorem~\ref{thm:gradient-obj-func} extends to the Gaussian
case, as well as Algorithm~\ref{alg:grad-est} for estimating the
partial derivative $\frac{\partial}{\partial\theta_{j}}F_{\erf}(\theta)$.
See Appendix~\ref{app:erf-grad-deriv} for details.

\subsection{Gaussian-smoothed ReLU activation observable}

Next we gaussify the smooth ReLU function, which we refer to as the
Gaussian-smoothed ReLU (GReLU):
\begin{equation}
\grelu_{T}(x)\coloneqq x\Phi\!\left(\frac{x}{T}\right)+T\phi\!\left(\frac{x}{T}\right).\label{eq:grelu-def}
\end{equation}
Note that the following identities hold:
\begin{align}
\grelu_{T}(x) & =\mathbb{E}\left[\relu\!\left(x+TZ\right)\right]\\
 & =T_{2}\left(\relu*\phi_{T_{1}}\right)\left(\frac{x}{T_{2}}\right),\label{eq:convolve-relu-varphi-Gauss}
\end{align}
where $Z\sim\mathcal{N}(0,1)$ and $T=T_{1}T_{2}$. Thus, GReLU arises
by replacing the logistic function $\ell_{T_{1}}$ in~\eqref{eq:smooth-relu-conv}
with $\phi_{T_{1}}$ in~\eqref{eq:normal-dens-mean-0-stddev-T}. See
Appendix~\ref{app:relu-normal-conv} for a proof of~\eqref{eq:convolve-relu-varphi-Gauss}.
The result is that the gaussified version of Algorithm~\ref{alg:ReLU-one-shot}
outputs a random variable with expected value
\begin{equation}
F_{\grelu_{T}}(\theta)\equiv\Tr\!\left[\grelu_{T}\!\left(H(\theta)\right)\rho\right].
\end{equation}
For training a GReLU neuron, we need the partial derivative $\frac{\partial}{\partial\theta_{j}}F_{\grelu_{T}}(\theta)$,
for which an expression is derived in Appendix~\ref{app:deriv-GReLU},
making use of the following equality:
\begin{equation}
\frac{\partial}{\partial x}\grelu_{T}(x)=\Phi\!\left(\frac{x}{T}\right).\label{eq:deriv-grelu}
\end{equation}
Thus, Theorem~\ref{thm:grad-softplus} extends to the Gaussian case,
as well as Algorithm~\ref{alg:grad-smooth-relu} for estimating the
partial derivative $\frac{\partial}{\partial\theta_{j}}F_{\grelu_{T}}(\theta)$.

\subsection{Gaussian error linear unit activation observable}

Finally, we gaussify the sigmoid linear unit (SiLU) function, which
becomes the Gaussian error linear unit (GeLU):
\begin{equation}
\gelu_{T}(x)\coloneqq x\Phi\!\left(\frac{x}{T}\right).\label{eq:gelu-obj}
\end{equation}
The gaussified version of Algorithm~\ref{alg:SiLU-one-shot} replaces
$\ell_{T_{1}}$ with $\phi_{T_{1}}$ in the state of the control qumode.
The resulting algorithm then outputs a random variable with expected
value
\begin{equation}
F_{\gelu_{T}}(\theta)\equiv\Tr\!\left[\gelu_{T}\!\left(H(\theta)\right)\rho\right].
\end{equation}
For training such a neuron, we need the partial derivative $\frac{\partial}{\partial\theta_{j}}F_{\gelu_{T}}(\theta)$,
for which an expression is derived in Appendix~\ref{app:deriv-GeLU},
making use of the following equality:
\begin{equation}
\frac{\partial}{\partial x}\left[x\Phi\!\left(\frac{x}{T}\right)\right]=\frac{1}{2}+\frac{1}{2}\erf\!\left(\frac{x}{\sqrt{2}T}\right)+\frac{x}{T}\phi\!\left(\frac{x}{T}\right).
\end{equation}
Thus, Theorem~\ref{thm:grad-silu} extends to the Gaussian case, as
well as Algorithm~\ref{alg:grad-silu} for estimating the partial
derivative $\frac{\partial}{\partial\theta_{j}}F_{\gelu_{T}}(\theta)$.

\subsection{Discussion}

``Fully Gaussian'' bosonic quantum computations involve the preparation
of bosonic Gaussian states, evolutions according to bosonic Gaussian
channels, and conclude with Gaussian measurements. Such protocols
can be simulated efficiently classically~\cite{Bartlett2002} (see
also~\cite{Veitch2013}), so that one cannot expect any kind of advantage
when using them. Our protocols outlined above prepare a control qumode
in a bosonic Gaussian state and perform a Gaussian measurement (i.e.,
homodyne detection) on the control qumode. However, these protocols
involve highly non-Gaussian evolutions, i.e., unitary interactions
such as $e^{i\hat{x}\otimes H(\theta)t}$ or the Hadamard test, so
that they fall outside the class of fully Gaussian protocols.

\section{Generalization to continuous quantum variables}

\label{sec:gen-CV-case}

In this section, we demonstrate that the whole formalism developed
in previous sections extends to continuous quantum variables (see
\cite{Serafini2023} for a review of this topic). In order to do so,
we return to the classical second-order neuron model from~\eqref{eq:second-order-neuron},
modify it to allow for continuous inputs, and then quantize this model.

\subsection{Canonical quantization of continuous neurons}

\label{subsec:Canonical-quantization-of-CV-neurons}

Let us now suppose that $x\coloneqq\left(x_{1},\ldots,x_{n}\right)\in\mathbb{R}^{n}$
and the activation function for a classical second-order neuron is
as follows:
\begin{equation}
\varphi(x^{T}Wx+w^{T}x+b),\label{eq:continuous-2nd-order-neuron}
\end{equation}
where $\varphi\colon\mathbb{R}\mapsto\mathbb{R}$, $W\in\mathbb{R}^{n\times n}$,
$w\in\mathbb{R}^{n}$, and $b\in\mathbb{R}$. Noting that the input
to $\varphi$ is a classical Hamiltonian of the following form:
\begin{equation}
x^{T}Wx+w^{T}x+b=\sum_{i,j=1}^{n}W_{ij}x_{i}x_{j}+\sum_{i=1}^{n}w_{i}x_{i}+b,
\end{equation}
we can quantize the model in~\eqref{eq:continuous-2nd-order-neuron}
by employing the canonical position quadrature operators from quantum
optics~\cite{Serafini2023}, leading to the following Hamiltonian:
\begin{equation}
H_{C}(\theta)\coloneqq\sum_{i,j=1}^{n}W_{ij}\hat{x}_{i}\otimes\hat{x}_{j}+\sum_{i=1}^{n}w_{i}\hat{x}_{i}+bI^{\otimes n},\label{eq:classical-continuous-ham}
\end{equation}
which acts on $n$ bosonic modes, with $\hat{x}_{i}$ being a position
quadrature operator acting on the $i$th bosonic mode. Then the activation
observable in this case is
\begin{equation}
\varphi\!\left(H_{C}(\theta)\right).
\end{equation}
Using the spectral representation of each position quadrature operator
as $\hat{x}_{i}=\int_{\mathbb{R}}dx_{i}\,x_{i}|x_{i}\rangle\!\langle x_{i}|$,
where $|x_{i}\rangle$ is a position quadrature eigenstate, and employing
the functional calculus, we can write
\begin{equation}
\varphi\!\left(H_{C}(\theta)\right)=\int_{\mathbb{R}^{n}}dx\,\varphi(x^{T}Wx+w^{T}x+b)|x\rangle\!\langle x|,
\end{equation}
where $|x\rangle\equiv|x_{1}\rangle\otimes\cdots\otimes|x_{n}\rangle$.
It then follows that the expectation of $\varphi\!\left(H_{C}(\theta)\right)$
with respect to an $n$-mode state $\rho$ is as follows:
\begin{equation}
\Tr\!\left[\varphi\!\left(H_{C}(\theta)\right)\rho\right]=\int_{\mathbb{R}^{n}}dx\,p(x)\,\varphi(x^{T}Wx+w^{T}x+b),
\end{equation}
where $p(x)\coloneqq\langle x|\rho|x\rangle$ is a probability density
function. As such, this model does not go beyond the classical case,
similar to what happened in~\eqref{eq:ham-q-class-1}--\eqref{eq:classical-second-order-gen-model}.

To go beyond the classical case, we can employ the canonical quadrature
operators for $n$ bosonic modes \cite[Chapter~3]{Serafini2023},
which we denote by
\begin{equation}
\left(\hat{r}_{1},\ldots,\hat{r}_{2n}\right)\coloneqq\left(\hat{x}_{1},\hat{p}_{1},\ldots,\hat{x}_{n},\hat{p}_{n}\right),
\end{equation}
where $\hat{p}_{i}$ is a momentum quadrature operator for all $i\in\left[n\right]$.
The quadrature operators are subject to the canonical commutation
relations:
\begin{equation}
\left[\hat{x}_{j},\hat{p}_{k}\right]=i\delta_{jk}I,\label{eq:CCR}
\end{equation}
holding for all $j,k\in\left[n\right]$. We can then write a general
parameterized (quadratic) quantum Hamiltonian as follows:
\begin{equation}
H_{Q}(\theta)\coloneqq\sum_{i,j=1}^{2n}\Omega_{i,j}\hat{r}_{i}\otimes\hat{r}_{j}+\sum_{i=1}^{n}\omega_{i}\hat{r}_{i}+bI^{\otimes n},\label{eq:bosonic-Ham-quad}
\end{equation}
where $\theta\coloneqq\left(\Omega,\omega,b\right)$, $\Omega\in\mathbb{R}^{2n\times2n}$,
$\omega\in\mathbb{R}^{2n}$, and $b\in\mathbb{R}$. A signature of
the Hamiltonian $H_{Q}(\theta)$ is that its summands need not commute,
thus going beyond the classical case in~\eqref{eq:classical-continuous-ham}.

Applying the activation function $\varphi$ to $H_{Q}(\theta)$ then
leads to the activation observable $\varphi(H_{Q}(\theta))$. In the
case that the weight matrix $\Omega$ is positive definite, the Hamiltonian
$H_{Q}(\theta)$ simplifies and admits a normal-mode decomposition
along the lines discussed in \cite[Chapter~3]{Serafini2023}. However,
if this is not the case, then $H_{Q}(\theta)$ no longer admits this
simple decomposition and its structure is more complex. More generally,
we can consider $k$-local bosonic Hamiltonians of the form in~\eqref{eq:general-k-local-ham},
denoted by $H_{Q}(\theta)$, but with each $H_{j}$ being a tensor
product of quadrature operators that act nontrivially on $k$ modes
or even being a degree-$k$ polynomial of the quadrature operators. Indeed, we can extend the continuous-variable setting to higher-order neurons~\cite{Giles1987}, by performing the same canonical quantization procedure on the following cubic-order neuron:
\begin{align}
    \varphi\!\left(\sum_{i,j,k} V_{ijk}x_ix_jx_k+\sum_{i,j}W_{ij}x_i x_j+\sum_{i}w_i x_i+b\right),
\end{align}
or higher-order versions. 

While it is known that having only Gaussian quadrature operators in the
Hamiltonian $H_Q(\theta)$ can imply that  simulating $\exp(iH_Q(\theta)t)$
classically is efficient~\cite{Bartlett2002}, this statement does not
apply to these continuous-variable Fermi--Dirac neurons. Indeed, while
$H_Q(\theta)$ may only contain Gaussian terms, the activation observable
$\varphi(H_Q(\theta))$ certainly is not limited to Gaussian operators.
In fact, the expansion $\tanh(x)=x-x^3/3+2x^5/15+O(x^7)$ means that
high-degree non-Gaussian observables are being simulated. However, a
finite-order polynomial cut-off may be sufficient if one works in an energy-limited subspace.

\subsection{Training and testing}

\label{subsec:Training-and-testing-CV-case}

In order to train and test canonical quantizations of neurons based
on bosonic Hamiltonians, we need quantum algorithms for estimating
the objective function $\Tr\!\left[\varphi(H_{Q}(\theta))\rho\right]$
and its gradient. Interestingly, many of our algorithms can be used
essentially unchanged for this purpose, whenever the state $\rho$
has bounded energy or a bounded higher moment of the photon number
operator; i.e., $\Tr[\hat{N}^{k}\rho]<+\infty$ for some $k\in\mathbb{N}$,
where $\hat{N}\coloneqq\sum_{i=1}^{n}\hat{n}_{i}$ is the total photon
number operator and $\hat{n}_{i}\coloneqq\frac{1}{2}\left(\hat{x}_{i}^{2}+\hat{p}_{i}^{2}-1\right)$
(the case of bounded energy corresponds to $k=1$). This includes
Algorithm~\ref{alg:FD-thermal-alg} for Fermi--Dirac neurons, Algorithm
\ref{alg:ReLU-one-shot} for smooth ReLU, and Algorithm~\ref{alg:SiLU-one-shot}
for SiLU. Here, one needs to use Hamiltonian simulation algorithms
suitable for bosonic systems~\cite{Peng2025,Becker2025}.

Many of our estimation guarantees in the finite-dimensional, qubit
case are based on the Hoeffding inequality (see, e.g.,~\eqref{eq:hoeffding-guarantee-FD-grad}).
A desirable feature of the Hoeffding inequality is that it implies
that the sample complexity needed to obtain an $\varepsilon$-accurate
estimate with failure probability $\leq\delta$ is $O\!\left(\left(\frac{M}{\varepsilon}\right)^{2}\ln\!\left(\frac{1}{\delta}\right)\right)$,
where $M$ is the size of the interval in which the random variable
being estimated takes values.

In the continuous-variable case, we are no longer guaranteed that
the random variable being estimated takes values in a bounded interval.
However, we often have a guarantee that the state being used in an
experiment has finite energy or finite higher moments of the total
photon number operator. As a consequence of the Chebyshev inequality,
estimating the mean of a random variable $X$ with finite variance
$\sigma^{2}$ to accuracy $\varepsilon$ and failure probability $\delta$
requires $n\geq O\!\left(\left(\frac{\sigma}{\varepsilon}\right)^{2}\frac{1}{\delta}\right)$
samples, when using the sample mean as an estimator. This dependence
of the sample complexity on $\delta$ is less favorable when compared
to the guarantees from the Hoeffding inequality. Interestingly, the
median-of-means estimator essentially restores the favorable dependence
of the sample complexity on $\delta$ to be like that from the Hoeffding
inequality; indeed, the median-of-means estimator requires only $O\!\left(\left(\frac{\sigma}{\varepsilon}\right)^{2}\ln\!\left(\frac{1}{\delta}\right)\right)$
samples \cite[Theorem~2]{Lugosi2019}, thus having a dramatic improvement
in its dependence on $\delta$.

If we want to estimate $\Tr\!\left[O\rho\right]$, where $O$ is an
observable, then we require the condition $\Tr\!\left[O^{2}\rho\right]<\infty$
to hold in order to invoke the sample complexity guarantees from the
Chebyshev inequality or the median-of-means estimator. Suppose now
that we use $H_{Q}(\theta)$ of the form in~\eqref{eq:bosonic-Ham-quad}.
Algorithms~\ref{alg:grad-est} and~\ref{alg:obj-func-est} for Fermi--Dirac
neurons aim to estimate a quantity of the form $\Re\!\left[\Tr\!\left[H_{j}U\mathcal{V}(\rho)\right]\right]$,
where $H_{j}$ is one of the terms in~\eqref{eq:bosonic-Ham-quad},
$U$ is a unitary, and $\mathcal{V}$ is a unitary channel applied
to $\rho$. Such a quantity can be estimated by the same Hadamard
test used in Algorithms~\ref{alg:grad-est} and~\ref{alg:obj-func-est},
but instead using homodyne detection at the end to measure $H_{j}$.
Here, one needs to use Hamiltonian simulation algorithms suitable
for bosonic systems~\cite{Peng2025,Becker2025}, and given that Algorithms
\ref{alg:grad-est} and~\ref{alg:obj-func-est} require a control
qubit, we need to use methods for interacting qubit and bosonic systems
\cite{Liu2026}. Since each $H_{j}$ term in~\eqref{eq:bosonic-Ham-quad}
is quadratic in the quadrature operators, we thus require that the
input state $\rho$ satisfies $\Tr\!\left[\hat{N}^{2}\rho\right]<\infty$
in order to estimate $\Re\!\left[\Tr\!\left[H_{j}U\mathcal{V}(\rho)\right]\right]$
reliably using either the sample mean or the median-of-means estimators.

Algorithms~\ref{alg:grad-smooth-relu} and~\ref{alg:softplus-obj-func-est}
for smooth ReLU neurons and Algorithms~\ref{alg:grad-silu} and~\ref{alg:silu-obj-func-est}
for SiLU neurons aim to estimate a quantity of the form $\Re\!\left[\Tr\!\left[H_{j}H_{k}U\mathcal{V}(\rho)\right]\right]$.
To estimate these quantities in the qubit case, we assumed that each
$H_{k}$ is a Pauli operator, implying that one can incorporate it
into a controlled unitary interaction to perform the estimation task.
For the Hamiltonian in~\eqref{eq:bosonic-Ham-quad}, this assumption
of $H_{k}$ being unitary no longer holds. However, the form of the
Hamiltonian in~\eqref{eq:bosonic-Ham-quad} implies that the products
$H_{j}H_{k}$ are linear combinations of terms of the form $\hat{r}_{i}\hat{r}_{j}\hat{r}_{k}\hat{r}_{\ell}$,
each of which can be measured by homodyne detection, after performing
algebraic manipulations of the quadrature operators based on~\eqref{eq:CCR}.
This special structure implies that we can again employ a Hadamard
test circuit similar to that in Figure~\ref{fig:Quantum-circuits-grad-obj-func},
but we use homodyne detection to measure $H_{j}H_{k}$ at the end
of the circuit. Since each $H_{j}$ term in~\eqref{eq:bosonic-Ham-quad}
is at most quadratic in the quadrature operators, the term $H_{j}H_{k}$
is at most quartic, and we thus require that the input state $\rho$
satisfies $\Tr\!\left[\hat{N}^{4}\rho\right]<\infty$, a stronger
assumption than before, in order to estimate $\Re\!\left[\Tr\!\left[H_{j}H_{k}U\mathcal{V}(\rho)\right]\right]$
reliably using either the sample mean or the median-of-means estimators.

\section{Numerical experiments}

\label{sec:Numerical-experiments}

In this section, we report the results of numerical experiments, comparing
classical neurons to our quantized neurons for binary classification
and function approximation tasks, using squared-loss and logistic-loss minimization for training, respectively. We restricted our experiments to Hamiltonians
involving just a few qubits (ranging from two to seven), leaving it open to scale up these simulations
to much larger numbers of qubits. Across all experiments, the results indicate that our quantized neurons outperform classical neurons for learning
functions generated by quantum data. The code used to perform our experiments is publicly available \cite{He2026}.

\begin{figure*}
\centering
\begin{subfigure}[b]{0.49\textwidth}
\centering
\includegraphics[width=\linewidth]{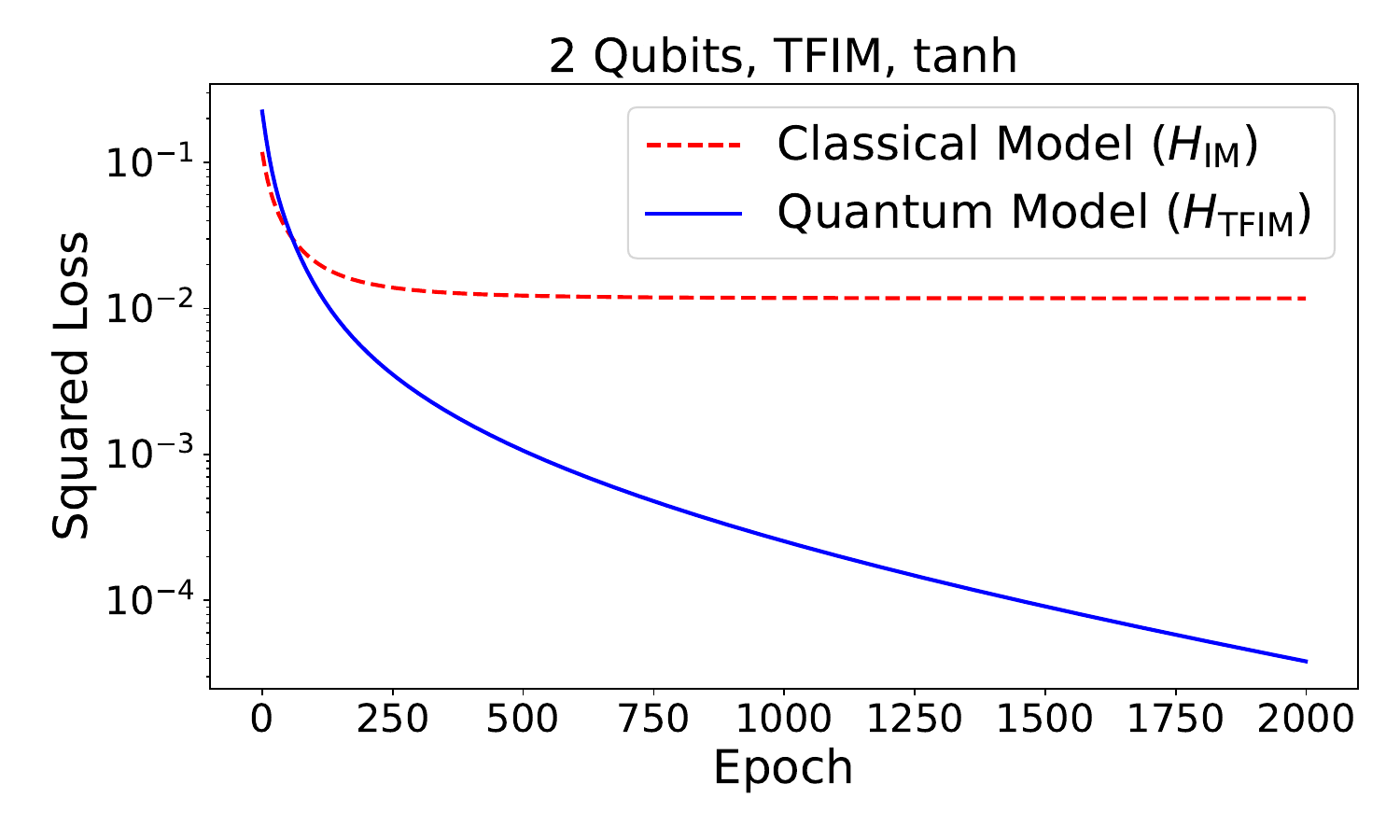}
\caption{}
\label{fig:sq_tanh_2_sq}
\end{subfigure}
\hfill
\begin{subfigure}[b]{0.49\textwidth}
\centering
\includegraphics[width=\linewidth]{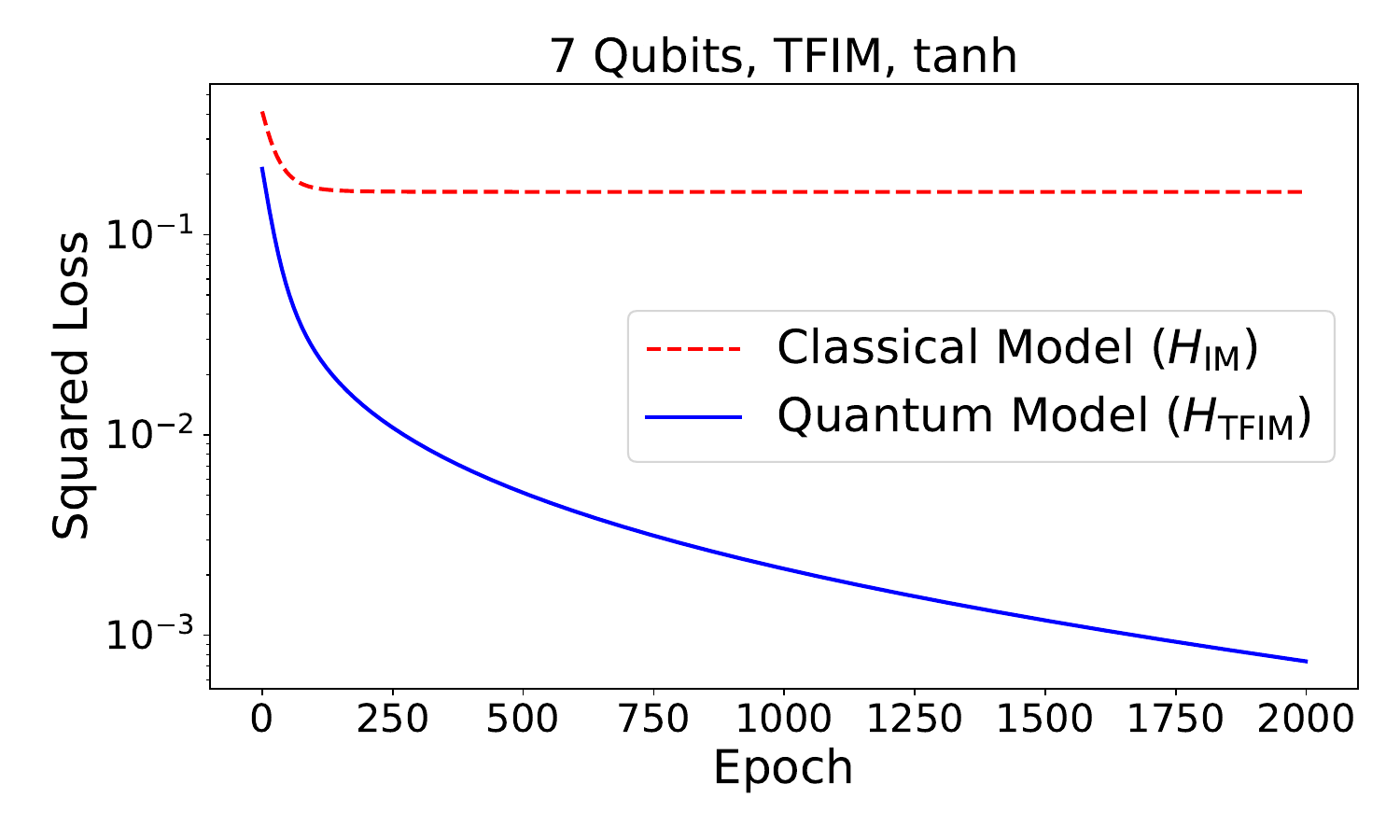}
\caption{}
\label{fig:sq_tanh_7_sq}
\end{subfigure}
\begin{subfigure}[b]{0.49\textwidth}
\centering
\includegraphics[width=\linewidth]{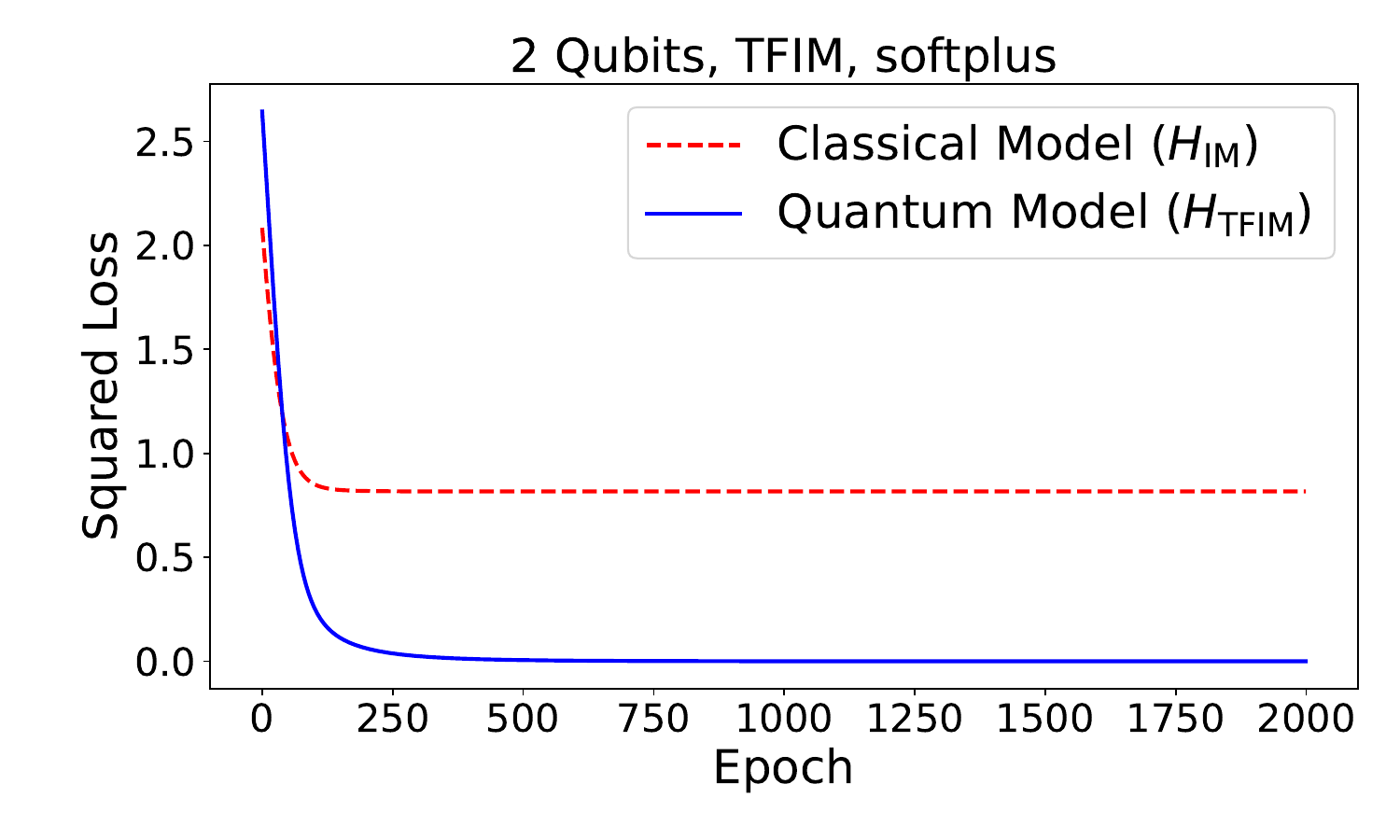}
\caption{}
\label{fig:sq_soft_2_sq}
\end{subfigure}
\hfill
\begin{subfigure}[b]{0.49\textwidth}
\centering
\includegraphics[width=\linewidth]{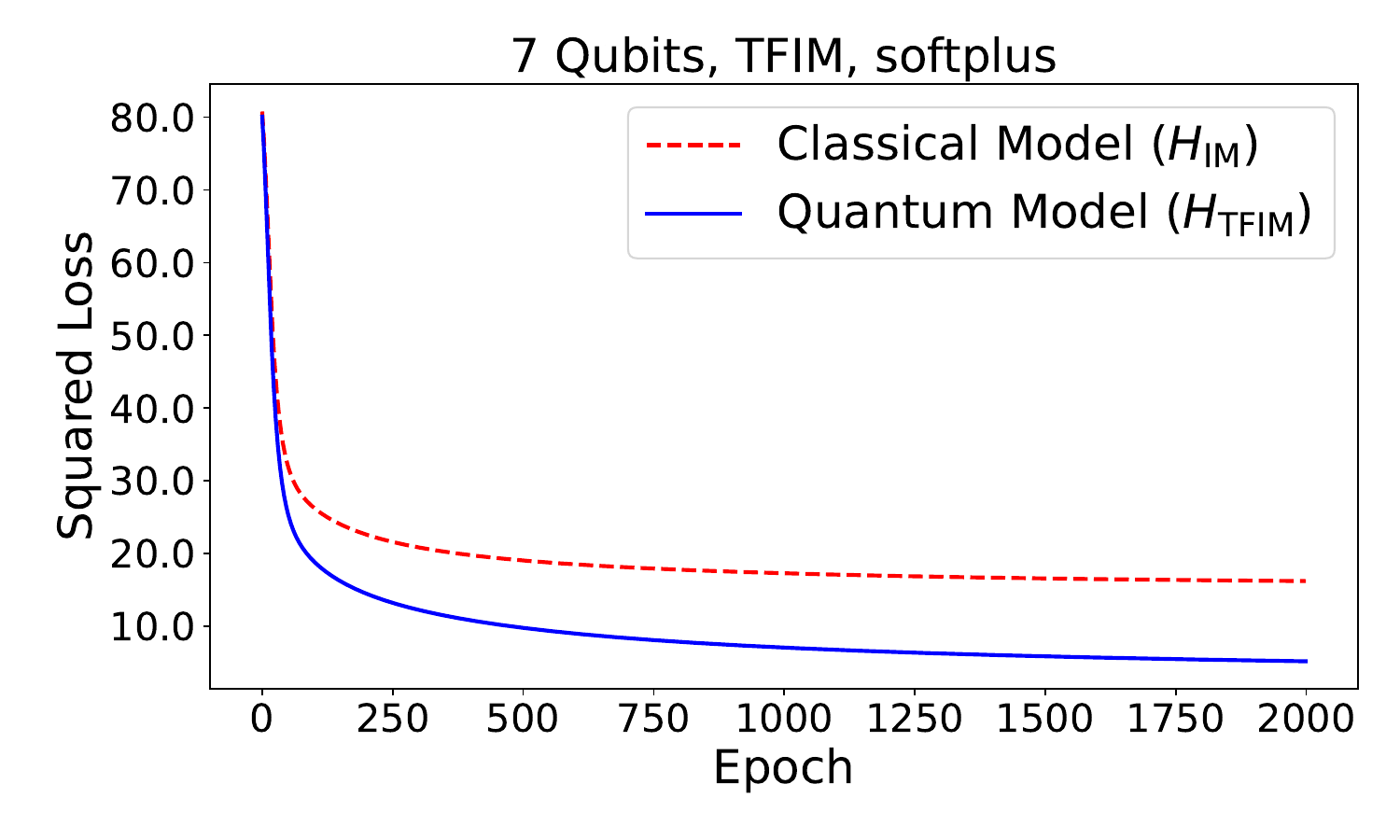}
\caption{}
\label{fig:sq_soft_7_sq}
\end{subfigure}
\caption{Results of function-approximation experiments for training using squared-loss minimization. The quantum Hamiltonian is the transverse-field Ising model (TFIM) and the classical Hamiltonian is the classical Ising model (IM). We conducted experiments using both tanh and softplus as activation observables for systems of two and seven qubits.}
\label{fig:numerics-sq-loss}
\end{figure*}

\begin{figure*}
\centering
\begin{subfigure}[b]{0.31\textwidth}
\centering
\includegraphics[width=\linewidth]{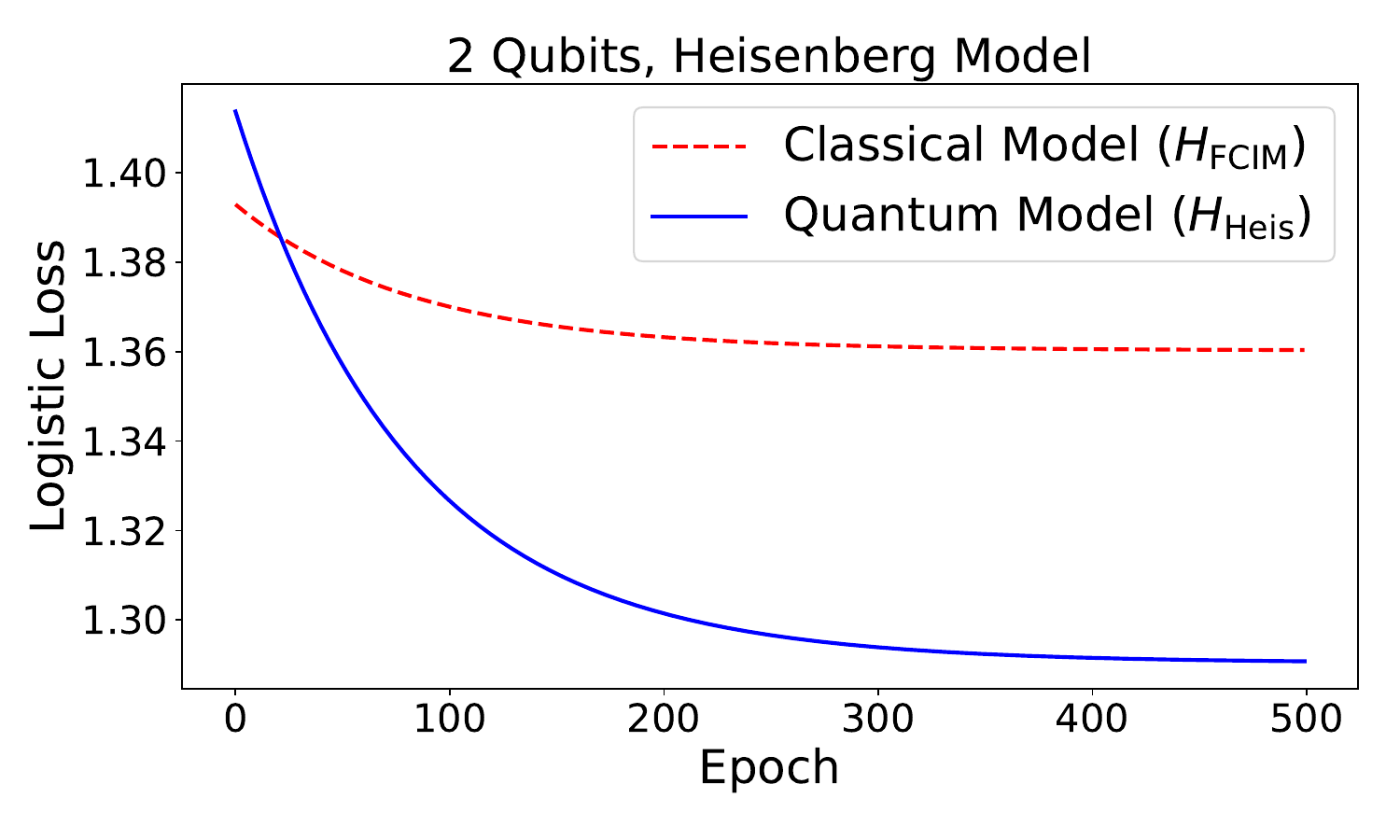}
\caption{}
\label{fig:logloss_2}
\end{subfigure}
\hfill
\begin{subfigure}[b]{0.31\textwidth}
\centering
\includegraphics[width=\linewidth]{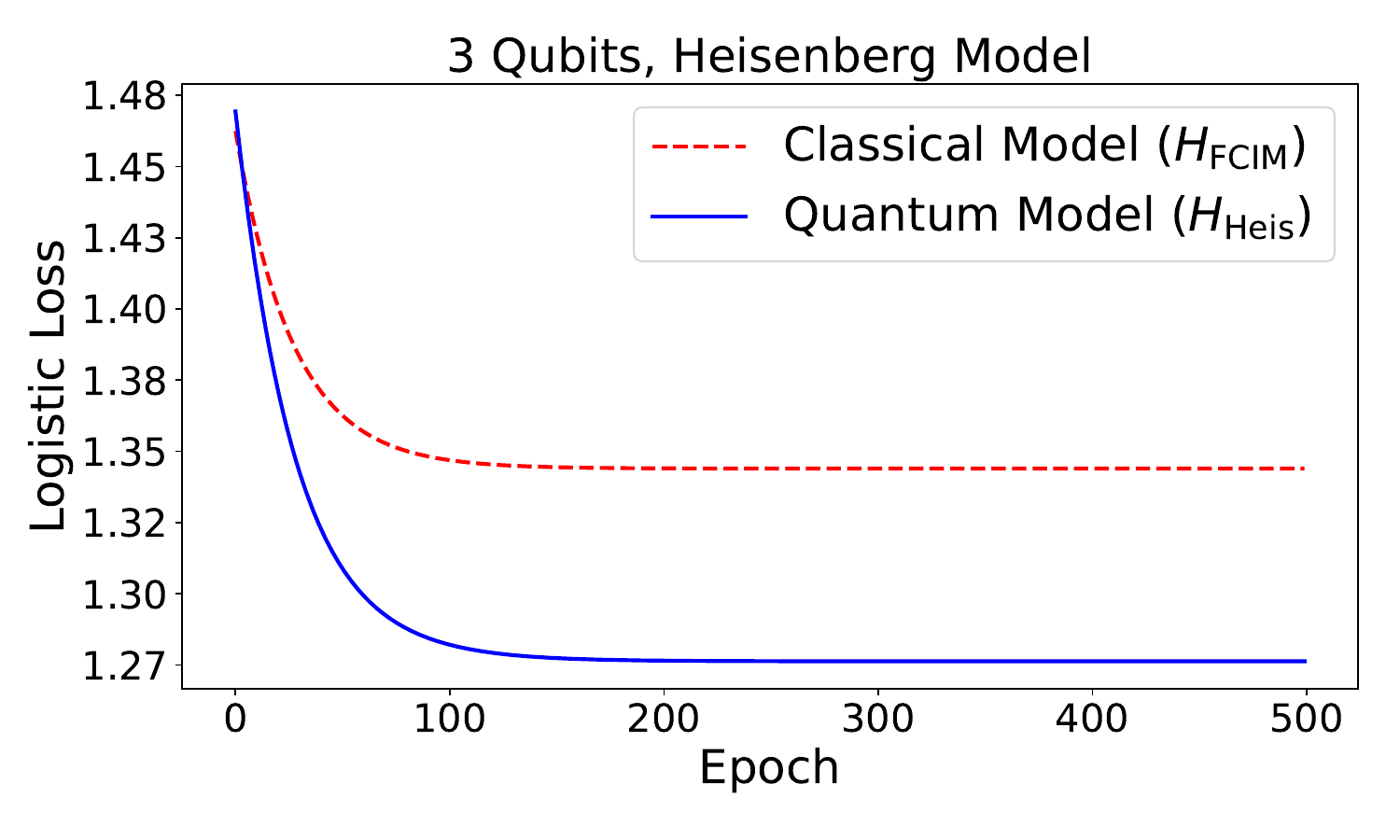}
\caption{}
\label{fig:logloss_3}
\end{subfigure}
\hfill
\begin{subfigure}[b]{0.31\textwidth}
\centering
\includegraphics[width=\linewidth]{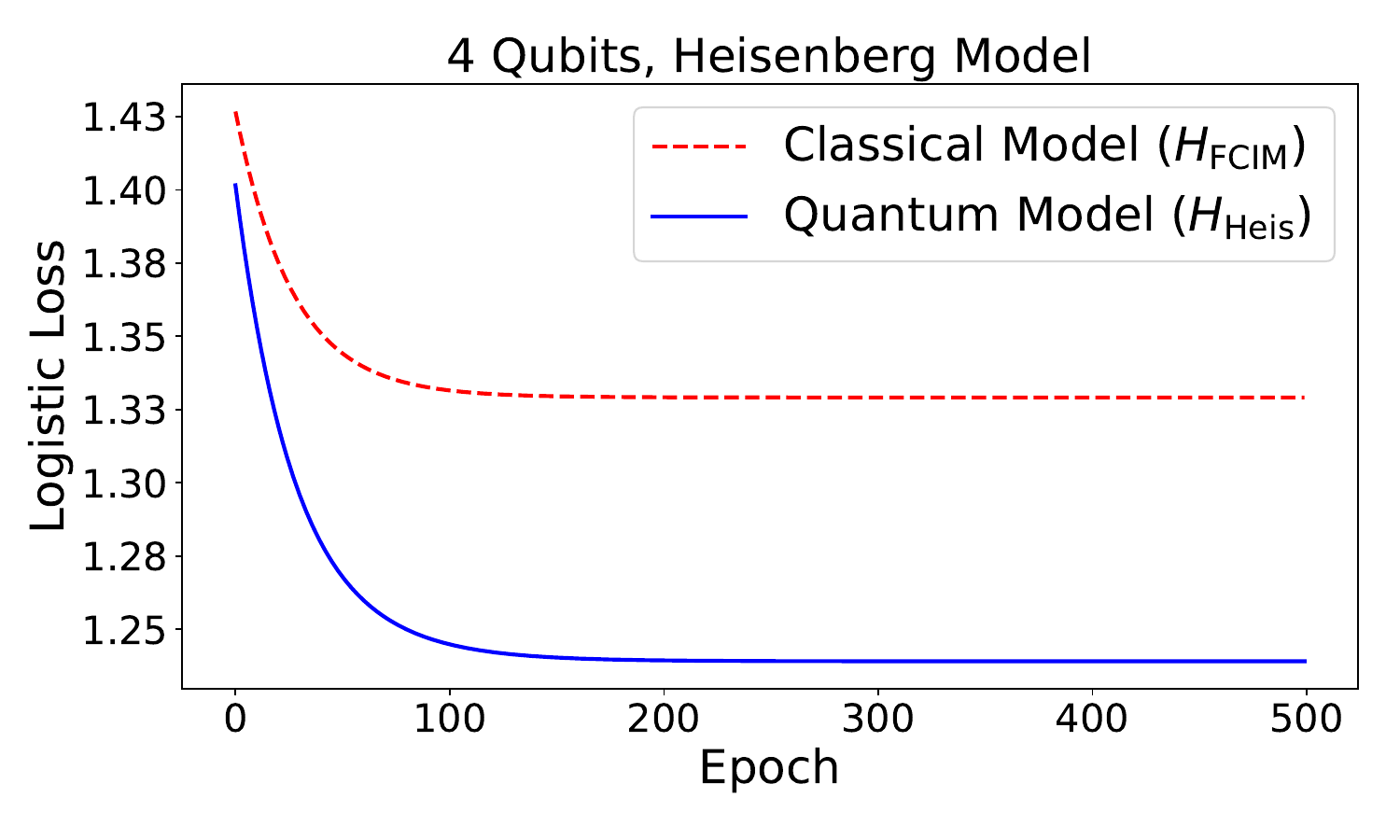}
\caption{}
\label{fig:logloss_4}
\end{subfigure}

\begin{subfigure}[b]{0.31\textwidth}
\centering
\includegraphics[width=\linewidth]{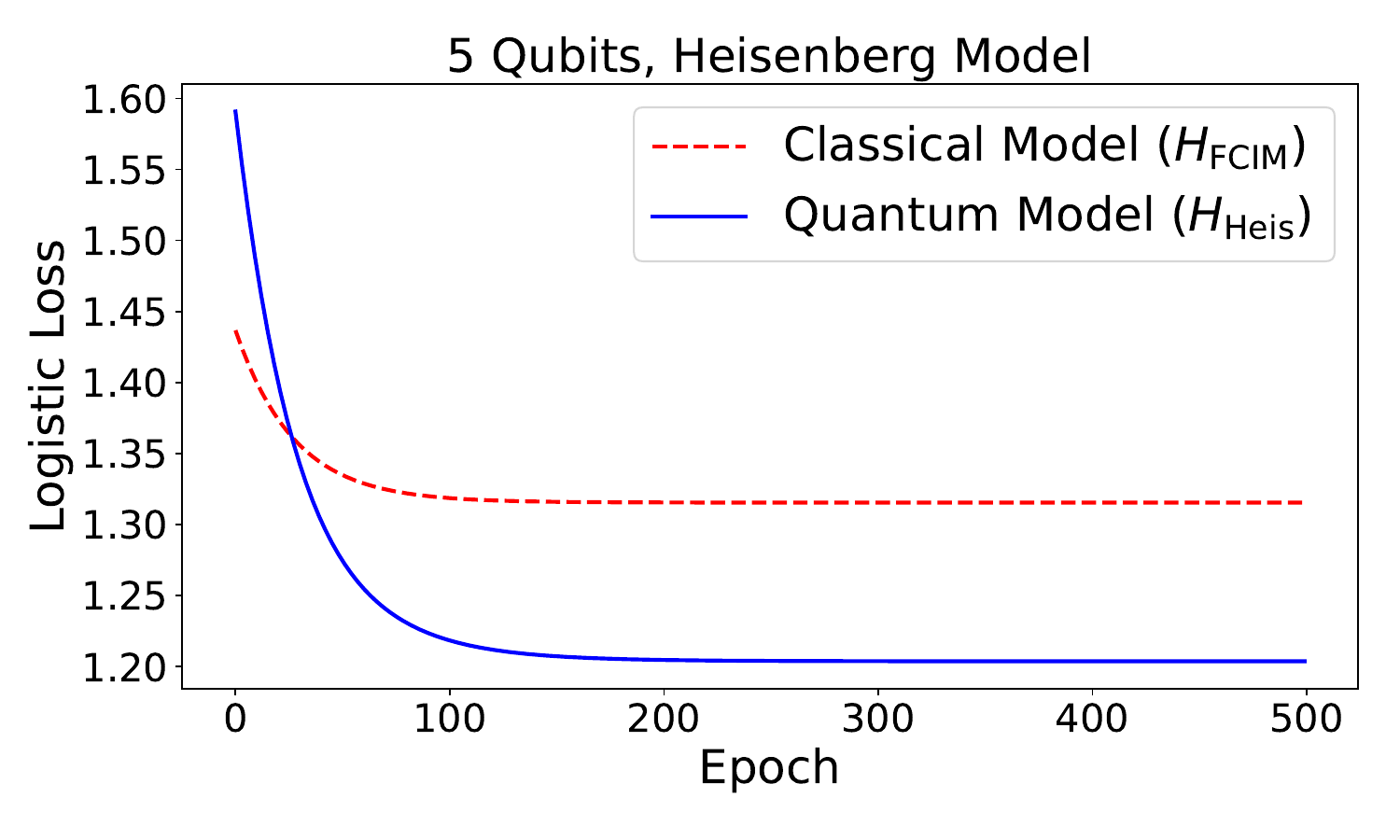}
\caption{}
\label{fig:logloss_5}
\end{subfigure}
\hfill
\begin{subfigure}[b]{0.31\textwidth}
\centering
\includegraphics[width=\linewidth]{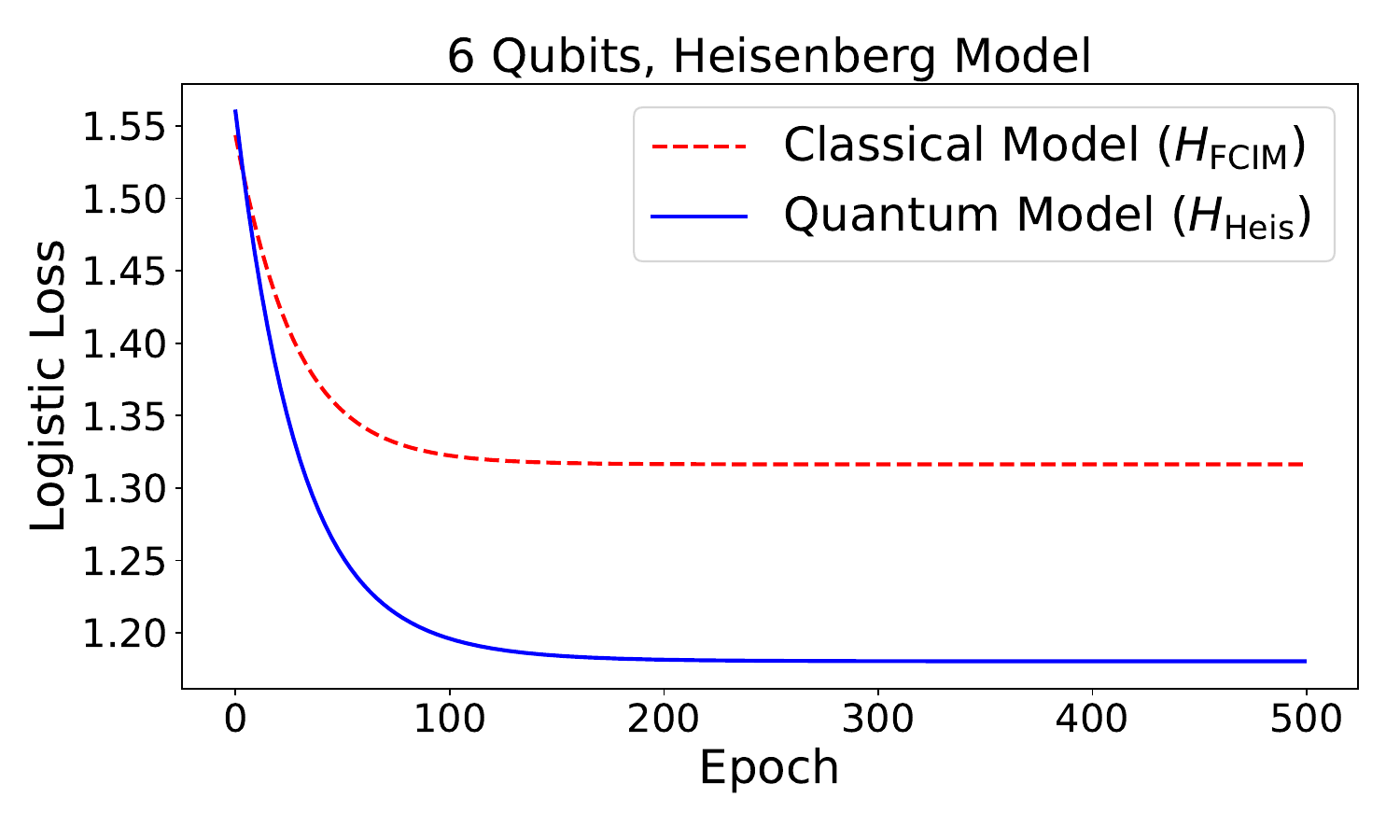}
\caption{}
\label{fig:logloss_6}
\end{subfigure}
\hfill
\begin{subfigure}[b]{0.31\textwidth}
\centering
\includegraphics[width=\linewidth]{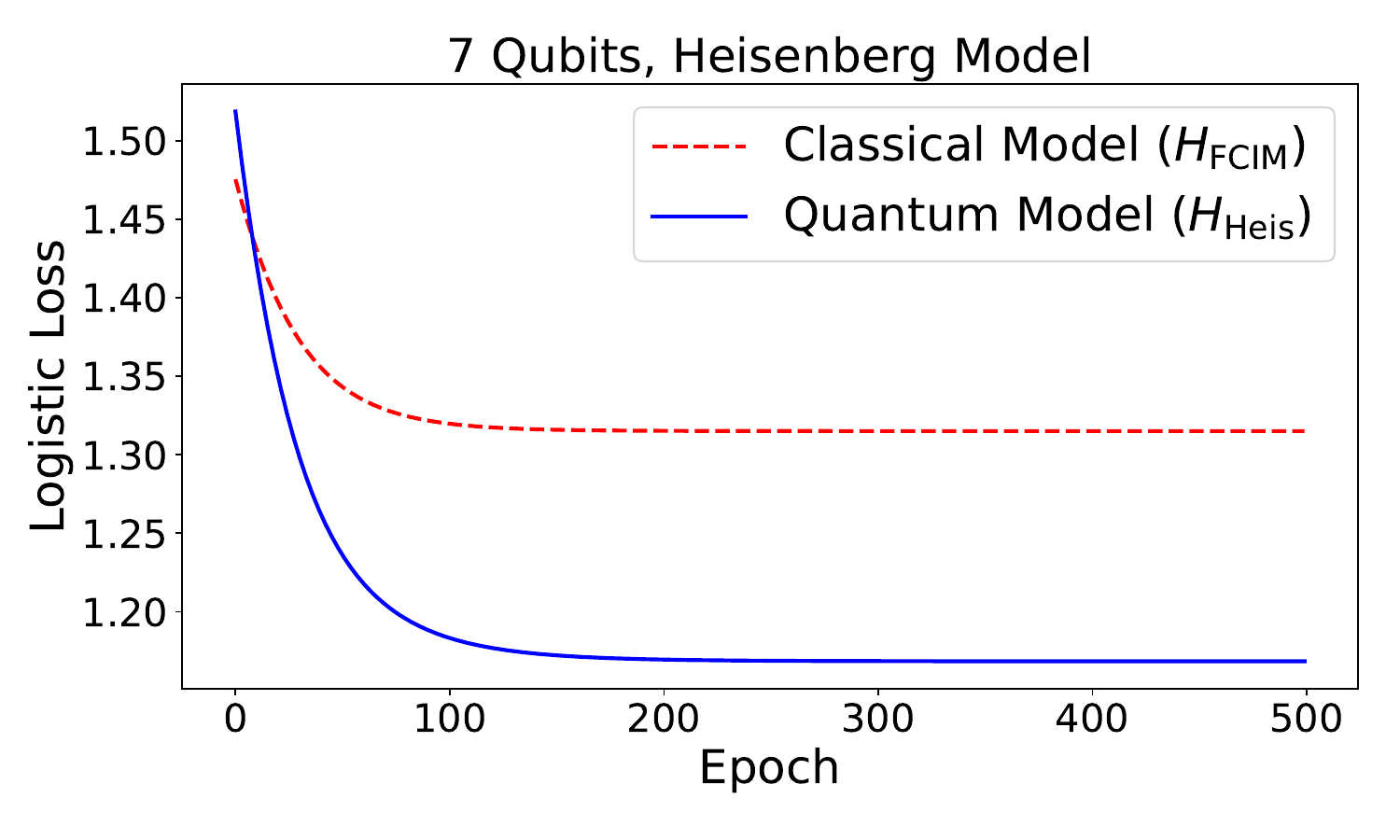}
\caption{}
\label{fig:logloss_7}
\end{subfigure}

\caption{Results of binary-classification experiments for training using logistic-loss minimization. The quantum Hamiltonian is the Heisenberg model (Heis), while the classical model is the fully-connected Ising model (FCIM). We conducted experiments using systems of 2-7 qubits.}
\label{fig:numerics-logloss}
\end{figure*}

\subsection{Data generation}

All of our experiments began by generating training data as follows. We assumed a target Hamiltonian $H^{\star}\equiv H^{\star}(\zeta)$
of the form in~\eqref{eq:fully-q-ham-1} with a parameter vector $\zeta\in\mathbb{R}^{J}$
selected at random, and we selected a set of quantum states
\begin{equation}
\rho_{1},\ldots,\rho_{M},
\end{equation}
where $M\in\mathbb{N}$. Then, for a fixed activation function $\varphi$,
we set $z_{m}\leftarrow\Tr\!\left[\varphi(H^{\star})\rho_{m}\right]$,
for all $m\in\left[M\right]$. If the goal was binary classification
and $\varphi=g_{T}$, then we set $y_{m}\leftarrow\signum(z_{m})$.
If instead the goal was function approximation, then we set $y_{m}\leftarrow z_{m}$.
This procedure generated the labeled training data
\begin{equation}
\left(\rho_{1},y_{1}\right),\ldots,\left(\rho_{M},y_{M}\right),
\end{equation}
which we used for training via squared-loss or logistic-loss minimization.

For binary classification, we also tested the performance of the trained model, using the following method. For $L\in \mathbb{N}$, we selected a separate set of quantum states 
\begin{equation}
    \rho_{M+1},\ldots,\rho_{M+L},
\end{equation}
for validation. For the same activation function $\varphi$ used in training, we set $y_{m}\leftarrow\signum\!\left(\Tr\!\left[\varphi(H^{\star})\rho_{m}\right]\right)$ as the ground truth. Then, we compared the output of the trained models to the ground truth and computed the validation accuracy based on this.


\subsection{Models}

Next, we selected a Hamiltonian $H_{Q}(\theta)$ of the form in~\eqref{eq:fully-q-ham-1}. We then picked
 a Hamiltonian $H_{C}(\theta)$ of the same form, but with every coefficient labeled by $\alpha,\beta\in\left\{ x,y\right\} $
set to zero. Thus, $H_{C}(\theta)$ is a classical Hamiltonian featuring
only $\sigma_{Z}$ terms, and it has a straightforward diagonalization
as in~\eqref{eq:diagonalization-2nd-order-neuron}. We then trained
both $H_{Q}(\theta)$ and $H_{C}(\theta)$ on the training set $\left(\rho_{1},y_{1}\right),\ldots,\left(\rho_{M},y_{M}\right)$,
using the gradient-descent algorithm outlined in~\eqref{eq:grad-descent},
until convergence or a maximum number of iterations was reached. For binary classification, we
finally tested the trained models $H_{Q}(\theta)$ and $H_{C}(\theta)$
on the validation set $\left(\rho_{M+1},y_{M+1}\right),\ldots,\left(\rho_{M+L},y_{M+L}\right)$,
and we computed both $\Tr\!\left[\varphi(H_{Q}(\theta))\rho_{m}\right]$
and $\Tr\!\left[\varphi(H_{C}(\theta))\rho_{m}\right]$ and compared them
to the true values $y_{M+1}$, $\ldots$, $y_{M+L}$.

For squared-loss minimization, we used the transverse-field Ising model
$H_{\text{TFIM}}(\theta)$ for both $H^{\star}(\zeta)$ and $H_{Q}(\theta)$
and the classical Ising model $H_{\text{IM}}(\theta)$ for $H_{C}(\theta)$.
Specifically, setting $\theta\equiv\left(W,w,b\right)$, the $n$-qubit
Hamiltonians $H_{\text{TFIM}}$ and $H_{\text{IM}}$ are specified
by
\begin{align}
H_{\text{TFIM}}(\theta) & \coloneqq\sum_{i=1}^{n-1}W_{i}\sigma_{Z}^{\left(i\right)}\otimes\sigma_{Z}^{\left(i+1\right)}+\sum_{i=1}^{n}w_{i}\sigma_{X}^{\left(i\right)}+bI^{\otimes n},\label{eq:TFIM}\\
H_{\text{IM}}(\theta) & \coloneqq\sum_{i=1}^{n-1}W_{i}\sigma_{Z}^{\left(i\right)}\otimes\sigma_{Z}^{\left(i+1\right)}+\sum_{i=1}^{n}w_{i}\sigma_{Z}^{\left(i\right)}+bI^{\otimes n},\label{eq:Ising-M}
\end{align}
with the key difference between the models being in the second sum
in each expression above. Observe that the number of parameters are the same for both the quantum model $H_{\text{TFIM}}$ and the classical model $H_{\text{IM}}$.

For logistic-loss minimization, we used the
quantum Heisenberg model for both $H^{\star}(\zeta)$ and $H_{Q}(\theta)$,
which has the following form:
\begin{multline}
H_{\text{Heis}}(\theta)\coloneqq\sum_{\alpha\in\left\{ x,y,z\right\} }\sum_{i=1}^{n-1}W_{i,\alpha}\sigma_{\alpha}^{\left(i\right)}\otimes\sigma_{\alpha}^{\left(i+1\right)}\\
+\sum_{\alpha\in\left\{ x,y,z\right\} }\sum_{i=1}^{n}w_{i,\alpha}\sigma_{\alpha}^{\left(i\right)}.
\label{eq:Heis-model}
\end{multline}
The number of parameters for $H_{\text{Heis}}(\theta)$ is $6n-3$.
For the classical model $H_C(\theta)$, we used the fully-connected Ising Model $H_{\text{FCIM}}(\theta)$, which includes $\sigma_Z^{(i)} \otimes \sigma_Z^{(j)}$ interactions on every pair of qubits. It is specified by
\begin{align}
    H_{\text{FCIM}}(\theta) \coloneqq \sum_{i, j=1}^n W_{i,j}\sigma_Z^{(i)} \otimes \sigma_Z^{(j)} + \sum_{i=1}^{n}w_{i}\sigma_{Z}^{\left(i\right)}.
    \label{eq:FCIM}
\end{align}
The number of parameters for $H_{\text{FCIM}}(\theta)$ is $\frac{n(n+1)}{2}$.
Thus, the number of parameters in the classical model $H_{\text{FCIM}}(\theta)$  is roughly comparable to that of the quantum model $H_{\text{Heis}}$ for moderate values of $n$.
 Note that the identity terms have been removed in both models; the justification for this choice is discussed in Section~\ref{sec:numerical-results}. 

\subsection{Training protocol}

We proceeded with training according to the following steps:
\begin{enumerate}
\item Select a learning rate $\eta=1/10$ and temperature $T = 2$.
\item Initialize $H_{Q}(\theta^{q})$ and $H_{C}(\theta^{c})$ with random
choices of $\theta^{q}$ and $\theta^{c}$.
\item For each training state $\rho_{m}$, calculate $f_{q}(\theta^{q},\rho_m)$ and $f_{c}(\theta^{c},\rho_m)$, where
\begin{align}
f_{q}(\theta^{q},\rho) & \coloneqq\Tr\!\left[\varphi(H_{Q}(\theta^{q}))\rho\right],\\
f_{c}(\theta^{c},\rho) & \coloneqq\Tr\!\left[\varphi(H_{C}(\theta^{c}))\rho\right],
\end{align}
and evaluate the squared (or logistic) losses $\mathcal{L}_{q}(\theta^{q})$
and $\mathcal{L}_{c}(\theta^{c})$, making use of each true value
$f(\rho_{m})=\Tr\!\left[\varphi(H^{\star})\rho_{m}\right]$ to calculate
both.
\item For $i\in\left[\operatorname{dim} \theta\right] $ and $k\in\left\{ q,c\right\} $,
calculate $\left.\frac{\partial}{\partial\theta_{i}}\mathcal{L}_{k}(\theta)\right|_{\theta=\theta^{k}}$
and perform the update:
\begin{equation}
    \theta^{k}\leftarrow\theta^{k}-\eta\left.\nabla_{\theta}\mathcal{L}_{k}(\theta)\right|_{\theta=\theta^{k}}.
\end{equation}
\item Repeat steps 3-4 until convergence or a maximum number of iterations
is reached.
\end{enumerate}

\subsection{Results}
\label{sec:numerical-results}

\subsubsection{Squared-loss minimization}

In our experiments for squared-loss minimization, we set
the target Hamiltonian $H^{\star}\equiv H^{\star}(\zeta)$ to be a
transverse-field Ising model of the form in~\eqref{eq:TFIM}, by sampling
each entry of $\zeta$ uniformly at random from the interval $\left[-2,2\right]$.
The target function to learn is 
\begin{equation}
f(\rho)\coloneqq\Tr\!\left[\varphi(H^{\star})\rho\right].\label{eq:training-function}
\end{equation}
We set the model quantum Hamiltonian $H_{Q}(\theta)$ to be of the
form in~\eqref{eq:TFIM} and the model classical Hamiltonian $H_{C}(\theta)$
to be of the form in~\eqref{eq:Ising-M}.

We performed training using various combinations of the following state vectors: computational
basis states (built from $|0\rangle$ and $|1\rangle$), Hadamard basis states (built from $|+\rangle$ and $|-\rangle$), 
Bell basis states $\left|\Phi^{+}\right\rangle $, $\left|\Phi^{-}\right\rangle $,
$\left|\Psi^{+}\right\rangle $, $\left|\Psi^{-}\right\rangle $, GHZ states,
and $\sqrt{p}\left|00\right\rangle +\sqrt{1-p}\left|11\right\rangle $,
for various choices of $p\in\left(0,1\right)$. We also incorporated
the maximally mixed state $I/2^n$ into training, where $n$ is the number of qubits. We then generated the training data
by calculating $y_{m}=f(\rho_{m})$ for every training state $\rho_{m}$.

Figure~\ref{fig:numerics-sq-loss} depicts the performance of training for squared-loss minimization. The experiments included learning using both tanh and softplus for two and seven-qubit systems. They indicate that activation observables built from quantum Hamiltonians outperform those built from classical Hamiltonians.

\subsubsection{Logistic-loss minimization}

We conducted similar experiments for logistic-loss minimization, but
with $H^{\star}\equiv H^{\star}(\zeta)$ and $H_{Q}(\theta)$ set
to be Heisenberg models of the form in~\eqref{eq:Heis-model}, $H_C(\theta)$ set to be the fully-connected Ising model of the form in~\eqref{eq:FCIM}, in both cases
using the logistic-loss function in~\eqref{eq:log-loss-quantum}. We excluded identity terms from both Hamiltonians, since we observed that including them led to an excessive weighting on the constant offset, causing the optimization to prioritize this term rather than learning meaningful structure in the remaining parameters. Incorporating regularization terms into the loss function may mitigate this issue in future experiments to penalize the excessive weighting. Since the target Hamiltonian additionally contains $\sigma_Y$ terms, the training states were expanded to also include $\sigma_Y$-basis states (generated by $|{+i}\rangle_Y$ and $|{-i}\rangle_Y$). The training labels were generated by calculating $y_m = \signum(f(\rho_m))$ for every training state $\rho_m$.
The validation set was generated by picking 500 Haar-random pure states.

The results of training are plotted in Figure~\ref{fig:numerics-logloss}.
The validation accuracies for binary classification are shown in Table~\ref{tab:val_acc}.
At all qubit numbers, there is a significant separation in the losses between the quantum and classical models. Since the classical model has access to $\sigma_Z$ interactions only, the training states in the $\sigma_X$ and $\sigma_Y$ bases have expectation value equal to zero, so that the model is effectively insensitive to these states. This effect is most pronounced in our current selection of training states, which has $\sigma_Z$-, $\sigma_X$-, and $\sigma_Y$- basis states chosen in equal part. Indeed, when we train on only $\sigma_Z$ and $\sigma_X$ states, the separation between the quantum and classical models is less pronounced. Thus, we again find that quantum Hamiltonians outperform classical ones for training for binary classification. 

\begin{table}[]
    \centering
    \begin{tabular}{|c|c|c|}
        \hline
         Qubit Count & Quantum Accuracy & Classical Accuracy \\
         \hline
         2 & 95.2\% & 64.0\% \\
         3 & 97.6\% & 65.8\% \\
         4 & 96.8\% & 68.6\% \\
         5 & 92.2\% & 63.8\% \\
         6 & 91.2\% & 73.6\% \\
         7 & 91.0\% & 69.6\% \\
         \hline
    \end{tabular}
    \caption{Results of binary classification with logistic-loss minimization on randomly generated datasets of 500 states. }
    \label{tab:val_acc}
\end{table}

\subsection{Discussion}

Across all experiments, the results indicate that noncommuting Hamiltonian models
provide a consistent advantage over commuting (diagonal) models when
learning observables generated by nonlinear functions of quantum Hamiltonians.
This suggests that the expressive power of the underlying operator
algebra plays a central role in the learnability of quantum-induced
functions in this setting. On one level, it is not difficult to understand why the noncommuting Hamiltonian model should offer more expressivity, since the Hamiltonian itself contains directions that are not contained in the commuting case. For example, while $H_{\text{TFIM}}$ contains $\sigma_X$ terms, $H_{\text{IM}}$ does not, and so one would not expect $H_{\text{IM}}$ to learn states in the $\sigma_X$ eigenbasis well. This explains the simulation result that a separation between the classical and the quantum linear model $\Tr[H(\theta)\rho]$ is present, as depicted in Figure~\ref{fig:numerics-lin-tanh-sq-loss}. However, from the same figure, the Fermi--Dirac neuron appears to learn much better than $\Tr[H(\theta)\rho]$ (we will come back to an exception for the two-qubit Heisenberg model case later). We will see that this better performance arises directly from the noncommutativity of the individual terms in the Hamiltonian itself. 

\begin{figure*}
\centering
\begin{subfigure}[b]{0.49\textwidth}
\centering
\includegraphics[width=\linewidth]{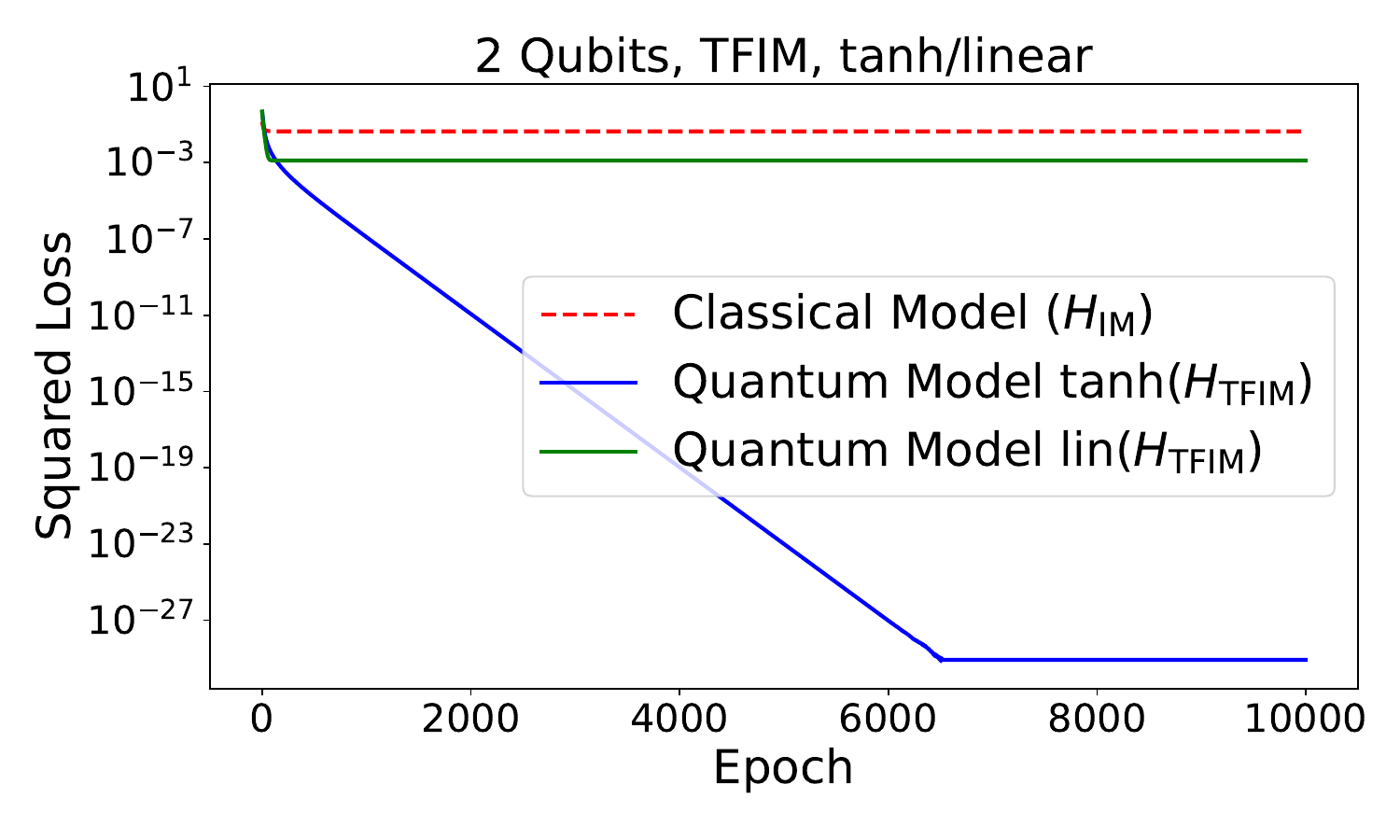}
\caption{}
\label{fig:sq_lin_tanh_2_sq_ising}
\end{subfigure}
\hfill
\begin{subfigure}[b]{0.49\textwidth}
\centering
\includegraphics[width=\linewidth]{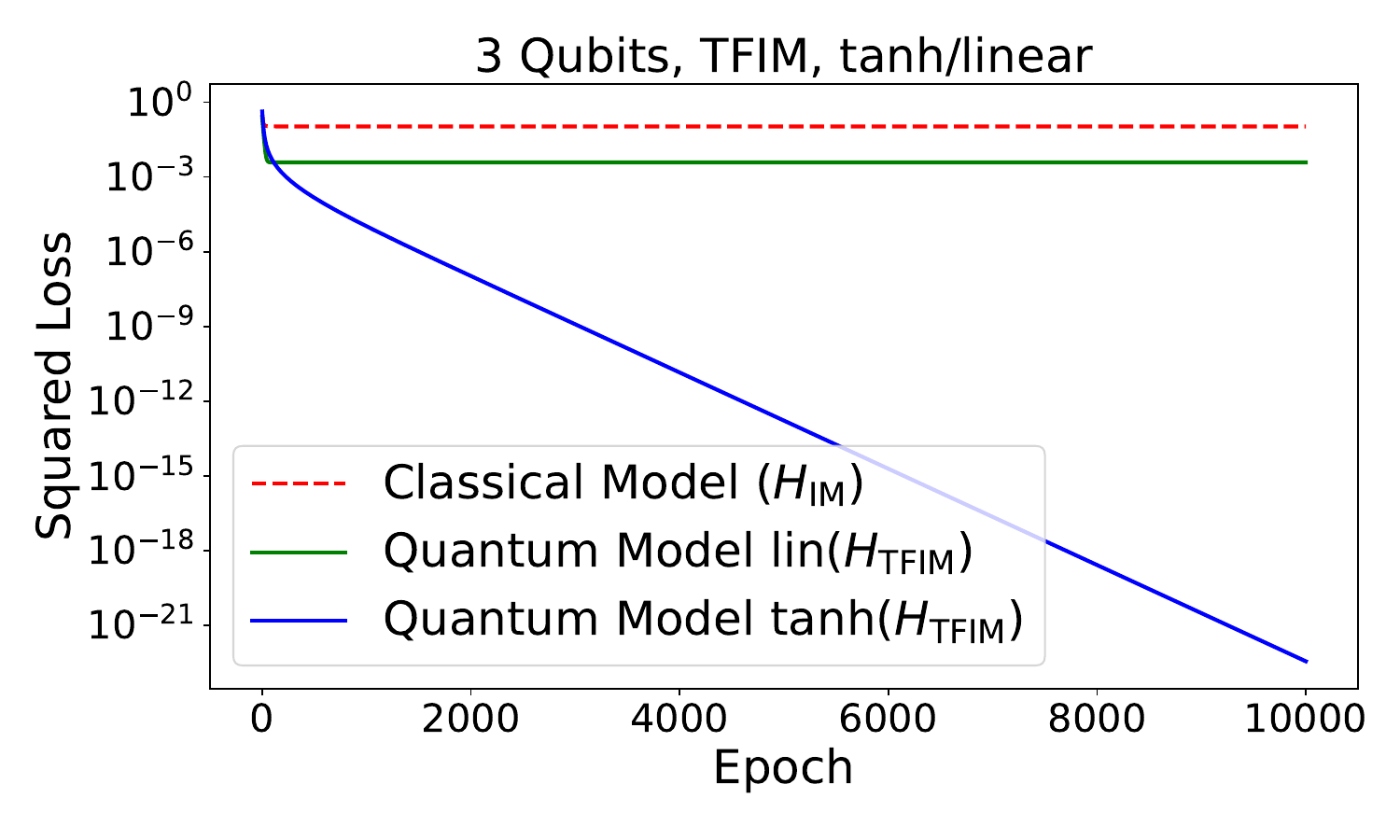}
\caption{}
\label{fig:sq_lin_tanh_3_sq_ising}
\end{subfigure}
\begin{subfigure}[b]{0.49\textwidth}
\centering
\includegraphics[width=\linewidth]{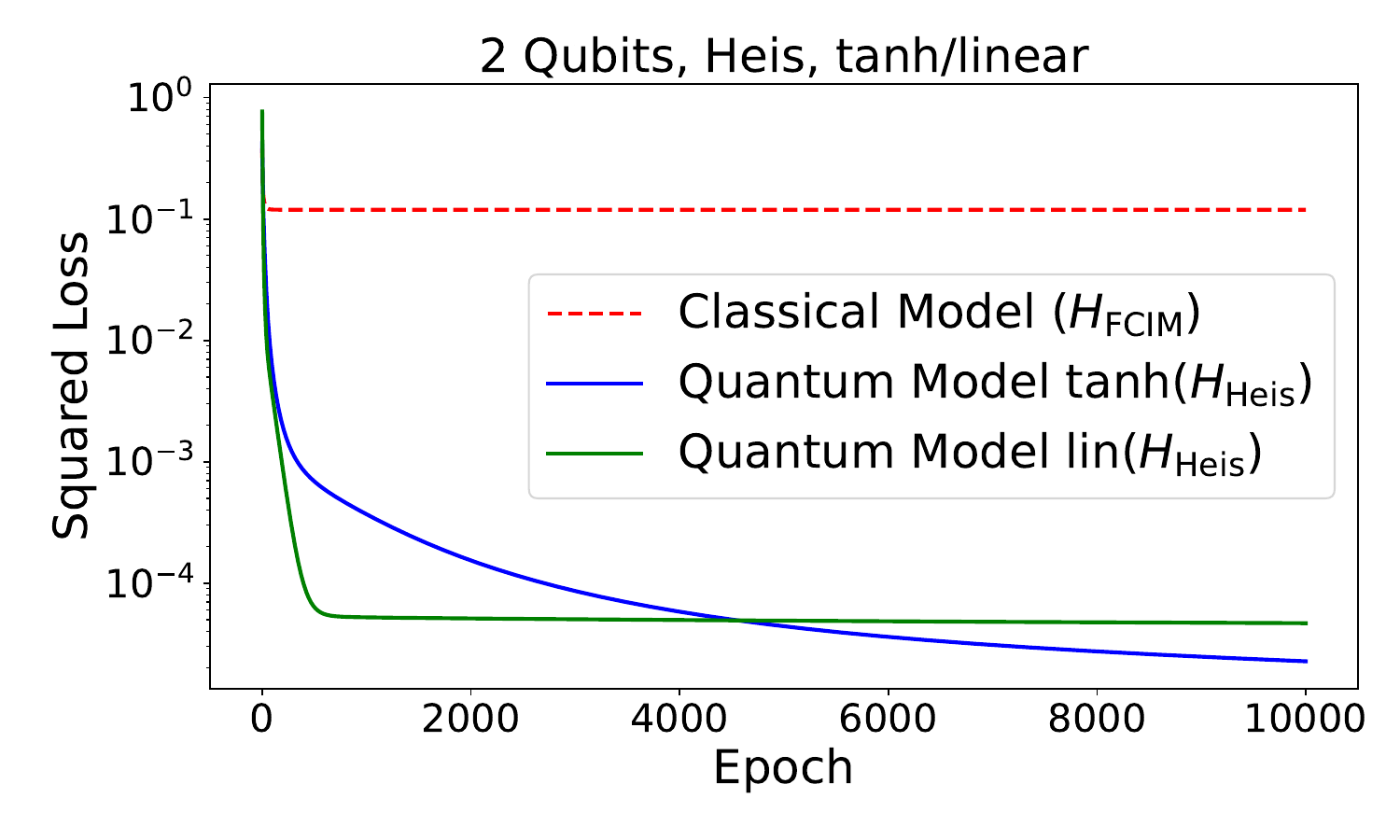}
\caption{}
\label{fig:sq_lin_tanh_2_sq_heis}
\end{subfigure}
\hfill
\begin{subfigure}[b]{0.49\textwidth}
\centering
\includegraphics[width=\linewidth]{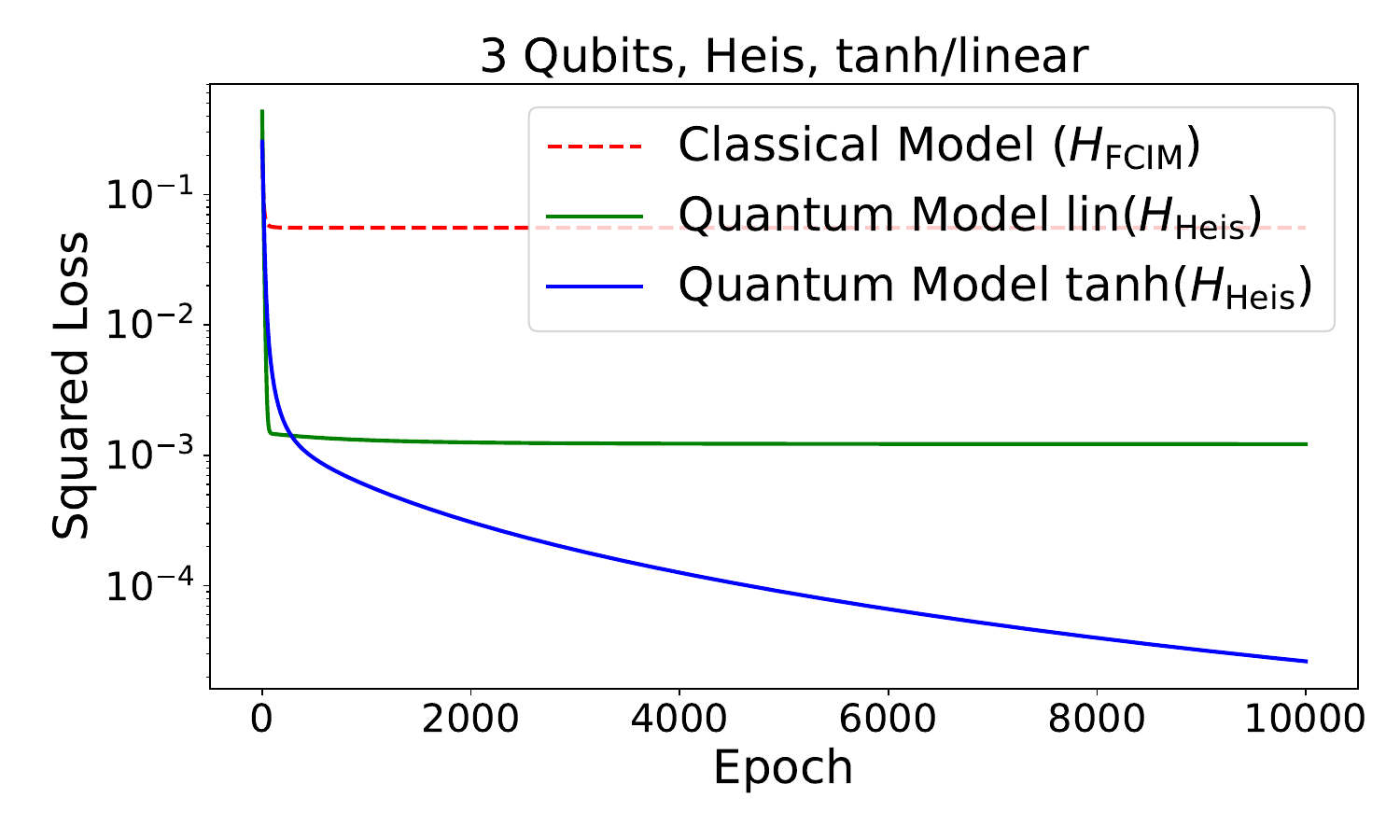}
\caption{}
\label{fig:sq_lin_tanh_3_sq_heis}
\end{subfigure}
\caption{Comparison of classical, quantum linear, and quantum nonlinear models for function approximation. Quantum linear models have expected value $\Tr[H_{\operatorname{TFIM}}(\theta)\rho_m]$, and quantum nonlinear models have expected value $\Tr[\tanh(H_{\operatorname{TFIM}}(\theta)/T)\rho_m]$. The plots depict the results of training for squared-loss minimization. The quantum Hamiltonian is the transverse-field Ising model (TFIM) and the classical Hamiltonian is the classical Ising model (IM). We conducted experiments using tanh as an activation observable for systems of two and three qubits.}
\label{fig:numerics-lin-tanh-sq-loss}
\end{figure*}

There is an intuitive way to understand the importance of this noncommutativity. We can first consider high-temperature expansions of the Fermi--Dirac neuron ($\left\|H(\theta)\right\|/T \ll 1$) with $H(\theta)=\sum_{j=1}^J \theta_j H_j$:
\begin{multline}
    \Tr\!\left[\tanh\!\left(\frac{H(\theta)}{T}\right) \rho\right] \approx \frac{1}{2T}\Tr[H(\theta)\rho]  \\
  -\frac{1}{24T^3}\Tr[H(\theta)^3 \rho]+O(T^{-5}). 
\end{multline}
If we only keep the term linear in $1/T$ (for very high temperatures), then the neuron can be mapped directly into a quantum kernel model with corresponding feature map vectors with components $\Tr[H_j \rho]$, and here the role of noncommutativity between the $H_j$ operators themselves do not yet appear. However, for lower temperatures where $1/T^3$ and higher orders are required, it is clear where the power of Fermi--Dirac neurons comes from: it can learn higher-order features like $\Re [\Tr[H_j H_k H_\ell \rho]]$. If the elements from the original set $\{H_j\}_j$ of Hamiltonians  do not commute, then these higher-order terms like $H_j H_k H_\ell$ can generate a greater diversity of new terms not in the original set. In fact, the lower the temperature, the Fermi--Dirac neuron might be expected to be more powerful, since it can capture patterns depending on even higher-order terms. However, as we will later see, beyond a certain critical temperature, its power might not increase further. Figure~\ref{fig:numerics-lin-tanh-sq-loss} supports this reasoning.

This advantage from noncommutativity only happens beyond single-qubit cases. In single-qubit cases, for instance,  if one chooses the original set $\{H_1=I, H_2=\sigma_X, H_3=\sigma_Z\}$, then $\tanh(H(\theta))$ only contains terms in the span of $\{I, \sigma_X, \sigma_Z\}$. However, for our two-qubit example using the transverse-field Ising model, we chose the noncommuting original set $\{H_1=\sigma_X^{(1)},  H_2=\sigma_X^{(2)}, H_3=\sigma_Z^{(1)} \otimes \sigma_Z^{(2)} \}$. Then $H(\theta)^3$ already contains new terms like $\sigma_Y^{(1)} \otimes \sigma_Y^{(2)}$, where no $\sigma_Y$ Hamiltonian appeared in the original set. This can explain why there are certain features in this example (i.e., those that correspond to $\sigma_Y$ directions) that a Fermi--Dirac neuron can learn, but the corresponding classical classifier, which depends only on a commuting Hamiltonian set, is not capable of learning. This inability is reflected in our simulation results. In this particular case, since higher-order terms like $H(\theta)^5$ do not generate new types of terms since the algebra closes, then it appears that it is sufficient to have a relatively high temperature Fermi--Dirac neuron, where higher-temperature versions could be more experimentally accessible; see the discussion, for example, in Section~\ref{subsec:one-shot-FD-neuron}.

However, as we increase the number of qubits, the algebra does not close at such low-order polynomials, and lower temperature limits are necessary. In the transverse-field Ising model, the original Hamiltonian set includes $\{I, \sigma_X^{(i)}, \sigma_Z^{(i)} \otimes \sigma_Z^{(i+1)}\}$. In this case, for $n$ qubits, each polynomial $H(\theta)^{2n-1} \supset \sigma_Y^{(1)} \otimes  \sigma_X^{(2)} \otimes \cdots \otimes \sigma_X^{(n-1)} \otimes \sigma_Y^{(n)}$; thus, we see that  higher powers generate new features and long-range Pauli strings that are missed by low-degree truncations. In the  presence of long-range phenomena, the Fermi--Dirac neuron therefore allows the exploration of all the higher-order features by a single training procedure without explicit estimation of these higher-order features themselves. 

For finite-dimensional systems, the algebra always closes at a finite polynomial order. More precisely, from the Cayley--Hamilton theorem, for a  $d$-dimensional Hamiltonian $H(\theta)$ with $D$ distinct eigenvalues, $H(\theta)^D$ can already be expressed in terms of a linear sum of lower-order polynomials of $H$, where $D_{\max}=d=2^n$ for $n$-qubit systems. This means that only terms up to but not including  $\text{Tr}[H(\theta)^D\rho]/T^D$ need to be kept. Thus a full exploration of the possible features generated by powers of $H(\theta)$ is only possible if the temperature $T$ is low enough, so that the corresponding term is not negligible. This can be considered like a critical temperature, since no features appear if the temperature is further lowered. 

For highly structured $H(\theta)$, generally $D$ is much smaller than $2^n$. Here it is expected that using Fermi--Dirac neurons is still more efficient, since with $J$ Hamiltonian terms in the original set for the Fermi--Dirac neuron, there are only $J$ free parameters to learn. On the other hand, if one were to separately learn the higher-order polynomial terms in $H(\theta)$, there are generally many more parameters to learn, in cases where $J \ll D$. 

Not only are Fermi--Dirac neurons advantageous in requiring less parameters, the fact that the coefficients corresponding to different Hamiltonian terms in the $\tanh(H(\theta))$ expansion are coupled to each other means that it is possible for the neuron to be sensitive to directions that are not present in the original training data. For instance, in the two-qubit transverse-field Ising model simulation, even though states sensitive to $\sigma_Y$ are not included in the original training data, the quantum model learns these states (which are not learned by the classical model). This can be explained as the coefficient in front of $\sigma_Y^{(1)} \otimes \sigma_Y^{(2)}$ is coupled to coefficients in the other directions that are included in the training data, and thus cannot be zero.

We can also look at the case of the Heisenberg model, where the original Hamiltonian set includes $\{\sigma_X^{(i)}\otimes \sigma_X^{(i+1)}, \sigma_Y^{(i)}\otimes \sigma_Y^{(i+1)}, \sigma_Z^{(i)}\otimes \sigma_Z^{(i+1)}, \sigma_X^{(i)}, \sigma_Y^{(i)}, \sigma_Z^{(i)}\}$. For two-qubit cases, the $\tanh(H(\theta))$ expansion does not introduce any new terms to the original set. This means that any quantum advantage that we see mostly arises from including both $\sigma_X$ and $\sigma_Y$ directions in the original Hamiltonian, which is absent in the corresponding classical model that contains only $\sigma_Z$ terms. This means that we do not expect a large difference between the Fermi--Dirac neuron and the linear model for the two-qubit case, and this is verified by Figure~\ref{fig:numerics-lin-tanh-sq-loss}. However, for three qubits we again see a separation between the linear model and the Fermi-Dirac neuron, as with the transverse-Ising model. This is because for three qubits,  $H(\theta)^3$ is sufficient to generate new terms belonging to the set $\mathcal{S}_{13} =\{\sigma_X^{(1)}\otimes \sigma_X^{(3)}, \sigma_Y^{(1)}\otimes \sigma_Y^{(3)}, \sigma_Z^{(1)}\otimes \sigma_Z^{(3)}\}$ (while $H(\theta)^5$ and higher polynomials generate no extra features), hence generating features on next-nearest neighbors that are not present in the original Hamiltonian. For larger numbers of qubits, longer-range two-qubit exchange terms are present, containing terms in $\mathcal{S}_{ij}$ such that $2<|i-j|<n$, as well as genuine $n$-body Pauli strings. Thus the Fermi--Dirac neuron with Heisenberg model selected as the Hamiltonian is expected to learn genuine $n$-body correlations better than the corresponding linear model. 

\section{Proposals for canonical quantizations of neural networks}

\label{sec:Proposals-for-q-neural-nets}

\subsection{Review of classical neural networks}

\label{subsec:Review-of-classical-NNs}

Composing classical neurons into multilayer neural networks leads
to powerful computational capabilities. We can express this composition
mathematically as follows. To begin with, recall from~\eqref{eq:1st-order-neuron}
that a single neuron acting on $z\in\left\{ -1,+1\right\} ^{n}$ has
the following form:
\begin{equation}
\varphi\!\left(\sum_{i=1}^{n}w_{i}z_{i}+b\right).\label{eq:single-neuron-in-network}
\end{equation}

In what follows, we consider fully connected feedforward neural networks.
Let $L\in\mathbb{N}$, $\ell\in\left[L\right]$, and $J_{\ell}\in\mathbb{N}$.
Suppose that we have $L$ layers of neurons, with layer $\ell$ consisting
of $J_{\ell}$ neurons. Then, letting $a_{i}^{(\ell)}$ denote the
output of neuron $i$ in layer $\ell$, its dependence on the previous
layer is given by
\begin{equation}
a_{i}^{(\ell)}=\varphi_{\ell}\!\left(\sum_{j=1}^{J_{\ell-1}}w_{ij}^{(\ell)}a_{j}^{(\ell-1)}+b_{i}^{(\ell)}\right),\label{eq:output-layer-l}
\end{equation}
where $w_{ij}^{(\ell)}\in\mathbb{R}$ denotes the weight from neuron
$j$ in layer $\ell-1$ to neuron $i$ in layer $\ell$ and $b_{i}^{(\ell)}$
denotes the bias of neuron $i$ in layer $\ell$. The full network
is defined by initializing
\begin{equation}
a_{j}^{(0)}=z_{j},
\end{equation}
and recursively defining $a_{i}^{(\ell)}$ as in~\eqref{eq:output-layer-l},
with the final output being 
\begin{equation}
F(z)=\left(a_{1}^{\left(L\right)},\ldots,a_{J_{L}}^{\left(L\right)}\right).
\end{equation}
Thus, the network output $F(z)\in\mathbb{R}^{J_{L}}$.

For a two-layer network, the general development above reduces to
\begin{align}
a_{j}^{(1)} & =\varphi_{1}\!\left(\sum_{i=1}^{J_{0}}w_{ji}^{(1)}z_{i}+b_{j}^{(1)}\right),\\
F_{k}(z) & =a_{k}^{\left(2\right)}\\
 & =\varphi_{2}\!\left(\sum_{j=1}^{J_{1}}w_{kj}^{(2)}a_{j}^{(1)}+b_{k}^{(2)}\right)\\
 & =\varphi_{2}\!\left(\sum_{j=1}^{J_{1}}w_{kj}^{(2)}\varphi_{1}\!\left(\sum_{i=1}^{J_{0}}w_{ji}^{(1)}z_{i}+b_{j}^{(1)}\right)+b_{k}^{(2)}\right).
\end{align}

In preparation for our quantum generalization in Section~\ref{subsec:q-observable-networks},
let us simplify the notation somewhat. A single neuron can be written
more generally as
\begin{equation}
\varphi\!\left(\sum_{j=1}^{J}\theta_{j}h_{j}(z)\right),\label{eq:classical-ham-neural-net-1}
\end{equation}
where each $\theta_{j}\in\mathbb{R}$ and $h_{j}(z)$ is an energy
function (classical Hamiltonian) interacting just a few of the variables
in $z$ (e.g., as in~\eqref{eq:neuron-hamiltonian} or~\eqref{eq:neuron-hamiltonian-second-order}).
We recover~\eqref{eq:single-neuron-in-network} by setting $J=n+1$,
$\theta_{j}=w_{j}$ and $h_{j}(z)=z_{j}$ for all $j\in\left[n\right]$,
$\theta_{n+1}=b$, and $h_{n+1}(z)=1$. As before, suppose that we
have $L$ layers of neurons, with layer $\ell$ consisting of $J_{\ell}$
neurons. Then, letting $a_{i}^{(\ell)}$ denote the output of neuron
$i$ in layer $\ell$, its dependence on the previous layer is given
by
\begin{equation}
a_{i}^{(\ell)}=\varphi_{\ell}\!\left(\sum_{j=1}^{J_{\ell-1}}\theta_{ij}^{(\ell)}a_{j}^{(\ell-1)}\right),\label{eq:output-layer-l-1}
\end{equation}
where $\theta_{ij}^{(\ell)}\in\mathbb{R}$ denotes the parameter from
neuron $j$ in layer $\ell-1$ to neuron $i$ in layer $\ell$. The
full network is defined by initializing
\begin{equation}
a_{j}^{(0)}=h_{j}(z),
\end{equation}
and recursively defining $a_{i}^{(\ell)}$ as in~\eqref{eq:output-layer-l-1},
with the final output being 
\begin{equation}
F(z)=\left(a_{1}^{\left(L\right)},\ldots,a_{J_{L}}^{\left(L\right)}\right).\label{eq:final-output-classical}
\end{equation}

For a two-layer network, the general development above reduces to
\begin{align}
a_{j}^{(1)} & =\varphi_{1}\!\left(\sum_{i=1}^{J_{0}}\theta_{ji}^{\left(1\right)}h_{i}(z)\right),\\
F_{k}(z) & =a_{k}^{\left(2\right)}\\
 & =\varphi_{2}\!\left(\sum_{j=1}^{J_{1}}\theta_{kj}^{(2)}a_{j}^{(1)}\right)\\
 & =\varphi_{2}\!\left(\sum_{j=1}^{J_{1}}\theta_{kj}^{(2)}\varphi_{1}\!\left(\sum_{i=1}^{J_{0}}\theta_{ji}^{\left(1\right)}h_{i}(z)\right)\right).\label{eq:classical-ham-neural-net-last}
\end{align}

\subsection{Quantum observable networks}

\label{subsec:q-observable-networks}

Extending the development in~\eqref{eq:classical-ham-neural-net-1}--\eqref{eq:classical-ham-neural-net-last}
to the quantum case, we quantize each variable $a_{i}^{(\ell)}$ by
replacing it with an observable $A_{i}^{(\ell)}$, and then we construct
a \textit{quantum observable network}, as defined below.

To begin with, recall from Section~\ref{subsec:Summary-of-contributions}
that the activation observable corresponding to a single neuron can
be written as
\begin{equation}
\varphi\!\left(\sum_{j=1}^{J}\theta_{j}H_{j}\right),
\end{equation}
where each $\theta_{j}\in\mathbb{R}$ and $H_{j}$ is a Hamiltonian
acting nontrivially on just $k$ qubits, while $\sum_{j=1}^{J}\theta_{j}H_{j}$
acts on $n$ qubits. As before, here and in what follows, the activation
function $\varphi$ maps $\mathbb{R}\to\mathbb{R}$, so that $\varphi\!\left(\sum_{j}\theta_{j}H_{j}\right)$
is a Hermitian operator. Although the Hamiltonian $\sum_{j=1}^{J}\theta_{j}H_{j}$
is $k$-local, the activation observable $\varphi\!\left(\sum_{j}\theta_{j}H_{j}\right)$
need not be.

Suppose that we have $L$ layers of neurons, with layer $\ell$ consisting
of $J_{\ell}$ neurons. Then, letting $A_{i}^{(\ell)}$ denote the
observable associated with neuron $i$ in layer $\ell$, its dependence
on the previous layer is given by
\begin{align}
A_{i}^{(\ell)} & =\varphi_{\ell}\!\left(B_{i}^{\left(\ell\right)}\right),\label{eq:output-layer-l-1-1}\\
B_{i}^{\left(\ell\right)} & \coloneqq\sum_{j=1}^{J_{\ell-1}}\theta_{ij}^{(\ell)}A_{j}^{(\ell-1)},
\end{align}
where $\theta_{ij}^{(\ell)}\in\mathbb{R}$ denotes the parameter from
neuron $j$ in layer $\ell-1$ to neuron $i$ in layer $\ell$ and
the activation function $\varphi_{\ell}$ acts via functional calculus
to produce the activation observable $A_{i}^{(\ell)}$. Since each
$A_{j}^{(\ell-1)}$ is Hermitian and each $\theta_{ij}^{(\ell)}\in\mathbb{R}$,
the functional calculus defining the action of $\varphi_{\ell}$ is
well defined at every layer. The full quantum observable network is defined
by initializing
\begin{equation}
A_{j}^{(0)}=H_{j},
\end{equation}
and recursively defining $A_{i}^{(\ell)}$ as in~\eqref{eq:output-layer-l-1-1},
with the final tuple of activation observables being 
\begin{equation}
F\equiv\left(A_{1}^{\left(L\right)},\ldots,A_{J_{L}}^{\left(L\right)}\right).
\label{eq:last-layer-general}
\end{equation}
Thus, the quantum observable network produces $J_{L}$ observables, each
of which acts on $n$ qubits.

For a two-layer network, the general development above reduces to
\begin{align}
A_{j}^{(1)} & =\varphi_{1}\!\left(\sum_{i=1}^{J_{0}}\theta_{ji}^{\left(1\right)}H_{i}\right),\label{eq:2-layer-ham-net-1}\\
F_{k} & =A_{k}^{\left(2\right)}\\
 & =\varphi_{2}\!\left(\sum_{j=1}^{J_{1}}\theta_{kj}^{(2)}A_{j}^{(1)}\right)\\
 & =\varphi_{2}\!\left(\sum_{j=1}^{J_{1}}\theta_{kj}^{(2)}\varphi_{1}\!\left(\sum_{i=1}^{J_{0}}\theta_{ji}^{\left(1\right)}H_{i}\right)\right).\label{eq:2-layer-ham-net-last}
\end{align}

The expectation of each observable in $F$ with respect to a state
$\rho$ is as follows:
\begin{equation}
\Tr\!\left[A_{k}^{\left(L\right)}\rho\right].
\end{equation}
Evaluating the quantum observable network on an input state $\rho$ then
amounts to estimating this expectation, and one can train it by using
labeled training data of the form in~\eqref{eq:quantum-training-data}.
For example, the squared-loss function is as follows:
\begin{equation}
\mathcal{L}_{k}^{\left(2\right)}(\theta)\coloneqq\frac{1}{M}\sum_{m=1}^{M}\left(\Tr\!\left[A_{k}^{\left(L\right)}\rho_{m}\right]-y_{m}\right)^{2},\label{eq:sq-loss-min-neural-net}
\end{equation}
where $\theta$ is a parameter vector containing all of the parameters
defined in~\eqref{eq:output-layer-l-1-1}. Alternatively, the logistic
loss function is as follows:
\begin{equation}
\mathcal{L}_{T,k}^{\log}(\theta)\coloneqq\frac{1}{M}\sum_{m=1}^{M}T\Tr\!\left[\ln\!\left(I+e^{-y_{m}A_{k}^{\left(L\right)}/T}\right)\rho_{m}\right].
\end{equation}

In order to train the model, we require the gradient of these functions
with respect to $\theta$. For the case of squared-loss minimization
in~\eqref{eq:sq-loss-min-neural-net}, calculating the gradient reduces
to calculating the gradient $\nabla_{\theta}\Tr[A_{k}^{\left(L\right)}\rho_{m}]$.
In Appendix~\ref{app:grad-ham-nets}, we establish formulas for the
gradient of a general $L$-layer network, after presenting the simpler cases
of two- and three-layer networks. The derived formulas can be expressed in terms of a backpropagation rule similar to the well known one from~\cite{Rumelhart1986a,Rumelhart1986}. However, we leave it open to determine if there
is an efficient quantum algorithm that generalizes the backpropagation
algorithm~\cite{Rumelhart1986a,Rumelhart1986} to quantum observable networks,
while noting that there have been difficulties in extending it to other
quantum machine learning models~\cite{Abbas2023}.

\subsection{Hybrid quantum--classical neural networks}

Another neural-network architecture consists of a hybrid quantum--classical
format, in which some of the neurons are quantized while others are
classical. Perhaps the most interesting such instantiation is one
in which the input data is quantum, the first layer consists of quantized
neurons, and all subsequent layers are classical. This model extends
the single-neuron setup studied in Sections~\ref{sec:Fermi=002013Dirac-machines-quants-neurs}--\ref{sec:Numerical-experiments}
of our paper, and it seems reasonable that it might have enhanced
performance for classification and function approximation tasks.

The basic idea behind this approach is that it is indeed a hybrid
of the networks proposed in Sections~\ref{subsec:Review-of-classical-NNs}
and~\ref{subsec:q-observable-networks}. The first layer
consists of the following $J_{1}$ activation observables:
\begin{equation}
A_{i}^{(1)}=\varphi_{1}\!\left(\sum_{j=1}^{J_{0}}\theta_{ij}^{\left(1\right)}H_{j}\right).
\end{equation}
However, in distinction to the proposal in Section \ref{subsec:q-observable-networks},
each of these observables is measured with respect to an input state
$\rho$, leading to an activation variable $a_{i}^{(1)}$. Each activation
variable $a_{i}^{(1)}$ is then fed into the second layer, leading
to
\begin{equation}
a_{i}^{(2)}=\varphi_{2}\!\left(\sum_{j=1}^{J_{1}}\theta_{ij}^{(2)}a_{j}^{(1)}\right).
\end{equation}
After that, the rest of the network continues as in~\eqref{eq:output-layer-l-1}
and~\eqref{eq:final-output-classical}.

There are at least two ways in which we can measure each activation
observable $A_{i}^{(1)}$ in the first layer. If the input state is
$\rho$, one can prepare it and perform one of our single-shot algorithms
(i.e., Algorithms~\ref{alg:FD-thermal-alg},~\ref{alg:ReLU-one-shot},
or~\ref{alg:SiLU-one-shot}) that outputs a random variable $a_{i}^{(1)}$
with expected value $\Tr[A_{i}^{(1)}\rho]$. Each random
variable $a_{i}^{(1)}$ would then be fed into the second layer. Alternatively,
one could perform repeated sampling, followed by averaging, to form
an estimate $\widetilde{a}_{i}^{(1)}$ of the expected value $\Tr[A_{i}^{(1)}\rho]$,
and then one could feed each estimate $\widetilde{a}_{i}^{(1)}$ into
the second layer. The single-shot approach thus induces stochastic
activation variables, analogous to stochastic neural networks in classical
machine learning~\cite{Mueller1995,Turchetti2004}.

The hybrid quantum--classical architecture seems like a practical
method for incorporating quantum algorithms into neural networks.
That is, the classical network reviewed in Section~\ref{subsec:Review-of-classical-NNs}
does not benefit from quantum processing, while the quantum observable network
proposed in Section~\ref{subsec:q-observable-networks}
might be too difficult to realize in the near term. So this hybrid
architecture represents a compromise between the two approaches, while
also extending the capabilities of a single quantized neuron.

Additionally, for training these hybrid networks according to squared-
or logistic-loss minimization, one can take advantage of our hybrid
quantum--classical algorithms for gradient estimation, as well as
the standard backpropagation algorithm~\cite{Rumelhart1986,Rumelhart1986a},
due to the hybrid nature of the network. We leave it open to future
work to simulate the performance and training of hybrid quantum--classical
neural networks.

\section{Complexity-theoretic evidence against classical simulation}

\label{sec:Complexity-theoretic-evidence}

In this section, we provide complexity-theoretic evidence that Fermi--Dirac
machines can efficiently solve problems that classical computers cannot.
To begin with, consider the following computational decision problem
based on Fermi--Dirac machines, recalling the definition of $g_T$ in \eqref{eq:tanh-g-T}:
\begin{problem}
\label{prob:FD-neuron-decision}Let $H=\sum_{j=1}^{J}H_{j}$ be a
$k$-local Hamiltonian such that $\left\Vert H_{j}\right\Vert \leq1$
for all $j\in\left[J\right]$, let $\rho$ be a state preparable by
a quantum circuit, and let $T>0$ be a temperature. Given real numbers
$\alpha$ and $\beta$ such that $-1\leq\beta<\alpha\leq1$, decide
which of the following holds:
\begin{align}
\text{YES}:\qquad\Tr\!\left[g_{T}(H)\rho\right] & \geq\alpha,\\
\text{NO}:\qquad\Tr\!\left[g_{T}(H)\rho\right] & \leq\beta,
\end{align}
under the promise that one of the two cases is true.
\end{problem}

We now state our main complexity-theoretic result:
\begin{thm}
\label{thm:BQP-complete-FD-decision}Let $n$ be the number of bits
needed to specify $H$, the circuit that generates $\rho$, the temperature
$T$, and the thresholds $\alpha$ and $\beta$. Then Problem~\ref{prob:FD-neuron-decision}
is BQP-complete for $k\geq5$ under polynomial-time reductions, for inverse polynomial temperature $T=\frac{1}{\poly(n)}$, provided
that the promise gap satisfies $\alpha-\beta\geq\frac{1}{\poly(n)}$.
\end{thm}

\begin{proof}
See Appendix~\ref{app:BQP-complete-FD-decision}.
\end{proof}

Let us summarize the reasoning why Theorem~\ref{thm:BQP-complete-FD-decision}
holds. Containment in BQP follows by applying Algorithm~\ref{alg:obj-func-est}
to estimate $\Tr\!\left[g_{T}(H)\rho\right]$ to inverse polynomial
additive accuracy.

Hardness follows by employing the same reasoning used in the proof
of \cite[Proposition~7]{Cade2023}, which established BQP hardness
for the guided local Hamiltonian problem. A critical aspect of this
approach, relevant for Problem~\ref{prob:FD-neuron-decision}, is
that a general BQP computation is encoded entirely in the Hamiltonian
$H$, via the standard Feynman--Kitaev circuit-to-Hamiltonian construction
\cite{Kitaev2002,Feynman1985} (see also~\cite{Aharonov2007,Gharibian2013}),
while the state $\rho$ is a state that is independent
of the computation (e.g., a simple computational basis state tensored
with a uniform superposition state of the clock register). Thus, this
approach illustrates that it is the activation observable $g_{T}(H)$
underlying the computational power of Fermi--Dirac machines. By choosing
$T=\frac{1}{\poly(n)}$, the function $g_{T}(x)=\tanh(x/T)$ acts
as a smooth approximation to the sign function, sharply distinguishing
positive from negative eigenvalues of $H$. As a result, $\Tr\!\left[g_{T}(H)\rho\right]$
encodes the acceptance probability of the underlying BQP computation,
up to inverse polynomial accuracy. This establishes BQP-hardness.

\section{Connections with quantum Boltzmann machines}

\label{sec:Connections-with-QBMs}

A central feature of our construction is that quantized neurons are
defined by applying an activation function~$\varphi$ to a parameterized
Hamiltonian $H(\theta)$ of the form in~\eqref{eq:fully-q-ham-1},
producing an activation observable $\varphi(H(\theta))$. This viewpoint
places our framework within a broader class of Hamiltonian-based quantum
machine learning models, where nonlinearities are induced via functional
calculus on operators. A particularly relevant choice is when $\varphi$
is the Fermi--Dirac function in~\eqref{eq:fd-function}. The resulting
activation observable is a measurement operator, as introduced in
\cite{liu2026fermidirac} and explored in more detail in the present
paper. As mentioned in Section~\ref{subsec:Canonical-quantization},
the notion of an activation observable distinguishes our approach
to quantizing neurons from all prior proposals~\cite{Kapoor2016,cao2017,Wan2017,Hu2018QuantumNeuron,Killoran2019,Yan2020,Kristensen2021,Monteiro2021,Singh2024,barney2025,Roncallo2025}.

Quantum Boltzmann machines (QBMs) provide another prominent example
of Hamiltonian-based learning models~\cite{Amin2018,Benedetti2017,Kieferova2017},
being realized by applying the exponential function to the Hamiltonian
$H(\theta)$ and normalizing the resulting operator, leading to a
parameterized thermal state of the form $e^{-H(\theta)}/\Tr\!\left[e^{-H(\theta)}\right]$.
As in the classical case, the exponential function induces nonlinearity
at the level of the induced probability distribution, rather than
producing an observable-valued activation.

From this perspective, both frameworks employ a Hamiltonian $H(\theta)$
followed by a nonlinear functional transformation; in fact, by consulting
Eqs.~(1), (2), (9), and (12) of~\cite{Amin2018}, one can see that
the way in which we quantized neurons in Sections~\ref{subsec:Canonical-quantization}
and~\ref{subsec:Quantizing-second-order-neurons} here is similar
to the way in which the authors of~\cite{Amin2018} quantized Boltzmann
machines. However, our quantized neurons and QBMs differ in the type
of object produced. QBMs map Hamiltonians to quantum states (density
operators), whereas quantized neurons map Hamiltonians to activation
observables. This distinction can be interpreted through the standard
state--observable duality of quantum mechanics, where expectation
values take the form $\Tr\!\left[A\rho\right]$ for an observable
$A$ and a state $\rho$.

The particular nonlinear functions arising for both quantum Boltzmann machines and Fermi--Dirac machines are not arbitrary, but instead arise from well motivated free-energy minimization problems; see, e.g., ~\cite{lindsey2023,Liu2025} and~\cite{lindsey2023,liu2026fermidirac}, respectively. Indeed, the nonlinear function
\begin{equation}
H \mapsto e^{-H/T}/\Tr[e^{-H/T}]  \label{eq:thermal-map}  
\end{equation}
for quantum Boltzmann machines arises by optimizing a semidefinite program over the set of density operators after augmenting the objective with a free-energy term involving the von Neumann entropy scaled by a temperature $T>0$. The nonlinear function
\begin{equation}
H\mapsto (e^{-H/T}+I)^{-1}   \label{eq:FD-map}  
\end{equation}
for Fermi--Dirac machines arises instead by optimizing a semidefinite program over the set of measurement operators after augmenting the objective with a free-energy term involving the Fermi--Dirac entropy scaled by $T$.
Thus, the two nonlinearities arise from parallel variational principles defined on different convex sets: the set of density operators in one case and the operator interval $0\le M\le I$ in the other. Moreover, the two nonlinear maps inherit the geometric constraints of their respective feasible sets: the thermal map in~\eqref{eq:thermal-map} is normalized to have unit trace because it produces density operators, whereas the Fermi--Dirac map in~\eqref{eq:FD-map} satisfies $0<(e^{-H/T}+I)^{-1}< I$, making it naturally observable-valued as a measurement operator.
In this sense, Fermi--Dirac machines are a natural counterpart to quantum Boltzmann machines, reflecting the broader duality between states and observables in quantum mechanics.

The structural analogy relating the approaches is reflected in training
algorithms. That is, our algorithms for training Fermi--Dirac and
other quantized neurons are broadly similar to the algorithms from
\cite{Patel2025a,Patel2025,Minervini2026,Liu2025} for training quantum
Boltzmann machines. In both cases, the training algorithms are used
to tune the values in the parameter vector $\theta$ and their core
constituents include classical random sampling, Hamiltonian simulation,
and the Hadamard test, due to the common form of the gradients in
both cases. The main distinction lies in the particular classical
random sampling conducted and in the state input to the quantum circuit.
The circuits used to train our quantized neurons take quantum training
data as input, while the circuits used for training quantum Boltzmann
machines require parameterized thermal states as input. Thus, training
quantum Boltzmann machines requires quantum algorithms for thermal
state preparation (see, e.g.,~\cite{Chen2025} and references therein),
while our algorithms for training quantized neurons do not. This difference
leads to a potential computational advantage of quantized neurons
in quantum learning settings, as it removes the need for thermal state
preparation while retaining a Hamiltonian-based nonlinear learning
model.

\section{Lagniappe: Hybrid quantum--classical algorithm for estimating cross entropy
and log-partition function}

\label{subsec:HQC-cross-entropy-log-part}

As an offshoot of our findings in this paper, we note here that our
approach in Theorems~\ref{thm:obj-func-formula},~\ref{thm:obj-func-softplus},
and~\ref{thm:silu-obj-func} and Algorithms~\ref{alg:obj-func-est},
\ref{alg:softplus-obj-func-est}, and~\ref{alg:silu-obj-func-est}
can be repurposed for estimating cross entropy and the log-partition
function. That is, these theorems and algorithms all have to do with
methods for estimating $\Tr\!\left[\varphi(H(\theta))\rho\right]$
that are derived from the fundamental theorem of calculus and expressions
for the gradient $\nabla_{\theta}\Tr\!\left[\varphi(H(\theta))\rho\right]$,
and the same idea can be used for the aforementioned purposes, which
are relevant for quantum Boltzmann machine learning~\cite{Kieferova2017,Coopmans2024}
and quantum thermodynamics more broadly.

To see this, let $H(\theta)$ be a parameterized Hamiltonian of the
form in~\eqref{eq:ham-q-class-1}, and define the parameterized thermal
state $\rho(\theta)$ as follows:
\begin{align}
\rho(\theta) & \coloneqq\frac{e^{-H(\theta)}}{Z(\theta)},\\
Z(\theta) & \coloneqq\Tr\!\left[e^{-H(\theta)}\right].
\end{align}
Let $\eta$ be a quantum state to which we have sample access. The
quantum relative entropy is defined as~\cite{Umegaki1962}
\begin{equation}
D(\eta\|\rho(\theta))\coloneqq\Tr\!\left[\eta\left(\ln\eta-\ln\rho(\theta)\right)\right]\label{eq:q-rel-entr-def}
\end{equation}
and is a standard measure of distinguishability well motivated in
quantum information theory~\cite{Hiai1991,Ogawa2000}. The cross entropy
corresponds to the second term in~\eqref{eq:q-rel-entr-def} and is
defined as follows:
\begin{align}
\Xi(\eta\|\rho(\theta)) & \coloneqq-\Tr\!\left[\eta\ln\rho(\theta)\right]\\
 & =\left\langle H(\theta)\right\rangle _{\eta}+\ln Z(\theta).\label{eq:cross-entr-2nd-formula}
\end{align}

The following expression for the $j$th partial derivative of $\Xi(\eta\|\rho(\theta))$
is known \cite[Eq.~(4)]{Kieferova2017}:
\begin{equation}
\frac{\partial}{\partial\theta_{j}}\Xi(\eta\|\rho(\theta))=\left\langle H_{j}\right\rangle _{\eta}-\left\langle H_{j}\right\rangle _{\rho(\theta)},\label{eq:grad-cross-entropy-QBMs}
\end{equation}
following from basic calculus and the fact that
\begin{equation}
\frac{\partial}{\partial\theta_{j}}Z(\theta)=-\Tr\!\left[H_{j}e^{-H(\theta)}\right].
\end{equation}

\begin{thm}
\label{thm:cross-entropy-est}The following equality holds:
\begin{align}
 & \Xi(\eta\|\rho(\theta))-\ln d\nonumber \\
 & =\left\langle H(\theta)\right\rangle _{\eta}-\left\Vert \theta\right\Vert _{1}\mathbb{E}_{j\sim q,\lambda\sim\upsilon}\!\left[\signum(\theta_{j})\left\langle H_{j}\right\rangle _{\rho(\theta^{\left(j\right)}(\lambda))}\right]\label{eq:cross-ent-formula-1}\\
 & =\left\Vert \theta\right\Vert _{1}\mathbb{E}_{j\sim q,\lambda\sim\upsilon}\!\left[\signum(\theta_{j})\left(\left\langle H_{j}\right\rangle _{\eta}-\left\langle H_{j}\right\rangle _{\rho(\theta^{\left(j\right)}(\lambda))}\right)\right],\label{eq:cross-ent-formula}
\end{align}
where $q(j)\coloneqq\frac{\left|\theta_{j}\right|}{\left\Vert \theta\right\Vert _{1}}$
is a probability distribution, $\upsilon$ is a uniform random variable
over the unit interval $\left[0,1\right]$, and
\begin{align}
H(\theta_{1},\ldots,\theta_{J}) & \equiv H(\theta)=\sum_{j=1}^{J}\theta_{j}H_{j},\\
\theta^{\left(j\right)}(\lambda) & \coloneqq\left(0,\ldots,0,\lambda\theta_{j},\theta_{j+1},\ldots,\theta_{J}\right),\\
\rho(\theta^{\left(j\right)}(\lambda)) & \coloneqq\frac{e^{-H(\theta^{\left(j\right)}(\lambda))}}{Z(\theta^{\left(j\right)}(\lambda))},\\
Z(\theta^{\left(j\right)}(\lambda)) & \coloneqq\Tr\!\left[e^{-H(\theta^{\left(j\right)}(\lambda))}\right].
\end{align}
\end{thm}

\begin{proof}
See Appendix~\ref{app:cross-entropy-est}.
\end{proof}
In Appendix~\ref{app:alg-cross-ent-log-part}, we use the formula
in~\eqref{eq:cross-ent-formula} to delineate a hybrid quantum--classical
algorithm for estimating the cross entropy $\Xi(\eta\|\rho(\theta))$.
Note that it requires methods for thermal state preparation (see,
e.g.,~\cite{Chen2025} and references therein for such methods).

By combining~\eqref{eq:cross-entr-2nd-formula} and~\eqref{eq:cross-ent-formula-1},
we arrive at the following formula for the log-partition function:
\begin{cor}
Using the same notation as in Theorem~\ref{thm:cross-entropy-est},
the following equality holds:
\begin{equation}
\ln Z(\theta)=\ln d-\left\Vert \theta\right\Vert _{1}\mathbb{E}_{j\sim q,\lambda\sim\upsilon}\!\left[\signum(\theta_{j})\left\langle H_{j}\right\rangle _{\rho(\theta^{\left(j\right)}(\lambda))}\right].\label{eq:log-part-formula}
\end{equation}
\end{cor}

In Appendix~\ref{app:alg-cross-ent-log-part}, we use the formula
in~\eqref{eq:log-part-formula} to delineate a hybrid quantum--classical
algorithm for estimating the log-partition function $\ln Z(\theta)$.

\section{Conclusion}

\label{sec:Conclusion}

In conclusion, the main conceptual contribution of our paper is a
novel method for quantizing neurons, rooted in canonical quantization.
Indeed, by viewing a classical neuron as an activation function applied
to an energy function (classical Hamiltonian), we quantize this model
by replacing the classical Hamiltonian with a quantum Hamiltonian.
Applying the activation function to the quantum Hamiltonian results
in an activation observable.

Our technical contributions demonstrate how to train and evaluate
these quantized neurons for many key examples of activation functions,
which include the Fermi--Dirac (sigmoid), hyperbolic tangent, smooth
rectified linear unit (ReLU), sigmoid linear unit (SiLU), Gaussian
error function (erf), Gaussian-smoothed ReLU, and Gaussian error linear
unit (GeLU). For all of these activation functions, we developed methods
for single-shot evaluation of them based on quantum convolution algorithms
(Algorithms~\ref{alg:convolution-on-A} and~\ref{alg:convolution-and-mult-on-A}),
and we developed formulas for their gradients that lead to hybrid
quantum--classical algorithms for efficiently training them.

We generalized the whole framework from finite-dimensional to continuous-variable
quantum systems, and we performed numerical experiments that support
a separation between the performance of classical neurons and our
quantized neurons for binary classification and function approximation.
Finally, we proposed canonical quantizations of neural networks based
on our quantized neurons, and we proved that a computational decision
problem based on Fermi--Dirac neurons is BQP-complete, thus providing
complexity-theoretic evidence that they cannot be efficiently simulated
classically.

We believe our contributions here lead to many fruitful directions
for future research. First, further work is needed to explore the
performance of our quantized neurons in a neural network architecture.
As mentioned, what seems promising is a hybrid architecture
in which the first layer consists of quantized neurons that accept
quantum data as inputs and output classical variables, while neurons
in later layers are classical and produce classical values at the
output of the final layer. It would also certainly be interesting
to conduct more numerical experiments that incorporate the effects
of finite sampling and scale to larger systems. Our algorithms that
realize single-shot evaluation of our quantized neurons all rely on
having a control qumode that is prepared and measured and interacts
with the data register via controlled Hamiltonian evolution; it is
possible to discretize these algorithms based on the methods of~\cite{McArdle2025,Guseynov2026,Claudon2026},
and it would be interesting to provide a detailed analysis of performance
when doing so.  

There are many more questions of interest for applications, including a better understanding of which target observables can be efficiently captured by a relatively small number of Fermi--Dirac neurons. This could be of particular interest to condensed-matter physics with many-body observables, for example, in order-parameter discovery and phase classification. This Hamiltonian-based framework lends itself readily to more interpretable learning that additionally includes inductive bias like symmetries. It is also important to know whether or not there are vanishing-gradient problems \cite{Hochreiter1998vanishing,goodfellow2016deep} or barren-plateau problems \cite{McClean2018,Larocca2025} in relevant settings for applications. Different activation observables also have different expansions, and the differences in learning can be investigated from the viewpoint of approximation theory. While very high-temperature limits of Fermi--Dirac neurons are expected to be less powerful than low-temperature limits (analogously to how high-temperature quantum thermal states are classically simulable but low temperatures ones need not be), more remains to be understood how this power changes as a function of temperature and whether there are deeper insights that can emerge from thermodynamical analogies. 
\medskip{}

\begin{acknowledgments}
NL acknowledges funding from the Science
and Technology Commission of Shanghai Municipality (STCSM) grant no.~24LZ1401200
(21JC1402900), NSFC grants no.~12471411 and no.~12341104, the Shanghai
Jiao Tong University 2030 Initiative, the Shanghai Pilot Program for
Basic Research, and the Fundamental Research Funds for the Central
Universities. MMW acknowledges support from the National Science Foundation
under grant no.~2329662.
\end{acknowledgments}

\bibliography{Ref}

\appendix
\clearpage
\onecolumngrid
\large

\section{Derivations for Fermi--Dirac machines}

\subsection{Proof of Theorem~\ref{thm:gradient-obj-func} (derivation of gradient)}

\label{app:grad-deriv}

In this appendix, we derive a formula for the gradient that is useful
for estimation on quantum computers. We begin by deriving a novel
formula for the derivative of the matrix hyperbolic tangent function.
\begin{lem}
\label{lem:grad-tanh}Let $x\in\mathbb{R}$, and let $x\mapsto A(x)$
be a Hermitian-valued function. Then the following equality holds:
\begin{equation}
\frac{\partial}{\partial x}\tanh(A(x))=\int_{-\infty}^{\infty}dt\,\mu(t)\int_{0}^{1}ds\,e^{itsA(x)}\left(\frac{\partial}{\partial x}A(x)\right)e^{it\left(1-s\right)A(x)},
\end{equation}
where $\mu(t)$ is the following probability density function:
\begin{equation}
\mu(t)\coloneqq\frac{t}{2\sinh\!\left(\frac{\pi t}{2}\right)}.\label{eq:high-peak-tent-2}
\end{equation}
\end{lem}

\begin{proof}
To begin with, let
\begin{equation}
A(x)=\sum_{k}\lambda_{k}\Pi_{k}
\end{equation}
be a spectral decomposition of $A(x)$, where we have omitted the
dependence on $x$ in the eigenvalues and eigenprojections. Consider
that the derivative of the matrix hyperbolic tangent function is given
by
\begin{equation}
\frac{\partial}{\partial x}\tanh(A(x))=\sum_{k,\ell}f_{\tanh}^{\left[1\right]}(\lambda_{k},\lambda_{\ell})\Pi_{k}\left(\frac{\partial}{\partial x}A(x)\right)\Pi_{\ell},\label{eq:1st-divided-diff-tanh}
\end{equation}
where $f_{\tanh}^{\left[1\right]}$ is the first divided difference
of the hyperbolic tangent function, defined as
\begin{equation}
f_{\tanh}^{\left[1\right]}(y_{1},y_{2})\coloneqq\begin{cases}
\frac{\tanh(y_{1})-\tanh(y_{2})}{y_{1}-y_{2}} & :y_{1}\neq y_{2}\\
\sech^{2}(y_{1}) & :y_{1}=y_{2}
\end{cases},
\end{equation}
where we used the fact that $\frac{\partial}{\partial x}\tanh(x)=\sech^{2}(x)$
(see, e.g., \cite[Theorem~42]{wilde2025} for a proof of~\eqref{eq:1st-divided-diff-tanh}).
Now observe that an alternative expression for $f_{\tanh}^{\left[1\right]}$
is as follows:
\begin{equation}
f_{\tanh}^{\left[1\right]}(y_{1},y_{2})=\int_{0}^{1}ds\,\sech^{2}(sy_{1}+\left(1-s\right)y_{2}),
\end{equation}
as a consequence of the fundamental theorem of calculus and that
\begin{equation}
\frac{d}{ds}\tanh(sy_{1}+\left(1-s\right)y_{2})=\sech^{2}(sy_{1}+\left(1-s\right)y_{2})\left(y_{1}-y_{2}\right).
\end{equation}
Then it follows that
\begin{equation}
\frac{\partial}{\partial x}\tanh(A(x))=\sum_{k,\ell}\int_{0}^{1}ds\,\sech^{2}(s\lambda_{k}+\left(1-s\right)\lambda_{\ell})\Pi_{k}\left(\frac{\partial}{\partial x}A(x)\right)\Pi_{\ell}.
\end{equation}
Now consider that the Fourier transform of $\sech^{2}(x)$ is as follows
\cite{Kaertner2005}:
\begin{equation}
\sech^{2}(x)=\int_{-\infty}^{\infty}dt\,\mu(t)e^{itx},
\end{equation}
where $\mu(t)$ is the following probability density function:
\begin{equation}
\mu(t)\coloneqq\frac{t}{2\sinh\!\left(\frac{\pi t}{2}\right)}.
\end{equation}
Then it follows that
\begin{align}
\frac{\partial}{\partial x}\tanh(A(x)) & =\sum_{k,\ell}\int_{0}^{1}ds\,\int_{-\infty}^{\infty}dt\,\mu(t)e^{it\left(s\lambda_{k}+\left(1-s\right)\lambda_{\ell}\right)}\Pi_{k}\left(\frac{\partial}{\partial x}A(x)\right)\Pi_{\ell}\\
 & =\int_{-\infty}^{\infty}dt\,\mu(t)\int_{0}^{1}ds\,\sum_{k,\ell}e^{it\left(s\lambda_{k}+\left(1-s\right)\lambda_{\ell}\right)}\Pi_{k}\left(\frac{\partial}{\partial x}A(x)\right)\Pi_{\ell}\\
 & =\int_{-\infty}^{\infty}dt\,\mu(t)\int_{0}^{1}ds\,\left(\sum_{k}e^{its\lambda_{k}}\Pi_{k}\right)\left(\frac{\partial}{\partial x}A(x)\right)\sum_{\ell}e^{it\left(1-s\right)\lambda_{\ell}}\Pi_{\ell}\\
 & =\int_{-\infty}^{\infty}dt\,\mu(t)\int_{0}^{1}ds\,e^{itsA(x)}\left(\frac{\partial}{\partial x}A(x)\right)e^{it\left(1-s\right)A(x)},
\end{align}
thus concluding the proof.
\end{proof}
\begin{rem}
For Hermitian matrices $A$ and $H$, we can write the Fr\'echet
derivative~\cite{Coleman2012} of $\tanh(A)$ at $H$ as follows:
\begin{equation}
D\tanh(A)[H]=\int_{-\infty}^{\infty}dt\,\mu(t)\int_{0}^{1}ds\,e^{itsA}He^{it\left(1-s\right)A}.
\end{equation}
\end{rem}

\begin{thm*}[Restatement of Theorem~\ref{thm:gradient-obj-func}]
The following equality holds:
\begin{equation}
\frac{\partial}{\partial\theta_{j}}\Tr\!\left[g_{T}(H(\theta))\rho\right]=\frac{1}{T}\mathbb{E}_{t\sim\mu,s\sim\upsilon}\!\left[\Re\!\left[\Tr\!\left[H_{j}e^{iH(\theta)t/T}\mathcal{U}_{st/T}^{H(\theta)}(\rho)\right]\right]\right],
\end{equation}
where the objective function $\Tr\!\left[g_{T}(H(\theta))\rho\right]$
is defined in~\eqref{eq:obj-func-tanh}, $\mu(t)$ is the probability
density function defined in~\eqref{eq:high-peak-tent-2}, $\upsilon$
is a uniform random variable on the unit interval $\left[0,1\right]$,
and $\mathcal{U}_{t}^{H(\theta)}$ is the following unitary quantum
channel:
\begin{equation}
\mathcal{U}_{t}^{H(\theta)}(\cdot)\coloneqq e^{-iH(\theta)t}(\cdot)e^{iH(\theta)t},
\end{equation}
\end{thm*}
\begin{proof}[Proof of Theorem~\ref{thm:gradient-obj-func}]
Consider that
\begin{align}
\frac{\partial}{\partial\theta_{j}}\Tr\!\left[g_{T}(H(\theta))\rho\right] & =\Tr\!\left[\left(\frac{\partial}{\partial\theta_{j}}g_{T}(H(\theta))\right)\rho\right]\\
 & =\int_{-\infty}^{\infty}dt\,\mu(t)\int_{0}^{1}ds\,\Tr\!\left[e^{itsH(\theta)/T}\left(\frac{\partial}{\partial\theta_{j}}\frac{H(\theta)}{T}\right)e^{it\left(1-s\right)H(\theta)/T}\rho\right]\\
 & =\frac{1}{T}\int_{-\infty}^{\infty}dt\,\mu(t)\int_{0}^{1}ds\,\Tr\!\left[e^{itsH(\theta)/T}H_{j}e^{it\left(1-s\right)H(\theta)/T}\rho\right]\\
 & =\frac{1}{T}\int_{-\infty}^{\infty}dt\,\mu(t)\int_{0}^{1}ds\,\Tr\!\left[H_{j}e^{it\left(1-s\right)H(\theta)/T}\rho e^{itsH(\theta)/T}\right].
\end{align}
Using the conventions in the statement of the lemma, we can rewrite
the expression once more as follows:
\begin{equation}
\frac{\partial}{\partial\theta_{j}}\Tr\!\left[g_{T}(H(\theta))\rho\right]=\frac{1}{T}\mathbb{E}_{t\sim\mu,s\sim\upsilon}\!\left[\Tr\!\left[H_{j}e^{iH(\theta)t/T}\mathcal{U}_{ts/T}^{H(\theta)}(\rho)\right]\right].
\end{equation}
We can now observe that $\frac{\partial}{\partial\theta_{j}}g_{T}(H(\theta))$
is a Hermitian matrix. Indeed, this is a consequence of the fact that
$H(\theta)$ is a Hermitian matrix and from the limit definition of
the derivative as a limit of a divided difference. Given that $\rho$
is Hermitian, we conclude that $\frac{\partial}{\partial\theta_{j}}\Tr\!\left[g_{T}(H(\theta))\rho\right]$
is real. So then it follows that
\begin{align}
\frac{\partial}{\partial\theta_{j}}\Tr\!\left[g_{T}(H(\theta))\rho\right] & =\Re\!\left[\frac{\partial}{\partial\theta_{j}}\Tr\!\left[g_{T}(H(\theta))\rho\right]\right]\\
 & =\frac{1}{T}\Re\!\left[\mathbb{E}_{t\sim\mu,s\sim\upsilon}\!\left[\Tr\!\left[H_{j}e^{iH(\theta)t/T}\mathcal{U}_{ts/T}^{H(\theta)}(\rho)\right]\right]\right]\\
 & =\frac{1}{T}\mathbb{E}_{t\sim\mu,s\sim\upsilon}\!\left[\Re\!\left[\Tr\!\left[H_{j}e^{iH(\theta)t/T}\mathcal{U}_{ts/T}^{H(\theta)}(\rho)\right]\right]\right],
\end{align}
where the last equality follows because the expectations are over
probability density functions, which are in turn real-valued.
\end{proof}

\subsection{Proof of Theorem~\ref{thm:obj-func-formula} (derivation of formula
for the objective function)}

\label{app:obj-func-deriv}

We can also rewrite the objective function itself in terms of its
partial derivatives, by employing Theorem~\ref{thm:gradient-obj-func}
and the fundamental theorem of calculus, leading to the following
theorem.
\begin{thm*}[Restatement of Theorem~\ref{thm:obj-func-formula}]
The following equality holds:
\begin{equation}
\Tr\!\left[g_{T}(H(\theta))\rho\right]=\frac{\left\Vert \theta\right\Vert _{1}}{T}\mathbb{E}_{\substack{j\sim q,t\sim\mu,\\
s\sim\upsilon,\lambda\sim\upsilon
}
}\!\left[\signum(\theta_{j})\Re\!\left[\Tr\!\left[H_{j}e^{iH(\theta^{\left(j\right)}(\lambda))t/T}\mathcal{U}_{ts/T}^{H(\theta^{\left(j\right)}(\lambda))}(\rho)\right]\right]\right],
\end{equation}
where $q(j)\coloneqq\frac{\left|\theta_{j}\right|}{\left\Vert \theta\right\Vert _{1}}$
is a probability distribution, and
\begin{align}
H(\theta_{1},\ldots,\theta_{J}) & \equiv H(\theta)=\sum_{j=1}^{J}\theta_{j}H_{j},\\
\theta^{\left(j\right)}(\lambda) & \coloneqq\left(0,\ldots,0,\lambda\theta_{j},\theta_{j+1},\ldots,\theta_{J}\right).
\end{align}
\end{thm*}
\begin{proof}[Proof of Theorem~\ref{thm:obj-func-formula}]
Consider that
\begin{align}
 & \Tr\!\left[g_{T}(H(\theta))\rho\right]-\Tr\!\left[g_{T}(H(0,\theta_{2},\ldots,\theta_{J}))\rho\right]\nonumber \\
 & =\Tr\!\left[g_{T}(H(\theta_{1},\theta_{2},\ldots,\theta_{J}))\rho\right]-\Tr\!\left[g_{T}(H(0,\theta_{2},\ldots,\theta_{J}))\rho\right]\\
 & =\Tr\!\left[g_{T}(H(\theta^{\left(1\right)}(1)))\rho\right]-\Tr\!\left[g_{T}(H(\theta^{\left(1\right)}(0)))\rho\right]\\
 & =\int_{0}^{1}d\lambda_{1}\frac{\partial}{\partial\lambda_{1}}\Tr\!\left[g_{T}(H(\theta^{\left(1\right)}(\lambda_{1})))\rho\right]\\
 & =\int_{0}^{1}d\lambda_{1}\,\frac{\theta_{1}}{T}\mathbb{E}_{t\sim\mu,s\sim\upsilon}\!\left[\Re\!\left[\Tr\!\left[H_{1}e^{iH(\theta^{\left(1\right)}(\lambda_{1}))t/T}\mathcal{U}_{ts/T}^{H(\theta^{\left(1\right)}(\lambda_{1}))}(\rho)\right]\right]\right]\\
 & =\frac{\theta_{1}}{T}\mathbb{E}_{t\sim\mu,s\sim\upsilon,\lambda_{1}\sim\upsilon}\!\left[\Re\!\left[\Tr\!\left[H_{1}e^{iH(\theta^{\left(1\right)}(\lambda_{1}))t/T}\mathcal{U}_{ts/T}^{H(\theta^{\left(1\right)}(\lambda_{1}))}(\rho)\right]\right]\right].
\end{align}
The penultimate equality follows from an application of Theorem~\ref{thm:gradient-obj-func},
while accounting for the fact that
\begin{equation}
\frac{\partial}{\partial\lambda_{1}}H(\theta^{\left(1\right)}(\lambda_{1}))=\theta_{1}H_{1}.
\end{equation}
Additionally,
\begin{align}
 & \Tr\!\left[g_{T}(H(0,\theta_{2},\ldots,\theta_{J}))\rho\right]-\Tr\!\left[g_{T}(H(0,0,\theta_{3},\ldots,\theta_{J}))\rho\right]\nonumber \\
 & =\Tr\!\left[g_{T}(H(\theta^{\left(2\right)}(1)))\rho\right]-\Tr\!\left[g_{T}(H(\theta^{\left(2\right)}(0)))\rho\right]\\
 & =\int_{0}^{1}d\lambda_{2}\frac{\partial}{\partial\lambda_{2}}\Tr\!\left[g_{T}(H(\theta^{\left(2\right)}(\lambda_{2})))\rho\right]\\
 & =\int_{0}^{1}d\lambda_{2}\,\frac{\theta_{2}}{T}\mathbb{E}_{t\sim\mu,s\sim\upsilon}\!\left[\Re\!\left[\Tr\!\left[H_{2}e^{iH(\theta^{\left(2\right)}(\lambda_{2}))t/T}\mathcal{U}_{ts/T}^{H(\theta^{\left(2\right)}(\lambda_{2}))}(\rho)\right]\right]\right]\\
 & =\frac{\theta_{2}}{T}\mathbb{E}_{t\sim\mu,s\sim\upsilon,\lambda_{2}\sim\upsilon}\!\left[\Re\!\left[\Tr\!\left[H_{2}e^{iH(\theta^{\left(2\right)}(\lambda_{2}))t/T}\mathcal{U}_{ts/T}^{H(\theta^{\left(2\right)}(\lambda_{2}))}(\rho)\right]\right]\right].
\end{align}
We continue iteratively along these lines and find that the last term
is given by
\begin{align}
 & \Tr\!\left[g_{T}(H(0,\ldots,0,\theta_{J}))\rho\right]\nonumber \\
 & =\Tr\!\left[g_{T}(H(0,\ldots,0,\theta_{J}))\rho\right]-\Tr\!\left[g_{T}(H(0,\ldots,0,0))\rho\right]\\
 & =\Tr\!\left[g_{T}(H(\theta^{\left(J\right)}(1)))\rho\right]-\Tr\!\left[g_{T}(H(\theta^{\left(J\right)}(0)))\rho\right]\\
 & =\frac{\theta_{J}}{T}\mathbb{E}_{t\sim\mu,s\sim\upsilon,\lambda_{J}\sim\upsilon}\!\left[\Re\!\left[\Tr\!\left[H_{J}e^{iH(\theta^{\left(J\right)}(\lambda_{J}))t/T}\mathcal{U}_{ts/T}^{H(\theta^{\left(J\right)}(\lambda_{J}))}(\rho)\right]\right]\right],
\end{align}
where we used the fact that
\begin{equation}
g_{T}(H(\theta^{\left(J\right)}(0)))=g_{T}(H(0,\ldots,0))=g_{T}(0)=0,
\end{equation}
given that $\tanh(0)=0$. So then we form a telescoping sum and conclude
that
\begin{align}
\Tr\!\left[g_{T}(H(\theta))\rho\right] & =\Tr\!\left[g_{T}(H(\theta))\rho\right]-\Tr\!\left[g_{T}(H(0,\theta_{2},\ldots,\theta_{J}))\rho\right]\nonumber \\
 & \qquad+\Tr\!\left[g_{T}(H(0,\theta_{2},\ldots,\theta_{J}))\rho\right]-\Tr\!\left[g_{T}(H(0,0,\theta_{3},\ldots,\theta_{J}))\rho\right]\nonumber \\
 & \qquad+\cdots+\Tr\!\left[g_{T}(H(0,\ldots,0,\theta_{J}))\rho\right]-\Tr\!\left[g_{T}(H(0,\ldots,0,0))\rho\right]\\
 & =\frac{\theta_{1}}{T}\mathbb{E}_{t\sim\mu,s\sim\upsilon,\lambda_{1}\sim\upsilon}\!\left[\Re\!\left[\Tr\!\left[H_{1}e^{iH(\theta^{\left(1\right)}(\lambda_{1}))t/T}\mathcal{U}_{ts/T}^{H(\theta^{\left(1\right)}(\lambda_{1}))}(\rho)\right]\right]\right]\nonumber \\
 & \qquad+\frac{\theta_{2}}{T}\mathbb{E}_{t\sim\mu,s\sim\upsilon,\lambda_{2}\sim\upsilon}\!\left[\Re\!\left[\Tr\!\left[H_{2}e^{iH(\theta^{\left(2\right)}(\lambda_{2}))t/T}\mathcal{U}_{ts/T}^{H(\theta^{\left(2\right)}(\lambda_{2}))}(\rho)\right]\right]\right]\nonumber \\
 & \qquad+\frac{\theta_{3}}{T}\mathbb{E}_{t\sim\mu,s\sim\upsilon,\lambda_{3}\sim\upsilon}\!\left[\Re\!\left[\Tr\!\left[H_{3}e^{iH(\theta^{\left(3\right)}(\lambda_{3}))t/T}\mathcal{U}_{ts/T}^{H(\theta^{\left(3\right)}(\lambda_{3}))}(\rho)\right]\right]\right]\nonumber \\
 & \qquad+\cdots+\frac{\theta_{J}}{T}\mathbb{E}_{t\sim\mu,s\sim\upsilon,\lambda_{J}\sim\upsilon}\!\left[\Re\!\left[\Tr\!\left[H_{J}e^{iH(\theta^{\left(J\right)}(\lambda_{J}))t/T}\mathcal{U}_{ts/T}^{H(\theta^{\left(J\right)}(\lambda_{J}))}(\rho)\right]\right]\right]\\
 & =\sum_{j=1}^{J}\frac{\theta_{j}}{T}\mathbb{E}_{t\sim\mu,s\sim\upsilon,\lambda_{j}\sim\upsilon}\!\left[\Re\!\left[\Tr\!\left[H_{j}e^{iH(\theta^{\left(j\right)}(\lambda_{j}))t/T}\mathcal{U}_{ts/T}^{H(\theta^{\left(j\right)}(\lambda_{j}))}(\rho)\right]\right]\right]\\
 & =\sum_{j=1}^{J}\frac{\theta_{j}}{T}\mathbb{E}_{t\sim\mu,s\sim\upsilon,\lambda\sim\upsilon}\!\left[\Re\!\left[\Tr\!\left[H_{j}e^{iH(\theta^{\left(j\right)}(\lambda))t/T}\mathcal{U}_{ts/T}^{H(\theta^{\left(j\right)}(\lambda))}(\rho)\right]\right]\right]\\
 & =\frac{\left\Vert \theta\right\Vert _{1}}{T}\sum_{j=1}^{J}\frac{\left|\theta_{j}\right|}{\left\Vert \theta\right\Vert _{1}}\mathbb{E}_{t\sim\mu,s\sim\upsilon,\lambda\sim\upsilon}\!\left[\signum(\theta_{j})\Re\!\left[\Tr\!\left[H_{j}e^{iH(\theta^{\left(j\right)}(\lambda))t/T}\mathcal{U}_{ts/T}^{H(\theta^{\left(j\right)}(\lambda))}(\rho)\right]\right]\right],\\
 & =\frac{\left\Vert \theta\right\Vert _{1}}{T}\mathbb{E}_{\substack{j\sim q,t\sim\mu,\\
s\sim\upsilon,\lambda\sim\upsilon
}
}\!\left[\signum(\theta_{j})\Re\!\left[\Tr\!\left[H_{j}e^{iH(\theta^{\left(j\right)}(\lambda))t/T}\mathcal{U}_{ts/T}^{H(\theta^{\left(j\right)}(\lambda))}(\rho)\right]\right]\right],
\end{align}
thus completing the proof.
\end{proof}

\subsection{Proof of Theorem~\ref{thm:gen-conv-alg} (expected value of quantum
convolution algorithm)}

\label{app:proof-q-conv-alg}

If $q$ is a probability density, then $|q\rangle$ is a state vector
because
\begin{align}
\langle q|q\rangle & =\left(\int_{-\infty}^{\infty}dp'\,\sqrt{q(p')}\langle p'|\right)\left(\int_{-\infty}^{\infty}dp\,\sqrt{q(p)}|p\rangle\right)\label{eq:normalize-state-1}\\
 & =\int_{-\infty}^{\infty}dp'\,\int_{-\infty}^{\infty}dp\,\sqrt{q(p')q(p)}\langle p'|p\rangle\\
 & =\int_{-\infty}^{\infty}dp'\,\int_{-\infty}^{\infty}dp\,\sqrt{q(p')q(p)}\delta(p-p')\\
 & =\int_{-\infty}^{\infty}dp\,q(p)\\
 & =1.\label{eq:normalize-state-last}
\end{align}
Let us first suppose that the state of the data register is pure and
given by $|\varphi\rangle\!\langle\varphi|$, where $|\varphi\rangle$
is a state vector. Suppose that a spectral decomposition of $A$ is
given by
\begin{equation}
A=\sum_{i}a_{i}|i\rangle\!\langle i|.
\end{equation}
This implies that
\begin{equation}
e^{i\hat{x}\otimes A}=\sum_{i}\int_{-\infty}^{\infty}dx\,e^{ixa_{i}}|x\rangle\!\langle x|\otimes|i\rangle\!\langle i|.
\end{equation}
The probability density that Step 3 of Algorithm~\ref{alg:convolution-on-A}
outputs $p\in\mathbb{R}$ is equal to
\begin{equation}
\left\Vert \left(\langle p|\otimes I\right)e^{i\hat{x}\otimes A}\left(|q\rangle\otimes|\varphi\rangle\right)\right\Vert ^{2}.
\end{equation}
Thus, the expected value of the output of Algorithm~\ref{alg:convolution-on-A}
is
\begin{equation}
\int_{-\infty}^{\infty}dp\,r(p)\left\Vert \left(\langle p|\otimes I\right)e^{i\hat{x}\otimes A}\left(|q\rangle\otimes|\varphi\rangle\right)\right\Vert ^{2}.
\end{equation}
Consider that
\begin{align}
 & \left(\langle p|\otimes I\right)e^{i\hat{x}\otimes A}\left(|q\rangle\otimes|\varphi\rangle\right)\nonumber \\
 & =\left(\langle p|\otimes I\right)\left(\sum_{i}\int_{-\infty}^{\infty}dx\,e^{ixa_{i}}|x\rangle\!\langle x|\otimes|i\rangle\!\langle i|\right)\left(\int_{-\infty}^{\infty}dp'\,\sqrt{q(p')}|p'\rangle\otimes|\varphi\rangle\right)\\
 & =\sum_{i}\int_{-\infty}^{\infty}dp'\int_{-\infty}^{\infty}dx\,e^{ixa_{i}}\langle p|x\rangle\langle x|p'\rangle\sqrt{q(p')}|i\rangle\!\langle i|\varphi\rangle\\
 & =\sum_{i}\int_{-\infty}^{\infty}dp'\,\delta(p'-p+a_{i})\sqrt{q(p')}|i\rangle\!\langle i|\varphi\rangle\\
 & =\sum_{i}\sqrt{q(p-a_{i})}|i\rangle\!\langle i|\varphi\rangle\\
 & =\sum_{i}\sqrt{q(a_{i}-p)}|i\rangle\!\langle i|\varphi\rangle
\end{align}
where the third equality follows because
\begin{align}
\int_{-\infty}^{\infty}dx\,e^{ixa_{i}}\langle p|x\rangle\langle x|p'\rangle & =\frac{1}{2\pi}\int_{-\infty}^{\infty}dx\,e^{ixa_{i}}e^{-ipx}e^{ip'x}\label{eq:fourier-delta-1}\\
 & =\frac{1}{2\pi}\int_{-\infty}^{\infty}dx\,e^{i\left(p'+a_{i}-p\right)x}\\
 & =\delta(p'-p+a_{i}),\label{eq:fourier-delta-last}
\end{align}
and the last equality from the assumption that $q$ is even. So then
\begin{align}
\left\Vert \left(\langle p|\otimes I\right)e^{i\hat{x}\otimes A}\left(|q\rangle\otimes|\varphi\rangle\right)\right\Vert ^{2} & =\left\Vert \sum_{i}\sqrt{q(a_{i}-p)}|i\rangle\!\langle i|\varphi\rangle\right\Vert ^{2}\\
 & =\langle\varphi|\left(\sum_{i}q(a_{i}-p)|i\rangle\!\langle i|\right)|\varphi\rangle.
\end{align}
Then we conclude that the expected value is given by
\begin{align}
\int_{-\infty}^{\infty}dp\,r(p)\left\Vert \left(\langle p|\otimes I\right)e^{i\hat{x}\otimes A}\left(|q\rangle\otimes|\varphi\rangle\right)\right\Vert ^{2} & =\int_{-\infty}^{\infty}dp\,r(p)\langle\varphi|\left(\sum_{i}q(a_{i}-p)|i\rangle\!\langle i|\right)|\varphi\rangle\\
 & =\langle\varphi|\left(\sum_{i}\int_{-\infty}^{\infty}dp\,r(p)q(a_{i}-p)|i\rangle\!\langle i|\right)|\varphi\rangle\\
 & =\langle\varphi|\left(\sum_{i}s(a_{i})|i\rangle\!\langle i|\right)|\varphi\rangle\\
 & =\langle\varphi|s(A)|\varphi\rangle.\label{eq:proof-conv-alg-pure-states}
\end{align}
The result generalizes to an arbitrary state $\rho$ because every
such state can be written as a convex combination of pure states as
\begin{equation}
\rho=\sum_{z}t(z)|\varphi_{z}\rangle\!\langle\varphi_{z}|,
\end{equation}
so that
\begin{align}
 & \int_{-\infty}^{\infty}dp\,r(p)\,\Tr\!\left[\left(|p\rangle\!\langle p|\otimes I\right)e^{i\hat{x}\otimes A}\left(|q\rangle\!\langle q|\otimes\rho\right)e^{-i\hat{x}\otimes A}\right]\nonumber \\
 & =\sum_{z}t(z)\int_{-\infty}^{\infty}dp\,r(p)\,\Tr\!\left[\left(|p\rangle\!\langle p|\otimes I\right)e^{i\hat{x}\otimes A}\left(|q\rangle\!\langle q|\otimes|\varphi_{z}\rangle\!\langle\varphi_{z}|\right)e^{-i\hat{x}\otimes A}\right]\\
 & =\sum_{z}t(z)\int_{-\infty}^{\infty}dp\,r(p)\,\left\Vert \left(\langle p|\otimes I\right)e^{i\hat{x}\otimes A}\left(|q\rangle\otimes|\varphi_{z}\rangle\right)\right\Vert ^{2}\\
 & =\sum_{z}t(z)\Tr\!\left[s(A)|\varphi_{z}\rangle\!\langle\varphi_{z}|\right]\\
 & =\Tr\!\left[s(A)\rho\right],
\end{align}
where the penultimate equality follows from~\eqref{eq:proof-conv-alg-pure-states}.

\subsection{Proof of Theorem~\ref{thm:FD-thermal-alg-proof} and Equation~\eqref{eq:convolution-FD-firing}}

\label{app:convolution-FD-firing}Recall that $T_{1},T_{2}>0$ such
that $T/2=T_{1}T_{2}$, 
\begin{align}
\ell_{T_{1}}(p) & \coloneqq\frac{e^{p/T_{1}}}{T_{1}\left(e^{p/T_{1}}+1\right)^{2}},\\
r(p) & =\mathbf{1}_{p\geq0}-\mathbf{1}_{p<0},
\end{align}
and observe that
\begin{equation}
\frac{d}{dp}f_{T_{1}}(p)=\ell_{T_{1}}(p),
\end{equation}
where $f_{T_{1}}(p)=\left(e^{-p/T_{1}}+1\right)^{-1}$ is the Fermi--Dirac
function defined in~\eqref{eq:fd-function}. Consider that
\begin{align}
(\ell_{T_{1}}*r)(p/T_{2}) & =(\ell_{T_{1}}*\mathbf{1}_{p\geq0})(p/T_{2})-(\ell_{T_{1}}*\mathbf{1}_{p<0})(p/T_{2})\label{eq:convolution-logistic-indicators-1}\\
 & =\int_{-\infty}^{\infty}dp'\,\mathbf{1}_{p'\geq0}\ell_{T_{1}}(p/T_{2}-p')-\int_{-\infty}^{\infty}dp'\,\mathbf{1}_{p'<0}\ell_{T_{1}}(p/T_{2}-p')\\
 & =\int_{0}^{\infty}dp'\,\ell_{T_{1}}(p/T_{2}-p')-\int_{-\infty}^{0}dp'\,\ell_{T_{1}}(p/T_{2}-p')\\
 & =\int_{0}^{\infty}dp'\,\ell_{T_{1}}(p'-p/T_{2})-\int_{-\infty}^{0}dp'\,\ell_{T_{1}}(p'-p/T_{2})\\
 & =\int_{-p/T_{2}}^{\infty}dp'\,\ell_{T_{1}}(p')-\int_{-\infty}^{-p/T_{2}}dp'\,\ell_{T_{1}}(p')\\
 & =\int_{-p/T_{2}}^{\infty}dp'\,\frac{d}{dp'}f_{T_{1}}(p')-\int_{-\infty}^{-p/T_{2}}dp'\,\frac{d}{dp'}f_{T_{1}}(p')\\
 & =f_{T_{1}}(\infty)-f_{T_{1}}(-p/T_{2})-f_{T_{1}}(-p/T_{2})+f_{T_{1}}(-\infty)\\
 & =1-f_{T_{1}T_{2}}(-p)-f_{T_{1}T_{2}}(-p)\\
 & =1-f_{T/2}(-p)-f_{T/2}(-p)\\
 & =f_{T/2}(p)-\left(1-f_{T/2}(p)\right)\\
 & =g_{T}(p).\label{eq:convolution-logistic-indicators-last}
\end{align}
The fourth equality follows because $\ell_{T_{1}}$ is an even function.

\section{Squared-loss minimization for function approximation}

\label{app:Squared-loss-function}

We first prove~\eqref{eq:loss-function}. Consider that
\begin{align}
\frac{1}{M}\sum_{m=1}^{M}\left(\Tr\!\left[g_{T}(H(\theta))\rho_{m}\right]-y_{m}\right)^{2} & =\frac{1}{M}\sum_{m=1}^{M}\left(\Tr\!\left[\left(g_{T}(H(\theta))-y_{m}I\right)\rho_{m}\right]\right)^{2}\\
 & =\frac{1}{M}\sum_{m=1}^{M}\Tr\!\left[\left(g_{T}(H(\theta))-y_{m}I\right)^{\otimes2}\rho_{m}^{\otimes2}\right]\\
 & =\Tr\!\left[\Delta^{\left(2\right)}(\theta,y)\overline{\rho}\right],
\end{align}
where the squared-loss observable is defined in~\eqref{eq:squared-loss-observable}
and the labeled classical--quantum training state $\overline{\rho}$
in~\eqref{eq:cq-training-state}.
\begin{cor}
\label{cor:sampling-exp-loss-func}For the squared-loss function $\mathcal{L}^{\left(2\right)}(\theta)$
defined in~\eqref{eq:loss-function}, for $j\in\left[J\right]$, the
$j$th partial derivative can be expressed as follows:
\begin{equation}
\frac{\partial}{\partial\theta_{j}}\mathcal{L}^{\left(2\right)}(\theta)=\frac{2\left\Vert \theta\right\Vert _{1}}{T^{2}}\mathbb{E}_{\substack{k\sim q,t_{1},t_{2}\sim\mu,\\
s_{1},s_{2},\lambda\sim\upsilon,\\
m\sim\left[M\right]
}
}\!\left[\signum(\theta_{k})\Re\!\left[\Tr\!\left[\begin{array}{c}
\left(\left(H_{k}-y_{m}I\right)\otimes H_{j}\right)\left(e^{iH(k,\lambda)t_{1}/T}\otimes e^{iH(\theta)t_{2}/T}\right)\times\\
\left(\mathcal{U}_{s_{1}t_{1}/T}^{H(k,\lambda)}\otimes\mathcal{U}_{s_{2}t_{2}/T}^{H(\theta)}\right)(\rho_{m}\otimes\rho_{m})
\end{array}\right]\right]\right].\label{eq:grad-squared-loss-app}
\end{equation}
\end{cor}

\begin{proof}
Consider that
\begin{align}
\frac{\partial}{\partial\theta_{j}}\mathcal{L}^{\left(2\right)}(\theta) & =\frac{2}{M}\sum_{m=1}^{M}\left(\Tr\!\left[g_{T}(H(\theta))\rho_{m}\right]-y_{m}\right)\frac{\partial}{\partial\theta_{j}}\Tr\!\left[g_{T}(H(\theta))\rho_{m}\right]\\
 & =\frac{2}{M}\sum_{m=1}^{M}\left(\Tr\!\left[\left(g_{T}(H(\theta))-y_{m}I\right)\rho_{m}\right]\right)\Tr\!\left[\frac{\partial}{\partial\theta_{j}}g_{T}(H(\theta))\rho_{m}\right]\\
 & =\frac{2}{M}\sum_{m=1}^{M}\left(\Tr\!\left[\left(\left(g_{T}(H(\theta))-y_{m}I\right)\otimes\frac{\partial}{\partial\theta_{j}}g_{T}(H(\theta))\right)\left(\rho_{m}\otimes\rho_{m}\right)\right]\right).
\end{align}
By substituting the expressions stated in Theorems~\ref{thm:gradient-obj-func}
and~\ref{thm:obj-func-formula}, we conclude the proof.
\end{proof}
We can then use the expression in~\eqref{eq:grad-squared-loss-app}
to develop the following hybrid quantum--classical algorithm for
estimating $\frac{\partial}{\partial\theta_{j}}\mathcal{L}^{\left(2\right)}(\theta)$:
\begin{lyxalgorithm}
\label{alg:grad-est-squared-loss}A hybrid quantum--classical algorithm
for estimating the $j$th partial derivative $\frac{\partial}{\partial\theta_{j}}\mathcal{L}^{\left(2\right)}(\theta)$
consists of the following steps:
\begin{enumerate}
\item Set $\ell\leftarrow1$, and set
\begin{equation}
L\leftarrow O\!\left(\left(\frac{\left\Vert \theta\right\Vert _{1}\left(H_{\max}+y_{\max}\right)H_{\max}}{T^{2}\varepsilon}\right)^{2}\ln\!\left(\frac{1}{\delta}\right)\right),
\end{equation}
where $H_{\max}\coloneqq\max_{j\in\left[J\right]}\left\Vert H_{j}\right\Vert $,
$y_{\max}\coloneqq\max_{m\in\left[M\right]}\left|y_{m}\right|$, $\varepsilon>0$
is the desired accuracy, and $\delta\in\left(0,1\right)$ is the desired
failure probability.
\item Sample $t_{1},t_{2}\sim\mu$, $s_{1},s_{2}\sim\upsilon$, $k\sim q$,
$\lambda\sim\upsilon$, and $m\sim\left[M\right]$, where the probability
densities $\mu$ and $\upsilon$ are defined in Theorem~\ref{thm:gradient-obj-func},
$q$ is defined in Theorem~\ref{thm:obj-func-formula}, and $m\sim\left[M\right]$
indicates that $m$ is selected uniformly at random from $\left[M\right]$.
\item Prepare the states $\mathcal{U}_{s_{1}t_{1}/T}^{H(k,\lambda)}(\rho_{m})$
and $\mathcal{U}_{s_{2}t_{2}/T}^{H(\theta)}(\rho_{m})$ using two
samples of $\rho_{m}$ and Hamiltonian simulation to realize the unitary
channels $\mathcal{U}_{s_{1}t_{1}/T}^{H(k,\lambda)}$ and $\mathcal{U}_{s_{2}t_{2}/T}^{H(\theta)}$.
\item Perform the quantum circuit depicted in Figure~\ref{fig:Quantum-circuits-grad-sq-loss},
with measurement outcomes $Z_{\ell}^{(1)},Z_{\ell}^{(2)}\in\left\{ -1,1\right\} $
for the $\sigma_{Z}$ measurements, $X_{\ell}^{(1)}\in\spec(H_{k})$
for the $H_{k}$ measurement, and $X_{\ell}^{(2)}\in\spec(H_{j})$
for the $H_{j}$ measurement. Set
\begin{equation}
Y_{\ell}\leftarrow\frac{2\left\Vert \theta\right\Vert _{1}}{T^{2}}\signum(\theta_{k})\left(X_{\ell}^{(1)}-y_{m}\right)\cdot X_{\ell}^{(2)}\cdot Z_{\ell}^{(1)}\cdot Z_{\ell}^{(2)}.
\end{equation}
Set $\ell\leftarrow\ell+1.$
\item Repeat Steps 2-4 $L-1$ more times. Compute the average $\overline{Y_{L}}\coloneqq\frac{1}{L}\sum_{\ell=1}^{L}Y_{\ell}$
and output this value as an estimate of $\frac{\partial}{\partial\theta_{j}}\mathcal{L}(\theta)$.
\end{enumerate}
\end{lyxalgorithm}

Figure~\ref{fig:Quantum-circuits-grad-sq-loss} depicts the quantum
circuit used in Algorithm~\ref{alg:grad-est-squared-loss}. By the
Hoeffding inequality, we are guaranteed that
\begin{equation}
\Pr\!\left[\left|\overline{Y_{L}}-\frac{\partial}{\partial\theta_{j}}\mathcal{L}^{\left(2\right)}(\theta)\right|\leq\varepsilon\right]\geq1-\delta.
\end{equation}

\begin{figure}
\begin{centering}
\includegraphics[width=4.5in]{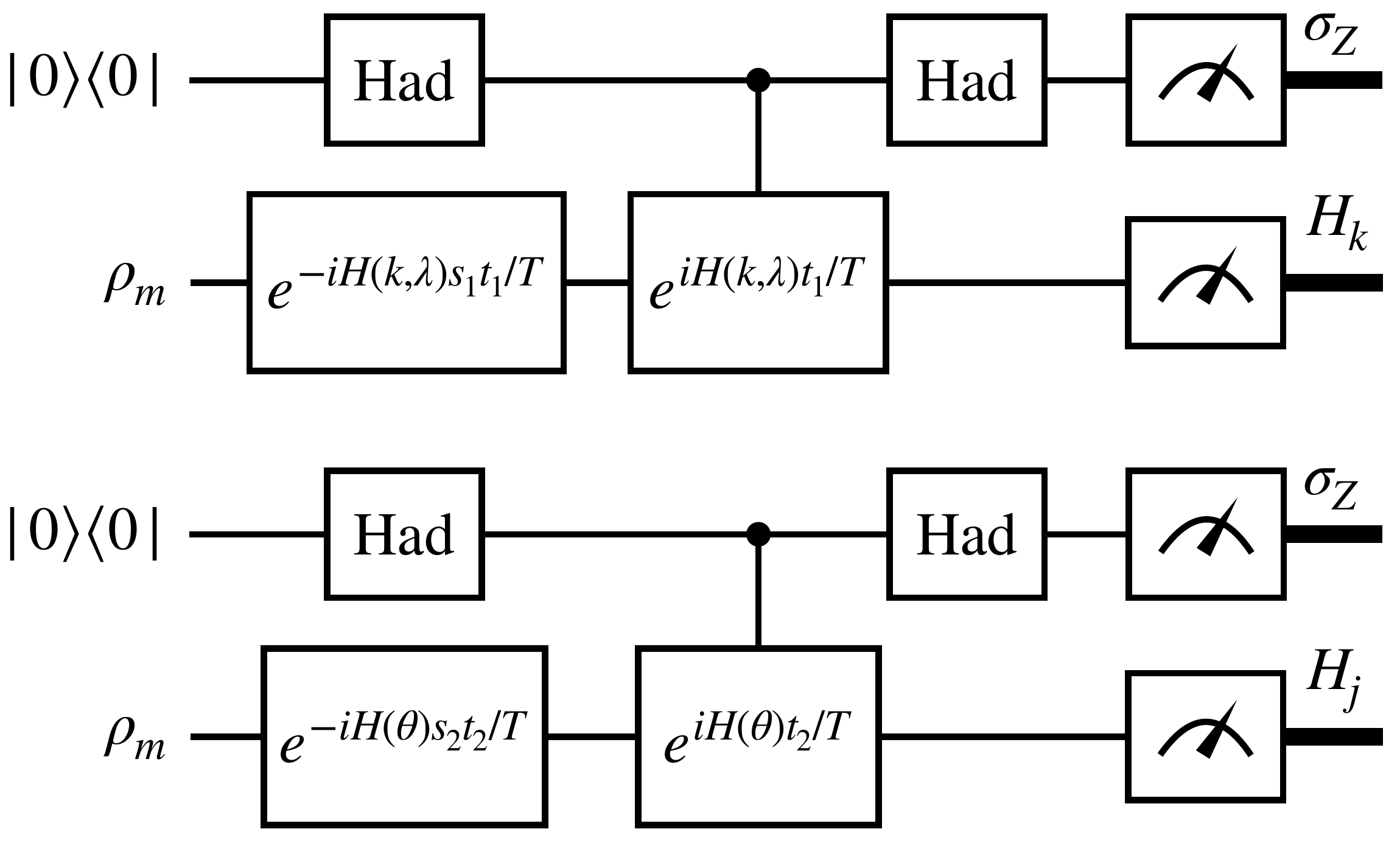}
\par\end{centering}
\caption{Quantum circuit used in Algorithm~\ref{alg:grad-est-squared-loss}
for estimating the $j$th partial derivative $\frac{\partial}{\partial\theta_{j}}\mathcal{L}^{\left(2\right)}(\theta)$.
The notation $H(k,\lambda)$ is defined in Theorem~\ref{thm:obj-func-formula}.
The circuit involves random classical sampling, Hamiltonian simulation,
and the Hadamard test.}\label{fig:Quantum-circuits-grad-sq-loss}
\end{figure}

\section{Logistic-loss function for binary classification}

\label{app:Logistic-loss-function}

\subsection{Derivative of matrix logistic-loss function}

We begin by deriving a novel formula for the derivative of the matrix
logistic-loss function $A\mapsto\ln\!\left(1+e^{-A}\right)$, where
$A$ is a Hermitian matrix.
\begin{lem}
Let $x\in\mathbb{R}$, and let $x\mapsto A(x)$ be a Hermitian-valued
function. Then the following equality holds:
\begin{equation}
\frac{\partial}{\partial x}\ln\!\left(I+e^{-A(x)}\right)=-\frac{1}{2}\frac{\partial}{\partial x}A(x)+\frac{1}{2}\mathbb{E}_{s\sim\upsilon,t\sim\gamma}\!\left[s\Re\!\left[A(x)e^{-iA(x)st}\left(\frac{\partial}{\partial x}A(x)\right)e^{-iA(x)\left(1-s\right)t}\right]\right],
\end{equation}
where the high-peak-tent probability density $\gamma(t)$ is defined
as
\begin{equation}
\gamma(t)\coloneqq\frac{2}{\pi}\ln\!\left|\coth\!\left(\frac{\pi t}{2}\right)\right|.\label{eq:high-peak-tent-gamma}
\end{equation}
\end{lem}

\begin{proof}
To begin with, let
\begin{equation}
A(x)=\sum_{k}\lambda_{k}\Pi_{k}
\end{equation}
be a spectral decomposition of $A(x)$, where we have omitted the
dependence on $x$ in the eigenvalues and eigenprojections. Defining
\begin{equation}
\text{logloss}(y)\coloneqq\ln\!\left(1+e^{-y}\right),
\end{equation}
consider that the derivative of the matrix logistic-loss function
is given by
\begin{equation}
\frac{\partial}{\partial x}\ln\!\left(I+e^{-A(x)}\right)=\sum_{k,\ell}f_{\text{logloss}}^{\left[1\right]}(\lambda_{k},\lambda_{\ell})\Pi_{k}\left(\frac{\partial}{\partial x}A(x)\right)\Pi_{\ell},\label{eq:1st-divided-diff-tanh-1}
\end{equation}
where $f_{\text{logloss}}^{\left[1\right]}$ is the first divided
difference of the logistic-loss function, defined as
\begin{equation}
f_{\text{logloss}}^{\left[1\right]}(y_{1},y_{2})\coloneqq\begin{cases}
\frac{\text{logloss}(y_{1})-\text{logloss}(y_{2})}{y_{1}-y_{2}} & :y_{1}\neq y_{2}\\
\frac{1}{2}\left(\tanh\!\left(\frac{y_{1}}{2}\right)-1\right) & :y_{1}=y_{2}
\end{cases},
\end{equation}
where we used the fact that
\begin{equation}
\frac{\partial}{\partial x}\text{logloss}(x)=-\frac{e^{-x}}{1+e^{-x}}=-\left(1+e^{x}\right)^{-1}=\frac{1}{2}\left(-1+\tanh\!\left(\frac{x}{2}\right)\right).
\end{equation}
(see, e.g., \cite[Theorem~42]{wilde2025} for a proof of~\eqref{eq:1st-divided-diff-tanh-1}).
Now observe that an alternative expression for $f_{\text{logloss}}^{\left[1\right]}$
is as follows:
\begin{align}
f_{\text{logloss}}^{\left[1\right]}(y_{1},y_{2}) & =\int_{0}^{1}ds\,\frac{1}{2}\left(-1+\tanh\!\left(\frac{sy_{1}+\left(1-s\right)y_{2}}{2}\right)\right)\\
 & =-\frac{1}{2}+\frac{1}{2}\int_{0}^{1}ds\,\tanh\!\left(\frac{sy_{1}+\left(1-s\right)y_{2}}{2}\right),
\end{align}
as a consequence of the fundamental theorem of calculus and that
\begin{equation}
\frac{d}{ds}\text{logloss}(sy_{1}+\left(1-s\right)y_{2})=\left(-\frac{1}{2}+\frac{1}{2}\int_{0}^{1}ds\,\tanh\!\left(\frac{sy_{1}+\left(1-s\right)y_{2}}{2}\right)\right)\left(y_{1}-y_{2}\right).
\end{equation}
Then it follows that
\begin{align}
\frac{\partial}{\partial x}\tanh(A(x)) & =\sum_{k,\ell}\left[-\frac{1}{2}+\frac{1}{2}\int_{0}^{1}ds\,\tanh\!\left(\frac{sy_{1}+\left(1-s\right)y_{2}}{2}\right)\right]\Pi_{k}\left(\frac{\partial}{\partial x}A(x)\right)\Pi_{\ell}.\\
 & =-\frac{1}{2}\sum_{k,\ell}\Pi_{k}\left(\frac{\partial}{\partial x}A(x)\right)\Pi_{\ell}\nonumber \\
 & \qquad+\frac{1}{2}\sum_{k,\ell}\int_{0}^{1}ds\,\tanh\!\left(\frac{s\lambda_{k}+\left(1-s\right)\lambda_{\ell}}{2}\right)\Pi_{k}\left(\frac{\partial}{\partial x}A(x)\right)\Pi_{\ell}\\
 & =-\frac{1}{2}\frac{\partial}{\partial x}A(x)\nonumber \\
 & \qquad+\frac{1}{2}\sum_{k,\ell}\int_{0}^{1}ds\,\left(\frac{s\lambda_{k}+\left(1-s\right)\lambda_{\ell}}{2}\right)\frac{\tanh\!\left(\frac{s\lambda_{k}+\left(1-s\right)\lambda_{\ell}}{2}\right)}{\frac{s\lambda_{k}+\left(1-s\right)\lambda_{\ell}}{2}}\Pi_{k}\left(\frac{\partial}{\partial x}A(x)\right)\Pi_{\ell}.\label{eq:proof-logistic-mid}
\end{align}
Now consider that the Fourier transform of $\omega\mapsto\frac{\tanh\left(\omega/2\right)}{\omega/2}$
is as follows~\cite{Ejima2019,Anshu2021,Patel2025a}:
\begin{equation}
\int_{-\infty}^{\infty}dt\,\gamma(t)e^{-i\omega t}=\frac{\tanh\!\left(\omega/2\right)}{\omega/2}.\label{eq:fourier-trans-high-peak-tent}
\end{equation}
Substituting~\eqref{eq:fourier-trans-high-peak-tent} into the second
term of~\eqref{eq:proof-logistic-mid}, we find that
\begin{align}
 & \frac{1}{2}\sum_{k,\ell}\int_{0}^{1}ds\,\left(\frac{s\lambda_{k}+\left(1-s\right)\lambda_{\ell}}{2}\right)\frac{\tanh\!\left(\frac{s\lambda_{k}+\left(1-s\right)\lambda_{\ell}}{2}\right)}{\frac{s\lambda_{k}+\left(1-s\right)\lambda_{\ell}}{2}}\Pi_{k}\left(\frac{\partial}{\partial x}A(x)\right)\Pi_{\ell}\nonumber \\
 & =\frac{1}{4}\sum_{k,\ell}\int_{0}^{1}ds\,\left(s\lambda_{k}+\left(1-s\right)\lambda_{\ell}\right)\int_{-\infty}^{\infty}dt\,\gamma(t)e^{-i\left(s\lambda_{k}+\left(1-s\right)\lambda_{\ell}\right)t}\Pi_{k}\left(\frac{\partial}{\partial x}A(x)\right)\Pi_{\ell}\\
 & =\frac{1}{4}\int_{0}^{1}ds\,\int_{-\infty}^{\infty}dt\,\gamma(t)\sum_{k,\ell}\left(s\lambda_{k}+\left(1-s\right)\lambda_{\ell}\right)e^{-i\left(s\lambda_{k}+\left(1-s\right)\lambda_{\ell}\right)t}\Pi_{k}\left(\frac{\partial}{\partial x}A(x)\right)\Pi_{\ell}\\
 & =\frac{1}{4}\int_{0}^{1}ds\,\int_{-\infty}^{\infty}dt\,\gamma(t)\sum_{k,\ell}\left(s\lambda_{k}+\left(1-s\right)\lambda_{\ell}\right)e^{-is\lambda_{k}t}\Pi_{k}\left(\frac{\partial}{\partial x}A(x)\right)e^{-i\left(1-s\right)\lambda_{\ell}t}\Pi_{\ell}\\
 & =\frac{1}{4}\int_{0}^{1}ds\,s\int_{-\infty}^{\infty}dt\,\gamma(t)\sum_{k,\ell}\lambda_{k}e^{-is\lambda_{k}t}\Pi_{k}\left(\frac{\partial}{\partial x}A(x)\right)e^{-i\left(1-s\right)\lambda_{\ell}t}\Pi_{\ell}\nonumber \\
 & \qquad+\frac{1}{4}\int_{0}^{1}ds\,\left(1-s\right)\int_{-\infty}^{\infty}dt\,\gamma(t)\sum_{k,\ell}e^{-is\lambda_{k}t}\Pi_{k}\left(\frac{\partial}{\partial x}A(x)\right)\lambda_{\ell}e^{-i\left(1-s\right)\lambda_{\ell}t}\Pi_{\ell}\\
 & =\frac{1}{4}\int_{0}^{1}ds\,s\int_{-\infty}^{\infty}dt\,\gamma(t)\left(\sum_{k}\lambda_{k}e^{-is\lambda_{k}t}\Pi_{k}\right)\left(\frac{\partial}{\partial x}A(x)\right)\left(\sum_{\ell}e^{-i\left(1-s\right)\lambda_{\ell}t}\Pi_{\ell}\right)\nonumber \\
 & \qquad+\frac{1}{4}\int_{0}^{1}ds\,\left(1-s\right)\int_{-\infty}^{\infty}dt\,\gamma(t)\left(\sum_{k}e^{-is\lambda_{k}t}\Pi_{k}\right)\left(\frac{\partial}{\partial x}A(x)\right)\left(\sum_{\ell}\lambda_{\ell}e^{-i\left(1-s\right)\lambda_{\ell}t}\Pi_{\ell}\right)\\
 & =\frac{1}{4}\int_{0}^{1}ds\,s\int_{-\infty}^{\infty}dt\,\gamma(t)A(x)e^{-iA(x)st}\left(\frac{\partial}{\partial x}A(x)\right)e^{-iA(x)\left(1-s\right)t}\nonumber \\
 & \qquad+\frac{1}{4}\int_{0}^{1}ds\,\left(1-s\right)\int_{-\infty}^{\infty}dt\,\gamma(t)e^{-iA(x)st}\left(\frac{\partial}{\partial x}A(x)\right)e^{-iA(x)\left(1-s\right)t}A(x)\\
 & \overset{(a)}{=}\frac{1}{4}\int_{0}^{1}ds\,s\int_{-\infty}^{\infty}dt\,\gamma(t)A(x)e^{-iA(x)st}\left(\frac{\partial}{\partial x}A(x)\right)e^{-iA(x)\left(1-s\right)t}\\
 & \qquad+\frac{1}{4}\int_{0}^{1}ds\,s\int_{-\infty}^{\infty}dt\,\gamma(t)e^{-iA(x)\left(1-s\right)t}\left(\frac{\partial}{\partial x}A(x)\right)e^{-iA(x)st}A(x)\\
 & \overset{(b)}{=}\frac{1}{4}\int_{0}^{1}ds\,s\int_{-\infty}^{\infty}dt\,\gamma(t)A(x)e^{-iA(x)st}\left(\frac{\partial}{\partial x}A(x)\right)e^{-iA(x)\left(1-s\right)t}\\
 & \qquad+\frac{1}{4}\int_{0}^{1}ds\,s\int_{-\infty}^{\infty}dt\,\gamma(t)e^{iA(x)\left(1-s\right)t}\left(\frac{\partial}{\partial x}A(x)\right)e^{iA(x)st}A(x)\\
 & =\frac{1}{2}\int_{0}^{1}ds\,s\int_{-\infty}^{\infty}dt\,\gamma(t)\Re\!\left[A(x)e^{-iA(x)st}\left(\frac{\partial}{\partial x}A(x)\right)e^{-iA(x)\left(1-s\right)t}\right].
\end{align}
The equality (a) follows from the substitution $s\to1-s$, and the
equality (b) follows from the substitution $t\to-t$, given that $\gamma(t)$
is an even function. The final equality follows from the definition
$\Re[B]=\left(B+B^{\dag}\right)/2$.
\end{proof}
\begin{rem}
For Hermitian matrices $A$ and $H$, we can write the Fr\'echet
derivative~\cite{Coleman2012} of $\ln\!\left(I+e^{-A}\right)$ at
$H$ as follows:
\begin{equation}
D\ln\!\left(I+e^{-A}\right)[H]=-\frac{1}{2}H+\frac{1}{2}\mathbb{E}_{s\sim\upsilon,t\sim\gamma}\!\left[s\Re\!\left[Ae^{-iAst}He^{-iA\left(1-s\right)t}\right]\right].
\end{equation}
\end{rem}

\begin{thm*}[Restatement of Theorem~\ref{thm:grad-logloss}]
The following equality holds:
\begin{multline}
\frac{\partial}{\partial\theta_{j}}T\Tr\!\left[\ln\!\left(I+e^{-y_{m}H(\theta)/T}\right)\rho_{m}\right]=-\frac{y_{m}}{2}\Tr\!\left[H_{j}\rho_{m}\right]+\\
\frac{1}{2T}\left\Vert \theta\right\Vert _{1}\mathbb{E}_{\substack{s\sim\upsilon,\\
k\sim q,\\
t\sim\gamma
}
}\!\left[s\Re\!\left[\Tr\!\left[\signum(\theta_{k})H_{k}H_{j}e^{iy_{m}H(\theta)t/T}\mathcal{U}_{st/T}^{y_{m}H(\theta)}\!\left(\rho_{m}\right)\right]\right]\right].\label{eq:partial-deriv-logloss-1}
\end{multline}
where $y_{m}\in\left\{ -1,+1\right\} $, $H(\theta)$ is a parameterized
Hamiltonian of the form in~\eqref{eq:general-k-local-ham}, $T>0$,
$\gamma(t)$ is the probability density function defined in~\eqref{eq:high-peak-tent-gamma},
$\upsilon$ is a uniform random variable on the unit interval $\left[0,1\right]$,
$q(k)\coloneqq\frac{\left|\theta_{k}\right|}{\left\Vert \theta\right\Vert _{1}}$
is a probability distribution, and $\mathcal{U}_{t}^{y_{m}H(\theta)}$
is the following unitary quantum channel:
\begin{equation}
\mathcal{U}_{t}^{y_{m}H(\theta)}(\cdot)\coloneqq e^{-iy_{m}H(\theta)t}(\cdot)e^{iy_{m}H(\theta)t}.
\end{equation}
\end{thm*}
\begin{proof}
Consider that
\begin{align}
 & \frac{\partial}{\partial\theta_{j}}T\Tr\!\left[\ln\!\left(I+e^{-y_{m}H(\theta)/T}\right)\rho_{m}\right]\nonumber \\
 & =T\Tr\!\left[\left(\frac{\partial}{\partial\theta_{j}}\ln\!\left(I+e^{-y_{m}H(\theta)/T}\right)\right)\rho_{m}\right]\\
 & =T\Tr\!\left[\left(\begin{array}{c}
-\frac{1}{2}\frac{\partial}{\partial\theta_{j}}\frac{y_{m}H(\theta)}{T}+\\
\frac{1}{2}\mathbb{E}_{s\sim\upsilon,t\sim\gamma}\left[s\Re\!\left[\frac{y_{m}H(\theta)}{T}e^{-iy_{m}H(\theta)st/T}\left(\frac{\partial}{\partial\theta_{j}}\frac{y_{m}H(\theta)}{T}\right)e^{-iy_{m}H(\theta)\left(1-s\right)t/T}\right]\right]
\end{array}\right)\rho_{m}\right]\\
 & =\Tr\!\left[\left(-\frac{y_{m}}{2}H_{j}+\frac{1}{2}\mathbb{E}_{s\sim\upsilon,t\sim\gamma}\left[s\Re\!\left[\frac{H(\theta)}{T}e^{-iy_{m}H(\theta)st/T}H_{j}e^{-iy_{m}H(\theta)\left(1-s\right)t/T}\right]\right]\right)\rho_{m}\right]\\
 & =-\frac{y_{m}}{2}\Tr\!\left[H_{j}\rho_{m}\right]\nonumber \\
 & \qquad+\frac{1}{2T}\mathbb{E}_{s\sim\upsilon,t\sim\gamma}\left[s\Tr\!\left[\Re\!\left[H(\theta)e^{-iy_{m}H(\theta)st/T}H_{j}e^{-iy_{m}H(\theta)\left(1-s\right)t/T}\right]\rho_{m}\right]\right].
\end{align}
Now consider that
\begin{align}
 & \mathbb{E}_{s\sim\upsilon,t\sim\gamma}\!\left[s\Tr\!\left[\Re\!\left[H(\theta)e^{-iy_{m}H(\theta)st/T}H_{j}e^{-iy_{m}H(\theta)\left(1-s\right)t/T}\right]\rho_{m}\right]\right]\nonumber \\
 & =\mathbb{E}_{s\sim\upsilon,t\sim\gamma}\!\left[s\Re\!\left[\Tr\!\left[H(\theta)e^{-iy_{m}H(\theta)st/T}H_{j}e^{-iy_{m}H(\theta)\left(1-s\right)t/T}\rho_{m}\right]\right]\right]\\
 & =\mathbb{E}_{s\sim\upsilon,t\sim\gamma}\!\left[s\Re\!\left[\Tr\!\left[H(\theta)H_{j}e^{-iy_{m}H(\theta)\left(1-s\right)t/T}\rho_{m}e^{-iy_{m}H(\theta)st/T}\right]\right]\right]\\
 & =\mathbb{E}_{s\sim\upsilon,t\sim\gamma}\!\left[s\Re\!\left[\Tr\!\left[H(\theta)H_{j}e^{iy_{m}H(\theta)\left(1-s\right)t/T}\rho_{m}e^{iy_{m}H(\theta)st/T}\right]\right]\right]\\
 & =\mathbb{E}_{s\sim\upsilon,t\sim\gamma}\!\left[s\Re\!\left[\Tr\!\left[H(\theta)H_{j}e^{iy_{m}H(\theta)t/T}e^{-iy_{m}H(\theta)st/T}\rho_{m}e^{iy_{m}H(\theta)st/T}\right]\right]\right]\\
 & =\mathbb{E}_{s\sim\upsilon,t\sim\gamma}\!\left[s\Re\!\left[\Tr\!\left[H(\theta)H_{j}e^{iy_{m}H(\theta)t/T}\mathcal{U}_{st/T}^{y_{m}H(\theta)}\!\left(\rho_{m}\right)\right]\right]\right]\\
 & =\mathbb{E}_{s\sim\upsilon,t\sim\gamma}\!\left[s\Re\!\left[\Tr\!\left[\sum_{k\in\left[J\right]}\theta_{k}H_{k}H_{j}e^{iy_{m}H(\theta)t/T}\mathcal{U}_{st/T}^{y_{m}H(\theta)}\!\left(\rho_{m}\right)\right]\right]\right]\\
 & =\left\Vert \theta\right\Vert _{1}\mathbb{E}_{s\sim\upsilon,t\sim\gamma}\!\left[s\Re\!\left[\Tr\!\left[\sum_{k\in\left[J\right]}\signum(\theta_{k})\frac{\left|\theta_{k}\right|}{\left\Vert \theta\right\Vert _{1}}H_{k}H_{j}e^{iy_{m}H(\theta)t/T}\mathcal{U}_{st/T}^{y_{m}H(\theta)}\!\left(\rho_{m}\right)\right]\right]\right]\\
 & =\left\Vert \theta\right\Vert _{1}\mathbb{E}_{s\sim\upsilon,t\sim\gamma,k\sim q}\!\left[s\Re\!\left[\Tr\!\left[\signum(\theta_{k})H_{k}H_{j}e^{iy_{m}H(\theta)t/T}\mathcal{U}_{st/T}^{y_{m}H(\theta)}\!\left(\rho_{m}\right)\right]\right]\right],
\end{align}
thus concluding the proof.
\end{proof}

\subsection{Hybrid quantum--classical algorithm for estimating derivative of
matrix logistic-loss function}

The first term of~\eqref{eq:partial-deriv-logloss} can be easily
estimated by preparing the state $\rho$ and measuring $H_{j}$. Under
the assumption that each $H_{j}$ in $H(\theta)$ is both Hermitian
and unitary (as in the common case when each $H_{j}$ is a Pauli string),
Algorithm~\ref{alg:grad-log-loss} below provides a method for estimating
the second term of~\eqref{eq:partial-deriv-logloss}, which we abbreviate
as follows:
\begin{equation}
\zeta_{j}\equiv\frac{1}{2T}\left\Vert \theta\right\Vert _{1}\mathbb{E}_{s\sim\upsilon,k\sim q,t\sim\gamma}\!\left[s\Re\!\left[\Tr\!\left[\signum(\theta_{k})H_{k}H_{j}e^{iy_{m}H(\theta)t/T}\mathcal{U}_{st/T}^{y_{m}H(\theta)}\!\left(\rho_{m}\right)\right]\right]\right].\label{eq:shorthand-2nd-term-grad-logloss}
\end{equation}

\begin{lyxalgorithm}
\label{alg:grad-log-loss}A hybrid quantum--classical algorithm for
estimating $\zeta_{j}$ in~\eqref{eq:shorthand-2nd-term-grad-logloss}
consists of the following steps:
\begin{enumerate}
\item Set $m\leftarrow1$, and set
\begin{equation}
M\leftarrow O\!\left(\left(\frac{\left\Vert \theta\right\Vert _{1}\max_{j\in\left[J\right]}\left\Vert H_{j}\right\Vert }{T\varepsilon}\right)^{2}\ln\!\left(\frac{1}{\delta}\right)\right),
\end{equation}
where $\varepsilon>0$ is the desired accuracy and $\delta\in\left(0,1\right)$
is the desired failure probability.
\item Sample $s\sim\upsilon$, $k\sim q$, and $t\sim\gamma$.
\item Prepare the state $\mathcal{U}_{st/T}^{y_{m}H(\theta)}\!\left(\rho_{m}\right)$
using one sample of $\rho_{m}$ and Hamiltonian simulation to realize
the unitary channel $\mathcal{U}_{st/T}^{y_{m}H(\theta)}$.
\item Perform the quantum circuit depicted in Figure~\ref{fig:Quantum-circuits-grad-log-loss},
with measurement outcomes $Z_{m}\in\left\{ -1,1\right\} $ for the
$\sigma_{Z}$ measurement and $X_{m}\in\spec(H_{k})$ for the $H_{k}$
measurement. Set $W_{m}\leftarrow\frac{\left\Vert \theta\right\Vert _{1}}{2T}s\cdot\signum(\theta_{k})Z_{m}\cdot X_{m}$.
Set $m\leftarrow m+1.$
\item Repeat Steps 2-4 $M-1$ more times. Compute the average $\overline{W_{M}}\coloneqq\frac{1}{M}\sum_{m=1}^{M}W_{m}$
and output this value as an estimate of $\zeta_{j}$.
\end{enumerate}
\end{lyxalgorithm}

By the Hoeffding inequality, we are guaranteed that
\begin{equation}
\Pr\!\left[\left|\overline{W_{M}}-\zeta_{j}\right|\leq\varepsilon\right]\geq1-\delta.
\end{equation}

\begin{figure}
\begin{centering}
\includegraphics[width=4.5in]{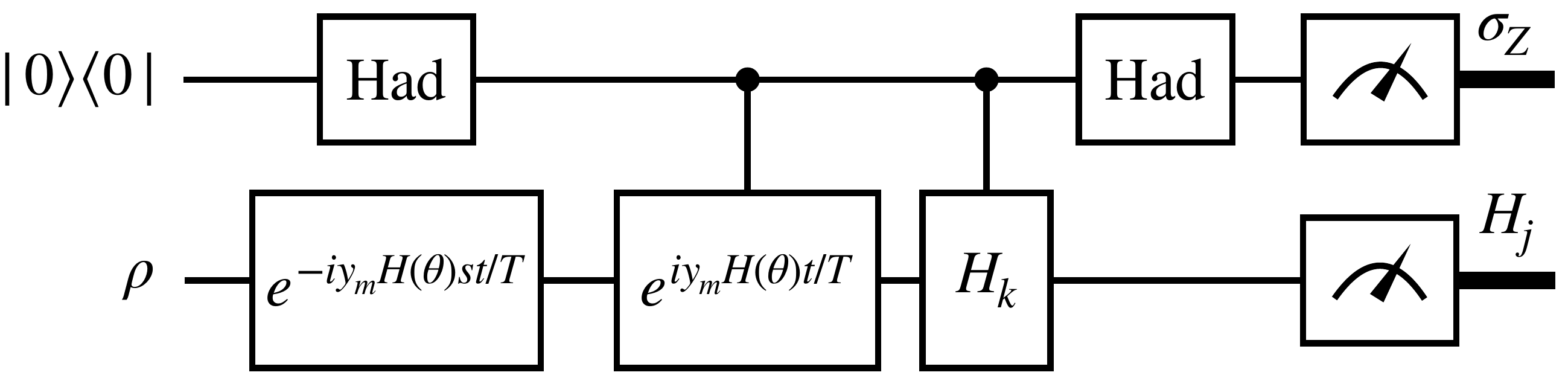}
\par\end{centering}
\caption{Quantum circuit used in Algorithm~\ref{alg:grad-log-loss} for estimating
the second term of~\eqref{eq:partial-deriv-logloss}, denoted by $\zeta_{j}$
in~\eqref{eq:shorthand-2nd-term-grad-logloss}.}\label{fig:Quantum-circuits-grad-log-loss}
\end{figure}

\subsection{Alternative formula for logistic-loss function}

Similar to the idea behind Theorem~\ref{thm:obj-func-formula}, we
can use Theorem~\ref{thm:grad-logloss} and the fundamental theorem
of calculus to derive an expression for the logistic-loss function
that can be evaluated by a hybrid quantum--classical algorithm.
\begin{thm}
\label{thm:log-loss-obj-func-formula-FTOC}The following equality
holds:
\begin{multline}
T\Tr\!\left[\ln\!\left(I+e^{-y_{m}H(\theta)/T}\right)\rho_{m}\right]=-\frac{y_{m}}{2}\Tr\!\left[H(\theta)\rho_{m}\right]+\\
\frac{J}{2T}\mathbb{E}_{\substack{j\sim\left[J\right],\\
\lambda,s\sim\upsilon,\\
t\sim\gamma
}
}\!\left[\left\Vert \theta^{\left(j\right)}(\lambda)\right\Vert _{1}\mathbb{E}_{k\sim q_{j,\lambda}}\!\left[s\Re\!\left[\Tr\!\left[\signum(\theta_{k})H_{k}H_{j}e^{iy_{m}H(\theta^{\left(j\right)}(\lambda))t/T}\mathcal{U}_{st/T}^{y_{m}H(\theta^{\left(j\right)}(\lambda))}\!\left(\rho_{m}\right)\right]\right]\right]\right],\label{eq:log-loss-FTOC}
\end{multline}
where $j$ is selected uniformly at random from $\left[J\right]$,
$\lambda$ and $s$ are selected uniformly at random from the unit
interval $\left[0,1\right]$, $t$ is selected according to $\gamma(t)$
in~\eqref{eq:high-peak-tent-gamma},
\begin{align}
\theta^{\left(j\right)}(\lambda) & \coloneqq\left(0,\ldots,0,\lambda\theta_{j},\theta_{j+1},\ldots,\theta_{J}\right),\\
\left\Vert \theta^{\left(j\right)}(\lambda)\right\Vert _{1} & \coloneqq\lambda\left|\theta_{j}\right|+\sum_{k=j+1}^{J}\left|\theta_{k}\right|,\\
H(\theta_{1},\ldots,\theta_{J}) & \equiv H(\theta)=\sum_{j=1}^{J}\theta_{j}H_{j},
\end{align}
and $k$ is selected according to the conditional probability distribution
$q_{j,\lambda}$, defined as
\begin{equation}
q_{j,\lambda}(k)\coloneqq\begin{cases}
\frac{\lambda\left|\theta_{j}\right|}{\left\Vert \theta^{\left(j\right)}(\lambda)\right\Vert _{1}} & :k=j\\
\frac{\left|\theta_{k}\right|}{\left\Vert \theta^{\left(j\right)}(\lambda)\right\Vert _{1}} & :k>j
\end{cases}.
\end{equation}
\end{thm}

\begin{proof}
Define
\begin{equation}
\mathcal{L}_{T}(y_{m}H(\theta))\coloneqq T\ln\!\left(I+e^{-y_{m}H(\theta)/T}\right).
\end{equation}
Consider that
\begin{align}
 & \Tr\!\left[\mathcal{L}_{T}(y_{m}H(\theta))\rho_{m}\right]-\Tr\!\left[\mathcal{L}_{T}(y_{m}H(0,\theta_{2},\ldots,\theta_{J}))\rho_{m}\right]\nonumber \\
 & =\Tr\!\left[\mathcal{L}_{T}(y_{m}H(\theta_{1},\theta_{2},\ldots,\theta_{J}))\rho_{m}\right]-\Tr\!\left[\mathcal{L}_{T}(y_{m}H(0,\theta_{2},\ldots,\theta_{J}))\rho_{m}\right]\\
 & =\Tr\!\left[\mathcal{L}_{T}(y_{m}H(\theta^{\left(1\right)}(1)))\rho_{m}\right]-\Tr\!\left[\mathcal{L}_{T}(y_{m}H(\theta^{\left(1\right)}(0)))\rho_{m}\right]\\
 & =\int_{0}^{1}d\lambda_{1}\frac{\partial}{\partial\lambda_{1}}\Tr\!\left[\mathcal{L}_{T}(y_{m}H(\theta^{\left(1\right)}(\lambda_{1})))\rho_{m}\right]\\
 & =\int_{0}^{1}d\lambda_{1}\,\left(\begin{array}{c}
-\frac{y_{m}\theta_{1}}{2}\Tr\!\left[H_{1}\rho_{m}\right]\\
+\mathbb{E}_{s\sim\upsilon,t\sim\gamma}\!\left[s\Re\!\left[\Tr\!\left[H(\theta^{\left(1\right)}(\lambda_{1}))H_{1}e^{iy_{m}H(\theta^{\left(1\right)}(\lambda_{1}))t/T}\mathcal{U}_{st/T}^{y_{m}H(\theta^{\left(1\right)}(\lambda_{1}))}\!\left(\rho_{m}\right)\right]\right]\right]
\end{array}\right)\\
 & =-\frac{y_{m}\theta_{1}}{2}\Tr\!\left[H_{1}\rho\right]\nonumber \\
 & \qquad+\frac{1}{2T}\mathbb{E}_{\lambda_{1},s\sim\upsilon,t\sim\gamma}\!\left[s\Re\!\left[\Tr\!\left[H(\theta^{\left(1\right)}(\lambda_{1}))H_{1}e^{iy_{m}H(\theta^{\left(1\right)}(\lambda_{1}))t/T}\mathcal{U}_{st/T}^{y_{m}H(\theta^{\left(1\right)}(\lambda_{1}))}\!\left(\rho_{m}\right)\right]\right]\right]\\
 & =-\frac{y_{m}\theta_{1}}{2}\Tr\!\left[H_{1}\rho\right]\nonumber \\
 & \qquad+\frac{1}{2T}\mathbb{E}_{\lambda,s\sim\upsilon,t\sim\gamma}\!\left[s\Re\!\left[\Tr\!\left[H(\theta^{\left(1\right)}(\lambda))H_{1}e^{iy_{m}H(\theta^{\left(1\right)}(\lambda))t/T}\mathcal{U}_{st/T}^{y_{m}H(\theta^{\left(1\right)}(\lambda))}\!\left(\rho_{m}\right)\right]\right]\right].
\end{align}
Additionally,
\begin{align}
 & \Tr\!\left[\mathcal{L}_{T}(y_{m}H(0,\theta_{2},\ldots,\theta_{J}))\rho_{m}\right]-\Tr\!\left[\mathcal{L}_{T}(y_{m}H(0,0,\theta_{3},\ldots,\theta_{J}))\rho_{m}\right]\nonumber \\
 & =\Tr\!\left[\mathcal{L}_{T}(y_{m}H(\theta^{\left(2\right)}(1)))\rho_{m}\right]-\Tr\!\left[\mathcal{L}_{T}(y_{m}H(\theta^{\left(2\right)}(0)))\rho_{m}\right]\\
 & =\int_{0}^{1}d\lambda_{2}\frac{\partial}{\partial\lambda_{2}}\Tr\!\left[\mathcal{L}_{T}(y_{m}H(\theta^{\left(2\right)}(\lambda_{2})))\rho_{m}\right]\\
 & =\int_{0}^{1}d\lambda_{1}\,\left(\begin{array}{c}
-\frac{y_{m}\theta_{2}}{2}\Tr\!\left[H_{2}\rho_{m}\right]\\
+\mathbb{E}_{s\sim\upsilon,t\sim\gamma}\!\left[s\Re\!\left[\Tr\!\left[H(\theta^{\left(2\right)}(\lambda_{2}))H_{2}e^{iy_{m}H(\theta^{\left(2\right)}(\lambda_{2}))t/T}\mathcal{U}_{st/T}^{y_{m}H(\theta^{\left(2\right)}(\lambda_{2}))}\!\left(\rho_{m}\right)\right]\right]\right]
\end{array}\right)\\
 & =-\frac{y_{m}\theta_{2}}{2}\Tr\!\left[H_{2}\rho_{m}\right]\nonumber \\
 & \qquad+\frac{1}{2T}\mathbb{E}_{\lambda_{2},s\sim\upsilon,t\sim\gamma}\!\left[s\Re\!\left[\Tr\!\left[H(\theta^{\left(2\right)}(\lambda_{2}))H_{2}e^{iy_{m}H(\theta^{\left(2\right)}(\lambda_{2}))t/T}\mathcal{U}_{st/T}^{y_{m}H(\theta^{\left(2\right)}(\lambda_{2}))}\!\left(\rho_{m}\right)\right]\right]\right]\\
 & =-\frac{y_{m}\theta_{2}}{2}\Tr\!\left[H_{2}\rho_{m}\right]\nonumber \\
 & \qquad+\frac{1}{2T}\mathbb{E}_{\lambda,s\sim\upsilon,t\sim\gamma}\!\left[s\Re\!\left[\Tr\!\left[H(\theta^{\left(2\right)}(\lambda))H_{2}e^{iy_{m}H(\theta^{\left(2\right)}(\lambda))t/T}\mathcal{U}_{st/T}^{y_{m}H(\theta^{\left(2\right)}(\lambda))}\!\left(\rho_{m}\right)\right]\right]\right].
\end{align}
We continue iteratively along these lines and find that the last term
is given by
\begin{align}
 & \Tr\!\left[\mathcal{L}_{T}(y_{m}H(0,\ldots,0,\theta_{J}))\rho_{m}\right]\nonumber \\
 & =\Tr\!\left[\mathcal{L}_{T}(y_{m}H(0,\ldots,0,\theta_{J}))\rho_{m}\right]-\Tr\!\left[\mathcal{L}_{T}(y_{m}H(0,\ldots,0,0))\rho_{m}\right]\\
 & =\Tr\!\left[\mathcal{L}_{T}(y_{m}H(\theta^{\left(J\right)}(1)))\rho_{m}\right]-\Tr\!\left[\mathcal{L}_{T}(y_{m}H(\theta^{\left(J\right)}(0)))\rho_{m}\right]\\
 & =-\frac{y_{m}\theta_{J}}{2}\Tr\!\left[H_{J}\rho\right]\\
 & \qquad+\frac{1}{2T}\mathbb{E}_{\lambda,s\sim\upsilon,t\sim\gamma}\!\left[s\Re\!\left[\Tr\!\left[H(\theta^{\left(J\right)}(\lambda))H_{J}e^{iy_{m}H(\theta^{\left(J\right)}(\lambda))t/T}\mathcal{U}_{st/T}^{y_{m}H(\theta^{\left(J\right)}(\lambda))}\!\left(\rho_{m}\right)\right]\right]\right].
\end{align}
where we used the fact that
\begin{equation}
\mathcal{L}_{T}(y_{m}H(\theta^{\left(J\right)}(0)))=\mathcal{L}_{T}(y_{m}H(0,\ldots,0))=\mathcal{L}_{T}(0)=0,
\end{equation}
given that $\left.T\ln\!\left(1+e^{-x/T}\right)\right|_{x=0}=0$.
So then we form a telescoping sum and conclude that
\begin{align}
 & \Tr\!\left[\mathcal{L}_{T}(y_{m}H(\theta))\rho_{m}\right]\nonumber \\
 & =\Tr\!\left[\mathcal{L}_{T}(y_{m}H(\theta))\rho_{m}\right]-\Tr\!\left[\mathcal{L}_{T}(y_{m}H(0,\theta_{2},\ldots,\theta_{J}))\rho_{m}\right]\nonumber \\
 & \qquad+\Tr\!\left[\mathcal{L}_{T}(y_{m}H(0,\theta_{2},\ldots,\theta_{J}))\rho_{m}\right]-\Tr\!\left[\mathcal{L}_{T}(y_{m}H(0,0,\theta_{3},\ldots,\theta_{J}))\rho_{m}\right]\nonumber \\
 & \qquad+\cdots+\Tr\!\left[\mathcal{L}_{T}(y_{m}H(0,\ldots,0,\theta_{J}))\rho_{m}\right]-\Tr\!\left[\mathcal{L}_{T}(y_{m}H(0,\ldots,0,0))\rho_{m}\right]\\
 & =-\frac{y_{m}\theta_{1}}{2}\Tr\!\left[H_{1}\rho_{m}\right]-\frac{y_{m}\theta_{2}}{2}\Tr\!\left[H_{2}\rho_{m}\right]-\ldots-\frac{y_{m}\theta_{J}}{2}\Tr\!\left[H_{J}\rho_{m}\right]\nonumber \\
 & \qquad+\frac{1}{2T}\mathbb{E}_{\lambda,s\sim\upsilon,t\sim\gamma}\!\left[s\Re\!\left[\Tr\!\left[H(\theta^{\left(1\right)}(\lambda))H_{1}e^{iy_{m}H(\theta^{\left(1\right)}(\lambda))t/T}\mathcal{U}_{st/T}^{y_{m}H(\theta^{\left(1\right)}(\lambda))}\!\left(\rho_{m}\right)\right]\right]\right]\nonumber \\
 & \qquad+\frac{1}{2T}\mathbb{E}_{\lambda,s\sim\upsilon,t\sim\gamma}\!\left[s\Re\!\left[\Tr\!\left[H(\theta^{\left(2\right)}(\lambda))H_{2}e^{iy_{m}H(\theta^{\left(2\right)}(\lambda))t/T}\mathcal{U}_{st/T}^{y_{m}H(\theta^{\left(2\right)}(\lambda))}\!\left(\rho_{m}\right)\right]\right]\right]\nonumber \\
 & \qquad+\cdots+\frac{1}{2T}\mathbb{E}_{\lambda,s\sim\upsilon,t\sim\gamma}\!\left[s\Re\!\left[\Tr\!\left[H(\theta^{\left(J\right)}(\lambda))H_{J}e^{iy_{m}H(\theta^{\left(J\right)}(\lambda))t/T}\mathcal{U}_{st/T}^{y_{m}H(\theta^{\left(J\right)}(\lambda))}\!\left(\rho_{m}\right)\right]\right]\right]\\
 & =-\frac{y_{m}}{2}\Tr\!\left[H(\theta)\rho_{m}\right]+\frac{1}{2T}\sum_{j=1}^{J}\mathbb{E}_{\lambda,s\sim\upsilon,t\sim\gamma}\!\left[s\Re\!\left[\Tr\!\left[H(\theta^{\left(j\right)}(\lambda))H_{j}e^{iy_{m}H(\theta^{\left(j\right)}(\lambda))t/T}\mathcal{U}_{st/T}^{y_{m}H(\theta^{\left(j\right)}(\lambda))}\!\left(\rho_{m}\right)\right]\right]\right].
\end{align}
Now observe that
\begin{align}
H(\theta^{\left(j\right)}(\lambda)) & =\lambda\theta_{j}H_{j}+\sum_{k=j+1}^{J}\theta_{k}H_{k}\\
 & =\left\Vert \theta^{\left(j\right)}(\lambda)\right\Vert _{1}\mathbb{E}_{k\sim q_{j,\lambda}}\left[\signum(\theta_{k})H_{k}\right],
\end{align}
where
\begin{align}
\left\Vert \theta^{\left(j\right)}(\lambda)\right\Vert _{1} & \coloneqq\lambda\left|\theta_{j}\right|+\sum_{k=j+1}^{J}\left|\theta_{k}\right|\\
q_{j,\lambda}(k) & \coloneqq\begin{cases}
\frac{\lambda\left|\theta_{j}\right|}{\left\Vert \theta^{\left(j\right)}(\lambda)\right\Vert _{1}} & :k=j\\
\frac{\left|\theta_{k}\right|}{\left\Vert \theta^{\left(j\right)}(\lambda)\right\Vert _{1}} & :k>j
\end{cases}.
\end{align}
Then
\begin{align}
 & \frac{1}{2T}\sum_{j=1}^{J}\mathbb{E}_{\lambda,s\sim\upsilon,t\sim\gamma}\!\left[s\Re\!\left[\Tr\!\left[H(\theta^{\left(j\right)}(\lambda))H_{j}e^{iy_{m}H(\theta^{\left(j\right)}(\lambda))t/T}\mathcal{U}_{st/T}^{y_{m}H(\theta^{\left(j\right)}(\lambda))}\!\left(\rho_{m}\right)\right]\right]\right]\nonumber \\
 & =\frac{1}{2T}\sum_{j=1}^{J}\mathbb{E}_{\substack{\lambda,s\sim\upsilon,\\
t\sim\gamma
}
}\!\left[\left\Vert \theta^{\left(j\right)}(\lambda)\right\Vert _{1}\mathbb{E}_{k\sim q_{j,\lambda}}\!\left[s\Re\!\left[\Tr\!\left[\signum(\theta_{k})H_{k}H_{j}e^{iy_{m}H(\theta^{\left(j\right)}(\lambda))t/T}\mathcal{U}_{st/T}^{y_{m}H(\theta^{\left(j\right)}(\lambda))}\!\left(\rho_{m}\right)\right]\right]\right]\right]\\
 & =\frac{J}{2T}\mathbb{E}_{\substack{j\sim\left[J\right],\\
\lambda,s\sim\upsilon,\\
t\sim\gamma
}
}\!\left[\left\Vert \theta^{\left(j\right)}(\lambda)\right\Vert _{1}\mathbb{E}_{k\sim q_{j,\lambda}}\!\left[s\Re\!\left[\Tr\!\left[\signum(\theta_{k})H_{k}H_{j}e^{iy_{m}H(\theta^{\left(j\right)}(\lambda))t/T}\mathcal{U}_{st/T}^{y_{m}H(\theta^{\left(j\right)}(\lambda))}\!\left(\rho_{m}\right)\right]\right]\right]\right],
\end{align}
thus completing the proof.
\end{proof}

\subsection{Hybrid quantum--classical algorithm for estimating logistic-loss
function}

\label{subsec:HQC-log-loss-obj-func}

The expression in~\eqref{eq:log-loss-FTOC} then leads to a hybrid
quantum--classical algorithm for estimating the logistic-loss objective
function. The first term $-\frac{y_{m}}{2}\Tr\!\left[H(\theta)\rho_{m}\right]$
in~\eqref{eq:log-loss-FTOC} can be easily estimated by writing
\begin{align}
-\frac{y_{m}}{2}\Tr\!\left[H(\theta)\rho_{m}\right] & =-\frac{y_{m}}{2}\Tr\!\left[\sum_{j=1}^{J}\theta_{j}H_{j}\rho_{m}\right]\label{eq:first-term-log-loss-1}\\
 & =-\frac{y_{m}}{2}\sum_{j=1}^{J}\theta_{j}\Tr\!\left[H_{j}\rho_{m}\right]\\
 & =-\frac{\left\Vert \theta\right\Vert _{1}y_{m}}{2}\sum_{j=1}^{J}\frac{\left|\theta_{j}\right|}{\left\Vert \theta\right\Vert _{1}}\signum(\theta_{j})\Tr\!\left[H_{j}\rho_{m}\right].\label{eq:first-term-log-loss-last}
\end{align}
Thus, in each iteration of an estimation algorithm, we can pick $j$
at random according to the probability distribution $\frac{\left|\theta_{j}\right|}{\left\Vert \theta\right\Vert _{1}}$,
prepare the state $\rho_{m}$, measure the observable $H_{j}$ on
this state (obtaining an outcome in $\left\{ -1,+1\right\} $), and
multiply the outcome by $-\frac{\left\Vert \theta\right\Vert _{1}y_{m}}{2}\signum(\theta_{j})$.
One then repeats this procedure a number of times and computes the
sample mean. The Hoeffding inequality guarantees that a desired accuracy
of $\varepsilon>0$ and a desired failure probability of $\delta\in\left(0,1\right)$
can be achieved using the following number of trials:
\begin{equation}
O\!\left(\left(\frac{\left\Vert \theta\right\Vert _{1}\max_{j\in\left[J\right]}\left\Vert H_{j}\right\Vert }{\varepsilon}\right)^{2}\ln\!\left(\frac{1}{\delta}\right)\right).
\end{equation}

The second term in~\eqref{eq:log-loss-FTOC}, which we abbreviate
as

\begin{equation}
\zeta\equiv\frac{J}{2T}\mathbb{E}_{\substack{j\sim\left[J\right],\\
\lambda,s\sim\upsilon,\\
t\sim\gamma
}
}\!\left[\left\Vert \theta^{\left(j\right)}(\lambda)\right\Vert _{1}\mathbb{E}_{k\sim q_{j,\lambda}}\!\left[s\Re\!\left[\Tr\!\left[\signum(\theta_{k})H_{k}H_{j}e^{iy_{m}H(\theta^{\left(j\right)}(\lambda))t/T}\mathcal{U}_{st/T}^{y_{m}H(\theta^{\left(j\right)}(\lambda))}\!\left(\rho_{m}\right)\right]\right]\right]\right],\label{eq:log-loss-2nd-term-abbr}
\end{equation}
can be estimated by the following algorithm:
\begin{lyxalgorithm}
\label{alg:log-loss-obj-func-est}A hybrid quantum--classical algorithm
for estimating $\zeta$ in~\eqref{eq:log-loss-2nd-term-abbr} consists
of the following steps:
\begin{enumerate}
\item Set $m\leftarrow1$, and set
\begin{align}
M & \leftarrow O\!\left(\left(\frac{J\theta^{\star}h^{\star}}{T\varepsilon}\right)^{2}\ln\!\left(\frac{1}{\delta}\right)\right),\\
\theta^{\star} & \coloneqq\max_{\lambda\in\left[0,1\right],j\in\left[J\right]}\left\Vert \theta^{\left(j\right)}(\lambda)\right\Vert _{1},\\
h^{\star} & \coloneqq\max_{j\in\left[J\right]}\left\Vert H_{j}\right\Vert ,
\end{align}
where $\varepsilon>0$ is the desired accuracy and $\delta\in\left(0,1\right)$
is the desired failure probability.
\item Sample $j$ according to the uniform distribution on $\left[J\right]$,
$s\sim\upsilon$, $\lambda\sim\upsilon$, $t\sim\gamma$, and $k\sim q_{j,\lambda}$. 
\item Prepare the state $\mathcal{U}_{st/T}^{y_{m}H(\theta^{\left(j\right)}(\lambda))}\!\left(\rho_{m}\right)$
using one sample of $\rho_{m}$ and Hamiltonian simulation to realize
the unitary channel $\mathcal{U}_{st/T}^{y_{m}H(\theta^{\left(j\right)}(\lambda))}$.
\item Perform the quantum circuit depicted in Figure~\ref{fig:Quantum-circuits-grad-log-loss},
with measurement outcomes $Z_{m}\in\left\{ -1,1\right\} $ for the
$\sigma_{Z}$ measurement and $X_{m}\in\spec(H_{k})$ for the $H_{k}$
measurement. Set
\begin{equation}
W_{m}\leftarrow\frac{J\left\Vert \theta^{\left(j\right)}(\lambda)\right\Vert _{1}}{2T}s\cdot\signum(\theta_{k})Z_{m}\cdot X_{m}.
\end{equation}
Set $m\leftarrow m+1.$
\item Repeat Steps 2-4 $M-1$ more times. Compute the average $\overline{W_{M}}\coloneqq\frac{1}{M}\sum_{m=1}^{M}W_{m}$
and output this value as an estimate of $\zeta$.
\end{enumerate}
\end{lyxalgorithm}

By the Hoeffding inequality, we are guaranteed that
\begin{equation}
\Pr\!\left[\left|\overline{W_{M}}-\zeta\right|\leq\varepsilon\right]\geq1-\delta.
\end{equation}

\section{Derivations for smooth rectified linear unit (ReLU)}

\subsection{Proof of Theorem~\ref{thm:grad-softplus} (derivative of smooth ReLU
function)}

\label{app:softplus-grad}

Recall from Theorem~\ref{thm:grad-logloss} that
\begin{multline}
\frac{\partial}{\partial\theta_{j}}T\Tr\!\left[\ln\!\left(I+e^{-y_{m}H(\theta)/T}\right)\rho_{m}\right]=-\frac{y_{m}}{2}\Tr\!\left[H_{j}\rho_{m}\right]+\\
\frac{\left\Vert \theta\right\Vert _{1}}{2T}\mathbb{E}_{\substack{s\sim\upsilon,\\
k\sim q,\\
t\sim\gamma
}
}\!\left[s\Re\!\left[\Tr\!\left[\signum(\theta_{k})H_{k}H_{j}e^{iy_{m}H(\theta)t/T}\mathcal{U}_{st/T}^{y_{m}H(\theta)}\!\left(\rho_{m}\right)\right]\right]\right].
\end{multline}
Now set $y_{m}=1$ and $\rho_{m}=\rho$ to get
\begin{equation}
\frac{\partial}{\partial\theta_{j}}T\Tr\!\left[\ln\!\left(I+e^{-H(\theta)/T}\right)\rho\right]=-\frac{1}{2}\Tr\!\left[H_{j}\rho\right]+\frac{\left\Vert \theta\right\Vert _{1}}{2T}\mathbb{E}_{\substack{s\sim\upsilon,\\
k\sim q,\\
t\sim\gamma
}
}\!\left[s\Re\!\left[\Tr\!\left[\signum(\theta_{k})H_{k}H_{j}e^{iH(\theta)t/T}\mathcal{U}_{st/T}^{H(\theta)}\!\left(\rho\right)\right]\right]\right].
\end{equation}
Under the substitution $T\to-T$, this becomes
\begin{align}
 & \frac{\partial}{\partial\theta_{j}}\left(-T\Tr\!\left[\ln\!\left(I+e^{-H(\theta)/-T}\right)\rho\right]\right)\nonumber \\
 & =\frac{\partial}{\partial\theta_{j}}\left(-T\Tr\!\left[\ln\!\left(I+e^{H(\theta)/T}\right)\rho\right]\right)\\
 & =-\frac{1}{2}\Tr\!\left[H_{j}\rho\right]-\frac{\left\Vert \theta\right\Vert _{1}}{2T}\mathbb{E}_{\substack{s\sim\upsilon,\\
k\sim q,\\
t\sim\gamma
}
}\!\left[s\Re\!\left[\Tr\!\left[\signum(\theta_{k})H_{k}H_{j}e^{-iH(\theta)t/T}\mathcal{U}_{-st/T}^{H(\theta)}\!\left(\rho\right)\right]\right]\right]\\
 & =-\frac{1}{2}\Tr\!\left[H_{j}\rho\right]-\frac{\left\Vert \theta\right\Vert _{1}}{2T}\mathbb{E}_{\substack{s\sim\upsilon,\\
k\sim q,\\
t\sim\gamma
}
}\!\left[s\Re\!\left[\Tr\!\left[\signum(\theta_{k})H_{k}H_{j}e^{iH(\theta)t/T}\mathcal{U}_{st/T}^{H(\theta)}\!\left(\rho\right)\right]\right]\right],
\end{align}
where the last equality follows because $\gamma$ is an even function.
Finally, we conclude that
\begin{align}
 & \frac{\partial}{\partial\theta_{j}}\left(T\Tr\!\left[\ln\!\left(I+e^{H(\theta)/T}\right)\rho\right]\right)\nonumber \\
 & =-\frac{\partial}{\partial\theta_{j}}\left(-T\Tr\!\left[\ln\!\left(I+e^{H(\theta)/T}\right)\rho\right]\right)\\
 & =\frac{1}{2}\Tr\!\left[H_{j}\rho\right]+\frac{\left\Vert \theta\right\Vert _{1}}{2T}\mathbb{E}_{\substack{s\sim\upsilon,\\
k\sim q,\\
t\sim\gamma
}
}\!\left[s\Re\!\left[\Tr\!\left[\signum(\theta_{k})H_{k}H_{j}e^{iH(\theta)t/T}\mathcal{U}_{st/T}^{H(\theta)}\!\left(\rho\right)\right]\right]\right].
\end{align}

\subsection{Hybrid quantum--classical algorithms for smooth ReLU}

\label{app:HQCs-softplus}

\subsubsection{Hybrid quantum--classical algorithm for estimating gradient of smooth
ReLU}

In this appendix, we briefly summarize a hybrid quantum--classical
algorithm for estimating the formula in~\eqref{eq:partial-deriv-softplus},
i.e., the $j$th partial derivative of $\Tr\!\left[r_{T}(H(\theta))\rho\right]$,
on a quantum computer.

The first term of~\eqref{eq:partial-deriv-softplus} can be easily
estimated by preparing the state $\rho$ and measuring $H_{j}$.

For the second term of~\eqref{eq:partial-deriv-softplus}, let us
make the assumption that each $H_{j}$ in $H(\theta)$ is both Hermitian
and unitary (as in the common case when each $H_{j}$ is a Pauli string).
Then Algorithm~\ref{alg:grad-log-loss} can be repurposed for estimating
the second term of~\eqref{eq:partial-deriv-softplus}, simply by setting
$y_{m}=1$ and $\rho_{m}=\rho$ therein. This is due to the expression
for the second term of~\eqref{eq:partial-deriv-softplus} being the
same as the second term of~\eqref{eq:partial-deriv-logloss} under
these substitutions. For completeness, we list the algorithm below.
Let us abbreviate the second term of~\eqref{eq:partial-deriv-softplus}
as follows:
\begin{equation}
\zeta_{j}\equiv\frac{1}{2T}\left\Vert \theta\right\Vert _{1}\mathbb{E}_{s\sim\upsilon,k\sim q,t\sim\gamma}\!\left[s\Re\!\left[\Tr\!\left[\signum(\theta_{k})H_{k}H_{j}e^{iH(\theta)t/T}\mathcal{U}_{st/T}^{H(\theta)}\!\left(\rho\right)\right]\right]\right].\label{eq:shorthand-2nd-term-grad-logloss-1}
\end{equation}

\begin{lyxalgorithm}
\label{alg:grad-smooth-relu}A hybrid quantum--classical algorithm
for estimating $\zeta_{j}$ in~\eqref{eq:shorthand-2nd-term-grad-logloss-1}
consists of the following steps:
\begin{enumerate}
\item Set $m\leftarrow1$, and set
\begin{equation}
M\leftarrow O\!\left(\left(\frac{\left\Vert \theta\right\Vert _{1}\max_{j\in\left[J\right]}\left\Vert H_{j}\right\Vert }{T\varepsilon}\right)^{2}\ln\!\left(\frac{1}{\delta}\right)\right),
\end{equation}
where $\varepsilon>0$ is the desired accuracy and $\delta\in\left(0,1\right)$
is the desired failure probability.
\item Sample $s\sim\upsilon$, $k\sim q$, and $t\sim\gamma$.
\item Prepare the state $\mathcal{U}_{st/T}^{H(\theta)}\!\left(\rho_{m}\right)$
using one sample of $\rho$ and Hamiltonian simulation to realize
the unitary channel $\mathcal{U}_{st/T}^{H(\theta)}$.
\item Perform the quantum circuit depicted in Figure~\ref{fig:Quantum-circuits-grad-log-loss},
with measurement outcomes $Z_{m}\in\left\{ -1,1\right\} $ for the
$\sigma_{Z}$ measurement and $X_{m}\in\spec(H_{k})$ for the $H_{k}$
measurement. Set $W_{m}\leftarrow\frac{\left\Vert \theta\right\Vert _{1}}{2T}s\cdot\signum(\theta_{k})Z_{m}\cdot X_{m}$.
Set $m\leftarrow m+1.$
\item Repeat Steps 2-4 $M-1$ more times. Compute the average $\overline{W_{M}}\coloneqq\frac{1}{M}\sum_{m=1}^{M}W_{m}$
and output this value as an estimate of $\zeta_{j}$.
\end{enumerate}
\end{lyxalgorithm}

By the Hoeffding inequality, we are guaranteed that
\begin{equation}
\Pr\!\left[\left|\overline{W_{M}}-\zeta_{j}\right|\leq\varepsilon\right]\geq1-\delta.
\end{equation}

\subsubsection{Hybrid quantum--classical algorithm for estimating smooth ReLU}

In this appendix, we detail an alternative formula for the objective
function $\Tr\!\left[r_{T}(H(\theta))\rho\right]$, which is amenable
to estimation by means of a hybrid quantum--classical algorithm.
The idea behind it is similar to that used for Theorem~\ref{thm:log-loss-obj-func-formula-FTOC},
which in turn is the same idea used in Theorem~\ref{thm:obj-func-formula}.
That is, we can use Theorem~\ref{thm:grad-softplus} and the fundamental
theorem of calculus to derive an expression for the smooth ReLU function
that can be estimated by a hybrid quantum--classical algorithm.
\begin{thm}
\label{thm:obj-func-softplus}The following equality holds:
\begin{multline}
\Tr\!\left[r_{T}(H(\theta))\rho\right]=\frac{1}{2}\Tr\!\left[H(\theta)\rho\right]+\\
\frac{J}{2T}\mathbb{E}_{\substack{j\sim\left[J\right],\\
\lambda,s\sim\upsilon,\\
t\sim\gamma
}
}\!\left[\left\Vert \theta^{\left(j\right)}(\lambda)\right\Vert _{1}\mathbb{E}_{k\sim q_{j,\lambda}}\!\left[s\Re\!\left[\Tr\!\left[\signum(\theta_{k})H_{k}H_{j}e^{iH(\theta^{\left(j\right)}(\lambda))t/T}\mathcal{U}_{st/T}^{H(\theta^{\left(j\right)}(\lambda))}\!\left(\rho\right)\right]\right]\right]\right],\label{eq:obj-func-softplus}
\end{multline}
where $r_{T}(H(\theta))$ is defined from~\eqref{eq:softplus-func},
$j$ is selected uniformly at random from $\left[J\right]$, $\lambda$
and $s$ are selected uniformly at random from the unit interval $\left[0,1\right]$,
$t$ is selected according to $\gamma(t)$ in~\eqref{eq:high-peak-tent-gamma},
\begin{align}
\theta^{\left(j\right)}(\lambda) & \coloneqq\left(0,\ldots,0,\lambda\theta_{j},\theta_{j+1},\ldots,\theta_{J}\right),\\
\left\Vert \theta^{\left(j\right)}(\lambda)\right\Vert _{1} & \coloneqq\lambda\left|\theta_{j}\right|+\sum_{k=j+1}^{J}\left|\theta_{k}\right|,\\
H(\theta_{1},\ldots,\theta_{J}) & \equiv H(\theta)=\sum_{j=1}^{J}\theta_{j}H_{j},
\end{align}
and $k$ is selected according to the conditional probability distribution
$q_{j,\lambda}$, defined as
\begin{equation}
q_{j,\lambda}(k)\coloneqq\begin{cases}
\frac{\lambda\left|\theta_{j}\right|}{\left\Vert \theta^{\left(j\right)}(\lambda)\right\Vert _{1}} & :k=j\\
\frac{\left|\theta_{k}\right|}{\left\Vert \theta^{\left(j\right)}(\lambda)\right\Vert _{1}} & :k>j
\end{cases}.
\end{equation}
\end{thm}

\begin{proof}
This is a direct consequence of the formula in~\ref{thm:grad-softplus}
and the same reasoning in the proof of Theorem~\ref{thm:log-loss-obj-func-formula-FTOC}
(note that we set $y_{m}=1$ and $\rho_{m}=\rho$ and the different
sign comes about from the same reasoning given in Appendix~\ref{app:softplus-grad}). 
\end{proof}
Given the similarity of the expressions in~\eqref{eq:obj-func-softplus}
and~\eqref{eq:log-loss-FTOC}, it follows that the two terms in~\eqref{eq:obj-func-softplus}
can estimated by similar algorithms. Indeed, set $y_{m}=1$ and $\rho_{m}=\rho$.
Then the algorithm outlined in~\eqref{eq:first-term-log-loss-1}--\eqref{eq:first-term-log-loss-last}
can be used to estimate the first term in~\eqref{eq:obj-func-softplus},
with the only change being the absence of the minus sign for the first
term in~\eqref{eq:obj-func-softplus}. The second term in~\eqref{eq:obj-func-softplus}
can be estimated by Algorithm~\ref{alg:log-loss-obj-func-est}, under
the same substitutions $y_{m}=1$ and $\rho_{m}=\rho$. For completeness,
we list the algorithm below.

The first term $\frac{1}{2}\Tr\!\left[H(\theta)\rho\right]$ in~\eqref{eq:obj-func-softplus}
can be easily estimated by writing
\begin{align}
\frac{1}{2}\Tr\!\left[H(\theta)\rho\right] & =\frac{1}{2}\Tr\!\left[\sum_{j=1}^{J}\theta_{j}H_{j}\rho\right]\label{eq:first-term-softplus-1}\\
 & =\frac{1}{2}\sum_{j=1}^{J}\theta_{j}\Tr\!\left[H_{j}\rho\right]\\
 & =\frac{\left\Vert \theta\right\Vert _{1}}{2}\sum_{j=1}^{J}\frac{\left|\theta_{j}\right|}{\left\Vert \theta\right\Vert _{1}}\signum(\theta_{j})\Tr\!\left[H_{j}\rho\right].\label{eq:first-term-softplus-last}
\end{align}
Thus, in each iteration of an estimation algorithm, we can pick $j$
at random according to the probability distribution $\frac{\left|\theta_{j}\right|}{\left\Vert \theta\right\Vert _{1}}$,
prepare the state $\rho$, measure the observable $H_{j}$ on this
state (obtaining an outcome in $\left\{ -1,+1\right\} $), and multiply
the outcome by $\frac{\left\Vert \theta\right\Vert _{1}}{2}\signum(\theta_{j})$.
One then repeats this procedure a number of times and computes the
sample mean. The Hoeffding inequality guarantees that a desired accuracy
of $\varepsilon>0$ and a desired failure probability of $\delta\in\left(0,1\right)$
can be achieved using the following number of trials:
\begin{equation}
O\!\left(\left(\frac{\left\Vert \theta\right\Vert _{1}\max_{j\in\left[J\right]}\left\Vert H_{j}\right\Vert }{\varepsilon}\right)^{2}\ln\!\left(\frac{1}{\delta}\right)\right).
\end{equation}

The second term in~\eqref{eq:obj-func-softplus}, which we abbreviate
aswhich we abbreviate as

\begin{equation}
\zeta\equiv\frac{J}{2T}\mathbb{E}_{\substack{j\sim\left[J\right],\\
\lambda,s\sim\upsilon,\\
t\sim\gamma
}
}\!\left[\left\Vert \theta^{\left(j\right)}(\lambda)\right\Vert _{1}\mathbb{E}_{k\sim q_{j,\lambda}}\!\left[s\Re\!\left[\Tr\!\left[\signum(\theta_{k})H_{k}H_{j}e^{iH(\theta^{\left(j\right)}(\lambda))t/T}\mathcal{U}_{st/T}^{H(\theta^{\left(j\right)}(\lambda))}\!\left(\rho\right)\right]\right]\right]\right],\label{eq:softplus-2nd-term-abbr}
\end{equation}
can be estimated by the following algorithm:
\begin{lyxalgorithm}
\label{alg:softplus-obj-func-est}A hybrid quantum--classical algorithm
for estimating $\zeta$ in~\eqref{eq:softplus-2nd-term-abbr} consists
of the following steps:
\begin{enumerate}
\item Set $m\leftarrow1$, and set
\begin{align}
M & \leftarrow O\!\left(\left(\frac{J\theta^{\star}h^{\star}}{T\varepsilon}\right)^{2}\ln\!\left(\frac{1}{\delta}\right)\right),\\
\theta^{\star} & \coloneqq\max_{\lambda\in\left[0,1\right],j\in\left[J\right]}\left\Vert \theta^{\left(j\right)}(\lambda)\right\Vert _{1},\\
h^{\star} & \coloneqq\max_{j\in\left[J\right]}\left\Vert H_{j}\right\Vert ,
\end{align}
where $\varepsilon>0$ is the desired accuracy and $\delta\in\left(0,1\right)$
is the desired failure probability.
\item Sample $j$ according to the uniform distribution on $\left[J\right]$,
$s\sim\upsilon$, $\lambda\sim\upsilon$, $t\sim\gamma$, and $k\sim q_{j,\lambda}$. 
\item Prepare the state $\mathcal{U}_{st/T}^{H(\theta^{\left(j\right)}(\lambda))}\!\left(\rho_{m}\right)$
using one sample of $\rho_{m}$ and Hamiltonian simulation to realize
the unitary channel $\mathcal{U}_{st/T}^{H(\theta^{\left(j\right)}(\lambda))}$.
\item Perform the quantum circuit depicted in Figure~\ref{fig:Quantum-circuits-grad-log-loss}
(with $y_{m}$ set to $1$), with measurement outcomes $Z_{m}\in\left\{ -1,1\right\} $
for the $\sigma_{Z}$ measurement and $X_{m}\in\spec(H_{k})$ for
the $H_{k}$ measurement. Set
\begin{equation}
W_{m}\leftarrow\frac{J\left\Vert \theta^{\left(j\right)}(\lambda)\right\Vert _{1}}{2T}s\cdot\signum(\theta_{k})Z_{m}\cdot X_{m}.
\end{equation}
Set $m\leftarrow m+1.$
\item Repeat Steps 2-4 $M-1$ more times. Compute the average $\overline{W_{M}}\coloneqq\frac{1}{M}\sum_{m=1}^{M}W_{m}$
and output this value as an estimate of $\zeta$.
\end{enumerate}
\end{lyxalgorithm}

By the Hoeffding inequality, we are guaranteed that
\begin{equation}
\Pr\!\left[\left|\overline{W_{M}}-\zeta\right|\leq\varepsilon\right]\geq1-\delta.
\end{equation}

\subsection{Proof of Theorem~\ref{thm:implementing-ReLU} (correctness of Algorithm
\ref{alg:ReLU-one-shot} for smooth ReLU)}

\label{app:ReLU-one-shot-proof}

Here we prove that
\begin{equation}
r_{T}(p)=T_{2}\left(\relu*\ell_{T_{1}}\right)\left(\frac{p}{T_{2}}\right),\label{eq:smooth-relu-conv-app}
\end{equation}
for all $p\in\mathbb{R}$, where $r_{T}$ is defined in~\eqref{eq:softplus-func},
$\relu$ is defined in~\eqref{eq:relu-def}, and $\ell_{T_{1}}$ is
defined in~\eqref{eq:logistic-prob-dens}.

Consider that, for all $a\in\mathbb{R}$ and $T>0$,
\begin{align}
\left(\relu*\ell_{T}\right)\left(a\right) & =\int_{-\infty}^{\infty}dp\,\relu(p)\ell_{T}(a-p)\\
 & =\int_{-\infty}^{\infty}dp\,\relu(p)\ell_{T}(p-a)\\
 & =\int_{-\infty}^{\infty}dp\,\relu(p)\frac{e^{\left(p-a\right)/T}}{T\left(e^{\left(p-a\right)/T}+1\right)^{2}}\\
 & =\lim_{R\to+\infty}\int_{-R}^{R}dp\,\relu(p)\frac{e^{\left(p-a\right)/T}}{T\left(e^{\left(p-a\right)/T}+1\right)^{2}}\\
 & =\lim_{R\to+\infty}\int_{-R}^{R}dp\,\relu(p)\frac{d}{dp}\left(\frac{1}{1+e^{-\left(p-a\right)/T}}\right)\\
 & =\lim_{R\to+\infty}\int_{0}^{R}dp\,p\frac{d}{dp}\left(\frac{1}{1+e^{-\left(p-a\right)/T}}\right)\\
 & \overset{(a)}{=}\lim_{R\to+\infty}\left[\left.\frac{p}{1+e^{-\left(p-a\right)/T}}\right|_{0}^{R}-\int_{0}^{R}dp\,\frac{1}{1+e^{-\left(p-a\right)/T}}\right]\\
 & =\lim_{R\to+\infty}\left[\frac{R}{1+e^{-\left(R-a\right)/T}}-\int_{0}^{R}dp\,\frac{1}{1+e^{-\left(p-a\right)/T}}\right]\\
 & \overset{(b)}{=}\lim_{R\to+\infty}\left[\frac{R}{1+e^{-\left(R-a\right)/T}}-\left[T\ln\!\left(1+e^{\left(R-a\right)/T}\right)-T\ln\!\left(1+e^{-a/T}\right)\right]\right]\\
 & =T\ln\!\left(1+e^{a/T}\right)+\lim_{R\to\infty}\left[\frac{R}{1+e^{-\left(R-a\right)/T}}-T\ln\!\left(1+e^{\left(R-a\right)/T}\right)\right]\\
 & \overset{(c)}{=}T\ln\!\left(1+e^{-a/T}\right)+a\\
 & =T\ln\!\left(1+e^{-a/T}\right)+T\ln\!\left(e^{a/T}\right)\\
 & =T\ln\!\left(1+e^{a/T}\right).\label{eq:smooth-ReLU-integral}
\end{align}
Now setting $T=T_{1}T_{2}$, for $T_{1},T_{2}>0$, this implies that
\begin{align}
T_{2}\left(\relu*\ell_{T_{1}}\right)\left(\frac{a}{T_{2}}\right) & =T_{1}T_{2}\ln\!\left(1+e^{a/\left(T_{1}T_{2}\right)}\right)\\
 & =T\ln\!\left(1+e^{a/T}\right)\\
 & =r_{T}(a),
\end{align}
thus establishing~\eqref{eq:smooth-relu-conv-app}. The equality (a)
follows from integration by parts. The equality (b) follows because,
by the substitution $u=\frac{p-a}{T}$,
\begin{align}
\int_{0}^{R}dp\,\frac{1}{1+e^{-\left(p-a\right)/T}} & =T\int_{-a/T}^{\left(R-a\right)/T}du\,\frac{1}{1+e^{-u}}\\
 & =T\left.\ln(1+e^{u})\right|_{-a/T}^{\left(R-a\right)/T}\\
 & =T\ln\!\left(1+e^{\left(R-a\right)/T}\right)-T\ln\!\left(1+e^{-a/T}\right).
\end{align}
The equality (c) follows because, by setting $S=\left(R-a\right)/T$,
\begin{align}
 & \lim_{R\to+\infty}\left[\frac{R}{1+e^{-\left(R-a\right)/T}}-T\ln\!\left(1+e^{\left(R-a\right)/T}\right)\right]\nonumber \\
 & =\lim_{S\to+\infty}\left[\frac{TS+a}{1+e^{-S}}-T\ln\!\left(1+e^{S}\right)\right]\\
 & =\lim_{S\to+\infty}\left[TS\left(\frac{1}{1+e^{-S}}\right)+a\left(\frac{1}{1+e^{-S}}\right)-T\ln\!\left(e^{S}\left(1+e^{-S}\right)\right)\right]\\
 & =\lim_{S\to+\infty}\left[TS\left(\frac{1}{1+e^{-S}}\right)+a\left(\frac{1}{1+e^{-S}}\right)-TS-T\ln\!\left(1+e^{-S}\right)\right]\\
 & =\lim_{S\to+\infty}\left[TS\left(\frac{1}{1+e^{-S}}-1\right)+a\left(\frac{1}{1+e^{-S}}\right)-T\ln\!\left(1+e^{-S}\right)\right]\\
 & =\lim_{S\to+\infty}\left[TS\left(\frac{e^{-S}}{1+e^{-S}}\right)+a\left(\frac{1}{1+e^{S}}\right)-T\ln\!\left(1+e^{-S}\right)\right]\\
 & =a.
\end{align}

\section{Derivations for sigmoid linear unit (SiLU)}

\subsection{Proof of Theorem~\ref{thm:grad-silu} (derivative of sigmoid linear
unit function)}

\label{app:grad-silu}Consider that
\begin{equation}
\silu_{T}(x)\coloneqq\frac{x}{1+e^{-x/T}}.
\end{equation}
The derivative of this function is given by
\begin{align}
k_{T}(x) & \coloneqq\frac{\partial}{\partial x}\left(\frac{x}{1+e^{-x/T}}\right)\\
 & =\frac{1}{1+e^{-x/T}}+\frac{xe^{-x/T}}{T\left(1+e^{-x/T}\right)^{2}}\\
 & =f_{T}(x)+x\ell_{T}(-x)\\
 & =f_{T}(x)+x\ell_{T}(x).
\end{align}
Observe also that
\begin{align}
k_{T}(x) & =\frac{1}{2}\left(1+g_{T}\!\left(\frac{x}{2}\right)\right)+x\ell_{T}(x)\\
 & =\frac{1}{2}+\frac{x}{4T}\frac{\tanh\!\left(\frac{x}{2T}\right)}{\frac{x}{2T}}+x\ell_{T}(x)\\
 & =\frac{1}{2}+x\left(\frac{1}{4T}\frac{\tanh\!\left(\frac{x}{2T}\right)}{\frac{x}{2T}}+\frac{1}{4T}\sech^{2}\!\left(\frac{x}{2T}\right)\right)\\
 & =\frac{1}{2}+\frac{x}{4T}\left(\frac{\tanh\!\left(\frac{x}{2T}\right)}{\frac{x}{2T}}+\sech^{2}\!\left(\frac{x}{2T}\right)\right).
\end{align}
where we used that
\begin{align}
\ell_{T}(x) & =\frac{e^{x/T}}{T\left(e^{x/T}+1\right)^{2}}\\
 & =\frac{1}{4T}\sech^{2}\!\left(\frac{x}{2T}\right).
\end{align}

\begin{lem}
\label{lem:grad-silu}Let $x\in\mathbb{R}$, let $x\mapsto A(x)$
be a Hermitian-valued function, and let $T>0$. Then the following
equality holds:
\begin{multline}
\frac{\partial}{\partial x}\silu_{T}(A(x))=\frac{1}{2}\frac{\partial}{\partial x}A(x)\\
+\frac{1}{2T}\int_{-\infty}^{\infty}dt\,\left(\frac{\gamma(t)+\mu(t)}{2}\right)\int_{0}^{1}ds\,s\,\Re\!\left[A(x)e^{-isA(x)t/\left(2T\right)}\left(\frac{\partial}{\partial x}A(x)\right)e^{-i\left(1-s\right)A(x)t/\left(2T\right)}\right].
\end{multline}
\end{lem}

\begin{proof}
To begin with, let
\begin{equation}
A(x)=\sum_{k}\lambda_{k}\Pi_{k}
\end{equation}
be a spectral decomposition of $A(x)$, where we have omitted the
dependence on $x$ in the eigenvalues and eigenprojections. Consider
that the derivative of the $\silu_{T}$ function is given by
\begin{equation}
\frac{\partial}{\partial x}\silu_{T}(A(x))=\sum_{k,\ell}f_{\silu_{T}}^{\left[1\right]}(\lambda_{k},\lambda_{\ell})\Pi_{k}\left(\frac{\partial}{\partial x}A(x)\right)\Pi_{\ell}.\label{eq:1st-divided-diff-silu}
\end{equation}
An expression for $f_{\silu_{T}}^{\left[1\right]}$ is as follows:
\begin{equation}
f_{\silu_{T}}^{\left[1\right]}(y_{1},y_{2})=\int_{0}^{1}ds\,k_{T}(sy_{1}+\left(1-s\right)y_{2}).
\end{equation}
Defining $\lambda_{s,k,\ell}\equiv s\lambda_{k}+\left(1-s\right)\lambda_{\ell}$,
then we find that
\begin{align}
 & \frac{\partial}{\partial x}\silu_{T}(A(x))\nonumber \\
 & =\sum_{k,\ell}\left(\int_{0}^{1}ds\,k_{T}(s\lambda_{k}+\left(1-s\right)\lambda_{\ell})\right)\Pi_{k}\left(\frac{\partial}{\partial x}A(x)\right)\Pi_{\ell}\\
 & =\sum_{k,\ell}\left(\int_{0}^{1}ds\,\left(\frac{1}{2}+\left(\frac{s\lambda_{k}+\left(1-s\right)\lambda_{\ell}}{4T}\right)\left(\frac{\tanh\!\left(\frac{\lambda_{s,k,\ell}}{2T}\right)}{\frac{\lambda_{s,k,\ell}}{2T}}+\sech^{2}\!\left(\frac{\lambda_{s,k,\ell}}{2T}\right)\right)\right)\right)\Pi_{k}\left(\frac{\partial}{\partial x}A(x)\right)\Pi_{\ell}\\
 & =\sum_{k,\ell}\left(\frac{1}{2}+\int_{0}^{1}ds\,\left(\frac{s\lambda_{k}+\left(1-s\right)\lambda_{\ell}}{4T}\right)\left(\frac{\tanh\!\left(\frac{\lambda_{s,k,\ell}}{2T}\right)}{\frac{\lambda_{s,k,\ell}}{2T}}+\sech^{2}\!\left(\frac{\lambda_{s,k,\ell}}{2T}\right)\right)\right)\Pi_{k}\left(\frac{\partial}{\partial x}A(x)\right)\Pi_{\ell}\\
 & =\frac{1}{2}\sum_{k,\ell}\Pi_{k}\left(\frac{\partial}{\partial x}A(x)\right)\Pi_{\ell}\nonumber \\
 & \qquad+\frac{1}{4T}\sum_{k,\ell}\left(\int_{0}^{1}ds\,\left(s\lambda_{k}+\left(1-s\right)\lambda_{\ell}\right)\left(\frac{\tanh\!\left(\frac{\lambda_{s,k,\ell}}{2T}\right)}{\frac{\lambda_{s,k,\ell}}{2T}}\right)+\sech^{2}\!\left(\frac{\lambda_{s,k,\ell}}{2T}\right)\right)\Pi_{k}\left(\frac{\partial}{\partial x}A(x)\right)\Pi_{\ell}\\
 & =\frac{1}{2}\frac{\partial}{\partial x}A(x)\nonumber \\
 & \qquad+\frac{1}{4T}\int_{0}^{1}ds\,\sum_{k,\ell}\left(s\lambda_{k}+\left(1-s\right)\lambda_{\ell}\right)\left(\int_{-\infty}^{\infty}dt\,\left(\gamma(t)+\mu(t)\right)e^{-\frac{i\left(s\lambda_{k}+\left(1-s\right)\lambda_{\ell}\right)t}{2T}}\right)\Pi_{k}\left(\frac{\partial}{\partial x}A(x)\right)\Pi_{\ell}.
\end{align}
We used that
\begin{align}
\int_{-\infty}^{\infty}dt\,\gamma(t)e^{-i\omega t} & =\frac{\tanh\!\left(\omega/2\right)}{\omega/2},\\
\int_{-\infty}^{\infty}dt\,\mu(t)e^{-i\omega t} & =\sech^{2}(\omega).
\end{align}
Now consider that
\begin{align}
 & \frac{1}{4T}\int_{0}^{1}ds\,\sum_{k,\ell}\left(s\lambda_{k}+\left(1-s\right)\lambda_{\ell}\right)\left(\int_{-\infty}^{\infty}dt\,\left(\gamma(t)+\mu(t)\right)e^{-\frac{i\left(s\lambda_{k}+\left(1-s\right)\lambda_{\ell}\right)t}{2T}}\right)\Pi_{k}\left(\frac{\partial}{\partial x}A(x)\right)\Pi_{\ell}\nonumber \\
 & =\frac{1}{4T}\int_{-\infty}^{\infty}dt\,\left(\gamma(t)+\mu(t)\right)\int_{0}^{1}ds\,\sum_{k,\ell}\left(s\lambda_{k}+\left(1-s\right)\lambda_{\ell}\right)\left(e^{-\frac{i\left(s\lambda_{k}+\left(1-s\right)\lambda_{\ell}\right)t}{2T}}\right)\Pi_{k}\left(\frac{\partial}{\partial x}A(x)\right)\Pi_{\ell}\\
 & =\frac{1}{4T}\int_{-\infty}^{\infty}dt\,\left(\gamma(t)+\mu(t)\right)\int_{0}^{1}ds\,\sum_{k,\ell}\left(s\lambda_{k}+\left(1-s\right)\lambda_{\ell}\right)e^{-\frac{is\lambda_{k}t}{2T}}\Pi_{k}\left(\frac{\partial}{\partial x}A(x)\right)e^{-\frac{i\left(1-s\right)\lambda_{\ell}t}{2T}}\Pi_{\ell}\\
 & =\frac{1}{4T}\int_{-\infty}^{\infty}dt\,\left(\gamma(t)+\mu(t)\right)\int_{0}^{1}ds\,s\sum_{k,\ell}\lambda_{k}e^{-\frac{is\lambda_{k}t}{2T}}\Pi_{k}\left(\frac{\partial}{\partial x}A(x)\right)e^{-\frac{i\left(1-s\right)\lambda_{\ell}t}{2T}}\Pi_{\ell}\nonumber \\
 & \qquad+\frac{1}{4T}\int_{-\infty}^{\infty}dt\,\left(\gamma(t)+\mu(t)\right)\int_{0}^{1}ds\,\left(1-s\right)\sum_{k,\ell}e^{-\frac{is\lambda_{k}t}{2T}}\Pi_{k}\left(\frac{\partial}{\partial x}A(x)\right)\lambda_{\ell}e^{-\frac{i\left(1-s\right)\lambda_{\ell}t}{2T}}\Pi_{\ell}\\
 & =\frac{1}{4T}\int_{-\infty}^{\infty}dt\,\left(\gamma(t)+\mu(t)\right)\int_{0}^{1}ds\,s\left(\sum_{k}\lambda_{k}e^{-\frac{is\lambda_{k}t}{2T}}\Pi_{k}\right)\left(\frac{\partial}{\partial x}A(x)\right)\left(\sum_{k}e^{-\frac{i\left(1-s\right)\lambda_{\ell}t}{2T}}\Pi_{\ell}\right)\nonumber \\
 & \qquad+\frac{1}{4T}\int_{-\infty}^{\infty}dt\,\left(\gamma(t)+\mu(t)\right)\int_{0}^{1}ds\,\left(1-s\right)\left(\sum_{k}e^{-\frac{is\lambda_{k}t}{2T}}\Pi_{k}\right)\left(\frac{\partial}{\partial x}A(x)\right)\left(\sum_{\ell}\lambda_{\ell}e^{-\frac{i\left(1-s\right)\lambda_{\ell}t}{2T}}\Pi_{\ell}\right)\\
 & =\frac{1}{4T}\int_{-\infty}^{\infty}dt\,\left(\gamma(t)+\mu(t)\right)\int_{0}^{1}ds\,s\,A(x)e^{-isA(x)t/\left(2T\right)}\left(\frac{\partial}{\partial x}A(x)\right)e^{-i\left(1-s\right)A(x)t/\left(2T\right)}\nonumber \\
 & \qquad+\frac{1}{4T}\int_{-\infty}^{\infty}dt\,\left(\gamma(t)+\mu(t)\right)\int_{0}^{1}ds\,\left(1-s\right)e^{-isA(x)t/\left(2T\right)}\left(\frac{\partial}{\partial x}A(x)\right)e^{-i\left(1-s\right)A(x)t/\left(2T\right)}A(x)\\
 & =\frac{1}{4T}\int_{-\infty}^{\infty}dt\,\left(\gamma(t)+\mu(t)\right)\int_{0}^{1}ds\,s\,A(x)e^{-isA(x)t/\left(2T\right)}\left(\frac{\partial}{\partial x}A(x)\right)e^{-i\left(1-s\right)A(x)t/\left(2T\right)}\nonumber \\
 & \qquad+\frac{1}{4T}\int_{-\infty}^{\infty}dt\,\left(\gamma(t)+\mu(t)\right)\int_{0}^{1}ds\,s\,e^{-i\left(1-s\right)A(x)t/\left(2T\right)}\left(\frac{\partial}{\partial x}A(x)\right)e^{-isA(x)t/\left(2T\right)}A(x)\\
 & =\frac{1}{4T}\int_{-\infty}^{\infty}dt\,\left(\gamma(t)+\mu(t)\right)\int_{0}^{1}ds\,s\,A(x)e^{-isA(x)t/\left(2T\right)}\left(\frac{\partial}{\partial x}A(x)\right)e^{-i\left(1-s\right)A(x)t/\left(2T\right)}\nonumber \\
 & \qquad+\frac{1}{4T}\int_{-\infty}^{\infty}dt\,\left(\gamma(t)+\mu(t)\right)\int_{0}^{1}ds\,s\,e^{i\left(1-s\right)A(x)t/\left(2T\right)}\left(\frac{\partial}{\partial x}A(x)\right)e^{isA(x)t/\left(2T\right)}A(x)\\
 & =\frac{1}{2T}\int_{-\infty}^{\infty}dt\,\left(\frac{\gamma(t)+\mu(t)}{2}\right)\int_{0}^{1}ds\,s\,\Re\!\left[A(x)e^{-isA(x)t/\left(2T\right)}\left(\frac{\partial}{\partial x}A(x)\right)e^{-i\left(1-s\right)A(x)t/\left(2T\right)}\right],
\end{align}
thus concluding the proof.
\end{proof}
\begin{thm*}[Restatement of Theorem~\ref{thm:grad-silu}]
The following equality holds:
\begin{equation}
\frac{\partial}{\partial\theta_{j}}\Tr\!\left[\silu_{T}(H(\theta))\rho\right]=\frac{1}{2}\Tr\!\left[H_{j}\rho\right]+\frac{\left\Vert \theta\right\Vert _{1}}{2T}\mathbb{E}_{\substack{t\sim\xi,\\
s\sim\upsilon,\\
k\sim q
}
}\left[s\,\Re\!\left[\Tr\!\left[\signum(\theta_{k})H_{k}H_{j}e^{iH(\theta)t/\left(2T\right)}\mathcal{U}_{st/\left(2T\right)}^{H(\theta)}\left(\rho\right)\right]\right]\right],
\end{equation}
where the objective function $\Tr\!\left[\silu_{T}(H(\theta))\rho\right]$
is defined from~\eqref{eq:silu-obj}, $\xi(t)$ is the following probability
density function
\begin{equation}
\xi(t)\coloneqq\frac{\gamma(t)+\mu(t)}{2}\label{eq:xi-prob-dens}
\end{equation}
$\upsilon$ is a uniform random variable on the unit interval $\left[0,1\right]$,
$q$ is the following probability distribution:
\begin{equation}
q(k)\coloneqq\frac{\left|\theta_{k}\right|}{\left\Vert \theta\right\Vert _{1}},
\end{equation}
and $\mathcal{U}_{t}^{H(\theta)}$ is the following unitary quantum
channel:
\begin{equation}
\mathcal{U}_{t}^{H(\theta)}(\cdot)\coloneqq e^{-iH(\theta)t}(\cdot)e^{iH(\theta)t}.
\end{equation}
\end{thm*}
\begin{proof}[Proof of Theorem~\ref{thm:grad-silu}]
Applying Lemma~\ref{lem:grad-silu}, consider that
\begin{align}
 & \frac{\partial}{\partial\theta_{j}}\Tr\!\left[\silu_{T}(H(\theta))\rho\right]\nonumber \\
 & =\Tr\!\left[\left(\frac{\partial}{\partial\theta_{j}}\silu_{T}(H(\theta))\right)\rho\right]\\
 & =\frac{1}{2}\Tr\!\left[H_{j}\rho\right]+\frac{1}{2T}\int_{-\infty}^{\infty}dt\,\xi(t)\int_{0}^{1}ds\,\Tr\!\left[s\,\Re\!\left[H(\theta)e^{-isH(\theta)t/\left(2T\right)}\left(\frac{\partial}{\partial\theta_{j}}H(\theta)\right)e^{-i\left(1-s\right)H(\theta)t/\left(2T\right)}\right]\rho\right]\\
 & =\frac{1}{2}\Tr\!\left[H_{j}\rho\right]+\frac{1}{2T}\mathbb{E}_{\substack{t\sim\xi,\\
s\sim\upsilon
}
}\left[s\,\Tr\!\left[\Re\!\left[H(\theta)e^{-isH(\theta)t/\left(2T\right)}\left(\frac{\partial}{\partial\theta_{j}}H(\theta)\right)e^{-i\left(1-s\right)H(\theta)t/\left(2T\right)}\right]\rho\right]\right]\\
 & =\frac{1}{2}\Tr\!\left[H_{j}\rho\right]+\frac{1}{2T}\mathbb{E}_{\substack{t\sim\xi,\\
s\sim\upsilon
}
}\left[s\,\Re\!\left[\Tr\!\left[e^{-isH(\theta)t/\left(2T\right)}H(\theta)H_{j}e^{-i\left(1-s\right)H(\theta)t/\left(2T\right)}\rho\right]\right]\right]\\
 & =\frac{1}{2}\Tr\!\left[H_{j}\rho\right]+\frac{1}{2T}\mathbb{E}_{\substack{t\sim\xi,\\
s\sim\upsilon
}
}\left[s\,\Re\!\left[\Tr\!\left[H(\theta)H_{j}e^{-iH(\theta)t/\left(2T\right)}e^{isH(\theta)t/\left(2T\right)}\rho e^{-istH(\theta)/\left(2T\right)}\right]\right]\right]\\
 & =\frac{1}{2}\Tr\!\left[H_{j}\rho\right]+\frac{1}{2T}\mathbb{E}_{\substack{t\sim\xi,\\
s\sim\upsilon
}
}\left[s\,\Re\!\left[\Tr\!\left[H(\theta)H_{j}e^{iH(\theta)t/\left(2T\right)}e^{-isH(\theta)t/\left(2T\right)}\rho e^{istH(\theta)/\left(2T\right)}\right]\right]\right]\\
 & =\frac{1}{2}\Tr\!\left[H_{j}\rho\right]+\frac{1}{2T}\mathbb{E}_{\substack{t\sim\xi,\\
s\sim\upsilon
}
}\left[s\,\Re\!\left[\Tr\!\left[H(\theta)H_{j}e^{iH(\theta)t/\left(2T\right)}\mathcal{U}_{st/\left(2T\right)}^{H(\theta)}\left(\rho\right)\right]\right]\right]\\
 & =\frac{1}{2}\Tr\!\left[H_{j}\rho\right]+\frac{\left\Vert \theta\right\Vert _{1}}{2T}\mathbb{E}_{\substack{t\sim\xi,\\
s\sim\upsilon,\\
k\sim q
}
}\left[s\,\Re\!\left[\Tr\!\left[\signum(\theta_{k})H_{k}H_{j}e^{iH(\theta)t/\left(2T\right)}\mathcal{U}_{st/\left(2T\right)}^{H(\theta)}\left(\rho\right)\right]\right]\right],
\end{align}
thus concluding the proof.
\end{proof}

\subsection{Hybrid quantum--classical algorithms for SiLU}

\label{app:HQCs-SiLU}

\subsubsection{Hybrid quantum--classical algorithm for estimating gradient of SiLU}

In this appendix, we briefly summarize a hybrid quantum--classical
algorithm for estimating the formula in~\eqref{eq:partial-deriv-silu},
i.e., the $j$th partial derivative of $\Tr\!\left[\silu_{T}(H(\theta))\rho\right]$,
on a quantum computer.

The first term of~\eqref{eq:partial-deriv-silu} can be easily estimated
by preparing the state $\rho$ and measuring $H_{j}$.

For the second term of~\eqref{eq:partial-deriv-silu}, let us make
the assumption that each $H_{j}$ in $H(\theta)$ is both Hermitian
and unitary (as in the common case when each $H_{j}$ is a Pauli string).
Then an approach similar to Algorithm~\ref{alg:grad-log-loss} can
be used. We list it below for completeness. Define

\begin{align}
\zeta_{j} & \equiv\frac{\left\Vert \theta\right\Vert _{1}}{2T}\mathbb{E}_{\substack{t\sim\xi,\\
s\sim\upsilon,\\
k\sim q
}
}\left[s\,\Re\left[\Tr\!\left[\signum(\theta_{k})H_{k}H_{j}e^{iH(\theta)t/\left(2T\right)}\mathcal{U}_{st/\left(2T\right)}^{H(\theta)}\left(\rho\right)\right]\right]\right].\label{eq:shorthand-2nd-term-grad-silu}
\end{align}

\begin{lyxalgorithm}
\label{alg:grad-silu}A hybrid quantum--classical algorithm for estimating
$\zeta_{j}$ in~\eqref{eq:shorthand-2nd-term-grad-silu} consists
of the following steps:
\begin{enumerate}
\item Set $m\leftarrow1$, and set
\begin{equation}
M\leftarrow O\!\left(\left(\frac{\left\Vert \theta\right\Vert _{1}\max_{j\in\left[J\right]}\left\Vert H_{j}\right\Vert }{T\varepsilon}\right)^{2}\ln\!\left(\frac{1}{\delta}\right)\right),
\end{equation}
where $\varepsilon>0$ is the desired accuracy and $\delta\in\left(0,1\right)$
is the desired failure probability.
\item Sample $s\sim\upsilon$, $k\sim q$, and $t\sim\xi$.
\item Prepare the state $\mathcal{U}_{st/(2T)}^{H(\theta)}\!\left(\rho\right)$
using one sample of $\rho$ and Hamiltonian simulation to realize
the unitary channel $\mathcal{U}_{st/(2T)}^{H(\theta)}$.
\item Perform the quantum circuit depicted in Figure~\ref{fig:Quantum-circuits-grad-silu},
with measurement outcomes $Z_{m}\in\left\{ -1,1\right\} $ for the
$\sigma_{Z}$ measurement and $X_{m}\in\spec(H_{k})$ for the $H_{k}$
measurement. Set $W_{m}\leftarrow\frac{\left\Vert \theta\right\Vert _{1}}{2T}s\cdot\signum(\theta_{k})Z_{m}\cdot X_{m}$.
Set $m\leftarrow m+1.$
\item Repeat Steps 2-4 $M-1$ more times. Compute the average $\overline{W_{M}}\coloneqq\frac{1}{M}\sum_{m=1}^{M}W_{m}$
and output this value as an estimate of $\zeta_{j}$.
\end{enumerate}
\end{lyxalgorithm}

By the Hoeffding inequality, we are guaranteed that
\begin{equation}
\Pr\!\left[\left|\overline{W_{M}}-\zeta_{j}\right|\leq\varepsilon\right]\geq1-\delta.
\end{equation}

\begin{figure}
\begin{centering}
\includegraphics[width=4.5in]{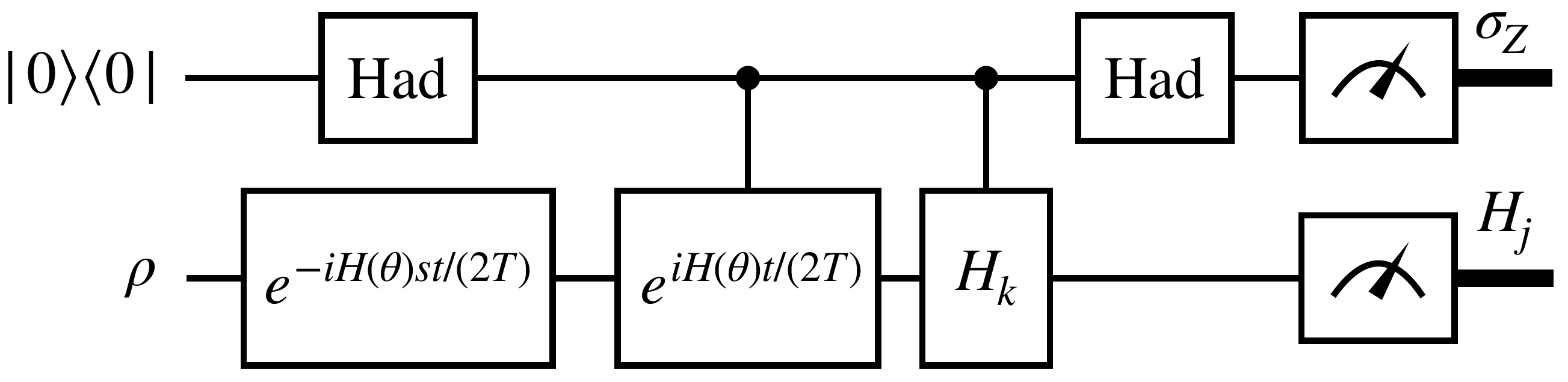}
\par\end{centering}
\caption{Quantum circuit used in Algorithm~\ref{alg:grad-silu} for estimating
the second term of~\eqref{eq:partial-deriv-silu}, denoted by $\zeta_{j}$
in~\eqref{eq:shorthand-2nd-term-grad-silu}.}\label{fig:Quantum-circuits-grad-silu}
\end{figure}

\subsubsection{Hybrid quantum--classical algorithm for estimating SiLU}

In this appendix, we detail an alternative formula for the objective
function $\Tr\!\left[\silu_{T}(H(\theta))\rho\right]$, which is amenable
to estimation by means of a hybrid quantum--classical algorithm.
The idea behind it is similar to that used for Theorem~\ref{thm:log-loss-obj-func-formula-FTOC},
which in turn is the same idea used in Theorem~\ref{thm:obj-func-formula}.
That is, we can use Theorem~\ref{thm:grad-silu} and the fundamental
theorem of calculus to derive an expression for SiLU that can be estimated
by a hybrid quantum--classical algorithm.
\begin{thm}
\label{thm:silu-obj-func}The following equality holds:
\begin{multline}
\Tr\!\left[\silu_{T}(H(\theta))\rho\right]=\frac{1}{2}\Tr\!\left[H(\theta)\rho\right]+\\
\frac{J}{2T}\mathbb{E}_{\substack{j\sim\left[J\right],\\
\lambda,s\sim\upsilon,\\
t\sim\xi
}
}\!\left[\left\Vert \theta^{\left(j\right)}(\lambda)\right\Vert _{1}\mathbb{E}_{k\sim q_{j,\lambda}}\!\left[s\Re\!\left[\Tr\!\left[\signum(\theta_{k})H_{k}H_{j}e^{iH(\theta^{\left(j\right)}(\lambda))t/\left(2T\right)}\mathcal{U}_{st/\left(2T\right)}^{H(\theta^{\left(j\right)}(\lambda))}\!\left(\rho\right)\right]\right]\right]\right],\label{eq:obj-func-silu}
\end{multline}
where $\silu_{T}(H(\theta))$ is defined from~\eqref{eq:silu-obj},
$j$ is selected uniformly at random from $\left[J\right]$, $\lambda$
and $s$ are selected uniformly at random from the unit interval $\left[0,1\right]$,
$t$ is selected according to $\xi(t)$ in~\eqref{eq:xi-prob-dens},
\begin{align}
\theta^{\left(j\right)}(\lambda) & \coloneqq\left(0,\ldots,0,\lambda\theta_{j},\theta_{j+1},\ldots,\theta_{J}\right),\\
\left\Vert \theta^{\left(j\right)}(\lambda)\right\Vert _{1} & \coloneqq\lambda\left|\theta_{j}\right|+\sum_{k=j+1}^{J}\left|\theta_{k}\right|,\\
H(\theta_{1},\ldots,\theta_{J}) & \equiv H(\theta)=\sum_{j=1}^{J}\theta_{j}H_{j},
\end{align}
and $k$ is selected according to the conditional probability distribution
$q_{j,\lambda}$, defined as
\begin{equation}
q_{j,\lambda}(k)\coloneqq\begin{cases}
\frac{\lambda\left|\theta_{j}\right|}{\left\Vert \theta^{\left(j\right)}(\lambda)\right\Vert _{1}} & :k=j\\
\frac{\left|\theta_{k}\right|}{\left\Vert \theta^{\left(j\right)}(\lambda)\right\Vert _{1}} & :k>j
\end{cases}.
\end{equation}
\end{thm}

\begin{proof}
The proof goes along the same lines as in the proof of Theorem~\ref{thm:log-loss-obj-func-formula-FTOC}
(see also Theorem~\ref{thm:obj-func-softplus}), but using the expression
in~\eqref{eq:partial-deriv-silu} instead of that in~\eqref{eq:partial-deriv-softplus}.
\end{proof}
For completeness, we delineate the hybrid quantum--classical algorithm
that estimates $\Tr\!\left[\silu_{T}(H(\theta))\rho\right]$, using
the expression in~\eqref{eq:obj-func-silu}. The first term $\frac{1}{2}\Tr\!\left[H(\theta)\rho\right]$
in~\eqref{eq:obj-func-silu} can be easily estimated by writing
\begin{align}
\frac{1}{2}\Tr\!\left[H(\theta)\rho\right] & =\frac{1}{2}\Tr\!\left[\sum_{j=1}^{J}\theta_{j}H_{j}\rho\right]\label{eq:first-term-silu-1}\\
 & =\frac{1}{2}\sum_{j=1}^{J}\theta_{j}\Tr\!\left[H_{j}\rho\right]\\
 & =\frac{\left\Vert \theta\right\Vert _{1}}{2}\sum_{j=1}^{J}\frac{\left|\theta_{j}\right|}{\left\Vert \theta\right\Vert _{1}}\signum(\theta_{j})\Tr\!\left[H_{j}\rho\right].\label{eq:first-term-silu-last}
\end{align}
Thus, in each iteration of an estimation algorithm, we can pick $j$
at random according to the probability distribution $\frac{\left|\theta_{j}\right|}{\left\Vert \theta\right\Vert _{1}}$,
prepare the state $\rho$, measure the observable $H_{j}$ on this
state (obtaining an outcome in $\left\{ -1,+1\right\} $), and multiply
the outcome by $\frac{\left\Vert \theta\right\Vert _{1}}{2}\signum(\theta_{j})$.
One then repeats this procedure a number of times and computes the
sample mean. The Hoeffding inequality guarantees that a desired accuracy
of $\varepsilon>0$ and a desired failure probability of $\delta\in\left(0,1\right)$
can be achieved using the following number of trials:
\begin{equation}
O\!\left(\left(\frac{\left\Vert \theta\right\Vert _{1}\max_{j\in\left[J\right]}\left\Vert H_{j}\right\Vert }{\varepsilon}\right)^{2}\ln\!\left(\frac{1}{\delta}\right)\right).
\end{equation}

The second term in~\eqref{eq:obj-func-silu}, which we abbreviate
as

\begin{equation}
\zeta\equiv\frac{J}{2T}\mathbb{E}_{\substack{j\sim\left[J\right],\\
\lambda,s\sim\upsilon,\\
t\sim\xi
}
}\!\left[\left\Vert \theta^{\left(j\right)}(\lambda)\right\Vert _{1}\mathbb{E}_{k\sim q_{j,\lambda}}\!\left[s\Re\!\left[\Tr\!\left[\signum(\theta_{k})H_{k}H_{j}e^{iH(\theta^{\left(j\right)}(\lambda))t/\left(2T\right)}\mathcal{U}_{st/\left(2T\right)}^{H(\theta^{\left(j\right)}(\lambda))}\!\left(\rho_{m}\right)\right]\right]\right]\right],\label{eq:silu-2nd-term-abbr}
\end{equation}
can be estimated by the following algorithm:
\begin{lyxalgorithm}
\label{alg:silu-obj-func-est}A hybrid quantum--classical algorithm
for estimating $\zeta$ in~\eqref{eq:silu-2nd-term-abbr} consists
of the following steps:
\begin{enumerate}
\item Set $m\leftarrow1$, and set
\begin{align}
M & \leftarrow O\!\left(\left(\frac{J\theta^{\star}h^{\star}}{T\varepsilon}\right)^{2}\ln\!\left(\frac{1}{\delta}\right)\right),\\
\theta^{\star} & \coloneqq\max_{\lambda\in\left[0,1\right],j\in\left[J\right]}\left\Vert \theta^{\left(j\right)}(\lambda)\right\Vert _{1},\\
h^{\star} & \coloneqq\max_{j\in\left[J\right]}\left\Vert H_{j}\right\Vert ,
\end{align}
where $\varepsilon>0$ is the desired accuracy and $\delta\in\left(0,1\right)$
is the desired failure probability.
\item Sample $j$ according to the uniform distribution on $\left[J\right]$,
$s\sim\upsilon$, $\lambda\sim\upsilon$, $t\sim\xi$, and $k\sim q_{j,\lambda}$. 
\item Prepare the state $\mathcal{U}_{st/\left(2T\right)}^{H(\theta^{\left(j\right)}(\lambda))}\!\left(\rho\right)$
using one sample of $\rho$ and Hamiltonian simulation to realize
the unitary channel $\mathcal{U}_{st/\left(2T\right)}^{H(\theta^{\left(j\right)}(\lambda))}$.
\item Perform the quantum circuit depicted in Figure~\ref{fig:Quantum-circuits-grad-silu},
with measurement outcomes $Z_{m}\in\left\{ -1,1\right\} $ for the
$\sigma_{Z}$ measurement and $X_{m}\in\spec(H_{k})$ for the $H_{k}$
measurement. Set
\begin{equation}
W_{m}\leftarrow\frac{J\left\Vert \theta^{\left(j\right)}(\lambda)\right\Vert _{1}}{2T}s\cdot\signum(\theta_{k})Z_{m}\cdot X_{m}.
\end{equation}
Set $m\leftarrow m+1.$
\item Repeat Steps 2-4 $M-1$ more times. Compute the average $\overline{W_{M}}\coloneqq\frac{1}{M}\sum_{m=1}^{M}W_{m}$
and output this value as an estimate of $\zeta$.
\end{enumerate}
\end{lyxalgorithm}

By the Hoeffding inequality, we are guaranteed that
\begin{equation}
\Pr\!\left[\left|\overline{W_{M}}-\zeta\right|\leq\varepsilon\right]\geq1-\delta.
\end{equation}

\subsection{Proof of Theorem~\ref{thm:gen-conv-and-mult-alg} (expected value
of quantum convolution and multiplication algorithm)}

\label{app:proof-q-conv-and-mult-alg}Following the same reasoning
as in~\eqref{eq:normalize-state-1}--\eqref{eq:normalize-state-last},
we conclude that, for $i\in\left\{ 1,2\right\} $, $|q_{i}\rangle$
is a state vector if $q_{i}$ is a probability density. As in Appendix
\ref{app:proof-q-conv-alg}, let us first suppose that the state of
the data register is pure and given by $|\varphi\rangle\!\langle\varphi|$,
where $|\varphi\rangle$ is a state vector. Suppose that a spectral
decomposition of $A$ is given by
\begin{equation}
A=\sum_{i}a_{i}|i\rangle\!\langle i|.
\end{equation}
This implies that
\begin{equation}
e^{i\hat{x}\otimes A}=\sum_{i}\int_{-\infty}^{\infty}dx\,e^{ixa_{i}}|x\rangle\!\langle x|\otimes|i\rangle\!\langle i|.
\end{equation}
We begin under the assumption that $t_{1}=t_{2}=1$. The probability
density that Step 3 of Algorithm~\ref{alg:convolution-and-mult-on-A}
outputs $p_{1},p_{2}\in\mathbb{R}$ is equal to
\begin{equation}
\left\Vert \left(\langle p_{1}|\otimes\langle p_{2}|\otimes I\right)e^{i\hat{x}_{2}\otimes A_{3}}e^{i\hat{x}_{1}\otimes A_{3}}\left(|q_{1}\rangle\otimes|q_{2}\rangle\otimes|\varphi\rangle\right)\right\Vert ^{2}.
\end{equation}
Thus, the expected value of the output of Algorithm~\ref{alg:convolution-and-mult-on-A}
is
\begin{equation}
\int_{-\infty}^{\infty}dp_{1}\int_{-\infty}^{\infty}dp_{2}\,r_{1}(p_{1})r_{2}(p_{2})\left\Vert \left(\langle p_{1}|\otimes\langle p_{2}|\otimes I\right)e^{i\hat{x}_{2}\otimes A_{3}}e^{i\hat{x}_{1}\otimes A_{3}}\left(|q_{1}\rangle\otimes|q_{2}\rangle\otimes|\varphi\rangle\right)\right\Vert ^{2}.
\end{equation}
Consider that
\begin{align}
 & \left(\langle p_{1}|\otimes\langle p_{2}|\otimes I\right)e^{i\hat{x}_{2}\otimes A_{3}}e^{i\hat{x}_{1}\otimes A_{3}}\left(|q_{1}\rangle\otimes|q_{2}\rangle\otimes|\varphi\rangle\right)\nonumber \\
 & =\left(\langle p_{1}|\otimes\langle p_{2}|\otimes I\right)\left(I\otimes\sum_{j}\int_{-\infty}^{\infty}dx_{2}\,e^{ix_{2}a_{j}}|x_{2}\rangle\!\langle x_{2}|\otimes|j\rangle\!\langle j|\right)\times\nonumber \\
 & \qquad\left(\sum_{i}\int_{-\infty}^{\infty}dx_{1}\,e^{ix_{1}a_{i}}|x_{1}\rangle\!\langle x_{1}|\otimes I\otimes|i\rangle\!\langle i|\right)\times\nonumber \\
 & \qquad\left(\int_{-\infty}^{\infty}dp_{1}'\,\sqrt{q_{1}(p_{1}')}|p_{1}'\rangle\otimes\int_{-\infty}^{\infty}dp_{2}'\,\sqrt{q_{2}(p_{2}')}|p_{2}'\rangle\otimes|\varphi\rangle\right)\\
 & =\sum_{i,j}\int_{-\infty}^{\infty}dx_{2}\int_{-\infty}^{\infty}dx_{1}\int_{-\infty}^{\infty}dp_{1}'\int_{-\infty}^{\infty}dp_{2}'e^{ix_{2}a_{j}}e^{ix_{1}a_{i}}\sqrt{q_{1}(p_{1}')q_{2}(p_{2}')}\times\nonumber \\
 & \qquad\langle p_{1}|x_{1}\rangle\langle x_{1}|p_{1}'\rangle\langle p_{2}|x_{2}\rangle\langle x_{2}|p_{2}'\rangle|j\rangle\langle j|i\rangle\langle i|\varphi\rangle\\
 & =\sum_{i}\int_{-\infty}^{\infty}dp_{1}'\int_{-\infty}^{\infty}dp_{2}'\left(\int_{-\infty}^{\infty}dx_{1}\,e^{ix_{1}a_{i}}\langle p_{1}|x_{1}\rangle\langle x_{1}|p_{1}'\rangle\right)\times\nonumber \\
 & \qquad\left(\int_{-\infty}^{\infty}dx_{2}e^{ix_{2}a_{i}}\langle p_{2}|x_{2}\rangle\langle x_{2}|p_{2}'\rangle\right)\sqrt{q_{1}(p_{1}')q_{2}(p_{2}')}|i\rangle\langle i|\varphi\rangle\\
 & =\sum_{i}\int_{-\infty}^{\infty}dp_{1}'\int_{-\infty}^{\infty}dp_{2}'\,\delta\!\left(p_{1}'-p_{1}+a_{i}\right)\delta\!\left(p_{2}'-p_{2}+a_{i}\right)\sqrt{q_{1}(p_{1}')q_{2}(p_{2}')}|i\rangle\langle i|\varphi\rangle\\
 & =\sum_{i}\sqrt{q_{1}(p_{1}-a_{i})q_{2}(p_{2}-a_{i})}|i\rangle\langle i|\varphi\rangle\\
 & =\sum_{i}\sqrt{q_{1}(a_{i}-p_{1})q_{2}(a_{i}-p_{2})}|i\rangle\langle i|\varphi\rangle.
\end{align}
The fourth equality follows from reasoning similar to that given in
\eqref{eq:fourier-delta-1}--\eqref{eq:fourier-delta-last}. The
last equality follows from the assumption that $q_{1}$ and $q_{2}$
are even functions. So then
\begin{align}
 & \left\Vert \left(\langle p_{1}|\otimes\langle p_{2}|\otimes I\right)e^{i\hat{x}_{2}\otimes A_{3}}e^{i\hat{x}_{1}\otimes A_{3}}\left(|q_{1}\rangle\otimes|q_{2}\rangle\otimes|\varphi\rangle\right)\right\Vert ^{2}\nonumber \\
 & =\left\Vert \sum_{i}\sqrt{q_{1}(a_{i}-p_{1})q_{2}(a_{i}-p_{2})}|i\rangle\langle i|\varphi\rangle\right\Vert ^{2}\\
 & =\langle\varphi|\left(\sum_{i}q_{1}(a_{i}-p_{1})q_{2}(a_{i}-p_{2})|i\rangle\!\langle i|\right)|\varphi\rangle.
\end{align}
Then we conclude that the expected value is given by
\begin{align}
 & \int_{-\infty}^{\infty}dp_{1}\int_{-\infty}^{\infty}dp_{2}\,r_{1}(p_{1})r_{2}(p_{2})\left\Vert \left(\langle p_{1}|\otimes\langle p_{2}|\otimes I\right)e^{i\hat{x}_{2}\otimes A_{3}}e^{i\hat{x}_{1}\otimes A_{3}}\left(|q_{1}\rangle\otimes|q_{2}\rangle\otimes|\varphi\rangle\right)\right\Vert ^{2}\nonumber \\
 & =\int_{-\infty}^{\infty}dp_{1}\int_{-\infty}^{\infty}dp_{2}\,r_{1}(p_{1})r_{2}(p_{2})\langle\varphi|\left(\sum_{i}q_{1}(a_{i}-p_{1})q_{2}(a_{i}-p_{2})|i\rangle\!\langle i|\right)|\varphi\rangle\\
 & =\langle\varphi|\left(\sum_{i}\left(\int_{-\infty}^{\infty}dp_{1}\,r_{1}(p_{1})q_{1}(a_{i}-p_{1})\right)\left(\int_{-\infty}^{\infty}dp_{2}\,r_{2}(p_{2})q_{2}(a_{i}-p_{2})\right)|i\rangle\!\langle i|\right)|\varphi\rangle\\
 & =\langle\varphi|\left(\sum_{i}s_{1}(a_{i})s_{2}(a_{i})|i\rangle\!\langle i|\right)|\varphi\rangle\\
 & =\langle\varphi|s_{1}(A)s_{2}(A)|\varphi\rangle.\label{eq:proof-conv-and-mult-alg-pure-states}
\end{align}
The result generalizes to an arbitrary state $\rho$ because every
such state can be written as a convex combination of pure states as
\begin{equation}
\rho=\sum_{z}t(z)|\varphi_{z}\rangle\!\langle\varphi_{z}|,
\end{equation}
so that
\begin{align}
 & \int_{-\infty}^{\infty}dp_{1}\int_{-\infty}^{\infty}dp_{2}\,r_{1}(p_{1})r_{2}(p_{2})\Tr\!\left[\begin{array}{c}
\left(|p_{1}\rangle\!\langle p_{1}|\otimes|p_{2}\rangle\!\langle p_{2}|\otimes I\right)e^{i\hat{x}_{2}\otimes A_{3}}e^{i\hat{x}_{1}\otimes A_{3}}\times\\
\left(|q_{1}\rangle\!\langle q_{1}|\otimes|q_{2}\rangle\!\langle q_{2}|\otimes\rho\right)e^{-i\hat{x}_{1}\otimes A_{3}}e^{-i\hat{x}_{2}\otimes A_{3}}
\end{array}\right]\nonumber \\
 & =\sum_{z}t(z)\int_{-\infty}^{\infty}dp_{1}\int_{-\infty}^{\infty}dp_{2}\,r_{1}(p_{1})r_{2}(p_{2})\Tr\!\left[\begin{array}{c}
\left(|p_{1}\rangle\!\langle p_{1}|\otimes|p_{2}\rangle\!\langle p_{2}|\otimes I\right)e^{i\hat{x}_{2}\otimes A_{3}}e^{i\hat{x}_{1}\otimes A_{3}}\times\\
\left(|q_{1}\rangle\!\langle q_{1}|\otimes|q_{2}\rangle\!\langle q_{2}|\otimes|\varphi_{z}\rangle\!\langle\varphi_{z}|\right)e^{-i\hat{x}_{1}\otimes A_{3}}e^{-i\hat{x}_{2}\otimes A_{3}}
\end{array}\right]\\
 & =\sum_{z}t(z)\int_{-\infty}^{\infty}dp_{1}\int_{-\infty}^{\infty}dp_{2}\,r_{1}(p_{1})r_{2}(p_{2})\left\Vert \left(\langle p_{1}|\otimes\langle p_{2}|\otimes I\right)e^{i\hat{x}_{2}\otimes A_{3}}e^{i\hat{x}_{1}\otimes A_{3}}\left(|q_{1}\rangle\otimes|q_{2}\rangle\otimes|\varphi_{z}\rangle\right)\right\Vert ^{2}\\
 & =\sum_{z}t(z)\Tr\!\left[s_{1}(A)s_{2}(A)|\varphi_{z}\rangle\!\langle\varphi_{z}|\right]\\
 & =\Tr\!\left[s_{1}(A)s_{2}(A)\rho\right],
\end{align}
where the penultimate equality follows from~\eqref{eq:proof-conv-and-mult-alg-pure-states}.

Finally, the claim for different times $t_{1}$ and $t_{2}$ follows
because
\begin{multline}
\int_{-\infty}^{\infty}dp_{1}\int_{-\infty}^{\infty}dp_{2}\,r_{1}(p_{1})r_{2}(p_{2})\Tr\!\left[\begin{array}{c}
\left(|p_{1}\rangle\!\langle p_{1}|\otimes|p_{2}\rangle\!\langle p_{2}|\otimes I\right)e^{i\hat{x}_{2}\otimes A_{3}t_{2}}e^{i\hat{x}_{1}\otimes A_{3}t_{1}}\times\\
\left(|q_{1}\rangle\!\langle q_{1}|\otimes|q_{2}\rangle\!\langle q_{2}|\otimes\rho\right)e^{-i\hat{x}_{1}\otimes A_{3}t_{1}}e^{-i\hat{x}_{2}\otimes A_{3}t_{2}}
\end{array}\right]\\
=\Tr\!\left[s_{1}(At_{1})s_{2}(At_{2})\rho\right].
\end{multline}

\subsection{Proof of Theorem~\ref{thm:silu-one-shot-alg} (correctness of Algorithm
\ref{alg:SiLU-one-shot} for SiLU)}

\label{app:SiLU-one-shot-proof}This is an application of Theorem
\ref{thm:gen-conv-and-mult-alg}. Given that the first control qumode
is prepared in the vacuum state, its representation in the momentum
eigenbasis is
\begin{align}
|0\rangle & =\int dp\,\sqrt{\phi(p)}|p\rangle,\\
\phi(p) & \coloneqq\frac{e^{-p^{2}}}{\sqrt{\pi}}.
\end{align}
This means that $q_{1}=\phi$. Additionally, $q_{2}=\ell_{T_{1}}$.
Since the output of Algorithm~\ref{alg:SiLU-one-shot} is $p_{1}\mathbf{1}_{p_{2}\geq0}$,
we can take
\begin{align}
r_{1}(p_{1}) & =p_{1},\\
r_{2}(p_{2}) & =\mathbf{1}_{p_{2}\geq0}.
\end{align}

By Theorem~\ref{thm:gen-conv-and-mult-alg}, the expected value equals
\begin{equation}
\Tr\!\left[\left(q_{1}*r_{1}\right)\left(H(\theta)\right)\left(q_{2}*r_{2}\right)\left(H(\theta)/T_{2}\right)\rho\right],
\end{equation}
because the second controlled evolution uses $H(\theta)/T_{2}$. Thus,
to conclude the proof, we need to show that
\begin{equation}
\left(q_{1}*r_{1}\right)\left(p\right)\left(q_{2}*r_{2}\right)\left(p/T_{2}\right)=\silu_{T}(p).
\end{equation}
To this end, consider that
\begin{align}
\left(q_{1}*r_{1}\right)\left(p\right) & =\int_{-\infty}^{\infty}dp'\,q_{1}(p')r_{1}(p-p')\\
 & =\int_{-\infty}^{\infty}dp'\,\phi(p')(p-p')\\
 & =\int_{-\infty}^{\infty}dp'\,\phi(p')p-\int_{-\infty}^{\infty}dp'\,\phi(p')p'\\
 & =p\int_{-\infty}^{\infty}dp'\,\phi(p')\\
 & =p.
\end{align}
The penultimate equality follows because $\phi$ is a probability
density function and because it has zero mean. Now consider that
\begin{align}
\left(q_{2}*r_{2}\right)\left(p/T_{2}\right) & =(\ell_{T_{1}}*\mathbf{1}_{p\geq0})(p/T_{2})\\
 & =\int_{-\infty}^{\infty}dp'\,\mathbf{1}_{p'\geq0}\ell_{T_{1}}(p/T_{2}-p')\\
 & =\int_{0}^{\infty}dp'\,\ell_{T_{1}}(p/T_{2}-p')\\
 & =\int_{0}^{\infty}dp'\,\ell_{T_{1}}(p'-p/T_{2})\\
 & =\int_{-p/T_{2}}^{\infty}dp'\,\ell_{T_{1}}(p')\\
 & =\int_{-p/T_{2}}^{\infty}dp'\,\frac{d}{dp'}f_{T_{1}}(p')\\
 & =f_{T_{1}}(\infty)-f_{T_{1}}(-p/T_{2})\\
 & =1-f_{T_{1}T_{2}}(-p)\\
 & =1-f_{T}(-p)\\
 & =f_{T}(p),
\end{align}
thus concluding the proof.

\section{Derivations for quantizing Gaussian activation functions}

Let us denote the cumulative Gaussian distribution function as
\begin{equation}
\Phi(x)\coloneqq\frac{1}{\sqrt{2\pi}}\int_{-\infty}^{x}dy\,e^{-\frac{1}{2}y^{2}}
\end{equation}
and let us denote the standard Gaussian probability density function as
\begin{equation}
\phi(x)\coloneqq\frac{1}{\sqrt{2\pi}}e^{-\frac{1}{2}x^{2}}.\label{eq:standard-Gaussian-app}
\end{equation}

\subsection{Proof of Equation~\eqref{eq:convolution-Gaussian-with-indicators}}

\label{app:erf-func-conv}

Our goal is to prove~\eqref{eq:convolution-Gaussian-with-indicators},
which we recall here:
\begin{equation}
(\phi_{T_{1}}*r)(p/T_{2})=2\Phi\!\left(\frac{2p}{T}\right)-1\eqqcolon\erf\!\left(\frac{\sqrt{2}p}{T}\right),\label{eq:convolution-Gaussian-with-indicators-app}
\end{equation}
for all $p\in\mathbb{R}$, with $r$ set as in~\eqref{eq:r-indicator}
and $T/2=T_{1}T_{2}$.

Following the same approach as in~\eqref{eq:convolution-logistic-indicators-1}--\eqref{eq:convolution-logistic-indicators-last}
and noting that
\begin{equation}
\Phi\!\left(\frac{p}{T_{1}}\right)=\int_{-\infty}^{p}dp'\,\phi_{T_{1}}(p'),
\end{equation}
consider that
\begin{align}
(\phi_{T_{1}}*r)(p/T_{2}) & =(\phi_{T_{1}}*\mathbf{1}_{p\geq0})(p/T_{2})-(\phi_{T_{1}}*\mathbf{1}_{p<0})(p/T_{2})\\
 & =\int_{-\infty}^{\infty}dp'\,\mathbf{1}_{p'\geq0}\phi_{T_{1}}(p/T_{2}-p')-\int_{-\infty}^{\infty}dp'\,\mathbf{1}_{p'<0}\phi_{T_{1}}(p/T_{2}-p')\\
 & =\int_{0}^{\infty}dp'\,\phi_{T_{1}}(p/T_{2}-p')-\int_{-\infty}^{0}dp'\,\phi_{T_{1}}(p/T_{2}-p')\\
 & =\int_{0}^{\infty}dp'\,\phi_{T_{1}}(p'-p/T_{2})-\int_{-\infty}^{0}dp'\,\phi_{T_{1}}(p'-p/T_{2})\\
 & =\int_{-p/T_{2}}^{\infty}dp'\,\phi_{T_{1}}(p')-\int_{-\infty}^{-p/T_{2}}dp'\,\phi_{T_{1}}(p')\\
 & =\int_{-p/T_{2}}^{\infty}dp'\,\frac{d}{dp'}\Phi\!\left(\frac{p'}{T_{1}}\right)-\int_{-\infty}^{-p/T_{2}}dp'\,\frac{d}{dp'}\Phi\!\left(\frac{p'}{T_{1}}\right)\\
 & =\Phi(\infty)-\Phi\!\left(-\frac{p}{T_{1}T_{2}}\right)-\Phi\!\left(-\frac{p}{T_{1}T_{2}}\right)+\Phi(-\infty)\\
 & =1-2\Phi\!\left(-\frac{2p}{T}\right)\\
 & =2\Phi\!\left(\frac{2p}{T}\right)-1\\
 & =\erf\!\left(\frac{\sqrt{2}p}{T}\right),
\end{align}
where the fourth equality follows because $\phi_{T_{1}}$ is an even
function and the penultimate equality follows because $\Phi(-x)=1-\Phi(x)$.

\subsection{Gradient of expectation of Gaussian error function activation observable}

\label{app:erf-grad-deriv}

In this appendix, we derive a formula for the gradient of $\frac{\partial}{\partial\theta_{j}}\Tr\!\left[\erf\!\left(\frac{\sqrt{2}H(\theta)}{T}\right)\rho\right]$
that is useful for estimation on quantum computers. The development
here mirrors that in Appendix~\ref{app:grad-deriv}.

We begin by deriving a formula for the derivative of the matrix Gaussian
error function. Recall that
\begin{equation}
\erf(x)=2\Phi(\sqrt{2}x)-1,
\end{equation}
which implies that
\begin{equation}
\frac{\partial}{\partial x}\erf(x)=2\sqrt{2}\phi\!\left(\sqrt{2}x\right).
\end{equation}

\begin{lem}
Let $x\in\mathbb{R}$, and let $x\mapsto A(x)$ be a Hermitian-valued
function. Then the following equality holds:
\begin{equation}
\frac{\partial}{\partial x}\erf(A(x))=\frac{2}{\sqrt{\pi}}\int_{-\infty}^{\infty}dt\,\phi(t)\int_{0}^{1}ds\,e^{i\sqrt{2}tsA(x)}\left(\frac{\partial}{\partial x}A(x)\right)e^{i\sqrt{2}t\left(1-s\right)A(x)},
\end{equation}
where $\phi(t)$ is the standard Gaussian probability density function
in~\eqref{eq:standard-Gaussian-app}.
\end{lem}

\begin{proof}
To begin with, let
\begin{equation}
A(x)=\sum_{k}\lambda_{k}\Pi_{k}
\end{equation}
be a spectral decomposition of $A(x)$, where we have omitted the
dependence on $x$ in the eigenvalues and eigenprojections. Consider
that the derivative of the matrix Gaussian error function is given
by
\begin{equation}
\frac{\partial}{\partial x}\erf(A(x))=\sum_{k,\ell}f_{\erf}^{\left[1\right]}(\lambda_{k},\lambda_{\ell})\Pi_{k}\left(\frac{\partial}{\partial x}A(x)\right)\Pi_{\ell},\label{eq:1st-divided-diff-erf}
\end{equation}
where $f_{\erf}^{\left[1\right]}$ is the first divided difference
of the Gaussian error function, which is given by
\begin{equation}
f_{\erf}^{\left[1\right]}(y_{1},y_{2})=2\sqrt{2}\int_{0}^{1}ds\,\phi\!\left(\sqrt{2}\left(sy_{1}+\left(1-s\right)y_{2}\right)\right),
\end{equation}
as a consequence of the fundamental theorem of calculus and that
\begin{equation}
\frac{d}{ds}\erf(sy_{1}+\left(1-s\right)y_{2})=2\sqrt{2}\phi\!\left(\sqrt{2}\left(sy_{1}+\left(1-s\right)y_{2}\right)\right)\left(y_{1}-y_{2}\right).
\end{equation}
Then it follows that
\begin{equation}
\frac{\partial}{\partial x}\erf(A(x))=2\sqrt{2}\sum_{k,\ell}\int_{0}^{1}ds\,\phi\!\left(\sqrt{2}\left(s\lambda_{k}+\left(1-s\right)\lambda_{\ell}\right)\right)\Pi_{k}\left(\frac{\partial}{\partial x}A(x)\right)\Pi_{\ell}.
\end{equation}
Now consider that the Fourier transform of $\phi(x)$ is as follows:
\begin{equation}
\phi(x)=\frac{1}{\sqrt{2\pi}}\int_{-\infty}^{\infty}dt\,\phi(t)e^{itx},
\end{equation}
where $\phi(t)$ is the standard Gaussian probability density function.
Then it follows that
\begin{align}
\frac{\partial}{\partial x}\erf(A(x)) & =2\sqrt{2}\left(\frac{1}{\sqrt{2\pi}}\right)\sum_{k,\ell}\int_{0}^{1}ds\,\int_{-\infty}^{\infty}dt\,\phi(t)e^{it\sqrt{2}\left(s\lambda_{k}+\left(1-s\right)\lambda_{\ell}\right)}\Pi_{k}\left(\frac{\partial}{\partial x}A(x)\right)\Pi_{\ell}\\
 & =\frac{2}{\sqrt{\pi}}\int_{-\infty}^{\infty}dt\,\phi(t)\int_{0}^{1}ds\,\sum_{k,\ell}e^{it\sqrt{2}\left(s\lambda_{k}+\left(1-s\right)\lambda_{\ell}\right)}\Pi_{k}\left(\frac{\partial}{\partial x}A(x)\right)\Pi_{\ell}\\
 & =\frac{2}{\sqrt{\pi}}\int_{-\infty}^{\infty}dt\,\phi(t)\int_{0}^{1}ds\,\left(\sum_{k}e^{i\sqrt{2}ts\lambda_{k}}\Pi_{k}\right)\left(\frac{\partial}{\partial x}A(x)\right)\sum_{\ell}e^{i\sqrt{2}t\left(1-s\right)\lambda_{\ell}}\Pi_{\ell}\\
 & =\frac{2}{\sqrt{\pi}}\int_{-\infty}^{\infty}dt\,\phi(t)\int_{0}^{1}ds\,e^{i\sqrt{2}tsA(x)}\left(\frac{\partial}{\partial x}A(x)\right)e^{i\sqrt{2}t\left(1-s\right)A(x)},
\end{align}
thus concluding the proof.
\end{proof}
\begin{rem}
For Hermitian matrices $A$ and $H$, we can write the Fr\'echet
derivative of $\erf(A)$ at $H$ as follows:
\begin{align}
D\erf(A)[H] & =\frac{2}{\sqrt{\pi}}\int_{-\infty}^{\infty}dt\,\phi(t)\int_{0}^{1}ds\,e^{i\sqrt{2}tsA}He^{i\sqrt{2}t\left(1-s\right)A}.
\end{align}
\end{rem}

\begin{thm}
The following equality holds:
\begin{equation}
\frac{\partial}{\partial\theta_{j}}\Tr\!\left[\erf\!\left(\frac{\sqrt{2}H(\theta)}{T}\right)\rho\right]=\frac{2}{T}\sqrt{\frac{2}{\pi}}\mathbb{E}_{t\sim\phi,s\sim\upsilon}\!\left[\Re\!\left[\Tr\!\left[H_{j}e^{iH(\theta)2t/T}\mathcal{U}_{2st/T}^{H(\theta)}(\rho)\right]\right]\right],
\end{equation}
where $\phi(t)$ is the standard Gaussian probability density function
in~\eqref{eq:standard-Gaussian-app}, $\upsilon$ is a uniform random
variable on the unit interval $\left[0,1\right]$, and $\mathcal{U}_{t}^{H(\theta)}$
is the following unitary quantum channel:
\begin{equation}
\mathcal{U}_{t}^{H(\theta)}(\cdot)\coloneqq e^{-iH(\theta)t}(\cdot)e^{iH(\theta)t},
\end{equation}
\end{thm}

\begin{proof}
Consider that
\begin{align}
 & \frac{\partial}{\partial\theta_{j}}\Tr\!\left[\erf\!\left(\frac{\sqrt{2}H(\theta)}{T}\right)\rho\right]\nonumber \\
 & =\Tr\!\left[\left(\frac{\partial}{\partial\theta_{j}}\erf\!\left(\frac{\sqrt{2}H(\theta)}{T}\right)\right)\rho\right]\\
 & =\Tr\!\left[\frac{2}{\sqrt{\pi}}\int_{-\infty}^{\infty}dt\,\phi(t)\int_{0}^{1}ds\,e^{i\sqrt{2}ts\left(\sqrt{2}H(\theta)/T\right)}\left(\frac{\partial}{\partial\theta_{j}}\frac{\sqrt{2}H(\theta)}{T}\right)e^{i\sqrt{2}t\left(1-s\right)\left(\sqrt{2}H(\theta)/T\right)}\rho\right]\\
 & =\frac{2}{T}\sqrt{\frac{2}{\pi}}\int_{-\infty}^{\infty}dt\,\phi(t)\int_{0}^{1}ds\,\Tr\!\left[e^{i2tsH(\theta)/T}H_{j}e^{i2t\left(1-s\right)H(\theta)/T}\rho\right]\\
 & =\frac{2}{T}\sqrt{\frac{2}{\pi}}\int_{-\infty}^{\infty}dt\,\phi(t)\int_{0}^{1}ds\,\Tr\!\left[H_{j}e^{i2t\left(1-s\right)H(\theta)/T}\rho e^{i2tsH(\theta)/T}\right].
\end{align}
Using the conventions in the statement of the lemma, we can rewrite
the expression once more as follows:
\begin{equation}
\frac{\partial}{\partial\theta_{j}}\Tr\!\left[\erf\!\left(\frac{\sqrt{2}H(\theta)}{T}\right)\rho\right]=\frac{2}{T}\sqrt{\frac{2}{\pi}}\mathbb{E}_{t\sim\phi,s\sim\upsilon}\!\left[\Tr\!\left[H_{j}e^{iH(\theta)2t/T}\mathcal{U}_{2st/T}^{H(\theta)}(\rho)\right]\right].
\end{equation}
We conclude the statement of the theorem by employing the same arguments
used at the end of the proof of Theorem~\ref{thm:gradient-obj-func}.
\end{proof}

\subsection{Proof of Equation~\eqref{eq:convolve-relu-varphi-Gauss}}

\label{app:relu-normal-conv}

Our goal is to prove~\eqref{eq:convolve-relu-varphi-Gauss}, which
we recall here:
\begin{align}
\grelu_{T}(x) & \coloneqq x\Phi\!\left(\frac{x}{T}\right)+T\phi\!\left(\frac{x}{T}\right)\\
 & =T_{2}\left(\relu*\phi_{T_{1}}\right)\left(\frac{x}{T_{2}}\right),\label{eq:convolve-relu-varphi-Gauss-1}
\end{align}
where $T=T_{1}T_{2}$.

Consider that, for all $a\in\mathbb{R}$ and $T>0$,
\begin{align}
\left(\relu*\phi_{T}\right)\left(a\right) & =\int_{-\infty}^{\infty}dp\,\relu(p)\phi_{T}(a-p)\\
 & =\int_{-\infty}^{\infty}dp\,\relu(p)\phi_{T}(p-a)\\
 & =\frac{1}{\sqrt{2\pi}T}\int_{-\infty}^{\infty}dp\,\relu(p)\exp\!\left(-\frac{\left(p-a\right)^{2}}{2T^{2}}\right)\\
 & =\frac{1}{\sqrt{2\pi}T}\int_{0}^{\infty}dp\,p\exp\!\left(-\frac{\left(p-a\right)^{2}}{2T^{2}}\right)\\
 & =\frac{1}{\sqrt{2\pi}T}\int_{-a}^{\infty}dp'\,\left(p'+a\right)\exp\!\left(-\frac{\left(p'\right)^{2}}{2T^{2}}\right)\\
 & =\frac{1}{\sqrt{2\pi}}\int_{-a/T}^{\infty}dz\,\left(Tz+a\right)\exp\!\left(-\frac{z^{2}}{2}\right)\\
 & =T\frac{1}{\sqrt{2\pi}}\int_{-a/T}^{\infty}dz\,z\exp\!\left(-\frac{z^{2}}{2}\right)+a\frac{1}{\sqrt{2\pi}}\int_{-a/T}^{\infty}dz\,\exp\!\left(-\frac{z^{2}}{2}\right)\\
 & \overset{(a)}{=}T\frac{1}{\sqrt{2\pi}}\left[\left.-\exp\!\left(-\frac{z^{2}}{2}\right)\right|_{-a/T}^{\infty}\right]+a\Phi(a/T)\\
 & =T\frac{1}{\sqrt{2\pi}}\exp\!\left(-\frac{\left(a/T\right)^{2}}{2}\right)+a\Phi(a/T)\\
 & =T\phi\!\left(\frac{a}{T}\right)+a\Phi\!\left(\frac{a}{T}\right),
\end{align}
where the equality (a) follows because
\begin{align}
\frac{1}{\sqrt{2\pi}}\int_{-a/T}^{\infty}dz\,\exp\!\left(-\frac{z^{2}}{2}\right) & =\frac{1}{\sqrt{2\pi}}\int_{-\infty}^{\infty}dz\,\exp\!\left(-\frac{z^{2}}{2}\right)-\frac{1}{\sqrt{2\pi}}\int_{-\infty}^{-a/T}dz\,\exp\!\left(-\frac{z^{2}}{2}\right)\\
 & =1-\Phi(-a/T)\\
 & =\Phi(a/T).
\end{align}
Now setting $T=T_{1}T_{2}$, for $T_{1},T_{2}>0$, this implies that
\begin{align}
T_{2}\left(\relu*\phi_{T_{1}}\right)\left(\frac{x}{T_{2}}\right) & =T_{2}\left[T_{1}\phi\!\left(\frac{x}{T_{1}T_{2}}\right)+\frac{x}{T_{2}}\Phi\!\left(\frac{x}{T_{1}T_{2}}\right)\right]\\
 & =T_{1}T_{2}\phi\!\left(\frac{x}{T_{1}T_{2}}\right)+x\Phi\!\left(\frac{x}{T_{1}T_{2}}\right)\\
 & =T\phi\!\left(\frac{x}{T}\right)+x\Phi\!\left(\frac{x}{T}\right).
\end{align}

\subsection{Derivative of Gaussian smoothed rectified linear unit (GReLU)}

\label{app:deriv-GReLU}

We begin by proving~\eqref{eq:deriv-grelu}. Consider that
\begin{align}
\frac{\partial}{\partial x}\left[\grelu_{T}(x)\right] & =\frac{\partial}{\partial x}\left[x\Phi\!\left(\frac{x}{T}\right)+T\phi\!\left(\frac{x}{T}\right)\right]\\
 & =\Phi\!\left(\frac{x}{T}\right)+x\left[\frac{\partial}{\partial x}\Phi\!\left(\frac{x}{T}\right)\right]+T\frac{\partial}{\partial x}\phi\!\left(\frac{x}{T}\right)\\
 & =\Phi\!\left(\frac{x}{T}\right)+\frac{x}{T}\phi\!\left(\frac{x}{T}\right)-T\left(\frac{x}{T}\right)\phi\!\left(\frac{x}{T}\right)\frac{1}{T}\\
 & =\Phi\!\left(\frac{x}{T}\right).
\end{align}

\begin{lem}
\label{lem:grelu-deriv}The following equality holds:
\begin{equation}
\frac{\partial}{\partial x}\left[\grelu_{T}(x)\right]=\frac{1}{2}+\frac{1}{\sqrt{2\pi}}\left(\frac{x}{T}\right)\mathbb{E}_{t\sim\nu_{T},v\sim\upsilon}\!\left[e^{-ixvt}\right],
\end{equation}
where $\nu_{T}(t)$ is the following Gaussian probability density
function on $t\in\mathbb{R}$:
\begin{equation}
\nu_{T}(t)\coloneqq\frac{T}{\sqrt{2\pi}}e^{-T^{2}t^{2}/2},\label{eq:gaussian-prob-dens-T-1}
\end{equation}
and $\upsilon$ is a uniform random variable over the unit interval
$\left[0,1\right]$.
\end{lem}

\begin{proof}
Consider that
\begin{align}
\frac{\partial}{\partial x}\left[\grelu_{T}(x)\right] & =\Phi\!\left(\frac{x}{T}\right)\\
 & =\frac{1}{2}+\frac{1}{2}\erf\!\left(\frac{x}{\sqrt{2}T}\right)\\
 & =\frac{1}{2}+\frac{x}{\sqrt{2\pi}T}\int_{0}^{1}dv\,e^{-\frac{x^{2}v^{2}}{2T^{2}}}\\
 & =\frac{1}{2}+\frac{x}{T}\int_{0}^{1}dv\,\left(\frac{T}{\sqrt{2\pi}}\int_{-\infty}^{\infty}dt\,e^{-T^{2}t^{2}/2}\left(\frac{1}{\sqrt{2\pi}}e^{-ixvt}\right)\right)\\
 & =\frac{1}{2}+\frac{x}{\sqrt{2\pi}T}\mathbb{E}_{t\sim\nu_{T},v\sim\upsilon}\!\left[e^{-ixvt}\right]\\
 & =\frac{1}{2}+\frac{1}{\sqrt{2\pi}}\left(\frac{x}{T}\right)\mathbb{E}_{t\sim\nu_{T},v\sim\upsilon}\!\left[e^{-ixvt}\right],
\end{align}
thus concluding the proof.
\end{proof}
\begin{lem}
\label{lem:grad-grelu}Let $x\in\mathbb{R}$, let $x\mapsto A(x)$
be a Hermitian-valued function, and let $T>0$. Then the following
equality holds:
\begin{align}
\frac{\partial}{\partial x}\grelu_{T}(A(x)) & =\frac{1}{2}\frac{\partial}{\partial x}A(x)+\sqrt{\frac{2}{\pi}}\left(\frac{1}{T}\right)\mathbb{E}_{\substack{t\sim\nu_{T},\\
v,s\sim\upsilon
}
}\!\left[s\,\Re\!\left[A(x)e^{-isA(x)vt}\left(\frac{\partial}{\partial x}A(x)\right)e^{-i\left(1-s\right)A(x)vt}\right]\right],
\end{align}
where $\nu_{T}$ is the probability measure defined in~\eqref{eq:gaussian-prob-dens-T-1}
and $\upsilon$ is a uniform random variable over the unit interval
$\left[0,1\right]$.
\end{lem}

\begin{proof}
To begin with, let
\begin{equation}
A(x)=\sum_{k}\lambda_{k}\Pi_{k}
\end{equation}
be a spectral decomposition of $A(x)$, where we have omitted the
dependence on $x$ in the eigenvalues and eigenprojections. Consider
that the derivative of the $\grelu_{T}$ function is given by
\begin{equation}
\frac{\partial}{\partial x}\grelu_{T}(A(x))=\sum_{k,\ell}f_{\grelu_{T}}^{\left[1\right]}(\lambda_{k},\lambda_{\ell})\Pi_{k}\left(\frac{\partial}{\partial x}A(x)\right)\Pi_{\ell},\label{eq:1st-divided-diff-gelu-1}
\end{equation}
where $f_{\grelu_{T}}^{\left[1\right]}$ denotes the first divided
difference of $\grelu_{T}$. An expression for $f_{\grelu_{T}}^{\left[1\right]}$
is as follows:
\begin{equation}
f_{\grelu_{T}}^{\left[1\right]}(y_{1},y_{2})=\int_{0}^{1}ds\,\left(\frac{1}{2}+\frac{\left(sy_{1}+\left(1-s\right)y_{2}\right)}{\sqrt{2\pi}T}\mathbb{E}_{t\sim\nu_{T},v\sim\upsilon}\!\left[e^{-i\left(sy_{1}+\left(1-s\right)y_{2}\right)vt}\right]\right),
\end{equation}
where we used Lemma~\ref{lem:grelu-deriv} and the fact that
\begin{equation}
f^{\left[1\right]}(y_{1},y_{2})=\int_{0}^{1}ds\,f'(sy_{1}+\left(1-s\right)y_{2}).
\end{equation}
Defining $\lambda_{s,k,\ell}\equiv s\lambda_{k}+\left(1-s\right)\lambda_{\ell}$,
then we find that
\begin{align}
 & \frac{\partial}{\partial x}\grelu_{T}(A(x))\nonumber \\
 & =\sum_{k,\ell}\left(\int_{0}^{1}ds\,\left(\frac{1}{2}+\frac{\lambda_{s,k,\ell}}{\sqrt{2\pi}T}\mathbb{E}_{t\sim\nu_{T},v\sim\upsilon}\!\left[e^{-i\lambda_{s,k,\ell}vt}\right]\right)\right)\Pi_{k}\left(\frac{\partial}{\partial x}A(x)\right)\Pi_{\ell}\\
 & =\sum_{k,\ell}\left(\frac{1}{2}+\int_{0}^{1}ds\,\frac{\lambda_{s,k,\ell}}{\sqrt{2\pi}T}\mathbb{E}_{t\sim\nu_{T},v\sim\upsilon}\!\left[e^{-i\lambda_{s,k,\ell}vt}\right]\right)\Pi_{k}\left(\frac{\partial}{\partial x}A(x)\right)\Pi_{\ell}\\
 & =\frac{1}{2}\sum_{k,\ell}\Pi_{k}\left(\frac{\partial}{\partial x}A(x)\right)\Pi_{\ell}\nonumber \\
 & \qquad+\frac{1}{\sqrt{2\pi}T}\sum_{k,\ell}\left(\int_{0}^{1}ds\,\lambda_{s,k,\ell}\mathbb{E}_{t\sim\nu_{T},v\sim\upsilon}\!\left[e^{-i\lambda_{s,k,\ell}vt}\right]\right)\Pi_{k}\left(\frac{\partial}{\partial x}A(x)\right)\Pi_{\ell}\\
 & =\frac{1}{2}\frac{\partial}{\partial x}A(x)\nonumber \\
 & \qquad+\frac{1}{\sqrt{2\pi}T}\int_{0}^{1}ds\,\sum_{k,\ell}\left(s\lambda_{k}+\left(1-s\right)\lambda_{\ell}\right)\mathbb{E}_{t\sim\nu_{T},v\sim\upsilon}\!\left[e^{-i\left(s\lambda_{k}+\left(1-s\right)\lambda_{\ell}\right)vt}\right]\Pi_{k}\left(\frac{\partial}{\partial x}A(x)\right)\Pi_{\ell}.
\end{align}
Now consider that
\begin{align}
 & \int_{0}^{1}ds\,\sum_{k,\ell}\left(s\lambda_{k}+\left(1-s\right)\lambda_{\ell}\right)\mathbb{E}_{t\sim\nu_{T},v\sim\upsilon}\!\left[e^{-i\left(s\lambda_{k}+\left(1-s\right)\lambda_{\ell}\right)vt}\right]\Pi_{k}\left(\frac{\partial}{\partial x}A(x)\right)\Pi_{\ell}\nonumber \\
 & =\mathbb{E}_{\substack{t\sim\nu_{T},\\
v\sim\upsilon
}
}\!\left[\int_{0}^{1}ds\,\sum_{k,\ell}\left(s\lambda_{k}+\left(1-s\right)\lambda_{\ell}\right)e^{-i\left(s\lambda_{k}+\left(1-s\right)\lambda_{\ell}\right)vt}\Pi_{k}\left(\frac{\partial}{\partial x}A(x)\right)\Pi_{\ell}\right]\label{eq:grelu-deriv-proof-1}\\
 & =\mathbb{E}_{\substack{t\sim\nu_{T},\\
v\sim\upsilon
}
}\!\left[\int_{0}^{1}ds\,\sum_{k,\ell}\left(s\lambda_{k}+\left(1-s\right)\lambda_{\ell}\right)e^{-is\lambda_{k}vt}\Pi_{k}\left(\frac{\partial}{\partial x}A(x)\right)e^{-i\left(1-s\right)\lambda_{\ell}vt}\Pi_{\ell}\right]\\
 & =\mathbb{E}_{\substack{t\sim\nu_{T},\\
v\sim\upsilon
}
}\!\left[\int_{0}^{1}ds\,s\sum_{k,\ell}\lambda_{k}e^{-is\lambda_{k}vt}\Pi_{k}\left(\frac{\partial}{\partial x}A(x)\right)e^{-i\left(1-s\right)\lambda_{\ell}vt}\Pi_{\ell}\right]\nonumber \\
 & \qquad+\mathbb{E}_{\substack{t\sim\nu_{T},\\
v\sim\upsilon
}
}\!\left[\int_{0}^{1}ds\,\left(1-s\right)\sum_{k,\ell}e^{-is\lambda_{k}vt}\Pi_{k}\left(\frac{\partial}{\partial x}A(x)\right)\lambda_{\ell}e^{-i\left(1-s\right)\lambda_{\ell}vt}\Pi_{\ell}\right]\\
 & =\mathbb{E}_{\substack{t\sim\nu_{T},\\
v\sim\upsilon
}
}\!\left[\int_{0}^{1}ds\,s\left(\sum_{k}\lambda_{k}e^{-is\lambda_{k}vt}\Pi_{k}\right)\left(\frac{\partial}{\partial x}A(x)\right)\left(\sum_{\ell}e^{-i\left(1-s\right)\lambda_{\ell}vt}\Pi_{\ell}\right)\right]\nonumber \\
 & \qquad+\mathbb{E}_{\substack{t\sim\nu_{T},\\
v\sim\upsilon
}
}\!\left[\int_{0}^{1}ds\,\left(1-s\right)\left(\sum_{k}e^{-is\lambda_{k}vt}\Pi_{k}\right)\left(\frac{\partial}{\partial x}A(x)\right)\left(\sum_{\ell}\lambda_{\ell}e^{-i\left(1-s\right)\lambda_{\ell}vt}\Pi_{\ell}\right)\right]\\
 & =\mathbb{E}_{\substack{t\sim\nu_{T},\\
v\sim\upsilon
}
}\!\left[\int_{0}^{1}ds\,s\,A(x)e^{-isA(x)vt}\left(\frac{\partial}{\partial x}A(x)\right)e^{-i\left(1-s\right)A(x)vt}\right]\nonumber \\
 & \qquad+\mathbb{E}_{\substack{t\sim\nu_{T},\\
v\sim\upsilon
}
}\!\left[\int_{0}^{1}ds\,\left(1-s\right)e^{-isA(x)vt}\left(\frac{\partial}{\partial x}A(x)\right)e^{-i\left(1-s\right)A(x)vt}A(x)\right]\\
 & =\mathbb{E}_{\substack{t\sim\nu_{T},\\
v\sim\upsilon,s\sim\upsilon
}
}\!\left[s\,A(x)e^{-isA(x)vt}\left(\frac{\partial}{\partial x}A(x)\right)e^{-i\left(1-s\right)A(x)vt}\right]\nonumber \\
 & \qquad+\mathbb{E}_{\substack{t\sim\nu_{T},\\
v\sim\upsilon
}
}\!\left[\int_{0}^{1}ds\,s\,e^{-i\left(1-s\right)A(x)vt}\left(\frac{\partial}{\partial x}A(x)\right)e^{-isA(x)vt}A(x)\right]\\
 & =\mathbb{E}_{\substack{t\sim\nu_{T},\\
v\sim\upsilon,s\sim\upsilon
}
}\!\left[s\,A(x)e^{-isA(x)vt}\left(\frac{\partial}{\partial x}A(x)\right)e^{-i\left(1-s\right)A(x)vt}\right]\nonumber \\
 & \qquad+\mathbb{E}_{\substack{t\sim\nu_{T},\\
v\sim\upsilon
}
}\!\left[\int_{0}^{1}ds\,s\,e^{i\left(1-s\right)A(x)vt}\left(\frac{\partial}{\partial x}A(x)\right)e^{isA(x)vt}A(x)\right]\\
 & =\mathbb{E}_{\substack{t\sim\nu_{T},\\
v\sim\upsilon,s\sim\upsilon
}
}\!\left[s\,A(x)e^{-isA(x)vt}\left(\frac{\partial}{\partial x}A(x)\right)e^{-i\left(1-s\right)A(x)vt}\right]\nonumber \\
 & \qquad+\mathbb{E}_{\substack{t\sim\nu_{T},\\
v\sim\upsilon,s\sim\upsilon
}
}\!\left[s\,e^{i\left(1-s\right)A(x)vt}\left(\frac{\partial}{\partial x}A(x)\right)e^{isA(x)vt}A(x)\right]\\
 & =2\mathbb{E}_{\substack{t\sim\nu_{T},\\
v\sim\upsilon,s\sim\upsilon
}
}\!\left[s\,\Re\!\left[A(x)e^{-isA(x)vt}\left(\frac{\partial}{\partial x}A(x)\right)e^{-i\left(1-s\right)A(x)vt}\right]\right],\label{eq:grelu-deriv-proof-last}
\end{align}
thus concluding the proof.
\end{proof}
\begin{thm}
The following equality holds:
\begin{multline}
\frac{\partial}{\partial\theta_{j}}\Tr\!\left[\grelu_{T}(H(\theta))\rho\right]=\frac{1}{2}\Tr\!\left[H_{j}\rho\right]\\
+\sqrt{\frac{2}{\pi}}\left(\frac{1}{T}\right)\left\Vert \theta\right\Vert _{1}\mathbb{E}_{\substack{t\sim\nu_{T},k\sim q,\\
v,s\sim\upsilon
}
}\!\left[s\,\Re\!\left[\Tr\!\left[\signum(\theta_{k})H_{k}H_{j}e^{iH(\theta)vt}\mathcal{U}_{svt}^{H(\theta)}\left(\rho\right)\right]\right]\right],
\end{multline}
where the objective function $\Tr\!\left[\grelu_{T}(H(\theta))\rho\right]$
is defined from~\eqref{eq:grelu-def}, the probability densities $\nu_{T}$
and $\upsilon$ are defined in Lemma~\ref{lem:grad-grelu}, $q$ is
the following probability distribution:
\begin{equation}
q(k)\coloneqq\frac{\left|\theta_{k}\right|}{\left\Vert \theta\right\Vert _{1}},
\end{equation}
and $\mathcal{U}_{t}^{H(\theta)}$ is the following unitary quantum
channel:
\begin{equation}
\mathcal{U}_{t}^{H(\theta)}(\cdot)\coloneqq e^{-iH(\theta)t}(\cdot)e^{iH(\theta)t}.
\end{equation}
\end{thm}

\begin{proof}
Applying Lemma~\ref{lem:grad-grelu}, consider that
\begin{align}
 & \frac{\partial}{\partial\theta_{j}}\Tr\!\left[\grelu_{T}(H(\theta))\rho\right]\nonumber \\
 & =\Tr\!\left[\left(\frac{\partial}{\partial\theta_{j}}\grelu_{T}(H(\theta))\right)\rho\right]\\
 & =\frac{1}{2}\Tr\!\left[H_{j}\rho\right]+\sqrt{\frac{2}{\pi}}\left(\frac{1}{T}\right)\mathbb{E}_{\substack{t\sim\nu_{T},\\
v,s\sim\upsilon
}
}\!\left[\Tr\!\left[s\,\Re\!\left[H(\theta)e^{-isH(\theta)vt}\left(\frac{\partial}{\partial\theta_{j}}H(\theta)\right)e^{-i\left(1-s\right)H(\theta)vt}\right]\rho\right]\right]\label{eq:grelu-more-proof-1}\\
 & =\frac{1}{2}\Tr\!\left[H_{j}\rho\right]+\sqrt{\frac{2}{\pi}}\left(\frac{1}{T}\right)\mathbb{E}_{\substack{t\sim\nu_{T},\\
v,s\sim\upsilon
}
}\!\left[s\,\Tr\!\left[\Re\!\left[H(\theta)e^{-isH(\theta)vt}H_{j}e^{-i\left(1-s\right)H(\theta)vt}\right]\rho\right]\right]\\
 & =\frac{1}{2}\Tr\!\left[H_{j}\rho\right]+\sqrt{\frac{2}{\pi}}\left(\frac{1}{T}\right)\mathbb{E}_{\substack{t\sim\nu_{T},\\
v,s\sim\upsilon
}
}\!\left[s\,\Re\!\left[\Tr\!\left[e^{-isH(\theta)vt}H(\theta)H_{j}e^{-i\left(1-s\right)H(\theta)vt}\rho\right]\right]\right]\\
 & =\frac{1}{2}\Tr\!\left[H_{j}\rho\right]+\sqrt{\frac{2}{\pi}}\left(\frac{1}{T}\right)\mathbb{E}_{\substack{t\sim\nu_{T},\\
v,s\sim\upsilon
}
}\!\left[s\,\Re\!\left[\Tr\!\left[H(\theta)H_{j}e^{-iH(\theta)vt}e^{isH(\theta)vt}\rho e^{-isvtH(\theta)}\right]\right]\right]\\
 & =\frac{1}{2}\Tr\!\left[H_{j}\rho\right]+\sqrt{\frac{2}{\pi}}\left(\frac{1}{T}\right)\mathbb{E}_{\substack{t\sim\nu_{T},\\
v,s\sim\upsilon
}
}\!\left[s\,\Re\!\left[\Tr\!\left[H(\theta)H_{j}e^{iH(\theta)vt}e^{-isH(\theta)vt}\rho e^{isvtH(\theta)}\right]\right]\right]\\
 & =\frac{1}{2}\Tr\!\left[H_{j}\rho\right]+\sqrt{\frac{2}{\pi}}\left(\frac{1}{T}\right)\mathbb{E}_{\substack{t\sim\nu_{T},\\
v,s\sim\upsilon
}
}\!\left[s\,\Re\!\left[\Tr\!\left[H(\theta)H_{j}e^{iH(\theta)vt}\mathcal{U}_{svt}^{H(\theta)}\left(\rho\right)\right]\right]\right]\\
 & =\frac{1}{2}\Tr\!\left[H_{j}\rho\right]+\sqrt{\frac{2}{\pi}}\left(\frac{1}{T}\right)\left\Vert \theta\right\Vert _{1}\mathbb{E}_{\substack{t\sim\nu_{T},\\
v,s\sim\upsilon,\\
k\sim q
}
}\!\left[s\,\Re\!\left[\Tr\!\left[\signum(\theta_{k})H_{k}H_{j}e^{iH(\theta)vt}\mathcal{U}_{svt}^{H(\theta)}\left(\rho\right)\right]\right]\right],\label{eq:grelu-more-proof-last}
\end{align}
thus concluding the proof.
\end{proof}

\subsection{Derivative of Gaussian error linear unit (GeLU)}

\label{app:deriv-GeLU}
\begin{lem}
\label{lem:gelu-deriv}The following equality holds:
\begin{equation}
\frac{\partial}{\partial x}\left[x\Phi\!\left(\frac{x}{T}\right)\right]=\frac{1}{2}+\sqrt{\frac{2}{\pi}}\left(\frac{x}{T}\right)\mathbb{E}_{t\sim\nu_{T},v\sim\kappa}\!\left[e^{-ixvt}\right],
\end{equation}
where $\nu_{T}(t)$ is the following Gaussian probability density
function on $t\in\mathbb{R}$:
\begin{equation}
\nu_{T}(t)\coloneqq\frac{T}{\sqrt{2\pi}}e^{-T^{2}t^{2}/2},\label{eq:gaussian-prob-dens-T}
\end{equation}
and $\kappa$ is the following probability measure:
\begin{equation}
d\kappa(v)\coloneqq\frac{1}{2}\delta_{1}(dv)+\frac{1}{2}\mathbf{1}_{\left[0,1\right]}(v)\,dv.\label{eq:kappa-funny-prob-meas}
\end{equation}
\end{lem}

\begin{proof}
Consider that
\begin{align}
\frac{\partial}{\partial x}\left[x\Phi\!\left(\frac{x}{T}\right)\right] & =\Phi\!\left(\frac{x}{T}\right)+\frac{x}{T}\phi\!\left(\frac{x}{T}\right)\\
 & =\frac{1}{2}\left(1+\erf\!\left(\frac{x}{\sqrt{2}T}\right)\right)+\frac{x}{T}\phi\!\left(\frac{x}{T}\right)\\
 & =\frac{1}{2}+\frac{1}{2}\erf\!\left(\frac{x}{\sqrt{2}T}\right)+\frac{x}{T}\phi\!\left(\frac{x}{T}\right).
\end{align}
Furthermore,
\begin{align}
\frac{x}{T}\phi\!\left(\frac{x}{T}\right) & =\frac{x}{T}\left(\frac{1}{\sqrt{2\pi}}e^{-\frac{x^{2}}{2T^{2}}}\right)\\
 & =\frac{x}{\sqrt{2\pi}T}\left(\frac{T}{\sqrt{2\pi}}\int_{-\infty}^{\infty}dt\,e^{-T^{2}t^{2}/2}e^{-ixt}\right)\\
 & =\frac{x}{\sqrt{2\pi}T}\mathbb{E}_{t\sim\nu_{T}}\!\left[e^{-ixt}\right].
\end{align}
Also,
\begin{align}
\Phi\!\left(\frac{x}{T}\right) & =\frac{1}{2}+\frac{1}{2}\erf\!\left(\frac{x}{\sqrt{2}T}\right)\\
 & =\frac{1}{2}+\frac{x}{\sqrt{2\pi}T}\int_{0}^{1}dv\,e^{-\frac{x^{2}v^{2}}{2T^{2}}}\\
 & =\frac{1}{2}+\frac{x}{T}\int_{0}^{1}dv\,\left(\frac{T}{\sqrt{2\pi}}\int_{-\infty}^{\infty}dt\,e^{-T^{2}t^{2}/2}\left(\frac{1}{\sqrt{2\pi}}e^{-ixvt}\right)\right)\\
 & =\frac{1}{2}+\frac{x}{\sqrt{2\pi}T}\mathbb{E}_{t\sim\nu_{T},v\sim\upsilon}\!\left[e^{-ixvt}\right].
\end{align}
So this implies that
\begin{equation}
\frac{\partial}{\partial x}\left[x\Phi\!\left(\frac{x}{T}\right)\right]=\frac{1}{2}+\frac{x}{\sqrt{2\pi}T}\left(\mathbb{E}_{t\sim\nu_{T},v\sim\upsilon}\!\left[e^{-ixvt}\right]+\mathbb{E}_{t\sim\nu_{T}}\!\left[e^{-ixt}\right]\right).
\end{equation}
Using the probability measure $\kappa$, it follows that
\begin{equation}
\mathbb{E}_{v\sim\kappa}\!\left[e^{-ixvt}\right]=\frac{1}{2}e^{-ixt}+\frac{1}{2}\int_{0}^{1}dv\,e^{-ixvt},
\end{equation}
so that
\begin{align}
\frac{\partial}{\partial x}\left[x\Phi\!\left(\frac{x}{T}\right)\right] & =\frac{1}{2}+\frac{2x}{\sqrt{2\pi}T}\mathbb{E}_{t\sim\nu_{T},v\sim\kappa}\!\left[e^{-ixvt}\right]\\
 & =\frac{1}{2}+\sqrt{\frac{2}{\pi}}\left(\frac{x}{T}\right)\mathbb{E}_{t\sim\nu_{T},v\sim\kappa}\!\left[e^{-ixvt}\right],
\end{align}
thus concluding the proof.
\end{proof}
\begin{lem}
\label{lem:grad-gelu}Let $x\in\mathbb{R}$, let $x\mapsto A(x)$
be a Hermitian-valued function, and let $T>0$. Then the following
equality holds:
\begin{equation}
\frac{\partial}{\partial x}\gelu_{T}(A(x))=\frac{1}{2}\frac{\partial}{\partial x}A(x)+\sqrt{\frac{2}{\pi}}\left(\frac{2}{T}\right)\mathbb{E}_{\substack{t\sim\nu_{T},\\
v\sim\kappa,\\
s\sim\upsilon
}
}\!\left[s\,\Re\!\left[A(x)e^{-isA(x)vt}\left(\frac{\partial}{\partial x}A(x)\right)e^{-i\left(1-s\right)A(x)vt}\right]\right],
\end{equation}
where $\nu_{T}$ is the probability measure defined in~\eqref{eq:gaussian-prob-dens-T},
$\kappa$ is the probability measure defined in~\eqref{eq:kappa-funny-prob-meas},
and $\upsilon$ is a uniform random variable over the unit interval
$\left[0,1\right]$.
\end{lem}

\begin{proof}
To begin with, let
\begin{equation}
A(x)=\sum_{k}\lambda_{k}\Pi_{k}
\end{equation}
be a spectral decomposition of $A(x)$, where we have omitted the
dependence on $x$ in the eigenvalues and eigenprojections. Consider
that the derivative of the $\gelu_{T}$ function is given by
\begin{equation}
\frac{\partial}{\partial x}\gelu_{T}(A(x))=\sum_{k,\ell}f_{\gelu_{T}}^{\left[1\right]}(\lambda_{k},\lambda_{\ell})\Pi_{k}\left(\frac{\partial}{\partial x}A(x)\right)\Pi_{\ell},\label{eq:1st-divided-diff-gelu}
\end{equation}
where $f_{\gelu_{T}}^{\left[1\right]}$ denotes the first divided
difference of $\gelu_{T}$. An expression for $f_{\gelu_{T}}^{\left[1\right]}$
is as follows:
\begin{equation}
f_{\gelu_{T}}^{\left[1\right]}(y_{1},y_{2})=\int_{0}^{1}ds\,\left(\frac{1}{2}+\frac{2\left(sy_{1}+\left(1-s\right)y_{2}\right)}{\sqrt{2\pi}T}\mathbb{E}_{t\sim\nu_{T},v\sim\kappa}\!\left[e^{-i\left(sy_{1}+\left(1-s\right)y_{2}\right)vt}\right]\right),
\end{equation}
where we used Lemma~\ref{lem:gelu-deriv} and the fact that
\begin{equation}
f^{\left[1\right]}(y_{1},y_{2})=\int_{0}^{1}ds\,f'(sy_{1}+\left(1-s\right)y_{2}).
\end{equation}
Defining $\lambda_{s,k,\ell}\equiv s\lambda_{k}+\left(1-s\right)\lambda_{\ell}$,
then we find that
\begin{align}
 & \frac{\partial}{\partial x}\gelu_{T}(A(x))\nonumber \\
 & =\sum_{k,\ell}\left(\int_{0}^{1}ds\,\left(\frac{1}{2}+\frac{2\lambda_{s,k,\ell}}{\sqrt{2\pi}T}\mathbb{E}_{t\sim\nu_{T},v\sim\kappa}\!\left[e^{-i\lambda_{s,k,\ell}vt}\right]\right)\right)\Pi_{k}\left(\frac{\partial}{\partial x}A(x)\right)\Pi_{\ell}\\
 & =\sum_{k,\ell}\left(\frac{1}{2}+\int_{0}^{1}ds\,\frac{2\lambda_{s,k,\ell}}{\sqrt{2\pi}T}\mathbb{E}_{t\sim\nu_{T},v\sim\kappa}\!\left[e^{-i\lambda_{s,k,\ell}vt}\right]\right)\Pi_{k}\left(\frac{\partial}{\partial x}A(x)\right)\Pi_{\ell}\\
 & =\frac{1}{2}\sum_{k,\ell}\Pi_{k}\left(\frac{\partial}{\partial x}A(x)\right)\Pi_{\ell}\nonumber \\
 & \qquad+\sqrt{\frac{2}{\pi}}\left(\frac{1}{T}\right)\sum_{k,\ell}\left(\int_{0}^{1}ds\,\lambda_{s,k,\ell}\mathbb{E}_{t\sim\nu_{T},v\sim\kappa}\!\left[e^{-i\lambda_{s,k,\ell}vt}\right]\right)\Pi_{k}\left(\frac{\partial}{\partial x}A(x)\right)\Pi_{\ell}\\
 & =\frac{1}{2}\frac{\partial}{\partial x}A(x)\nonumber \\
 & \qquad+\sqrt{\frac{2}{\pi}}\left(\frac{1}{T}\right)\int_{0}^{1}ds\,\sum_{k,\ell}\left(s\lambda_{k}+\left(1-s\right)\lambda_{\ell}\right)\mathbb{E}_{t\sim\nu_{T},v\sim\kappa}\!\left[e^{-i\left(s\lambda_{k}+\left(1-s\right)\lambda_{\ell}\right)vt}\right]\Pi_{k}\left(\frac{\partial}{\partial x}A(x)\right)\Pi_{\ell}.
\end{align}
Following the same steps as in~\eqref{eq:grelu-deriv-proof-1}--\eqref{eq:grelu-deriv-proof-last},
but with $v\sim\kappa$ instead of $v\sim\upsilon$, we conclude that
\begin{multline}
\int_{0}^{1}ds\,\sum_{k,\ell}\left(s\lambda_{k}+\left(1-s\right)\lambda_{\ell}\right)\mathbb{E}_{t\sim\nu_{T},v\sim\kappa}\!\left[e^{-i\left(s\lambda_{k}+\left(1-s\right)\lambda_{\ell}\right)vt}\right]\Pi_{k}\left(\frac{\partial}{\partial x}A(x)\right)\Pi_{\ell}\\
=2\mathbb{E}_{\substack{t\sim\nu_{T},\\
v\sim\kappa,s\sim\upsilon
}
}\!\left[s\,\Re\!\left[A(x)e^{-isA(x)vt}\left(\frac{\partial}{\partial x}A(x)\right)e^{-i\left(1-s\right)A(x)vt}\right]\right],
\end{multline}
thus completing the proof.
\end{proof}
\begin{thm}
The following equality holds:
\begin{multline}
\frac{\partial}{\partial\theta_{j}}\Tr\!\left[\gelu_{T}(H(\theta))\rho\right]=\frac{1}{2}\Tr\!\left[H_{j}\rho\right]\\
+\sqrt{\frac{2}{\pi}}\left(\frac{2}{T}\right)\left\Vert \theta\right\Vert _{1}\mathbb{E}_{\substack{t\sim\nu_{T},v\sim\kappa,\\
s\sim\upsilon,k\sim q
}
}\!\left[s\,\Re\!\left[\Tr\!\left[\signum(\theta_{k})H_{k}H_{j}e^{iH(\theta)vt}\mathcal{U}_{svt}^{H(\theta)}\left(\rho\right)\right]\right]\right],
\end{multline}
where the objective function $\Tr\!\left[\gelu_{T}(H(\theta))\rho\right]$
is defined from~\eqref{eq:gelu-obj}, the probability densities $\nu_{T}$,
$\kappa$, and $\upsilon$ are defined in Lemma~\ref{lem:grad-gelu},
$q$ is the following probability distribution:
\begin{equation}
q(k)\coloneqq\frac{\left|\theta_{k}\right|}{\left\Vert \theta\right\Vert _{1}},
\end{equation}
and $\mathcal{U}_{t}^{H(\theta)}$ is the following unitary quantum
channel:
\begin{equation}
\mathcal{U}_{t}^{H(\theta)}(\cdot)\coloneqq e^{-iH(\theta)t}(\cdot)e^{iH(\theta)t}.
\end{equation}
\end{thm}

\begin{proof}
Applying Lemma~\ref{lem:grad-gelu} and following the same steps as
in~\eqref{eq:grelu-more-proof-1}--\eqref{eq:grelu-more-proof-last},
but with $v\sim\kappa$ instead of $v\sim\upsilon$, consider that
\begin{align}
 & \frac{\partial}{\partial\theta_{j}}\Tr\!\left[\gelu_{T}(H(\theta))\rho\right]\nonumber \\
 & =\Tr\!\left[\left(\frac{\partial}{\partial\theta_{j}}\gelu_{T}(H(\theta))\right)\rho\right]\\
 & =\frac{1}{2}\Tr\!\left[H_{j}\rho\right]+\sqrt{\frac{2}{\pi}}\left(\frac{2}{T}\right)\mathbb{E}_{\substack{t\sim\nu_{T},\\
v\sim\kappa,\\
s\sim\upsilon
}
}\!\left[\Tr\!\left[s\,\Re\!\left[H(\theta)e^{-isH(\theta)vt}\left(\frac{\partial}{\partial\theta_{j}}H(\theta)\right)e^{-i\left(1-s\right)H(\theta)vt}\right]\rho\right]\right]\\
 & =\frac{1}{2}\Tr\!\left[H_{j}\rho\right]+\sqrt{\frac{2}{\pi}}\left(\frac{2}{T}\right)\left\Vert \theta\right\Vert _{1}\mathbb{E}_{\substack{t\sim\nu_{T},\\
v\sim\kappa,\\
s\sim\upsilon,k\sim q
}
}\!\left[s\,\Re\!\left[\Tr\!\left[\signum(\theta_{k})H_{k}H_{j}e^{iH(\theta)vt}\mathcal{U}_{svt}^{H(\theta)}\left(\rho\right)\right]\right]\right],
\end{align}
thus concluding the proof.
\end{proof}

\section{Gradient formulas for quantum observable networks}

\label{app:grad-ham-nets}

In this appendix, we derive formulas for the gradient of two-, three-, and $L$-layer quantum observable networks. In our formulas for general $L$-layer networks, we show how the gradient formulas can be expressed in terms of a backpropagation rule, generalizing that of \cite{rumelhart1986parallel,Rumelhart1986a}. While our derivations focus on hyperbolic tangent activation functions, we note that they can be generalized to an arbitrary activation function.

\subsection{Two-layer gradient formulas}

In this appendix, we show how to calculate the gradient for a two-layer
quantum observable network of the form in~\eqref{eq:2-layer-ham-net-1}--\eqref{eq:2-layer-ham-net-last},
specializing to the case when $\varphi_{1}$ and $\varphi_{2}$ are
both $\tanh$.
\begin{thm}
The following formulas hold for the partial derivatives of a quantum observable
network of the form in \eqref{eq:2-layer-ham-net-1}--\eqref{eq:2-layer-ham-net-last}:
\begin{align}
\frac{\partial}{\partial\theta_{kj}^{\left(2\right)}}\Tr\!\left[\varphi_{2}\!\left(\sum_{j'=1}^{J_{1}}\theta_{kj'}^{(2)}A_{j'}^{(1)}\right)\rho\right] & =\Tr\!\left[\mathcal{D}_{B_{k}^{\left(2\right)}\!\left(\theta^{\left(2\right)}\right)}\!\left(A_{j}^{(1)}\right)\rho\right],\\
\frac{\partial}{\partial\theta_{ji}^{\left(1\right)}}\Tr\!\left[\varphi_{2}\!\left(\sum_{j'=1}^{J_{1}}\theta_{kj'}^{(2)}\varphi_{1}\!\left(\sum_{i'=1}^{J_{0}}\theta_{j'i'}^{\left(1\right)}H_{i'}\right)\right)\rho\right] & =\theta_{kj}^{(2)}\Tr\!\left[\left(\mathcal{D}_{B_{k}^{\left(2\right)}\!\left(\theta^{\left(2\right)}\right)}\circ\mathcal{D}_{B_{j}^{\left(1\right)}\!\left(\theta^{\left(1\right)}\right)}\right)\left(H_{i}\right)\rho\right],
\end{align}
where $\varphi_{1}$ and $\varphi_{2}$ are both $\tanh$, $\rho$
is a quantum state, the superoperator $\mathcal{D}_{H}$ is defined
for a Hamiltonian $H$ as
\begin{equation}
\mathcal{D}_{H}(X)\coloneqq\mathbb{E}_{t\sim\mu,s\sim\upsilon}\!\left[e^{iHst}Xe^{iH\left(1-s\right)t}\right],
\end{equation}
$\mu(t)$ is the probability density function defined in \eqref{eq:high-peak-tent-2},
$\upsilon$ is a uniform random variable on the unit interval $\left[0,1\right]$,
\begin{align}
B_{k}^{\left(2\right)}\!\left(\theta^{\left(2\right)}\right) & \coloneqq\sum_{j=1}^{J_{1}}\theta_{kj}^{(2)}A_{j}^{(1)},\\
B_{j}^{\left(1\right)}\!\left(\theta^{\left(1\right)}\right) & \coloneqq\sum_{i=1}^{J_{0}}\theta_{ji}^{(1)}A_{i}^{(0)}=\sum_{i=1}^{J_{0}}\theta_{ji}^{(1)}H_{i}.
\end{align}
\end{thm}

\begin{proof}
Applying Lemma \ref{lem:grad-tanh}, i.e.,
\begin{align}
\frac{\partial}{\partial x}\tanh(A(x)) & =\mathbb{E}_{t\sim\mu,s\sim\upsilon}\!\left[e^{itsA(x)}\left(\frac{\partial}{\partial x}A(x)\right)e^{it\left(1-s\right)A(x)}\right],\\
 & =\mathcal{D}_{A(x)}\!\left(\frac{\partial}{\partial x}A(x)\right)
\end{align}
we find that
\begin{align}
\frac{\partial}{\partial\theta_{kj}^{\left(2\right)}}\Tr\!\left[\varphi_{2}\!\left(\sum_{j'=1}^{J_{1}}\theta_{kj'}^{(2)}A_{j'}^{(1)}\right)\rho\right] & =\frac{\partial}{\partial\theta_{kj}^{\left(2\right)}}\Tr\!\left[\tanh\!\left(\sum_{j'=1}^{J_{1}}\theta_{kj'}^{(2)}A_{j'}^{(1)}\right)\rho\right]\\
 & =\mathbb{E}_{t\sim\mu,s\sim\upsilon}\!\left[\Tr\!\left[\mathcal{D}_{B_{k}^{\left(2\right)}\!\left(\theta^{\left(2\right)}\right)}\!\left(\frac{\partial}{\partial\theta_{kj}^{\left(2\right)}}\sum_{j'=1}^{J_{1}}\theta_{kj'}^{(2)}A_{j'}^{(1)}\right)\rho\right]\right]\\
 & =\Tr\!\left[\mathcal{D}_{B_{k}^{\left(2\right)}\!\left(\theta^{\left(2\right)}\right)}\!\left(A_{j}^{(1)}\right)\rho\right].
\end{align}

Recalling from \eqref{eq:2-layer-ham-net-1}--\eqref{eq:2-layer-ham-net-last}
that
\begin{equation}
A_{k}^{\left(2\right)}\!\left(\theta^{\left(2\right)}\right)=\tanh\!\left(\sum_{i=1}^{J_{1}}\theta_{ki}^{(2)}A_{i}^{\left(1\right)}\!\left(\theta^{\left(1\right)}\right)\right),
\end{equation}
and applying Lemma \ref{lem:grad-tanh}, the other partial derivative
to calculate is as follows:

\begin{align}
 & \frac{\partial}{\partial\theta_{ji}^{\left(1\right)}}\Tr\!\left[\varphi_{2}\!\left(\sum_{j'=1}^{J_{1}}\theta_{kj'}^{(2)}\varphi_{1}\!\left(\sum_{i'=1}^{J_{0}}\theta_{j'i'}^{\left(1\right)}H_{i'}\right)\right)\rho\right]\nonumber \\
 & =\frac{\partial}{\partial\theta_{ji}^{\left(1\right)}}\Tr\!\left[\tanh\!\left(\sum_{j'=1}^{J_{1}}\theta_{kj'}^{(2)}\tanh\!\left(B_{j'}^{\left(1\right)}\!\left(\theta^{\left(1\right)}\right)\right)\right)\rho\right]\\
 & =\Tr\!\left[\frac{\partial}{\partial\theta_{ji}^{\left(1\right)}}\tanh\!\left(B_{k}^{\left(2\right)}\!\left(\theta^{\left(2\right)}\right)\right)\rho\right]\\
 & =\Tr\!\left[\mathcal{D}_{B_{k}^{\left(2\right)}\!\left(\theta^{\left(2\right)}\right)}\!\left(\frac{\partial}{\partial\theta_{ji}^{\left(1\right)}}B_{k}^{\left(2\right)}\!\left(\theta^{\left(2\right)}\right)\right)\rho\right]\\
 & =\Tr\!\left[\mathcal{D}_{B_{k}^{\left(2\right)}\!\left(\theta^{\left(2\right)}\right)}\!\left(\frac{\partial}{\partial\theta_{ji}^{\left(1\right)}}\sum_{j'=1}^{J_{1}}\theta_{kj'}^{(2)}A_{j'}^{(1)}\right)\rho\right]\\
 & =\sum_{j'=1}^{J_{1}}\theta_{kj'}^{(2)}\Tr\!\left[\mathcal{D}_{B_{k}^{\left(2\right)}\!\left(\theta^{\left(2\right)}\right)}\!\left(\frac{\partial}{\partial\theta_{ji}^{\left(1\right)}}A_{j'}^{(1)}\right)\rho\right]\\
 & =\sum_{j'=1}^{J_{1}}\theta_{kj'}^{(2)}\Tr\!\left[\mathcal{D}_{B_{k}^{\left(2\right)}\!\left(\theta^{\left(2\right)}\right)}\!\left(\frac{\partial}{\partial\theta_{ji}^{\left(1\right)}}\tanh\!\left(B_{j'}^{\left(1\right)}\!\left(\theta^{\left(1\right)}\right)\right)\right)\rho\right]\\
 & =\sum_{j'=1}^{J_{1}}\theta_{kj'}^{(2)}\Tr\!\left[\mathcal{D}_{B_{k}^{\left(2\right)}\!\left(\theta^{\left(2\right)}\right)}\!\left(\mathcal{D}_{B_{j}^{\left(1\right)}\!\left(\theta^{\left(1\right)}\right)}\left(\frac{\partial}{\partial\theta_{ji}^{\left(1\right)}}B_{j'}^{\left(1\right)}\!\left(\theta^{\left(1\right)}\right)\right)\right)\rho\right]\\
 & =\sum_{j'=1}^{J_{1}}\theta_{kj'}^{(2)}\Tr\!\left[\left(\mathcal{D}_{B_{k}^{\left(2\right)}\!\left(\theta^{\left(2\right)}\right)}\circ\mathcal{D}_{B_{j}^{\left(1\right)}\!\left(\theta^{\left(1\right)}\right)}\right)\left(\frac{\partial}{\partial\theta_{ji}^{\left(1\right)}}\sum_{i'=1}^{J_{0}}\theta_{j'i'}^{(1)}H_{i'}\right)\rho\right]\\
 & =\theta_{kj}^{(2)}\Tr\!\left[\left(\mathcal{D}_{B_{k}^{\left(2\right)}\!\left(\theta^{\left(2\right)}\right)}\circ\mathcal{D}_{B_{j}^{\left(1\right)}\!\left(\theta^{\left(1\right)}\right)}\right)\left(H_{i}\right)\rho\right].
\end{align}
The last equality follows because
\begin{align}
\frac{\partial}{\partial\theta_{ji}^{\left(1\right)}}\sum_{i'=1}^{J_{0}}\theta_{j'i'}^{(1)}H_{i'} & =\sum_{i'=1}^{J_{0}}\left(\frac{\partial}{\partial\theta_{ji}^{\left(1\right)}}\theta_{j'i'}^{(1)}\right)H_{i'}\\
 & =\sum_{i'=1}^{J_{0}}\delta_{jj'}\delta_{ii'}H_{i'}\\
 & =\delta_{jj'}H_{i},
\end{align}
thus concluding the proof.
\end{proof}

\begin{rem}
In order to estimate these partial derivatives, we require the ability
to perform Hamiltonian simulation according to the Hamiltonian $B_{k}^{\left(2\right)}\!\left(\theta^{\left(2\right)}\right)$.
Given that $B_{k}^{\left(2\right)}\!\left(\theta^{\left(2\right)}\right)$
is itself a composition of Hamiltonians, it is not generally a local
Hamiltonian, and so standard approaches do not apply. One can potentially
make use of the formalism of block encoding and quantum singular value
transformation~\cite{Gilyen2019} to address this problem, but we
leave it as the topic of future research.
\end{rem}

\subsection{Three-layer gradient formulas}

Now we consider the three-layer case, where the objective function
is
\begin{equation}
\Tr\!\left[\varphi_{3}\!\left(B_{k}^{\left(3\right)}\!\left(\theta^{\left(3\right)}\right)\right)\rho\right],
\end{equation}
and
\begin{align}
B_{k}^{\left(3\right)}\!\left(\theta^{\left(3\right)}\right) & \coloneqq\sum_{j_{2}=1}^{J_{2}}\theta_{kj_{2}}^{(3)}A_{j_{2}}^{\left(2\right)}\!\left(\theta^{\left(2\right)}\right),\\
A_{j_{2}}^{\left(2\right)}\!\left(\theta^{\left(2\right)}\right) & \coloneqq\varphi_{2}\!\left(B_{j_{2}}^{\left(2\right)}\!\left(\theta^{\left(2\right)}\right)\right),\\
B_{j_{2}}^{\left(2\right)}\!\left(\theta^{\left(2\right)}\right) & \coloneqq\sum_{j_{1}=1}^{J_{1}}\theta_{j_{2}j_{1}}^{(2)}A_{j_{1}}^{\left(1\right)}\!\left(\theta^{\left(1\right)}\right)\\
A_{j_{1}}^{\left(1\right)}\!\left(\theta^{\left(1\right)}\right) & \coloneqq\varphi_{1}\!\left(B_{j_{1}}^{\left(1\right)}\!\left(\theta^{\left(1\right)}\right)\right),\\
B_{j_{1}}^{\left(1\right)}\!\left(\theta^{\left(1\right)}\right) & \coloneqq\sum_{j_{0}=1}^{J_{0}}\theta_{j_{1}j_{0}}^{\left(1\right)}A_{j_{1}}^{\left(0\right)}=\sum_{j_{0}=1}^{J_{0}}\theta_{j_{1}j_{0}}^{\left(1\right)}H_{j_{0}}.
\end{align}

\begin{thm}
The gradient formulas for a three-layer quantum observable
network of the form in \eqref{eq:output-layer-l-1-1}--\eqref{eq:last-layer-general}
are given by
\begin{align}
\frac{\partial}{\partial\theta_{kj_{2}}^{(3)}}\Tr\!\left[\varphi_{3}\!\left(B_{k}^{\left(3\right)}\!\left(\theta^{\left(3\right)}\right)\right)\rho\right] & =\Tr\!\left[\mathcal{D}_{B_{k}^{\left(3\right)}\!\left(\theta^{\left(3\right)}\right)}\!\left(A_{j_{2}}^{\left(2\right)}\!\left(\theta^{\left(2\right)}\right)\right)\rho\right],\\
\frac{\partial}{\partial\theta_{j_{2}j_{1}}^{(2)}}\Tr\!\left[\varphi_{3}\!\left(B_{k}^{\left(3\right)}\!\left(\theta^{\left(3\right)}\right)\right)\rho\right] & =\theta_{kj_{2}}^{(3)}\Tr\!\left[\left(\mathcal{D}_{B_{k}^{\left(3\right)}\!\left(\theta^{\left(3\right)}\right)}\circ\mathcal{D}_{B_{j_{2}}^{\left(2\right)}\!\left(\theta^{\left(2\right)}\right)}\right)\!\left(A_{j_{1}}^{\left(1\right)}\!\left(\theta^{\left(1\right)}\right)\right)\rho\right],\\
\frac{\partial}{\partial\theta_{j_{1}j_{0}}^{\left(1\right)}}\Tr\!\left[\varphi_{3}\!\left(B_{k}^{\left(3\right)}\!\left(\theta^{\left(3\right)}\right)\right)\rho\right] & =\sum_{j'_{2}=1}^{J_{2}}\theta_{kj'_{2}}^{(3)}\theta_{j'_{2}j_{1}}^{(2)}\Tr\!\left[\!\left(\mathcal{D}_{B_{k}^{\left(3\right)}\!\left(\theta^{\left(3\right)}\right)}\circ\mathcal{D}_{B_{j'_{2}}^{\left(2\right)}\!\left(\theta^{\left(2\right)}\right)}\circ\mathcal{D}_{B_{j_{1}}^{\left(1\right)}\!\left(\theta^{\left(1\right)}\right)}\right)\left(H_{j_{0}}\right)\rho\right].
\end{align}
\end{thm}

\begin{proof}
Consider that
\begin{align}
\frac{\partial}{\partial\theta_{kj_{2}}^{(3)}}\Tr\!\left[\varphi_{3}\!\left(B_{k}^{\left(3\right)}\!\left(\theta^{\left(3\right)}\right)\right)\rho\right] & =\Tr\!\left[\mathcal{D}_{B_{k}^{\left(3\right)}\!\left(\theta^{\left(3\right)}\right)}\!\left(\frac{\partial}{\partial\theta_{kj_{2}}^{(3)}}B_{k}^{\left(3\right)}\!\left(\theta^{\left(3\right)}\right)\right)\rho\right]\\
 & =\Tr\!\left[\mathcal{D}_{B_{k}^{\left(3\right)}\!\left(\theta^{\left(3\right)}\right)}\!\left(\frac{\partial}{\partial\theta_{kj_{2}}^{(3)}}\sum_{j'_{2}=1}^{J_{2}}\theta_{kj'_{2}}^{(3)}A_{j'_{2}}^{\left(2\right)}\!\left(\theta^{\left(2\right)}\right)\right)\rho\right]\\
 & =\Tr\!\left[\mathcal{D}_{B_{k}^{\left(3\right)}\!\left(\theta^{\left(3\right)}\right)}\!\left(\sum_{j'_{2}=1}^{J_{2}}\left(\frac{\partial}{\partial\theta_{kj_{2}}^{(3)}}\theta_{kj'_{2}}^{(3)}\right)A_{j'_{2}}^{\left(2\right)}\!\left(\theta^{\left(2\right)}\right)\right)\rho\right]\\
 & =\Tr\!\left[\mathcal{D}_{B_{k}^{\left(3\right)}\!\left(\theta^{\left(3\right)}\right)}\!\left(A_{j_{2}}^{\left(2\right)}\!\left(\theta^{\left(2\right)}\right)\right)\rho\right].
\end{align}
Also,
\begin{align}
 & \frac{\partial}{\partial\theta_{j_{2}j_{1}}^{(2)}}\Tr\!\left[\varphi_{3}\!\left(B_{k}^{\left(3\right)}\!\left(\theta^{\left(3\right)}\right)\right)\rho\right]\nonumber \\
 & =\Tr\!\left[\mathcal{D}_{B_{k}^{\left(3\right)}\!\left(\theta^{\left(3\right)}\right)}\!\left(\frac{\partial}{\partial\theta_{j_{2}j_{1}}^{(2)}}B_{k}^{\left(3\right)}\!\left(\theta^{\left(3\right)}\right)\right)\rho\right]\\
 & =\Tr\!\left[\mathcal{D}_{B_{k}^{\left(3\right)}\!\left(\theta^{\left(3\right)}\right)}\!\left(\frac{\partial}{\partial\theta_{j_{2}j_{1}}^{(2)}}\sum_{j'_{2}=1}^{J_{2}}\theta_{kj'_{2}}^{(3)}A_{j'_{2}}^{\left(2\right)}\!\left(\theta^{\left(2\right)}\right)\right)\rho\right]\\
 & =\sum_{j'_{2}=1}^{J_{2}}\theta_{kj'_{2}}^{(3)}\Tr\!\left[\mathcal{D}_{B_{k}^{\left(3\right)}\!\left(\theta^{\left(3\right)}\right)}\!\left(\frac{\partial}{\partial\theta_{j_{2}j_{1}}^{(2)}}A_{j'_{2}}^{\left(2\right)}\!\left(\theta^{\left(2\right)}\right)\right)\rho\right]\\
 & =\sum_{j'_{2}=1}^{J_{2}}\theta_{kj'_{2}}^{(3)}\Tr\!\left[\mathcal{D}_{B_{k}^{\left(3\right)}\!\left(\theta^{\left(3\right)}\right)}\!\left(\frac{\partial}{\partial\theta_{j_{2}j_{1}}^{(2)}}\varphi_{2}\!\left(B_{j'_{2}}^{\left(2\right)}\!\left(\theta^{\left(2\right)}\right)\right)\right)\rho\right]\\
 & =\sum_{j'_{2}=1}^{J_{2}}\theta_{kj'_{2}}^{(3)}\Tr\!\left[\mathcal{D}_{B_{k}^{\left(3\right)}\!\left(\theta^{\left(3\right)}\right)}\!\left(\mathcal{D}_{B_{j'_{2}}^{\left(2\right)}\!\left(\theta^{\left(2\right)}\right)}\left(\frac{\partial}{\partial\theta_{j_{2}j_{1}}^{(2)}}B_{j'_{2}}^{\left(2\right)}\!\left(\theta^{\left(2\right)}\right)\right)\right)\rho\right]\\
 & =\sum_{j'_{2}=1}^{J_{2}}\theta_{kj'_{2}}^{(3)}\Tr\!\left[\mathcal{D}_{B_{k}^{\left(3\right)}\!\left(\theta^{\left(3\right)}\right)}\!\left(\mathcal{D}_{B_{j'_{2}}^{\left(2\right)}\!\left(\theta^{\left(2\right)}\right)}\left(\frac{\partial}{\partial\theta_{j_{2}j_{1}}^{(2)}}\sum_{j'_{1}=1}^{J_{1}}\theta_{j'_{2}j'_{1}}^{(2)}A_{j'_{1}}^{\left(1\right)}\!\left(\theta^{\left(1\right)}\right)\right)\right)\rho\right]\\
 & =\sum_{j'_{2}=1}^{J_{2}}\theta_{kj'_{2}}^{(3)}\Tr\!\left[\mathcal{D}_{B_{k}^{\left(3\right)}\!\left(\theta^{\left(3\right)}\right)}\!\left(\mathcal{D}_{B_{j'_{2}}^{\left(2\right)}\!\left(\theta^{\left(2\right)}\right)}\left(\sum_{j'_{1}=1}^{J_{1}}\left(\frac{\partial}{\partial\theta_{j_{2}j_{1}}^{(2)}}\theta_{j'_{2}j'_{1}}^{(2)}\right)A_{j'_{1}}^{\left(1\right)}\!\left(\theta^{\left(1\right)}\right)\right)\right)\rho\right]\\
 & =\theta_{kj_{2}}^{(3)}\Tr\!\left[\mathcal{D}_{B_{k}^{\left(3\right)}\!\left(\theta^{\left(3\right)}\right)}\!\left(\mathcal{D}_{B_{j_{2}}^{\left(2\right)}\!\left(\theta^{\left(2\right)}\right)}\left(A_{j_{1}}^{\left(1\right)}\!\left(\theta^{\left(1\right)}\right)\right)\right)\rho\right]\\
 & =\theta_{kj_{2}}^{(3)}\Tr\!\left[\left(\mathcal{D}_{B_{k}^{\left(3\right)}\!\left(\theta^{\left(3\right)}\right)}\circ\mathcal{D}_{B_{j_{2}}^{\left(2\right)}\!\left(\theta^{\left(2\right)}\right)}\right)\!\left(A_{j_{1}}^{\left(1\right)}\!\left(\theta^{\left(1\right)}\right)\right)\rho\right].
\end{align}
Finally, consider that
\begin{align}
 & \frac{\partial}{\partial\theta_{j_{1}j_{0}}^{\left(1\right)}}\Tr\!\left[\varphi_{3}\!\left(B_{k}^{\left(3\right)}\!\left(\theta^{\left(3\right)}\right)\right)\rho\right]\nonumber \\
 & =\Tr\!\left[\mathcal{D}_{B_{k}^{\left(3\right)}\!\left(\theta^{\left(3\right)}\right)}\!\left(\frac{\partial}{\partial\theta_{j_{1}j_{0}}^{\left(1\right)}}B_{k}^{\left(3\right)}\!\left(\theta^{\left(3\right)}\right)\right)\rho\right]\\
 & =\Tr\!\left[\mathcal{D}_{B_{k}^{\left(3\right)}\!\left(\theta^{\left(3\right)}\right)}\!\left(\frac{\partial}{\partial\theta_{j_{1}j_{0}}^{\left(1\right)}}\sum_{j'_{2}=1}^{J_{2}}\theta_{kj'_{2}}^{(3)}A_{j'_{2}}^{\left(2\right)}\!\left(\theta^{\left(2\right)}\right)\right)\rho\right]\\
 & =\sum_{j'_{2}=1}^{J_{2}}\theta_{kj'_{2}}^{(3)}\Tr\!\left[\mathcal{D}_{B_{k}^{\left(3\right)}\!\left(\theta^{\left(3\right)}\right)}\!\left(\frac{\partial}{\partial\theta_{j_{1}j_{0}}^{\left(1\right)}}\varphi_{2}\!\left(B_{j'_{2}}^{\left(2\right)}\!\left(\theta^{\left(2\right)}\right)\right)\right)\rho\right]\\
 & =\sum_{j'_{2}=1}^{J_{2}}\theta_{kj'_{2}}^{(3)}\Tr\!\left[\mathcal{D}_{B_{k}^{\left(3\right)}\!\left(\theta^{\left(3\right)}\right)}\!\left(\mathcal{D}_{B_{j'_{2}}^{\left(2\right)}\!\left(\theta^{\left(2\right)}\right)}\left(\frac{\partial}{\partial\theta_{j_{1}j_{0}}^{\left(1\right)}}B_{j'_{2}}^{\left(2\right)}\!\left(\theta^{\left(2\right)}\right)\right)\right)\rho\right]\\
 & =\sum_{j'_{2}=1}^{J_{2}}\theta_{kj'_{2}}^{(3)}\Tr\!\left[\mathcal{D}_{B_{k}^{\left(3\right)}\!\left(\theta^{\left(3\right)}\right)}\!\left(\mathcal{D}_{B_{j'_{2}}^{\left(2\right)}\!\left(\theta^{\left(2\right)}\right)}\left(\frac{\partial}{\partial\theta_{j_{1}j_{0}}^{\left(1\right)}}\sum_{j'_{1}=1}^{J_{1}}\theta_{j'_{2}j'_{1}}^{(2)}A_{j'_{1}}^{\left(1\right)}\!\left(\theta^{\left(1\right)}\right)\right)\right)\rho\right]\\
 & =\sum_{j'_{2}=1}^{J_{2}}\theta_{kj'_{2}}^{(3)}\sum_{j'_{1}=1}^{J_{1}}\theta_{j'_{2}j'_{1}}^{(2)}\Tr\!\left[\mathcal{D}_{B_{k}^{\left(3\right)}\!\left(\theta^{\left(3\right)}\right)}\!\left(\mathcal{D}_{B_{j'_{2}}^{\left(2\right)}\!\left(\theta^{\left(2\right)}\right)}\left(\frac{\partial}{\partial\theta_{j_{1}j_{0}}^{\left(1\right)}}\varphi_{1}\!\left(B_{j'_{1}}^{\left(1\right)}\!\left(\theta^{\left(1\right)}\right)\right)\right)\right)\rho\right]\\
 & =\sum_{j'_{2}=1}^{J_{2}}\theta_{kj'_{2}}^{(3)}\sum_{j'_{1}=1}^{J_{1}}\theta_{j'_{2}j'_{1}}^{(2)}\Tr\!\left[\mathcal{D}_{B_{k}^{\left(3\right)}\!\left(\theta^{\left(3\right)}\right)}\!\left(\mathcal{D}_{B_{j'_{2}}^{\left(2\right)}\!\left(\theta^{\left(2\right)}\right)}\left(\mathcal{D}_{B_{j'_{1}}^{\left(1\right)}\!\left(\theta^{\left(1\right)}\right)}\frac{\partial}{\partial\theta_{j_{1}j_{0}}^{\left(1\right)}}B_{j'_{1}}^{\left(1\right)}\!\left(\theta^{\left(1\right)}\right)\right)\right)\rho\right]\\
 & =\sum_{j'_{2}=1}^{J_{2}}\theta_{kj'_{2}}^{(3)}\sum_{j'_{1}=1}^{J_{1}}\theta_{j'_{2}j'_{1}}^{(2)}\Tr\!\left[\!\left(\mathcal{D}_{B_{k}^{\left(3\right)}\!\left(\theta^{\left(3\right)}\right)}\circ\mathcal{D}_{B_{j'_{2}}^{\left(2\right)}\!\left(\theta^{\left(2\right)}\right)}\circ\mathcal{D}_{B_{j'_{1}}^{\left(1\right)}\!\left(\theta^{\left(1\right)}\right)}\right)\left(\frac{\partial}{\partial\theta_{j_{1}j_{0}}^{\left(1\right)}}B_{j'_{1}}^{\left(1\right)}\!\left(\theta^{\left(1\right)}\right)\right)\rho\right]\\
 & =\sum_{j'_{2}=1}^{J_{2}}\theta_{kj'_{2}}^{(3)}\sum_{j'_{1}=1}^{J_{1}}\theta_{j'_{2}j'_{1}}^{(2)}\Tr\!\left[\!\left(\mathcal{D}_{B_{k}^{\left(3\right)}\!\left(\theta^{\left(3\right)}\right)}\circ\mathcal{D}_{B_{j'_{2}}^{\left(2\right)}\!\left(\theta^{\left(2\right)}\right)}\circ\mathcal{D}_{B_{j'_{1}}^{\left(1\right)}\!\left(\theta^{\left(1\right)}\right)}\right)\left(\frac{\partial}{\partial\theta_{j_{1}j_{0}}^{\left(1\right)}}\sum_{j'_{0}=1}^{J_{0}}\theta_{j'_{1}j'_{0}}^{\left(1\right)}H_{j'_{0}}\right)\rho\right]\\
 & =\sum_{j'_{2}=1}^{J_{2}}\theta_{kj'_{2}}^{(3)}\theta_{j'_{2}j_{1}}^{(2)}\Tr\!\left[\!\left(\mathcal{D}_{B_{k}^{\left(3\right)}\!\left(\theta^{\left(3\right)}\right)}\circ\mathcal{D}_{B_{j'_{2}}^{\left(2\right)}\!\left(\theta^{\left(2\right)}\right)}\circ\mathcal{D}_{B_{j_{1}}^{\left(1\right)}\!\left(\theta^{\left(1\right)}\right)}\right)\left(H_{j_{0}}\right)\rho\right],
\end{align}
thus concluding the proof.
\end{proof}

\subsection{$L$-layer gradient formulas}

The gradient formulas for an $L$-layer quantum observable
network of the form in \eqref{eq:output-layer-l-1-1}--\eqref{eq:last-layer-general}
are then given by
\begin{equation}
\frac{\partial}{\partial\theta_{kj_{L-1}}^{(L)}}\Tr\!\left[\varphi_{L}\!\left(B_{k}^{\left(L\right)}\!\left(\theta^{\left(L\right)}\right)\right)\rho\right]=\Tr\!\left[\mathcal{D}_{B_{k}^{\left(L\right)}\!\left(\theta^{\left(L\right)}\right)}\!\left(A_{j_{L-1}}^{\left(L-1\right)}\!\left(\theta^{\left(L-1\right)}\right)\right)\rho\right].
\end{equation}
For $\ell=L-1,\ldots,2$, we have that
\begin{multline}
\frac{\partial}{\partial\theta_{j_{\ell}j_{\ell-1}}^{(\ell)}}\Tr\!\left[\varphi_{L}\!\left(B_{k}^{\left(L\right)}\!\left(\theta^{\left(L\right)}\right)\right)\rho\right]=\sum_{j'_{L-1}=1}^{J_{L-1}}\sum_{j'_{L-2}=1}^{J_{L-2}}\cdots\sum_{j'_{\ell+1}=1}^{J_{\ell+1}}\theta_{kj'_{L-1}}^{(L)}\theta_{j'_{L-1}j'_{L-2}}^{(L-1)}\cdots\theta_{j'_{\ell+1}j_{\ell}}^{(\ell+1)}\times\\
\Tr\!\left[\!\left(\mathcal{D}_{B_{k}^{\left(L\right)}\!\left(\theta^{\left(L\right)}\right)}\circ\mathcal{D}_{B_{j'_{L-1}}^{\left(L-1\right)}\!\left(\theta^{\left(L-1\right)}\right)}\circ\cdots\circ\mathcal{D}_{B_{j_{\ell}}^{\left(\ell\right)}\!\left(\theta^{\left(\ell\right)}\right)}\right)\left(A_{j_{\ell-1}}^{\left(\ell-1\right)}\!\left(\theta^{\left(\ell-1\right)}\right)\right)\rho\right].
\end{multline}
Finally, we have that
\begin{multline}
\frac{\partial}{\partial\theta_{j_{1}j_{0}}^{\left(1\right)}}\Tr\!\left[\varphi_{L}\!\left(B_{k}^{\left(L\right)}\!\left(\theta^{\left(L\right)}\right)\right)\rho\right]
=\sum_{j'_{L-1}=1}^{J_{L-1}}\sum_{j'_{L-2}=1}^{J_{L-2}}\cdots\sum_{j'_{2}=1}^{J_{2}}\theta_{kj'_{L-1}}^{(L)}\theta_{j'_{L-1}j'_{L-2}}^{(L-1)}\cdots\theta_{j'_{2}j_{1}}^{(2)}\times\\
\Tr\!\left[\!\left(\mathcal{D}_{B_{k}^{\left(L\right)}\!\left(\theta^{\left(L\right)}\right)}\circ\mathcal{D}_{B_{j'_{L-1}}^{\left(L-1\right)}\!\left(\theta^{\left(L-1\right)}\right)}\circ\cdots\circ\mathcal{D}_{B_{j_{1}}^{\left(1\right)}\!\left(\theta^{\left(1\right)}\right)}\right)\left(H_{j_{0}}\right)\rho\right].
\end{multline}

From these formulas, we can derive a recursive formulation similar
to backpropagation \cite{Rumelhart1986a,Rumelhart1986}. Define
\begin{align}
\Gamma_{kj_{L-1}}^{\left(L\right)} & \coloneqq\theta_{kj_{L-1}}^{(L)}\mathcal{D}_{B_{k}^{\left(L\right)}\!\left(\theta^{\left(L\right)}\right)}, \label{eq:back-prop-1}\\
\Gamma_{kj_{L-2}}^{\left(L-1\right)} & \coloneqq\sum_{j_{L-1}=1}^{J_{L-1}}\theta_{j_{L-1}j{}_{L-2}}^{(L-1)}\left(\Gamma_{kj_{L-1}}^{\left(L\right)}\circ\mathcal{D}_{B_{j_{L-1}}^{\left(L-1\right)}\!\left(\theta^{\left(L-1\right)}\right)}\right),\\
\Gamma_{kj_{L-3}}^{\left(L-2\right)} & \coloneqq\sum_{j_{L-2}=1}^{J_{L-2}}\theta_{j_{L-2}j{}_{L-3}}^{(L-2)}\left(\Gamma_{kj_{L-2}}^{\left(L-1\right)}\circ\mathcal{D}_{B_{j_{L-2}}^{\left(L-2\right)}\!\left(\theta^{\left(L-2\right)}\right)}\right),\\
\vdots\\
\Gamma_{kj_{\ell}}^{\left(\ell+1\right)} & \coloneqq\sum_{j_{\ell+1}=1}^{J_{\ell+1}}\theta_{j_{\ell+1}j_{\ell}}^{(\ell+1)}\left(\Gamma_{kj_{\ell+1}}^{\left(\ell+2\right)}\circ\mathcal{D}_{B_{j_{\ell+1}}^{\left(\ell+1\right)}\!\left(\theta^{\left(\ell+1\right)}\right)}\right). \label{eq:back-prop-last}
\end{align}
So then, for $\ell=L-1,\ldots,2,1$, we have that
\begin{equation}
\frac{\partial}{\partial\theta_{j_{\ell}j_{\ell-1}}^{(\ell)}}\Tr\!\left[\varphi_{L}\!\left(B_{k}^{\left(L\right)}\!\left(\theta^{\left(L\right)}\right)\right)\rho\right]=\Tr\!\left[\left(\Gamma_{kj_{\ell}}^{\left(\ell+1\right)}\circ\mathcal{D}_{B_{j_{\ell}}^{\left(\ell\right)}\!\left(\theta^{\left(\ell\right)}\right)}\right)\left(A_{j_{\ell-1}}^{\left(\ell-1\right)}\!\left(\theta^{\left(\ell-1\right)}\right)\right)\rho\right].
\end{equation}

\medskip

\begin{rem}
If one is performing a classical simulation of quantum observable
networks, then the recursive relations in
\eqref{eq:back-prop-1}--\eqref{eq:back-prop-last}
can be used to evaluate gradients efficiently through a reverse-mode
propagation procedure analogous to the classical backpropagation
algorithm~\cite{rumelhart1986parallel,Rumelhart1986a}. In particular,
the superoperators
$
\Gamma_{kj_{\ell}}^{(\ell+1)}
$
play the role of recursively propagated error maps, with the Fr\'echet
derivative superoperators
$
\mathcal{D}_{B_{j_{\ell}}^{(\ell)}(\theta^{(\ell)})}
$
replacing the scalar derivatives that appear in classical neural
networks.

However, it is less clear how to exploit this recursive structure
in order to obtain an efficient quantum implementation of backpropagation.
Unlike the classical setting, the quantities propagated backward through
the network are superoperators acting on observables, rather than
classical vectors or tensors. Furthermore, implementing compositions
of Fr\'echet derivative superoperators coherently on a quantum computer
appears to require nontrivial higher-order transformations of quantum
states and observables. Determining whether there exists an efficient
quantum algorithm for reverse-mode differentiation in quantum observable
networks therefore remains an open problem.
\end{rem}

\section{Proof of Theorem~\ref{thm:BQP-complete-FD-decision} (complexity
theoretic-evidence against classical simulation)}

\label{app:BQP-complete-FD-decision}

Containment in BQP follows from Algorithm \ref{alg:obj-func-est}.

To see that Problem \ref{prob:FD-neuron-decision} is BQP-hard for
$k\geq5$, we can use the Feynman--Kitaev circuit-to-Hamiltonian
construction \cite{Kitaev2002,Feynman1985} (see also \cite{Aharonov2007,Gharibian2013}).
Let $A\equiv(A_{\text{yes}},A_{\text{no}})$ be a promise problem
in BQP, and let $x\in\left\{ 0,1\right\} ^{n}$ be an input. Let $U_{x}\coloneqq U_{M}U_{M-1}\cdots U_{2}U_{1}$
be a sequence of $M=\poly(n)$ 2-local unitary gates, acting on $|0\rangle^{\otimes p(n)}$,
where $p(n)=\poly(n)$, such that the acceptance probability
\begin{equation}
p_{\text{acc}}\equiv\left\Vert \left(\langle1|\otimes I^{\otimes p(n)-1}\right)U_{x}|0\rangle^{\otimes p(n)}\right\Vert ^{2}
\end{equation}
satisfies $p_{\text{acc}}\geq1-\varepsilon$ in the case of a YES
instance and $p_{\text{acc}}\leq\varepsilon$ in the case of a NO
instance, where $\varepsilon\in\left[0,1/2\right)$. By the standard
error reduction for BQP (i.e., repeating the circuit polynomially
many times and taking a majority vote), we can take $\varepsilon=2^{-q(n)}$
for some $q(n)=\poly(n)$. Without loss of generality, we then suppose
that $U_{x}$ achieves the error probability $\varepsilon=2^{-q(n)}$.

We now pre-idle the verifier, by performing $L=\poly(n)$ identity
gates before applying $U_{x}$, and define $W_{x}\coloneqq W_{M+L}W_{M+L-1}\cdots W_{2}W_{1}$,
where
\begin{equation}
W_{k}\coloneqq\begin{cases}
I & :k\in\left\{ 1,\ldots,L\right\} \\
U_{k-L} & :k\in\left\{ L+1,\ldots,M+L\right\} 
\end{cases}.
\end{equation}
Following \cite{Kitaev2002} and setting $K\equiv M+L$, we define
the Hamiltonian $H$ as follows:
\begin{equation}
H\coloneqq\Delta\left(H_{\text{in}}+H_{\text{prop}}+H_{\text{stab}}\right)+H_{\text{out}},
\label{eq:FK-Ham-with-Delta}
\end{equation}
where $\Delta$,
\begin{align}
H_{\text{in}} & \coloneqq\sum_{i=1}^{p(n)}|1\rangle\langle1|_{i}\otimes|0\rangle\!\langle0|_{C},\\
H_{\text{prop}} & \coloneqq\frac{1}{2}\sum_{k=1}^{K}\left(I\otimes\left(|k\rangle\!\langle k|_{C}+|k-1\rangle\!\langle k-1|_{C}\right)-W_{k}\otimes|k\rangle\langle k-1|_{C} - W_{k}^\dag \otimes|k-1\rangle\langle k|_{C}\right),\\
H_{\text{stab}} & \coloneqq I\otimes\sum_{j=1}^{K-1}|0\rangle\!\langle0|_{C_{j}}\otimes|1\rangle\!\langle1|_{C_{j+1}},\\
H_{\text{out}} & \coloneqq|0\rangle\!\langle0|\otimes I^{\otimes p(n)-1}\otimes|K\rangle\!\langle K|_{C},
\end{align}
where the clock states are encoded in unary, as done in the standard
construction \cite{Kitaev2002}. Here we use the shorthand $|k\rangle_C$
 for the unary clock state corresponding to time step $k$. 
 Each term in the construction acts nontrivially on at most five qubits,
so the Hamiltonian is 5-local.
 Each local term appearing in the construction has operator norm at most
$O(\Delta)$, and we return to this issue at the end of the proof. 

The history state corresponding to the computation $W_{x}$ is as
follows:
\begin{equation}
|\eta\rangle\coloneqq\frac{1}{\sqrt{K+1}}\sum_{k=0}^{K}W_{k}\cdots W_{1}|0\rangle|k\rangle_{C},
\end{equation}
satisfying
\begin{align}
\langle\eta|\left(H_{\text{in}}+H_{\text{prop}}+H_{\text{stab}}\right)|\eta\rangle & =0,\\
\langle\eta|H_{\text{out}}|\eta\rangle & \leq\frac{\varepsilon}{K+1}=\frac{2^{-q(n)}}{K+1},\label{eq:Hout-hist-small}
\end{align}
in the case of a YES instance. Also, in the case of a NO instance,
we have that
\begin{equation}
\langle\eta|H_{\text{out}}|\eta\rangle\geq\frac{1-\varepsilon}{K+1}=\frac{1-2^{-q(n)}}{K+1}.\label{eq:Hout-hist-small-NO}
\end{equation}

While the history state depends on the computation, the following
state, featuring a uniform superposition over the clock register,
does not:
\begin{equation}
|\upsilon\rangle\coloneqq|0\rangle^{\otimes p(n)}\otimes\left(\frac{1}{\sqrt{L+1}}\sum_{k=0}^{L}|k\rangle_{C}\right).
\end{equation}
Observe that
\begin{align}
\left|\langle\upsilon|\eta\rangle\right|^{2} & =\left|\left(\langle0|^{\otimes p(n)}\otimes\left(\frac{1}{\sqrt{L+1}}\sum_{k'=0}^{L}\langle k'|{}_{C}\right)\right)\left(\frac{1}{\sqrt{K+1}}\sum_{k=0}^{K}W_{k}\cdots W_{1}|0\rangle^{\otimes p(n)}|k\rangle_{C}\right)\right|^{2}\\
 & =\frac{1}{\left(L+1\right)\left(K+1\right)}\left|\sum_{k'=0}^{L}\sum_{k=0}^{K}\langle0|^{\otimes p(n)}W_{k}\cdots W_{1}|0\rangle^{\otimes p(n)}\langle k'|k\rangle_{C}\right|^{2}\\
 & =\frac{1}{\left(L+1\right)\left(K+1\right)}\left|\sum_{k'=0}^{L}\langle0|0\rangle^{\otimes p(n)}\right|^{2}\\
 & =\frac{L+1}{K+1}\\
 & =\frac{L+1}{L+M+1}.
\end{align}
We can then take $L=\poly$(n) such that $\left|\langle\upsilon|\eta\rangle\right|^{2}\geq1-\frac{1}{r(n)}$,
where $r(n)=\poly(n)$.

Eq.~(12) of \cite{Cade2023} guarantees that
\begin{equation}
\left\Vert |g\rangle-|\eta\rangle\right\Vert =O\!\left(\left(\frac{\Delta}{K^{2}}\right)^{-1}\right),\label{eq:g-to-eta-inv-poly-n}
\end{equation}
where $|g\rangle$ is the ground state of $H$, and we can choose
$\Delta$ such that
\begin{equation}
O\!\left(\left(\frac{\Delta}{K^{2}}\right)^{-1}\right)=O\!\left(\frac{1}{\poly(n)}\right).\label{eq:choice-Delta-inv-pol}
\end{equation}
The first-order effective Hamiltonian is then given by
\begin{equation}
H_{\text{eff}}\equiv|\eta\rangle\!\langle\eta|H_{\text{out}}|\eta\rangle\!\langle\eta|.
\end{equation}
Combining
\eqref{eq:Hout-hist-small},
\eqref{eq:Hout-hist-small-NO},
\eqref{eq:g-to-eta-inv-poly-n},
and
\eqref{eq:choice-Delta-inv-pol},
we conclude that the ground-state energy of $H$
differs from
$\langle\eta|H_{\text{out}}|\eta\rangle$
by at most $O(1/\poly(n))$.
Appealing now to \eqref{eq:Hout-hist-small} and \eqref{eq:Hout-hist-small-NO},
the ground-state energy of $H$ lies in the range $\pm O\!\left(\frac{1}{\poly(n)}\right)$
if $x\in A_{\text{yes}}$ and in the range of $\frac{1}{K+1}\pm O\!\left(\frac{1}{\poly(n)}\right)$
if $x\in A_{\text{no}}$.

We then shift the spectrum of the Hamiltonian by taking $H'\coloneqq H-\delta I$,
where $\delta\coloneqq\frac{1}{2\left(K+1\right)}$, while noting
that this modification does not change the locality of the Hamiltonian.
Thus, the ground-state energy $E_{0}$ of $H'$ lies in the range
$-\frac{1}{2\left(K+1\right)}\pm O\!\left(\frac{1}{\poly(n)}\right)$
if $x\in A_{\text{yes}}$ and in the range of $\frac{1}{2\left(K+1\right)}\pm O\!\left(\frac{1}{\poly(n)}\right)$
if $x\in A_{\text{no}}$. In fact, we have that
\begin{align}
E_{0} & \leq-\frac{1}{3\left(K+1\right)}\qquad\text{if}\,x\in A_{\text{yes}},\label{eq:ground-YES}\\
E_{0} & \geq\frac{1}{3\left(K+1\right)}\qquad\text{if}\,x\in A_{\text{no}}.\label{eq:ground-NO}
\end{align}
Now consider the following. The Hamiltonian $H'$ has spectral decomposition
\begin{equation}
H'=E_{0}|g\rangle\!\langle g|+\sum_{j\geq1}E_{j}|E_{j}\rangle\!\langle E_{j}|
\end{equation}
where, for all $j\ge1$,
$
E_j \ge \frac{1}{\poly(n)}.
$
This follows from the spectral gap of the Feynman--Kitaev Hamiltonian together with the choice of spectral shift $\delta$.
Recall that $\left|\langle\upsilon|g\rangle\right|^{2}\geq1-\varepsilon_{\upsilon}$,
where $\varepsilon_{\upsilon}=\frac{1}{\poly(n)}$. Writing
\begin{equation}
|\upsilon\rangle=\alpha_{0}|g\rangle+\sum_{j\geq1}\alpha_{j}|E_{j}\rangle,
\end{equation}
we know that
\begin{equation}
\left|\alpha_{0}\right|^{2}\geq1-\varepsilon_{\upsilon},\qquad\sum_{j\geq1}\left|\alpha_{j}\right|^{2}\leq\varepsilon_{\upsilon}.
\end{equation}
Recall that $g_T(x)\coloneqq \tanh(x/T)$. Then
\begin{equation}
\langle\upsilon|g_{T}(H')|\upsilon\rangle=\left|\alpha_{0}\right|^{2}g_{T}(E_{0})+\sum_{j\geq1}\left|\alpha_{j}\right|^{2}g_{T}(E_{j}),
\end{equation}
which implies that
\begin{equation}
\left|\langle\upsilon|g_{T}(H')|\upsilon\rangle-g_{T}(E_{0})\right|\leq2\varepsilon_{\upsilon},\label{eq:closeness-upsilon-state}
\end{equation}
given that
\begin{equation}
\sum_{j\geq1}\left|\alpha_{j}\right|^{2}g_{T}(E_{j})\leq\sum_{j\geq1}\left|\alpha_{j}\right|^{2}\left|g_{T}(E_{j})\right|\leq\sum_{j\geq1}\left|\alpha_{j}\right|^{2}\leq\varepsilon_{\upsilon},
\end{equation}
where we used that $\left|g_{T}(y)\right|\leq1$ for all $y\in\mathbb{R}$.

Now choose the temperature
\begin{equation}
T=\frac{1}{10\left(K+1\right)}.\label{eq:temp-choice}
\end{equation}
For $x\in A_{\text{yes}}$, \eqref{eq:ground-YES} and \eqref{eq:temp-choice}
imply that $\frac{E_{0}}{T}\leq-\frac{10}{3}$. Then
\begin{equation}
g_{T}(E_{0})=\tanh(E_{0}/T)\leq\tanh(-10/3)\leq-0.99.
\end{equation}
Combined with \eqref{eq:closeness-upsilon-state} and choosing parameters
such that $2\varepsilon_{\upsilon}\leq0.01$, we conclude that
\begin{equation}
\langle\upsilon|g_{T}(H')|\upsilon\rangle\leq-0.98.
\end{equation}
Similarly, for $x\in A_{\text{no}}$, \eqref{eq:ground-NO} and \eqref{eq:temp-choice}
imply that $\frac{E_{0}}{T}\geq\frac{10}{3}$. Then
\begin{equation}
g_{T}(E_{0})=\tanh(E_{0}/T)\geq\tanh(10/3)\geq0.99.
\end{equation}
Combined with \eqref{eq:closeness-upsilon-state} and the same choice
of parameters such that $2\varepsilon_{\upsilon}\leq0.01$, we conclude
that
\begin{equation}
\langle\upsilon|g_{T}(H')|\upsilon\rangle\geq0.98.
\end{equation}

Finally, note that the local terms appearing in the construction in \eqref{eq:FK-Ham-with-Delta}
have operator norm at most $O(\Delta)$, where
$\Delta=\poly(n)$.
To normalize the terms, define the rescaled Hamiltonian
$
\widetilde H \coloneqq H/\Delta
$
and the rescaled temperature
$
\widetilde T \coloneqq T/\Delta.
$
Then
$
\tanh(H/T)=\tanh(\widetilde H/\widetilde T),
$
and therefore
$
\Tr[\tanh(H/T)\rho]
=
\Tr[\tanh(\widetilde H/\widetilde T)\rho].
$
Since $\Delta=\poly(n)$, the temperature $\widetilde T$ remains
inverse-polynomially bounded, while every local term of
$\widetilde H$ has norm $O(1)$. After rescaling by an additional
constant factor if necessary, we may therefore assume that
$\|H_j\|\le1$ for all local terms.

\section{Derivations for cross entropy and log partition function estimation}

\subsection{Proof of Theorem~\ref{thm:cross-entropy-est} (derivation of formula
for cross entropy)}

\label{app:cross-entropy-est}

Using the notation in Theorem~\ref{thm:cross-entropy-est}, consider
that
\begin{align}
  \Xi(\eta\|\rho(\theta))-\Xi(\eta\|\rho(0,\theta_{2},\ldots,\theta_{J}))
 & =\Xi(\eta\|\rho(\theta_{1},\theta_{2},\ldots,\theta_{J}))-\Xi(\eta\|\rho(0,\theta_{2},\ldots,\theta_{J}))\\
 & =\Xi(\eta\|\rho(\theta^{\left(1\right)}(1)))-\Xi(\eta\|\rho(\theta^{\left(1\right)}(0)))\\
 & =\int_{0}^{1}d\lambda_{1}\frac{\partial}{\partial\lambda_{1}}\Xi(\eta\|\rho(\theta^{\left(1\right)}(\lambda_{1})))\\
 & =\int_{0}^{1}d\lambda_{1}\,\theta_{1}\left(\left\langle H_{1}\right\rangle _{\eta}-\left\langle H_{1}\right\rangle _{\rho(\theta^{\left(1\right)}(\lambda_{1}))}\right)\\
 & =\theta_{1}\left(\left\langle H_{1}\right\rangle _{\eta}-\int_{0}^{1}d\lambda\,\left\langle H_{1}\right\rangle _{\rho(\theta^{\left(1\right)}(\lambda))}\right)\\
 & =\theta_{1}\left(\left\langle H_{1}\right\rangle _{\eta}-\mathbb{E}_{\lambda\sim\upsilon}\!\left[\left\langle H_{1}\right\rangle _{\rho(\theta^{\left(1\right)}(\lambda))}\right]\right).
\end{align}
The fourth equality follows from~\eqref{eq:grad-cross-entropy-QBMs},
while accounting for the fact that
\begin{equation}
\frac{\partial}{\partial\lambda_{1}}H(\theta^{\left(1\right)}(\lambda_{1}))=\theta_{1}H_{1}.
\end{equation}
Additionally,
\begin{align}
  \Xi(\eta\|\rho(0,\theta_{2},\ldots,\theta_{J}))-\Xi(\eta\|\rho(0,0,\theta_{3},\ldots,\theta_{J}))
 & =\Xi(\eta\|\rho(\theta^{\left(2\right)}(1)))-\Xi(\eta\|\rho(\theta^{\left(2\right)}(0)))\\
 & =\int_{0}^{1}d\lambda_{2}\frac{\partial}{\partial\lambda_{2}}\Xi(\eta\|\rho(\theta^{\left(2\right)}(\lambda_{2})))\\
 & =\int_{0}^{1}d\lambda_{2}\,\theta_{2}\left(\left\langle H_{2}\right\rangle _{\eta}-\left\langle H_{2}\right\rangle _{\rho(\theta^{\left(2\right)}(\lambda_{2}))}\right)\\
 & =\theta_{2}\left(\left\langle H_{2}\right\rangle _{\eta}-\mathbb{E}_{\lambda\sim\upsilon}\!\left[\left\langle H_{2}\right\rangle _{\rho(\theta^{\left(2\right)}(\lambda))}\right]\right).
\end{align}
We continue iteratively along these lines and find that the last term
is given by
\begin{align}
  \Xi(\eta\|\rho(0,\ldots,0,\theta_{J}))-\ln d
 & =\Xi(\eta\|\rho(0,\ldots,0,\theta_{J}))-\Xi(\eta\|\rho(0,\ldots,0,0))\\
 & =\Xi(\eta\|\rho(\theta^{\left(J\right)}(1)))-\Xi(\eta\|\rho(\theta^{\left(J\right)}(0)))\\
 & =\theta_{J}\left(\left\langle H_{J}\right\rangle _{\eta}-\mathbb{E}_{\lambda\sim\upsilon}\!\left[\left\langle H_{J}\right\rangle _{\rho(\theta^{\left(J\right)}(\lambda))}\right]\right),
\end{align}
where we used the fact that
\begin{equation}
\rho(0,\ldots,0)=\frac{I}{d},
\end{equation}
which implies that
\begin{equation}
\Xi(\eta\|\rho(0,\ldots,0))=-\Tr\!\left[\eta\ln\left(\rho(0,\ldots,0)\right)\right]=-\Tr\!\left[\eta\ln\left(\frac{I}{d}\right)\right]=\ln d.
\end{equation}
So then we form a telescoping sum and conclude that
\begin{align}
\Xi(\eta\|\rho(\theta))-\ln d & =\Xi(\eta\|\rho(\theta))-\Xi(\eta\|\rho(0,\theta_{2},\ldots,\theta_{J}))\nonumber \\
 & \qquad+\Xi(\eta\|\rho(0,\theta_{2},\ldots,\theta_{J}))-\Xi(\eta\|\rho(0,0,\theta_{3},\ldots,\theta_{J}))\nonumber \\
 & \qquad+\cdots+\Xi(\eta\|\rho(0,\ldots,0,\theta_{J}))-\Xi(\eta\|\rho(0,\ldots,0,0))\\
 & =\theta_{1}\left(\left\langle H_{1}\right\rangle _{\eta}-\mathbb{E}_{\lambda\sim\upsilon}\!\left[\left\langle H_{1}\right\rangle _{\rho(\theta^{\left(1\right)}(\lambda))}\right]\right)\nonumber \\
 & \qquad+\theta_{2}\left(\left\langle H_{2}\right\rangle _{\eta}-\mathbb{E}_{\lambda\sim\upsilon}\!\left[\left\langle H_{2}\right\rangle _{\rho(\theta^{\left(2\right)}(\lambda))}\right]\right)\nonumber \\
 & \qquad+\cdots+\theta_{J}\left(\left\langle H_{J}\right\rangle _{\eta}-\mathbb{E}_{\lambda\sim\upsilon}\!\left[\left\langle H_{J}\right\rangle _{\rho(\theta^{\left(J\right)}(\lambda))}\right]\right)\\
 & =\sum_{j=1}^{J}\theta_{j}\left(\left\langle H_{j}\right\rangle _{\eta}-\mathbb{E}_{\lambda\sim\upsilon}\!\left[\left\langle H_{j}\right\rangle _{\rho(\theta^{\left(j\right)}(\lambda))}\right]\right)\\
 & =\left\langle H(\theta)\right\rangle _{\eta}-\sum_{j=1}^{J}\theta_{j}\mathbb{E}_{\lambda\sim\upsilon}\!\left[\left\langle H_{j}\right\rangle _{\rho(\theta^{\left(j\right)}(\lambda))}\right]\\
 & =\left\langle H(\theta)\right\rangle _{\eta}-\left\Vert \theta\right\Vert _{1}\sum_{j=1}^{J}\frac{\left|\theta_{j}\right|}{\left\Vert \theta\right\Vert _{1}}\mathbb{E}_{\lambda\sim\upsilon}\!\left[\signum(\theta_{j})\left\langle H_{j}\right\rangle _{\rho(\theta^{\left(j\right)}(\lambda))}\right]\\
 & =\left\langle H(\theta)\right\rangle _{\eta}-\left\Vert \theta\right\Vert _{1}\mathbb{E}_{j\sim q,\lambda\sim\upsilon}\!\left[\signum(\theta_{j})\left\langle H_{j}\right\rangle _{\rho(\theta^{\left(j\right)}(\lambda))}\right].
\end{align}
Note that we can also write
\begin{align}
 & \left\langle H(\theta)\right\rangle _{\eta}-\left\Vert \theta\right\Vert _{1}\mathbb{E}_{j\sim q,\lambda\sim\upsilon}\!\left[\signum(\theta_{j})\left\langle H_{j}\right\rangle _{\rho(\theta^{\left(j\right)}(\lambda))}\right]\nonumber \\
 & =\left\Vert \theta\right\Vert _{1}\mathbb{E}_{j\sim q}\!\left[\signum(\theta_{j})\left\langle H_{j}\right\rangle _{\eta}\right]-\left\Vert \theta\right\Vert _{1}\mathbb{E}_{j\sim q,\lambda\sim\upsilon}\!\left[\signum(\theta_{j})\left\langle H_{j}\right\rangle _{\rho(\theta^{\left(j\right)}(\lambda))}\right]\\
 & =\left\Vert \theta\right\Vert _{1}\mathbb{E}_{j\sim q,\lambda\sim\upsilon}\!\left[\signum(\theta_{j})\left(\left\langle H_{j}\right\rangle _{\eta}-\left\langle H_{j}\right\rangle _{\rho(\theta^{\left(j\right)}(\lambda))}\right)\right].
\end{align}

\subsection{Hybrid quantum--classical algorithm for cross entropy and log partition
function estimation}

\label{app:alg-cross-ent-log-part}

This leads to the following hybrid quantum--classical algorithm for
estimating the cross entropy:
\begin{lyxalgorithm}
\label{alg:obj-func-est-1}A hybrid quantum--classical algorithm
for estimating the cross entropy $\Xi(\eta\|\rho(\theta))$ consists
of the following steps:
\begin{enumerate}
\item Set $k\leftarrow1$, and set
\begin{equation}
K\leftarrow O\!\left(\left(\frac{\left\Vert \theta\right\Vert _{1}\max_{j\in\left[J\right]}\left\Vert H_{j}\right\Vert }{\varepsilon}\right)^{2}\ln\!\left(\frac{1}{\delta}\right)\right),
\end{equation}
where $\varepsilon>0$ is the desired accuracy and $\delta\in\left(0,1\right)$
is the desired failure probability.
\item Sample $j\sim q$ and $\lambda\sim\upsilon$.
\item Prepare the state $\eta$ and the thermal state $\rho(\theta^{\left(j\right)}(\lambda))$.
\item Measure the observable $H_{j}$ on each state, with measurement outcomes
$X_{k}^{\eta},X_{k}^{\rho}\in\spec(H_{j})$. Set $W_{k}\leftarrow\left\Vert \theta\right\Vert _{1}\cdot\signum(\theta_{j})\left(X_{k}^{\eta}-X_{k}^{\rho}\right)$.
Set $k\leftarrow k+1.$
\item Repeat Steps 2-4 $K-1$ more times. Compute the average $\overline{W_{K}}\coloneqq\frac{1}{K}\sum_{k=1}^{K}W_{k}$
and output $\ln d+\overline{W_{K}}$ as an estimate of $\Xi(\eta\|\rho(\theta))$.
\end{enumerate}
\end{lyxalgorithm}

By the Hoeffding inequality, we are guaranteed that
\begin{equation}
\Pr\!\left[\left|\ln d+\overline{W_{K}}-\Xi(\eta\|\rho(\theta))\right|\leq\varepsilon\right]\geq1-\delta.
\end{equation}

Let us finally note that one can estimate the log-partition function
$\ln Z(\theta)$ by means of the same algorithm, but with steps 3
and 4 replaced by the following:
\begin{enumerate}
\item Prepare the thermal state $\rho(\theta^{\left(j\right)}(\lambda))$.
\item Measure the observable $H_{j}$ on this state, with measurement outcome
$X_{k}^{\rho}\in\spec(H_{j})$. Set $W_{k}\leftarrow-\left\Vert \theta\right\Vert _{1}\cdot\signum(\theta_{j})X_{k}^{\rho}$.
Set $k\leftarrow k+1.$
\end{enumerate}

\end{document}